\documentclass[
superscriptaddress,
amsmath,amssymb,
aps,
prx,
twocolumn,
longbibliography,
footinbib,
]{revtex4-2}

\usepackage{tikz} 
\usepackage{amsmath,amsfonts,amsthm, mathtools}
\usepackage{float}
\usepackage{dsfont}
\usepackage{csquotes}
\usepackage{amssymb}
\usepackage{verbatim}
\usepackage{rotating}
\usepackage{thmtools, thm-restate}
\usepackage{bm}
\usepackage{soul}
\usepackage{hyperref}
\usepackage{mathbbol}
\usepackage{pifont}
\usepackage{appendix}
\usepackage[capitalize]{cleveref}
\usepackage{blkarray}
\usepackage{adjustbox}
\usepackage{graphicx}
\usepackage{dcolumn}

\definecolor{darkred}  {rgb}{0.5,0,0}
\definecolor{darkblue} {rgb}{0,0,0.5}
\definecolor{darkgreen}{rgb}{0,0.5,0}
\hypersetup{
  colorlinks = true,
  urlcolor  = blue,         
  linkcolor = darkblue,     
  citecolor = darkgreen,    
  filecolor = darkred       
}
\theoremstyle{definition}

\newtheorem{corollary}{Corollary}
\newtheorem{definition}{Definition}
\newtheorem{conjecture}{Conjecture}
\newtheorem{lemma}{Lemma}
\newtheorem{proposition}{Proposition}
\newtheorem{theorem}{Theorem}

\newtheorem{example}{Example}

\newtheorem*{remark}{Remark}

\newcommand{\mbb}{\mathbb}
\newcommand{\mc}{\mathcal}

\newcommand{\msf}{\mathsf}

\newcommand{\tr}{\textrm{Tr}}

\usepackage{multirow}
\newcommand{\ket}[1]{|#1\rangle}
\newcommand{\bra}[1]{\langle #1|}
\newcommand{\op}[2]{|#1\rangle\langle#2|}

\definecolor{cool_green}{rgb}{0.0, 0.5, 0.0}

\newcommand{\cmt}{\color{black}}

\begin{document}

\title{Reassessing the boundary between classical and nonclassical \\ for individual quantum processes}

\author{Yujie Zhang}
\email{yujie4physics@gmail.com}
\affiliation{Institute for Quantum Computing, University of Waterloo, Waterloo, Ontario Canada N2L 3G1}
\affiliation{Perimeter Institute for Theoretical Physics, 31 Caroline Street North, Waterloo, Ontario Canada N2L 2Y5}
\affiliation{Department of Physics \& Astronomy, University of Waterloo, Waterloo, Ontario, Canada, N2L 3G1}
\author{David Schmid} 
\email{dschmid1@pitp.ca}
\affiliation{Perimeter Institute for Theoretical Physics, 31 Caroline Street North, Waterloo, Ontario Canada N2L 2Y5}
\author{Y{\`i}l{\`e} Y{\=\i}ng}
\email{yying@pitp.ca}
\affiliation{Perimeter Institute for Theoretical Physics, 31 Caroline Street North, Waterloo, Ontario Canada N2L 2Y5}
\affiliation{Department of Physics \& Astronomy, University of Waterloo, Waterloo, Ontario, Canada, N2L 3G1}
\author{Robert W. Spekkens} 
\affiliation{Perimeter Institute for Theoretical Physics, 31 Caroline Street North, Waterloo, Ontario Canada N2L 2Y5}

\date{\today}
\begin{abstract}
There is a received wisdom about where to draw the boundary between classical and nonclassical for various types of quantum processes. For multipartite states, it is the divide between separable and entangled; for channels, the divide between entanglement-breaking and not; for sets of measurements, the divide between compatible and incompatible; for assemblages, the divide between unsteerable and steerable. However, these choices have not been motivated by any unified notion of what it means to be classically explainable. One well-motivated notion of classical explainability is the one based on generalized noncontextuality: a set of circuits is classically explainable if a generalized-noncontextual ontological model can realize the statistics they generate.  In this work, we show that this notion can be leveraged to define a classical-nonclassical divide for individual quantum processes of arbitrary type. A set of measurements, for example, is judged to be classical if and only if a particular set of circuits---the one obtained by contracting these measurements with every possible quantum state---is classically explainable in the sense just articulated.  We begin the task of characterizing where the classical-nonclassical divide lies according to this proposal for a variety of different types of processes. In particular, we show that all of the following are judged to be nonclassical: every entangled state, every set of incompatible measurements, every non-entanglement-breaking channel, and every steerable assemblage. Our proposal differs from the received wisdom, however, insofar as it also judges certain subsets of the complementary classes to be nonclassical, including certain separable states, compatible sets of measurements, entanglement-breaking channels, and unsteerable assemblages. Finally, we prove structure theorems characterizing the classical-nonclassical divide based on whether a process admits of a specific type of frame representation.  
\end{abstract}
\maketitle

\tableofcontents   
\section{Introduction}

A recurring question of interest in quantum information and foundations is to determine which processes (states, measurements, sets of measurements, channels, and so on) can be reasonably said to be classical, and which must be acknowledged to be nonclassical (that is, genuinely quantum). 

In much of the literature, proposals for the placement of the classical-nonclassical divide have been introduced on a case-by-case basis for different types of processes. For instance, it is commonly taken that for bipartite states, it is the divide between separable and entangled, for channels, it is the divide between entanglement-breaking and non-entanglement-breaking, for multi-measurements, it is the divide between the set of measurements being compatible and being incompatible, and so on. 

These proposals are often based on intuitive notions of what should count as a classical kind of correlation or process.  However, the question of whether all such proposals are consistent with a single principle has remained largely unexplored~\footnote{An exception to this is that a unified proposal has been made for studying nonclassicality of common causes in multipartite scenarios~\cite{Wolfe2020quantifyingbell, Schmid2020typeindependent, Schmid2023understanding}.}. Indeed, the resulting conceptions of nonclassicality can be distinct for different types of processes, and we will even show that they are sometimes in conflict with one another.

Seeking to define classicality across different types of processes using a unified notion of classical explainability is especially important, given that many widely held intuitions about which features of quantum theory count as evidence of genuine quantumness are inconsistent with well-motivated notions of classical explainability. For example, there exist classical theories~\cite{spekkens2007evidence,epistricted, Catani2023whyinterference}  that satisfy various principles, such as locality and generalized noncontextuality, and that nonetheless contain analogues of entanglement, non-entanglement-breaking channels, and incompatible measurements, not to mention steering, interference, no-cloning, and a wide variety of other phenomena that are commonly thought to be challenging to explain classically.

The question of where the classical-nonclassical divide lies is sometimes mistaken for one that is merely about semantics, that is, merely about how we use the term ``classical".  This is not the case because the answer to this question has a significant impact on how one pursues foundational and applied research in quantum theory. For instance, the answer impacts how one pursues the project of identifying quantum-over-classical advantages across various tasks (computational, cryptographic, metrological, thermodynamic, etc.), because one expects to find such advantages only in protocols that leverage nonclassical resources. In other words, a desideratum for a good proposal on where to place the classical-nonclassical divide is that any advantage should require resources that are deemed nonclassical by the proposal. Similarly, the placement of the classical-nonclassical divide is important for identifying the manner in which the quantum worldview differs from the classical worldview and hence for correctly identifying the true innovation of quantum theory over its classical predecessors. 

Here, we introduce an approach that defines the classical-nonclassical divide for a process based on a single principle, namely the idea that classical explainability of operational statistics corresponds to realizability by a generalized-noncontextual ontological model
~\cite{gencontext,Schmid2024structuretheorem}. 

Since our framework is based on the notion of a generalized-noncontextual ontological model, it is worth reiterating the motivations for why realizability in such a model constitutes a good notion of classical explainability for operational statistics. The ontological models framework captures the conventional manner of providing a {\em realist} account of an operational theory~\cite{schmid2020unscrambling}.  Demanding that a realist account should satisfy the principle of generalized noncontextuality can be motivated by a methodological principle for theory-construction due to Leibniz. This principle has proven its utility in physics through Einstein's use of it in his development of the special and general theories of relativity~\cite{Leibniz}.  Finally, we note that a generalized-noncontextual ontological representation is equivalent to a positive quasiprobability representation~\cite{negativity, SchmidGPT, Schmid2024structuretheorem}, and coincides with the natural notion of classical explainability in the framework of generalized probabilistic theories (GPTs~\cite{Hardy, GPT_Barrett}), namely, linear embeddability into the strictly classical (simplicial) GPT~\cite{SchmidGPT, Schmid2024structuretheorem}. 

This notion of classical explainability also has pragmatic advantages over some of its competitors. For one, it is noise-robust and is therefore testable in real experiments where the ideal of noiselessness is never realized~\cite{mazurek2016experimental}.  This is in contrast to the Kochen-Specker notion of noncontextuality~\cite{KS}, as argued in Ref.~\cite{mazurek2016experimental, Spekkens2005, kunjwal2018from, kunjwal2015kochen}.

It also has {\em universal applicability}: given a circuit with arbitrary causal structure, one can ask whether there exists a generalized-noncontextual ontological model that can reproduce the statistics generated by that circuit. In contrast, standard Bell inequalities~\cite{Bell} apply only to circuits having one specific causal structure. Due to its universal applicability, one can ask: for \textit{any} operational phenomenon, precisely which aspects of it fail to admit of a generalized-noncontextual ontological model and hence which aspects resist classical explanation? 
Results establishing facts about the boundary between what is classically explainable and what is not are now known for the phenomenology of: quantum computation~\cite{Schmid2022Stabilizer,shahandeh2021quantum}, state discrimination~\cite{schmid2018contextual,flatt2021contextual,mukherjee2021discriminating,Shin2021}, interference~\cite{Catani2023whyinterference,catani2022reply,catani2023aspects,giordani2023experimental}, compatibility~\cite{selby2023incompatibility,selby2023accessible,PhysRevResearch.2.013011}, uncertainty relations~\cite{catani2022nonclassical}, metrology~\cite{contextmetrology}, thermodynamics~\cite{contextmetrology,comar2024contextuality,lostaglio2018}, weak values~\cite{AWV, KLP19}, coherence~\cite{rossi2023contextuality,Wagner2024coherence,wagner2024inequalities}, quantum Darwinism~\cite{baldijao2021noncontextuality}, information processing and communication~\cite{POM,RAC,RAC2,Saha_2019,Yadavalli2020,PhysRevLett.119.220402,fonseca2024robustness}, cloning~\cite{cloningcontext}, broadcasting~\cite{jokinen2024nobroadcasting}, pre- and post-selection paradoxes~\cite{PP1}, randomness certification~\cite{Roch2021}, psi-epistemicity~\cite{Leifer}, macroscopic realism~\cite{Schmid2024reviewreformulation}, and Bell~\cite{Wright2023invertible,schmid2020unscrambling} and Kochen-Specker scenarios~\cite{operationalks,kunjwal2018from,Kunjwal16,Kunjwal19,Kunjwal20,specker,Gonda2018almostquantum}. 
Aiding this research is a growing set of powerful techniques for certifying the failure to admit of a noncontextual ontological model~\cite{Schmid2018, Chaturvedi2021characterising, schmid2024PTM, Schmid2024structuretheorem, SchmidGPT,catani2024resource, PuseydelRio,schmid2024shadows,schmid2024addressing,mazurek2016experimental,grabowecky2021experimentally}.

However, in all work to date, the existence of a generalized-noncontextual ontological model provides a notion of classical explainability {\em for a closed circuit}, since only such a circuit can generate operational statistics, and it is on the basis of these statistics that classical explainability is assessed. In contrast, no notion of classicality has been given for a single process within a circuit, such as a single quantum state, measurement, or channel. In this work, we propose a way of doing so.

The intuitive idea is straightforward: a given quantum process is nonclassical if and only if there exists a quantum circuit that leverages it \emph{nontrivially} to generate statistics that are inconsistent with a generalized-noncontextual ontological model. In Section~\ref{basicdefns}, we formalize this definition. In particular, we explain what it means to ``leverage in a nontrivial way'' the process, and then we simplify the definition considerably, by showing that rather than considering \emph{all} quantum circuits containing the process, one needs only consider circuits where product states are prepared on its inputs and product measurements are implemented on its outputs.

We then undertake the task of characterizing which quantum preparations, measurements, and transformations are nonclassical, and which are classical. (This paper focuses only on quantum theory, although our approach could be applied to an arbitrary generalized probabilistic theory.)

We prove structure theorems characterizing the structure of any classical process in terms of the existence of a particular kind of product dual-frame representation. We then use this to prove a number of interesting results, including the following:
\begin{enumerate}
    \item A single state on a unipartite system is always classical. 
    \item Any multi-state consisting of  a set of distinct {\em pure} states on a unipartite 
    system is nonclassical if and only if the states are 
    linearly dependent in the operator space (and therefore any  \textit{minimal}  tomographically complete set of pure states is classical).
    \item Any rank-1 POVM is nonclassical if and only if its POVM elements are linearly dependent in the operator space (and therefore any symmetric informationally complete (SIC)-POVM is classical). 
    \item The set of classical multi-states (respectively measurements) is nonconvex.
    \item  Every steerable assemblage is nonclassical, as are certain unsteerable assemblages. 
    \item Every multi-measurement consisting of a set of incompatible measurements  is nonclassical, as are certain sets of compatible measurements (including certain singleton sets, i.e., single measurements).
    \item Every entangled state is nonclassical, as are certain separable states.  
    \item For two-qubit systems, a state is nonclassical if and only if it is entangled. 
    \item All non-entanglement-breaking channels are nonclassical, as are certain entanglement-breaking channels.
\end{enumerate}
Note that ``multi-state'' and ``multi-measurement'' are terms we use in this paper to denote a type of single process associated to a set of states or a set of measurements, respectively (see Sec.~\ref{sec:preliminary} ). 

In addition, we show that the task of deciding nonclassicality for a quantum process \emph{of arbitrary type} (including channels and instruments) can be reduced to the corresponding task for an associated multipartite state.

Together, these results provide a useful characterization of the nonclassicality of arbitrary processes. The essential results of the paper are summarized in \cref{fig:outline} below. 

\begin{figure*}
    \centering
    \includegraphics[width=\linewidth]{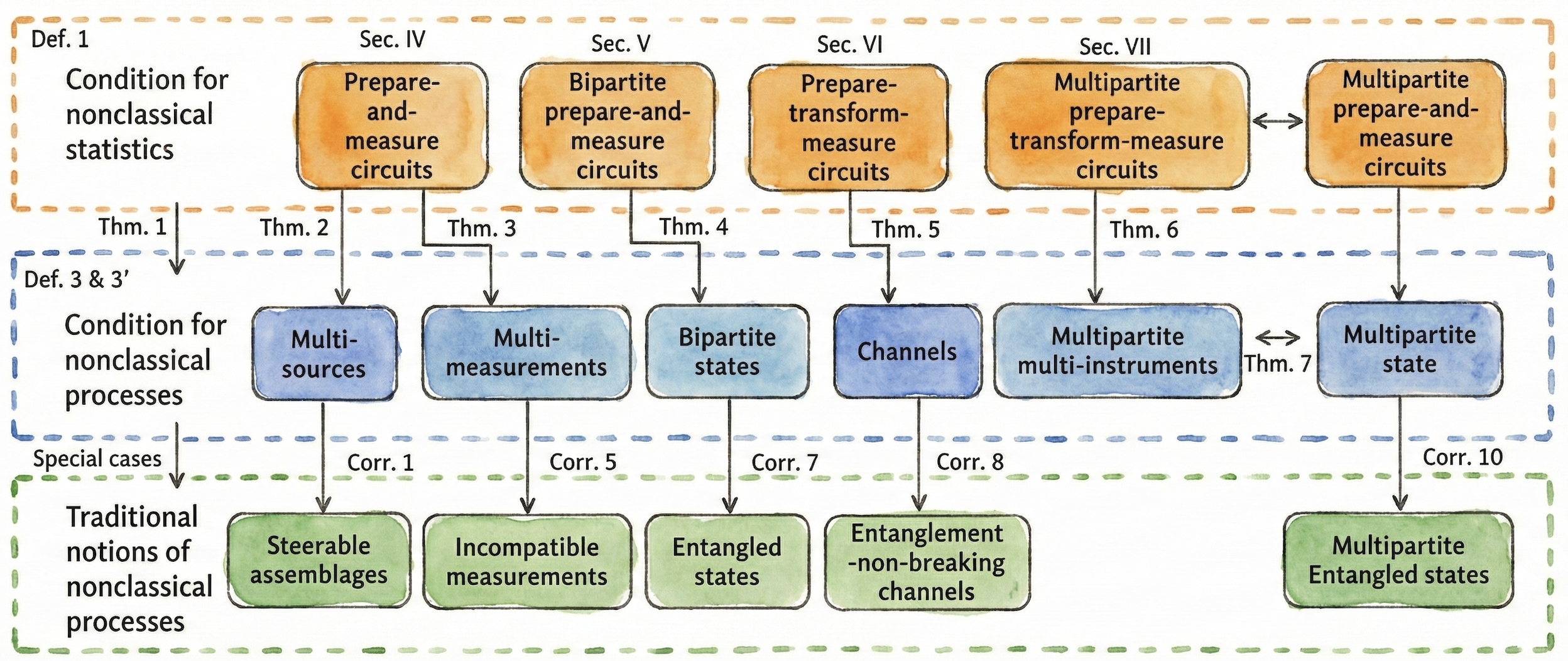}
    \caption{Schematic showing the main results of the paper. Yellow: in each section, we study the class of quantum circuits built from the type of target process of interest together with dual processes (Definition~\ref{def:dualprocess}) that close the circuit; the condition  for the resulting statistics to be classically explainable is realizability in a noncontextual ontological model (Definition~\ref{defn:classical}).  Blue: a process is then defined to be classical if the resulting statistics from  the set of circuits obtained by contracting it with {\em all possible} dual processes are classically explainable (Definitions~\ref{def:og_process} and~\ref{maindefn}); the classical–nonclassical divide for processes of each type is then characterized by the  corresponding structure theorems. Green: traditional notions of nonclassicality are recovered as special cases, i.e., as sufficient conditions for the proposed notion of nonclassicality.}
    \label{fig:outline}
\end{figure*}

In a companion paper~\cite{zhang2024parellel}, we provide further tools based on semidefinite programming that help one to witness and quantify the nonclassicality of single states and measurements. We do so both in a theory-dependent way (analogous to entanglement witnesses) and in a theory-independent way (using noncontextuality inequalities).

In another companion paper~\cite{zhang2025forthcoming}, we focus on bipartite states and investigate how our notion of nonclassicality relates to entanglement, steering, and nonlocality. We provide tools to certify these quantum resources via a hierarchy of noncontextuality inequalities. Moreover, we present an experimental demonstration of this novel entanglement-certification approach and contrast it with existing methods. 

We end with two points about terminology.

First, we wish to head off a possible point of confusion: the term ``quantum'' is sometimes taken as a synonym for nonclassical, and other times taken as an enveloping term that includes classical and nonclassical as special cases.~\footnote{This is analogous to how the term ``rectangular'' is sometimes used to mean non-square, and other times taken in the more inclusive sense that includes squares as a special case.} In this work, we use the term ``quantum" in the more inclusive sense.  This is because we are studying the question of which subsets of phenomena described by quantum theory are {\em classically explainable} and which are not.  We use the term ``nonclassical'' or ``intrinsically quantum'' to signal the impossibility of classical explanation~\footnote{The analogue for the geometric example would be to use the term ``rectangular'' in the more inclusive sense (as the mathematicians do) and to use the terms ``nonsquare'' (or even ``intrinsically rectangular'') to describe the property of failing to be square.}.

Second, the term ``classical'' is quite overloaded in the literature. Although we shall continue to use the term throughout this article (where its meaning should be clear from context and from our explicit definitions), we suggest the term {\em Leibniz-classical} be used by anyone who seeks a more specific term (e.g., if referring to our notion in a broader setting where it might not be clear from context to what the word ``classical'' refers).

\section{Preliminaries}

\subsection{Types of quantum processes}
\label{sec:preliminary}

A generic {\em process} in quantum theory is simply a transformation with various input and output systems, which may be quantum or classical. A classical system that is an input is generally called a \emph{setting variable} and one that is an output is generally called an \emph{outcome variable}. 
The most common examples of processes in quantum theory are states, effects, measurements, channels, and instruments. However, it is often useful to consider more general processes, including those listed in the \cref{tab:my_label}.

\begin{table}[h]
\begin{ruledtabular}
    \centering
\begin{tabular}{ccc}
 Process Type & Quantum Rep'n & Ontological Rep'n     \\
 \hline
        State & $\rho^A$ &  $p(\lambda_A|\rho^A)$\\
 \hline 
        Effect & $M^{A}$ &  $p(M^A|\lambda_A)$\\
 \hline 
 
        Channel &$\mc{E}^{B|A}$  &  $p(\lambda_B|\lambda_A)$ \\
 \hline
        Source & $ \{p(a)\rho_a^A\}_a$ &  $\{p(a\lambda_A)\}_a$ \\
 \hline
        Measurement & $\{M_{b}^{A}\}_b$  &  $\{p(b|\lambda_A)\}_b$ \\
\hline 
        Instrument & 
         $ \{\mc{F}_{c}^{B|A}\}_c$
        & $\{p(c\lambda_B|\lambda_A)\}_c$\\
\hline
        Multi-state & $\{\rho_{\cdot|x}^A\}_x$ &  $\{p(\lambda_A|x)\}_x$ \\
\hline
 Multi-channel & $ \{\mc{E}^{B|A}_{\cdot|z}\}_z$ &  $\{p(\lambda_B|z\lambda_A)\}_z$  \\
 \hline       
        Multi-source &  $\{\{p(a|x)\rho_{a|x}^A\}_{a}\}_x$ &  $\{\{p(a\lambda_A|x)\}_a\}_x$ \\
 \hline       
        Multi-measurement & $\{\{M_{b|y}^{A}\}_b\}_y$ & $\{\{p(b|y\lambda_A)\}_b\}_{y}$ \\
 \hline
        Multi-instrument & 
         $\{\{\mc{F}_{c|z}^{B|A}\}_{c}\}_z$ 
        & $\{\{p(c\lambda_B|z\lambda_A)\}_{c}\}_z$  \\
\hline
Multipartite state & $\rho^{A_1\cdots A_n}$ & $p(\lambda_1,\cdots,\lambda_n)$
    \end{tabular}
    \caption{ Examples of various types of process, their representations in quantum theory, and their representations as substochastic matrices in an ontological model. Throughout this paper, we denote classical outcome variables by $a,b,c$ and setting variables by $x,y,z$, while quantum inputs and quantum outputs are denoted by the capital letters $A, B, C$.  
    }
    \label{tab:my_label}
\end{ruledtabular}
\end{table}\par 

Most of these processes will be familiar to the quantum information theorist. However, the terminology of sources, multi-sources, multi-measurements, and multi-channels is relatively new, so we explain these terms below. These processes are easily understood via a diagrammatic representation, as in Fig.~\ref{processes}.
\begin{figure}[htb!]
\centering
\includegraphics[width=0.5\textwidth]{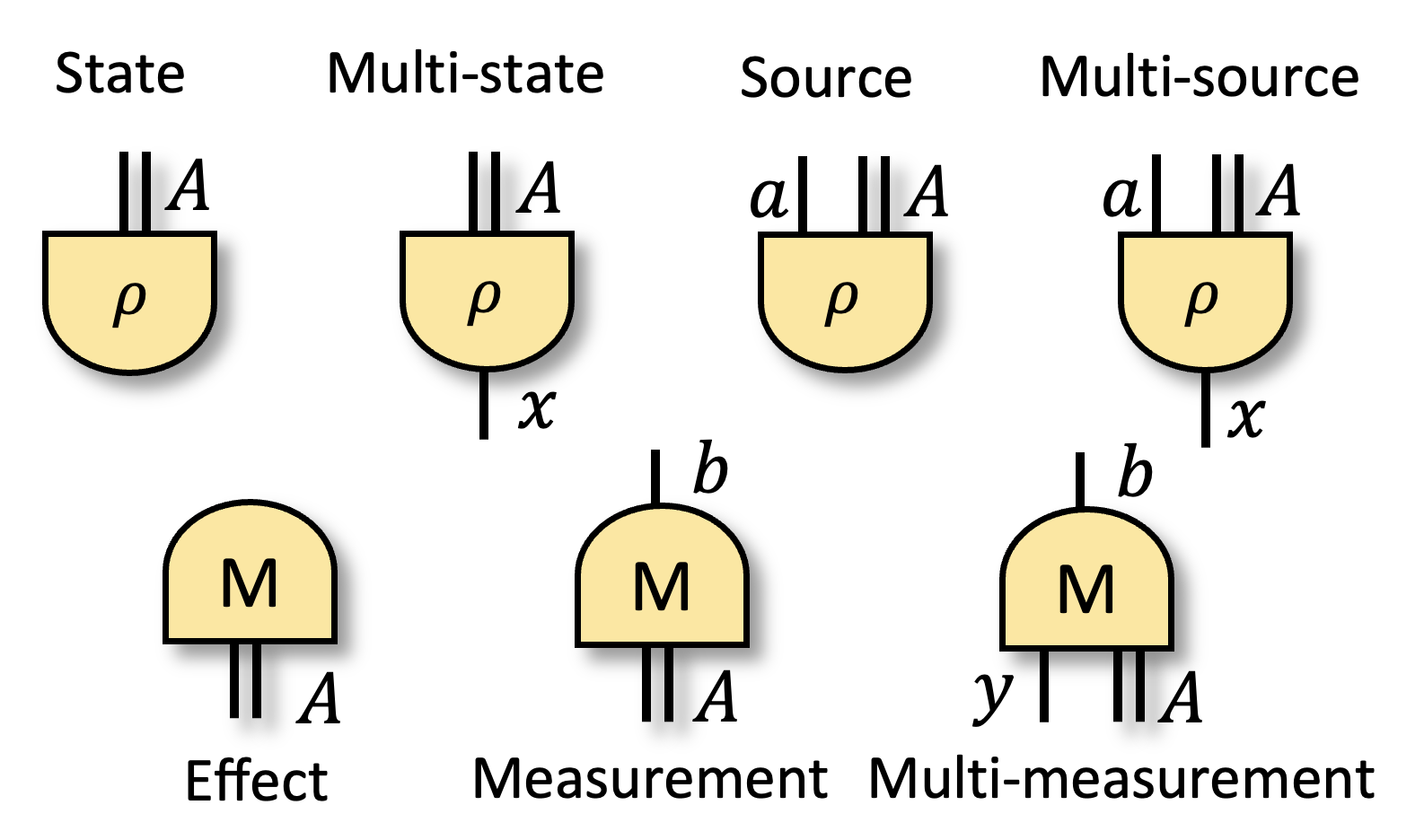}
\caption{Examples of processes with a single quantum output $A$ (top) and with a single quantum input $A$ (bottom) on a single quantum system. Throughout this paper, quantum systems are represented by double wires, whereas classical systems are represented by single wires.}
\label{processes}
\end{figure}

The {\em type} of a process is simply a specification of its input and output systems, as well as a specification of whether each of these is classical or quantum. Diagrammatically, we will depict quantum systems by double wires and classical systems by single wires (see, e.g., Fig.~\ref{processes}). In this approach, each individual input and output of a process (depicted by a wire in the diagram) is considered a fundamental system, having no subsystem structure, so that a composite system is always described by a collection of inputs or a collection of outputs.

Each quantum input and output for a given process is associated with a Hilbert space.  For a quantum system $A$, we denote this Hilbert space by $\mc{H}^A$ with superscript $A$. We will denote the set of all linear operators on $\mc{H}$ as $\mc{L}(\mc{H})$, the set of all Hermitian operators on $\mc{H}$ as $\text{Herm}(\mc{H})$, the set of all density operators (i.e., all $\rho\in \text{Herm}(\mc{H})$ such that $\rho\ge 0$ and $\tr[\rho]=1$) as $\mc{S}(\mc{H})$, and 
the set of all effects (i.e., all $M \in \text{Herm}(\mc{H})$ such that $0\le M\le \mbb{1}$) as $\mc{M}(\mc{H})$. {Formally, an effect on system $A$ is a map $\mc{M}^{\cdot|A}(X^A)\coloneqq\tr[M^AX^A]$, i.e., a channel with no output. Nonetheless, we will often simply refer to the operator $M^A$ itself as an effect.}  

It is useful to classify the types of processes by their {\em quantum} inputs and outputs, such that the only differences among the processes in a given class are in their {\em classical} outcome and setting variables.  For instance,  processes with a single quantum output and no quantum input form a class. We will refer to these as {\em processes of the preparation variety}.  In a similar fashion, the class of processes with a single quantum input and no quantum output will be termed {\em processes of the measurement variety}, and the class of processes with both a quantum input and a quantum output will be termed {\em processes of the transformation variety}.  We will sometimes qualify the terms describing these varieties to stipulate the subsystem structure of the quantum inputs or the quantum outputs; for example, we will speak of processes of the {\em bipartite} preparation variety.

We begin with processes of the preparation variety, i.e., those having a single quantum output $A$. Four types of processes in this class are depicted in the top row of Fig.~\ref{processes}. A single deterministic preparation on a quantum system $A$ is represented by a density operator (or state) $\rho\in \mc{S}(\mc{H}^A)$. A process with a classical input $x$ (a setting variable) that determines which of a set of preparations is implemented on $A$ can be represented by a set of states $\{\rho_{\cdot|x}^A\}_x$.  We refer to such a process as a {\em multi-state} process. Next, consider a process in which a variable $a$ is sampled at random from a distribution $p(a)$ and then the preparation represented by the state $\rho_a^A$ is implemented on $A$, and the variable $a$ becomes a classical output of the process (an outcome variable) that flags which state was prepared. Such a process is represented by a set of unnormalized states $\{p({a})\rho_a^A\}_a$, and we refer to it as a \emph{source} (it is sometimes termed an \emph{ensemble}).  Finally, a process with a classical input $x$ (a setting variable) that determines which of a set of sources is implemented can be represented by a set of sets of unnormalized states $\{\{p(a|x)\rho_{a|x}^A\}_{a}\}_x$.  We refer to such a process as a \emph{multi-source}. We will sometimes write $\msf{P}\coloneqq\{\{p(a|x)\rho_{a|x}^A\}_{a}\}_x$.  Note that when such processes arise in steering scenarios, they are also termed assemblages.

Next, we turn to processes of the measurement variety, having a single quantum input $A$. Three types of processes in this class are depicted in the bottom row of Fig.~\ref{processes}.  (Note that the term \emph{measurement} is used here to refer specifically to processes with no quantum output, such that there is only a {\em retrodictive aspect} to be characterized.  Processes with a classical outcome that have {\em both} a quantum input and a quantum output will here be termed {\em instruments} and are discussed below.) 
In quantum theory, a measurement with outcome labeled by $b$ is represented by a positive-operator valued measure (POVM) $\{M_{b}^{A}\}_b$ where each $M_{b}^{A}\in\mc{M}(\mc{H}^A)$.  The process that describes implementing this measurement and obtaining the outcome $b$ (so that it can be represented as a process with {\em no} classical outcome) is termed an {\em effect} and is represented by the single operator $M^{A}\in\mc{M}(\mc{H}^A)$. 
We use the term \emph{multi-measurement} to refer to a process with a classical input $y$ (a setting variable) that determines which of a set of measurements is implemented.  It can be represented by a set of POVMs, $\msf{M}\coloneqq\{\{M_{b|y}^A\}_{b}\}_y$.  Note that the notion of a measurement is subsumed as a special case of a multi-measurement, where the setting variable is trivial. 

Finally, we consider processes of the transformation variety.  A process that has a quantum input $A$ and a quantum output $B$ but that has neither a classical input nor a classical output is termed a {\em channel} and is represented by a map $\mc{E}^{B|A}: \mc{L}(\mc{H}^A)\mapsto \mc{L}(\mc{H}^B)$ that is completely positive and trace-preserving~\cite{nielsen2001quantum, Schmidcausal}. A process of the transformation variety that also has a classical input (a setting variable $z$) describes a set of channels $\{\mc{E}^{B|A}_{\cdot|z}\}_z$ and will be termed a {\em multi-channel}.  
A process of the transformation variety that has a classical output (an outcome variable) is termed an {\em instrument}.  It can be conceived of as a \textit{nondestructive} measurement procedure, so that the outcome is informative not only about the quantum input (the retrodictive aspect) but also about the quantum output and the correlations between them (the transformative aspect).  An instrument from $\mc{L}(\mc{H}^A)$ to $\mc{L}(\mc{H}^B)$ with classical outcome $c$ is represented by a set of trace-nonincreasing channels, denoted $\{\mc{F}_{c}^{B|A}\}_c$, that sum to a trace-preserving channel, i.e., $\sum_c \mc{F}_{c}^{B|A}$ defines a CPTP map. We use the term \emph{multi-instrument} to refer to a process with a classical input $z$ (a setting variable) that determines which of a set of instruments is implemented. We denote such a set of instruments (multi-instrument) by $\{ \{ \mc{F}_{c|z}^{B|A}\}_c \}_z$.

We note that any classical inputs (setting variables) and classical outputs (outcome variables) of a process can be reconceptualized, respectively, as quantum inputs and quantum outputs that are dephased in a particular basis. For example, a multi-measurement can be viewed as a channel that has two quantum inputs, one of which is dephased (corresponding to the setting variable), and that has a single dephased quantum output (corresponding to the outcome variable). As we will show in Section~\ref{consistency}, all of our results are invariant under this reconceptualization of the type of a process.

\subsection{Composition of quantum processes}

\begin{table}[h]
\begin{ruledtabular}
    \centering
    \begin{tabular}{ccc}
     Composition & Quantum Rep'n & Ontological Rep'n  \\ \hline \noalign{\vskip 1pt}
{Parallel}&   {$\mc{E}^{B_1|A_1}\otimes \mc{D}^{B_2|A_2}$}  &   {$p(\lambda_{B_1}|\lambda_{A_1})p(\lambda_{B_2}|\lambda_{A_2})$} \\
\hline \noalign{\vskip 1pt}
{Sequential }&    {$\mc{F}^{C|B}\circ\mc{E}^{B|A}$}  &   {$\sum_{\lambda_B}p(\lambda_C|\lambda_B)p(\lambda_B|\lambda_A) $} 
    \end{tabular}
    \caption{ Different kinds of composition (parallel and sequential) and how they are represented in quantum theory and in an ontological model thereof. }
    \label{tab:placeholder}

\end{ruledtabular}
\end{table}

Quantum circuits are formed via the composition of quantum processes, such as states, channels, and measurements. To discuss composition, it suffices to focus on channels, since states and effects can be treated as special cases of channels when discussing compositions, specifically:
\begin{itemize}
    \item  A state $\rho^A$ defines a channel from the trivial system to $A$;
    \item An effect $M^A$ defines a channel $\mc{M}^{\cdot|A}(X^A)\coloneqq\tr[M^AX^A]$ from system $A$ to the trivial system.
\end{itemize}
\par 
We denote the operation of implementing processes on distinct systems, i.e., parallel composition of the processes, by $\otimes$.
Given channels $\mc{E}^{B_1|A_1}:\mc L(\mc H^{A_1}) \to \mc L(\mc H^{B_1})$ and
$\mc{D}^{B_2|A_2}:\mc L(\mc H^{A_2}) \to \mc L(\mc H^{B_2})$, their parallel composition is a bipartite channel $
\mc{E}^{B_1|A_1}\otimes \mc{D}^{B_2|A_2}:\mc L(\mc H^{A_1} \otimes \mc H^{A_2})\to \mc L(\mc H^{B_1} \otimes \mc H^{B_2})$ (See table~\ref{tab:placeholder}). As a special case, given states $\rho^{A}\in\mc S(\mc H^{A})$ and $\sigma^{B}\in\mc S(\mc H^{B})$, their parallel composition defines a bipartite state $\rho^{A}\otimes\sigma^{B}\in\mc S(\mc H^{A} \otimes \mc H^{B}).$ 

We denote the operation of wiring the output of one process to the input of the other process (where the system types match), i.e., sequential composition of processes,  by $\circ$. Given channels $\mc{E}^{B|A}$: $\mc{L}(\mc{H}^A)\mapsto \mc{L}(\mc{H}^B)$ and $\mc{F}^{C|B}$: $\mc{L}(\mc{H}^B)\mapsto \mc{L}(\mc{H}^C)$, their sequential composition defines a new channel  $\mc{F}^{C|B}\circ\mc{E}^{B|A}$: $\mc{L}(\mc{H}^A)\mapsto \mc{L}(\mc{H}^C)$ (See table~\ref{tab:placeholder}). As a special case, the application of a channel $\mc{E}^{B|A}$: $\mc{L}(\mc{H}^A)\mapsto \mc{L}(\mc{H}^B)$ to a state $\rho^A\in \mc{S}(\mc{H}^A)$ defines a new state  $\mc{E}^{B|A}\circ \rho^A\coloneqq \mc{E}^{B|A}(\rho^A)\in \mc{S}(\mc{H}^B)$. We will also use $\circ$ to denote more general types of ``wiring together'' of processes.  For instance, a channel $\mc{E}^{B|A}$ from $A$ to $B$ can be slotted into a comb~\cite{Chiribella2008} $\mc{C}^{A|B}$ with output $A$ and input $B$ (as in Fig.~\ref{duals}(d) below), and we denote the resulting closed circuit by $\mc{C}^{A|B} \circ \mc{E}^{B|A}$.

Discarding (trace and partial trace) can be regarded as composition with the deterministic effect ${u}^{\cdot|A}: A\to I$, where $I$ is the trivial system. Consequently, the partial trace is a special case of sequential composition. For a bipartite state $\rho^{A_1A_2}$, one has ${u}^{\cdot|A}\circ \rho^{A_1A_2}=\tr_{A_1}[\rho^{A_1A_2}]$. 

Consider a quantum circuit built up from the parallel and sequential composition of a set of quantum processes. Quantum theory provides a prescription for computing the operational statistics generated by any such circuit. Here and throughout this article, we use the term {\em operational statistics} (or simply \emph{statistics}) to refer to the conditional probability distribution over the outcome variables in the circuit given the setting variables (see Fig.~\ref{PTM} and Eq.~\eqref{eq: opstatistics} for an example).  Moreover, we refer to any 
circuit that is the sequential composition of a process of the preparation variety with a process of the measurement variety as a {\em prepare-measure circuit/scenario}. We refer to any circuit that is the sequential composition of a process of the preparation variety with a process of the transformation variety, followed by a process of the measurement variety, as a {\em prepare-transform-measure circuit/scenario}. The same terminological convention applies to circuits on multipartite rather than unipartite systems.

\subsection{Classical explainability of the statistics generated by a set of quantum circuits} \label{classicalexplainability}

As noted in the introduction, we take {\em classical explainability of operational statistics} to mean the possibility of reproducing these operational statistics in a generalized-noncontextual ontological model~\cite{Spekkens2005, Schmid2024structuretheorem}. 

Here, we wish to conceptualize quantum theory as an operational theory wherein one has quotiented the set of laboratory procedures of a given type with respect to an operational equivalence relation~\cite{chiribella2010probabilistic}. (This is the standard conceptualization in the field of quantum information, where a state is a density operator, a channel is a completely-positive trace-preserving map, and a measurement outcome is a POVM element.) We digress briefly to explain this.  In an {\em un}quotiented operational theory, a given procedure (preparation, measurement, transformation, etc.) is conceptualized as a list of laboratory instructions~\cite{Spekkens2005}, so that the theoretical description includes details that can be irrelevant to the operational statistics obtained in circuits within which the procedure is embedded.  Despite their irrelevance to the operational statistics, such details---which are termed {\em contexts}---might still be relevant to the ontological representation of the procedure in a \emph{contextual} ontological model~\cite{Spekkens2005}.  Sameness of operational statistics defines an equivalence relation over procedures.  A quotiented operational theory is one that only includes operational equivalence classes of procedures, and does not include the details (i.e., contexts) that distinguish the procedures within a given class. For example, an equivalence class of preparation procedures in quantum theory corresponds to a density operator.

Conceptualizing quantum theory as a quotiented operational theory introduces a subtlety regarding how to understand the notion of classical explainability, namely, that the definition of a generalized-noncontextual ontological model was originally given only for unquotiented operational theories. Because a quotiented operational theory has quotiented away the notion of context that would be leveraged in a context-dependent ontological representation, it becomes a category mistake to talk about context-dependence, and therefore it becomes inappropriate to define classical explainability of such a theory as the possibility of reproducing its operational statistics in terms of a generalized-noncontextual ontological model.

What is required is a translation of the notion of classical explainability of unquotiented operational theories described above into the language of quotiented operational theories. Such a translation was presented in Ref.~\cite{SchmidGPT, Schmid2024structuretheorem}.  It is this: a quotiented operational theory is classically explainable if and only if its predictions can be reproduced by an ontological theory that is obtained from the quotiented operational theory by a linear and diagram-preserving ontological representation map, where each operational process is mapped to a substochastic map.

Consequently, this is our notion of classical explainability:
\begin{definition}\label{defn:classical}
The statistics generated by a set of quantum circuits of a given structure are said to be classically explainable if and only if they admit of a linear and diagram-preserving ontological representation---that is, if and only if 
there is a linear and diagram-preserving map from the quantum systems appearing in the set of circuits to random variables, and from the quantum processes appearing in the set of circuits to substochastic matrices such that the map preserves the statistical predictions~\cite{Schmid2024structuretheorem}.
\end{definition}

Note that we will sometimes simply say that the set of circuits
itself (rather than its statistics) is classically explainable, and we will often consider situations where the set of circuits in question contains only one element, in which case we will speak of the classical explainability of a circuit rather than of a set of circuits (an example of this follows). This notion of classical explainability of a set of circuits is equivalent to 
the notion defined by the possibility of a generalized-noncontextual ontological model of an unquotiented version of the set of circuits~\cite{Schmid2024structuretheorem}, the motivations for which
were discussed in the introduction. 

This definition relies on an abstract notion of a linear and diagram-preserving (or functorial~\cite{mellies2006functorial}) map from quantum processes in the circuit (s) to substochastic matrices over random variables. The advantage of this definition is that it is simple to state, simple to visualize (see Figure~\ref{PTM}), and fully general, in the sense that it applies to arbitrary circuits. The disadvantage is that writing it down more formally and explicitly requires a diagrammatic formalism like that presented in Ref.~\cite{Schmid2024structuretheorem}. As many readers may not be familiar with this diagrammatic notation, we will not repeat the general definitions here, but will instead only give the explicit form for specific examples (for which standard algebraic notations suffice). 

As a first illustrative example, imagine one implements an experiment with a prepare-transform-measure circuit of the sort shown in Figure~\ref{PTM}, consisting of a multi-state $\{\rho_{\cdot|x}^A\}_x$, multi-channel $\{\mc{E}_{\cdot|z}^{B|A}\}_z$, and measurement $\{M_{b}^{B}\}_b$.

\begin{figure}[htb!]
\centering
\includegraphics[width=0.42\textwidth]{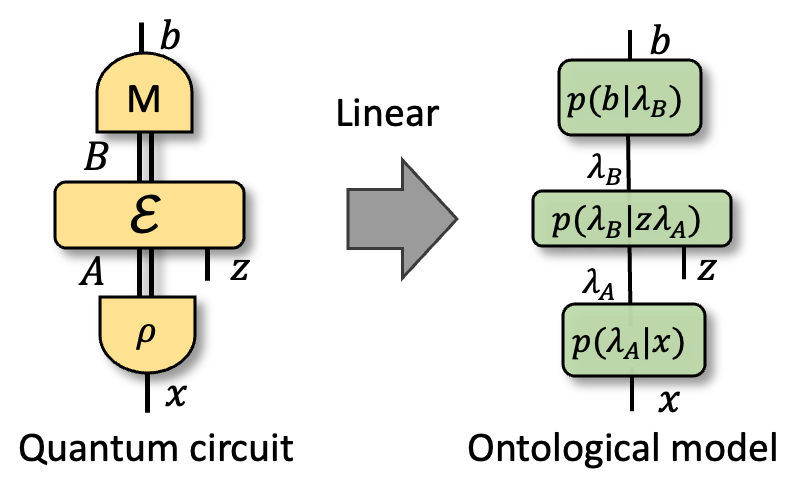}
\caption{A prepare-transform-measure circuit (left) and an ontological representation of it (right). The circuit is classically explainable if and only if there exists a linear and diagram-preserving ontological representation taking each quantum system to a classical random variable and each quantum process to a substochastic matrix. The map must, moreover, respect the circuit's topology and reproduce the quantum predictions.}
\label{PTM}
\end{figure}

The operational statistics one can generate within such a circuit are 
\begin{equation}
    p(b|xz)={\rm Tr}[M_{b}^{B} \mc{E}_{\cdot|z}^{B|A}(\rho_{\cdot|x}^A)]. 
\label{eq: opstatistics}
\end{equation}

A diagram-preserving ontological representation for such a circuit provides a realist explanation of these statistics in terms of underlying spaces of ontic states for systems $A$ and  $B$ that encapsulate the fundamental properties of the system in question and that ultimately determine the outcomes of measurements.
In particular, it associates with each system $S$ in the circuit an ontic state space $\Lambda_S$ stipulating the set of all possible ontic states $\lambda_S$. (We assume this set is finite for simplicity of the presentation.) Each state $\rho_{\cdot|x}^A$ is associated with a probability distribution over $\lambda_A$, namely $p(\lambda_A|x)$, where $p(\lambda_A|x) \geq 0$ and $\sum_{\lambda_A}p(\lambda_A|x)= 1.$
Each channel $\mc{E}_{\cdot|z}^{B|A}$ is represented by a stochastic matrix, the elements of which are conditional probabilities
$p(\lambda_B|z\lambda_A)$, giving the probability of transitioning from ontic state $\lambda_A$ to $\lambda_B$; this stochastic matrix
must satisfy
$p(\lambda_B|z\lambda_A) \geq 0$ for all $(\lambda_A,\lambda_B)\in \Lambda_A\times\Lambda_B$ and $\sum_{\lambda_B}p(\lambda_B|z\lambda_A)= 1$ for all $\lambda_A \in \Lambda_A$.
Each measurement effect $ M_{b}^{B}$ is represented by a response function $p(b|\lambda_B)$, describing the probability of obtaining outcome $b$ given the ontic state $\lambda_B$; each response function satisfies
$1\ge p(b|\lambda_B) \geq 0$, and for a set of effects $\{ M_{b}^{B}\}_b$ summing to the identity (and so forming a valid measurement), one must have
$\sum_b p(b|\lambda_B) = 1$ for all $\lambda_B\in \Lambda_B$. 

The requirement that  the ontological model  reproduces  the empirical predictions of quantum theory for this prepare–transform–measure scenario of \cref{PTM} implies that  
\begin{align}
  p(b|xz)&={\rm Tr}[M_{b}^{B} \mc{E}_{\cdot|z}^{B|A}(\rho_{\cdot|x}^A)] \notag\\
  &= \sum_{\lambda_A\lambda_B} p(b|\lambda_B)  p(\lambda_B|z\lambda_A) p(\lambda_A|x).
\end{align}

Thus far, we have described the implications of the ontological representation being diagram-preserving.  We now consider the additional conditions implied by the linearity of the representation. 
First of all, there are generic constraints among the ontological representations of the states in the circuit, since the mapping from the set of states $\{\rho_{\cdot|x}^A\}_x$ to their ontological representations $\{p(\lambda_A|x)\}_x$ must be linear. So if an identity such as
\begin{equation}\label{eq:opidentity}
\sum_{x} \alpha_{x} \rho_{\cdot|x}^A=\mbb{0}^A 
\end{equation}
holds among the states (where here and elsewhere in the article, $\mbb{0}$ denotes the zero operator or zero channel and $\{\alpha_x\}$ is a set of real numbers), the corresponding identity must hold among the ontological representations of these states, namely,
\begin{equation}
   \sum_{x} \alpha_{x} p(\lambda_A|x)=0,\quad \forall~\lambda_A  
   \label{eq:op-PTM-p}
\end{equation}
We will call constraints of the type \cref{eq:opidentity} {\em operational identities}, and constraints of the type \cref{eq:op-PTM-p} {\em ontological identities}. 

Similarly, any operational identity among transformations, such as
\begin{equation}
\sum_{z}{\alpha}_{z} \mc{E}_{\cdot|z}^{B|A} =\mbb{0}^{B|A},\label{eq:op-PTM-t}
\end{equation}
implies the corresponding ontological identity 
\begin{align}
    \sum_{z} \alpha_{z}  p(\lambda_B|z\lambda_A)=0,\quad \forall~\lambda_A, \lambda_B,  
\end{align}
and any operational identity
\begin{equation}
\sum_{b} \alpha_{b} M_{b}^{B}=\mbb{0}^{B} \label{eq:op-PTM-m}
\end{equation}
among effects implies the corresponding ontological identity 
\begin{align}
    \sum_{b} \alpha_{b} p(b|\lambda_B)=0,\quad \forall~\lambda_B.
\end{align}
\begin{figure}[htb!]
\centering
\includegraphics[width=0.5\textwidth]{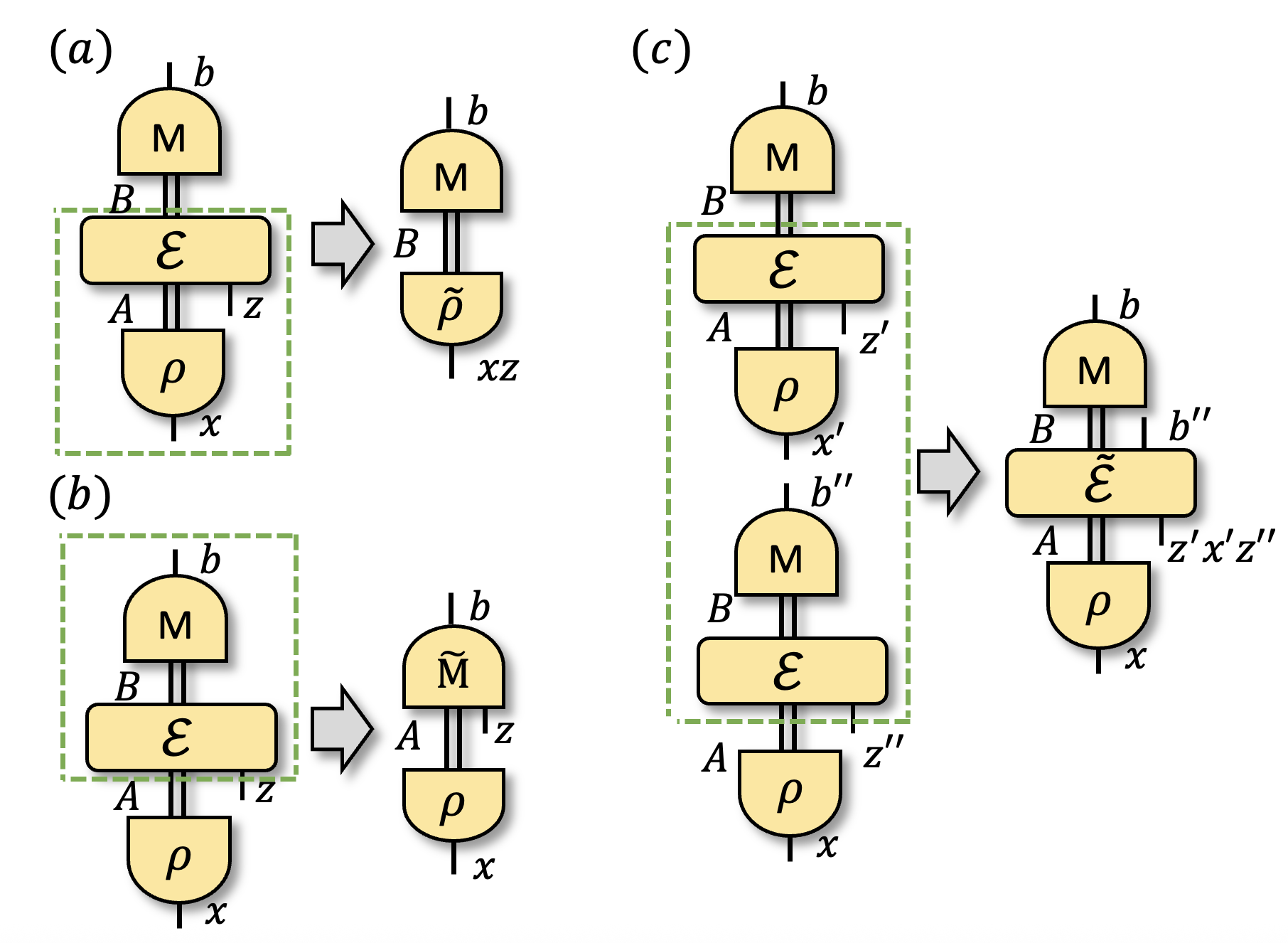}
\caption{ In the  prepare-transform-measure circuit, one can define different composite processes. (a) Boxing the multi-channel and the multi-state defines a new multi-state on $\mc{H}^B$, $\tilde{\rho}_{\cdot|zx}^B\coloneqq\mc{E}_{\cdot|z}^{B|A}(\rho_{\cdot|x}^A)$; (b) Boxing the multi-channel and the measurement defines a multi-measurement on $\mc{H}^A$, $\tilde{M}_{b|z}^A\coloneqq[\mc{E}^{B|A}]^{\dagger }_{\cdot|z}(M_{b}^{B})$; (c) More subtly, boxing the (composite) multi-state $\{\tilde{\rho}_{\cdot|zx}^B\}_{zx}$ with the (composite) multi-measurement $\{\{\tilde{M}_{b|z}^A\}_b\}_z$ defines a measure-and-prepare multi-instrument from $\mc{H}^A$ to $\mc{H}^B$ with $\tilde{\mc{E}}^{B|A}_{b''|z'x'z''}(\cdot)\coloneqq\tilde{\rho}_{\cdot|z'x'}^B\tr[\tilde{M}_{b''|z''}^A(\cdot)]$.    
}
\label{fig:PTMcomposition}
\end{figure}

Linearity of an ontological representation, moreover, imposes constraints on how the {\em compositions} of different processes are represented.

Consider the composition of the multi-state and multi-channel from Fig.~\ref{PTM}, yielding a new multi-state $\{\tilde{\rho}_{\cdot|zx}^B\}_{zx}$, depicted in \cref{fig:PTMcomposition}(a), and consider the composition of the measurement and multi-channel from Fig.~\ref{PTM}, yielding a new multi-measurement $\{\{\tilde{M}_{b|z}^A\}_{b}\}_z$, depicted in \cref{fig:PTMcomposition}(b). The corresponding ontological representations of these compositions are
\begin{subequations}
\begin{align}
\tilde p(\lambda_B|xz)\coloneqq\sum_{\lambda_A}p(\lambda_B|z\lambda_A) p(\lambda_A|x),\\
\tilde p(b|z\lambda_A)\coloneqq\sum_{\lambda_B} p(b|\lambda_B)p(\lambda_B|z\lambda_A).
\end{align}
\end{subequations}

Suppose that the operational identities that the new processes satisfy are 
\begin{subequations}
\begin{align}
&\sum_{zx}{\alpha}_{zx}\tilde{\rho}_{\cdot|zx}^B =\mbb{0}^B,  \\ 
 &\sum_{bz}{\alpha}_{bz} \tilde{M}_{b|z}^A =\mbb{0}^{A},
\end{align}
 \label{eq:op-PTM-s}
\end{subequations}
linearity then implies that the same identities must also hold for their ontological representations, namely, 
\begin{subequations}
\begin{align}
&\sum_{zx}{\alpha}_{zx}\tilde p(\lambda_B|xz)=0,\quad \forall~\lambda_B,  \\
&\sum_{bz}{\alpha}_{bz}\tilde p(b|z\lambda_A)=0, \quad \forall~\lambda_A.  \end{align}
\end{subequations}

A more subtle kind of operational identity that one must also consider~\footnote{ The constraints arising from this kind of operational identity must be included in order to ensure that the ontological model explains not only the way the processes {\em are} used in one's circuit, but also any possible way those processes {\em could} be used. Alternatively, one can view their inclusion as ensuring that the model explains the statistics arising in {\em pairs} of runs of the experiment. Note that running multiple copies of such an experiment is always necessary to gather statistics. (See also Appendix~\ref{sec:gauge}.)} 
is the type illustrated in Fig.~\ref{fig:PTMcomposition} (c). Here, a multi-state  $\{\tilde{\rho}_{\cdot|zx}^B\}_{zx}$ on system $B$ (defined in \cref{fig:PTMcomposition}(a)) is composed with a multi-measurement $\{\{\tilde{M}_{b|z}^A\}_b\}_z$ on system $A$ (defined in \cref{fig:PTMcomposition}(b)), in a manner giving rise to a multi-instrument from system $A$ to $B$, namely
\begin{equation}
\tilde{\mc{E}}^{B|A}_{b''|z'x'z''}(\cdot)\coloneqq\tilde{\rho}_{\cdot|z'x'}^B\tr[\tilde{M}_{b''|z''}^A(\cdot)]. \label{eq:op-PTM-pr}
\end{equation}
(Note that this type of composition operation can be considered as an instance of either parallel or sequential composition.)
The resulting multi-instrument, i.e., $\tilde{\mc{E}}^{B|A}_{b''|z'x'z''}(\cdot)$, is of the same type as the given multi-channel $\{\mc{E}_{\cdot|z}^{B|A}\}_z$, so there can be new operational identities  of the form
\begin{equation}
\label{eq:tranpa}
\sum_z \alpha_z \mc{E}_{\cdot|z}^{B|A} +\sum_{b''x'z'z''} \alpha_{b''x'z'z''}\tilde{\mc{E}}^{B|A}_{b''|x'z'z''}=\mbb{0}^{B|A},
\end{equation}
which implies the corresponding ontological identity
\begin{align}
&\sum_z \alpha_z p(\lambda_B|z\lambda_A)\notag \\
&+\sum_{b''x'z'z''} \alpha_{b''x'z'z''}\tilde p(\lambda_B|x'z')\tilde p(b''|\lambda_Az'')=0
\end{align}
for all $\lambda_A$ and $\lambda_B$. Of course, this operational identity is nontrivial (i.e., not implied by the operational identities from~\cref{eq:op-PTM-t}) only if $\alpha_z\ne 0$ for some $z$ while $\alpha_{b''x'z'z''}\ne 0$ for some $b''x'z'z''$.

If we move beyond prepare–transform–measure scenarios, the prescription is similar, but one must consider all possible ways of composing together processes within the circuit, and then demand that the operational identities on all of these effective processes are preserved in the ontological representation, in a manner analogous to what was done above. 
We again refer the reader to Ref.~\cite{Schmid2024structuretheorem} for a diagrammatic framework that greatly simplifies the explicit treatment of general circuits.

It is worth emphasizing that assessments of classical explainability are always given {\em relative to a circuit}, that is, relative to a factorization of systems into subsystems and a hypothesis about the causal structure. 
This follows from two consequences of diagram preservation: the fact that an ontological model associates an ontic state space to each system in the given circuit, and the fact that the stochastic matrices in the model must be wired together with the same connectivity as in the quantum circuit. Although this choice is quite innocuous in simple cases like a prepare-measure scenario or a Bell scenario~\cite{wood2015lesson}, it is a bit more substantive in general. We will see this in Section~\ref{sec:bipartite} when we discuss the representation of bipartite systems. The reader can also find more discussion of this point in Refs.~\cite{schmid2020unscrambling,schmid2024addressing}.

\subsection{Two notions of variability in the identity of a process}\label{variation}

There are two ways in which one can imagine variability in the identity of a process within a circuit.  On the one hand, this variability might be induced by the variability in the value of a setting variable and/or in the value of an outcome variable. This is a variability that gets explored physically across different runs of the experiment.  On the other hand, we sometimes wish to contemplate a theoretical notion of variability.  For instance, for a process that has neither classical inputs nor classical outputs (such as a single state or transformation), we might wish to consider all possible choices for this process according to quantum theory.  We refer to the former sort of variability as {\em realized variability} and the latter as {\em theoretical variability}.

In the following, when we consider a specific quantum circuit, there will be some realized variability within that circuit for a given process.  This type of variability is considered part of the specification of the circuit, and is captured by the ranges of the classical setting and outcome variables attached to that process.

However, in our assessments of classical explainability of a given type of process, we will often consider \emph{all} the processes that are \emph{dual} to it (in a sense to be specified below). The variability in the dual process is of the theoretical variety.  We are considering what quantum theory predicts for every possible choice of this dual process. 
For the case where the process is a state, multi-state, source, or multi-source, the dual process is an effect. Quantum theory stipulates the probability assigned to all effects, and these probability assignments do not depend on which measurement, or multi-measurement, this effect is considered part of. Similarly, for the case where the process is a measurement or a multi-measurement, the dual process is a state. Quantum theory stipulates the probability generated by all states, and these probabilities do not depend on what source or multi-source this state is considered a part of.

\section{Defining the classical-nonclassical divide for an individual quantum process}\label{basicdefns}

When should a given quantum process be deemed to be nonclassical? Intuitively, it is nonclassical if and only if (i) it can be used in some quantum circuit to realize  operational statistics that are not classically explainable (in the sense of Definition~\ref{defn:classical}), and 
(ii) it is {\em implicated} in the lack of classical explainability.
We begin by clarifying what we mean by ``being implicated in the lack of classical explainability".  Essentially, if a circuit realizes statistics that are not classically explainable {\em independently of the identity of the process in question}, the process is not implicated in the lack of classical explainability.

Consider the question of whether a given multi-source is nonclassical.  Suppose one has the circuit in Figure~\ref{aux}(a), where the multi-source on system $A$ appears in yellow, and suppose that an ancilla system labeled $B$ is also prepared in some state (in blue), and the composite system and ancilla are subjected to a joint measurement.
In this case, it is possible to observe a failure of classical explainability regardless of the identity of the multi-source on the system.
This can be done, for instance, simply by taking the multipartite measurements to be factorized across the system-ancilla divide, and by choosing any set of states and any set of effects on the ancilla sufficient to generate statistics that are not classically explainable (in the sense of Definition~\ref{defn:classical}).

\begin{figure}[htb!]
\centering
\includegraphics[width=0.45\textwidth]{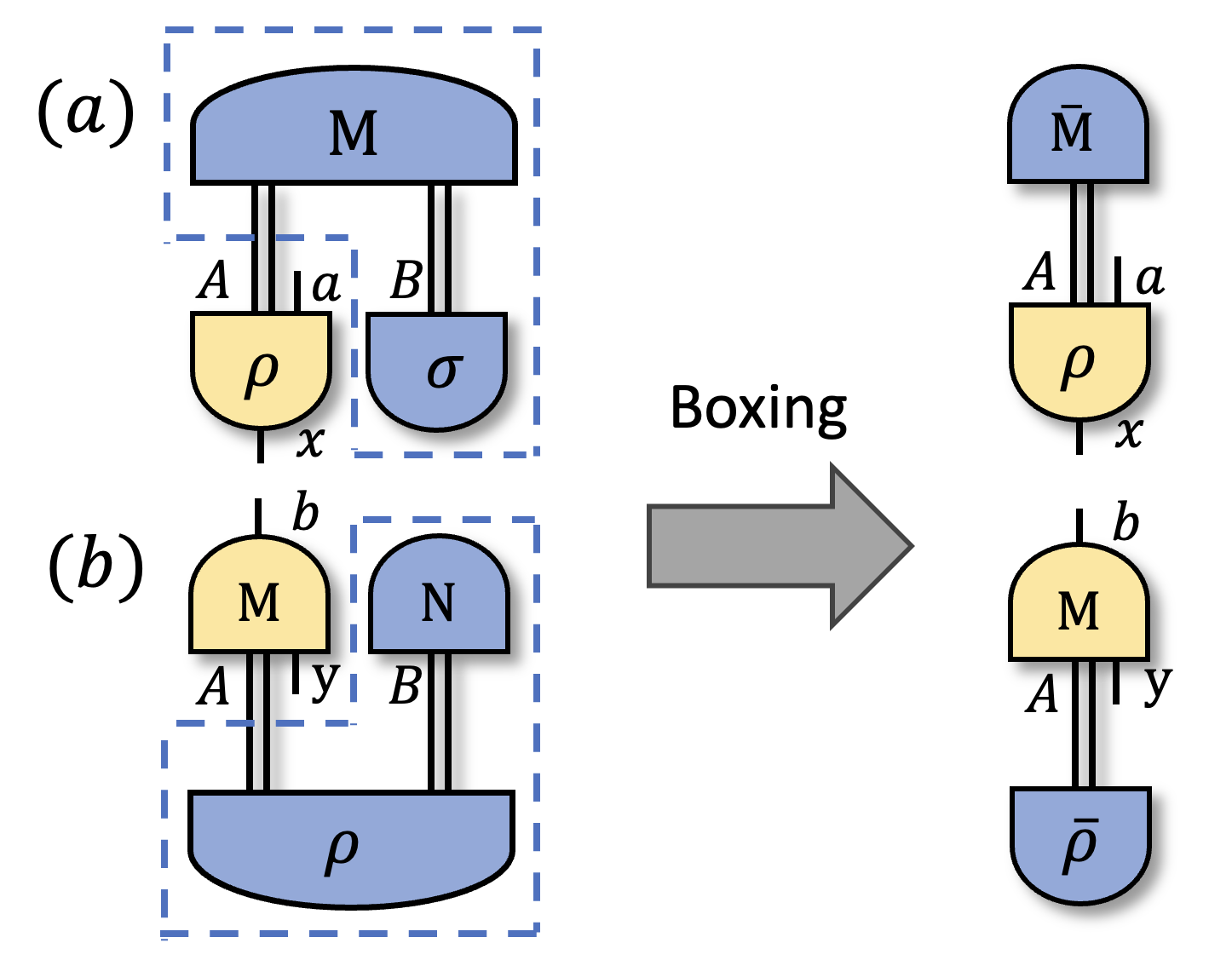}
\caption{ 
Left: Probing schemes (in blue) that can generate statistics that are not classically explainable {\em independently} of the identity of the given process of interest (in yellow). In (a), the process of interest is a multi-source. In (b), the process of interest is a multi-measurement. Right: Treating these same circuits in a manner that coarse-grains over all processes involving systems that are not connected directly to the process of interest 
leads to probing schemes (in blue) such that whether the statistics can be explained classically now always depends nontrivially on the identity of the process of interest. In (a), the effective process is simply an effect $\bar{M}_A$; in (b), the effective process is simply a state $\bar{\rho}_A$. } \label{aux}
\end{figure}

Similarly, consider whether a given multi-measurement is nonclassical. 
Suppose one has the circuit in Figure~\ref{aux}(b), where the multi-measurement appears in yellow, and suppose that the composite of system and ancilla is prepared in some joint state (in blue) and subsequently a factorizing measurement is made, with some local measurement implemented on the ancilla. 
Again, it is possible to observe the failure of classical explainability regardless of the identity of the multi-measurement on the system: simply take the joint state to be factorizing across the system-ancilla divide, and choose any set of states and set of effects on the ancilla is sufficient to generate statistics that are not classical explainable (in the sense of Definition~\ref{defn:classical}).

 In both cases, then, one cannot conclude that the process in question is implicated in the fact that the overall circuit admits of no classical explanation. 

To ensure that the lack of classical explainability of the circuit implies the nonclassicality of the process of interest, it suffices to stipulate that for any circuit that includes the process of interest, one must box together (i.e., compose together) all processes that act on ancillary quantum systems. In other words, this circuit should not introduce additional causal structure that is internal to the probing part of the circuit.

In the above examples, this entails boxing together the blue processes to define an effective process:  in Figure~\ref{aux}(a), one must compose the effect $M^{AB}$ with the state $\sigma^{B}$ to define an effective effect $\bar{M}^A\coloneqq M^{AB}\circ \sigma^{B}$ (i.e., $\tr[\bar{M}^A(\cdot)]\coloneqq \tr[M^{AB}(\cdot\otimes\sigma^B)] $) on system $A$ alone; in Figure~\ref{aux}(b), one must compose the effect $N^{B}$ with the state $\rho^{AB}$ to define an effective state $\bar{\rho}^A\coloneqq N^B\circ\rho^{AB}$ (i.e., $\bar{\rho}^A\coloneqq \tr[(\mbb{1}^A\otimes N^B)\rho^{AB}] $) on system $A$ alone. By boxing processes in this manner, one ensures that any nonclassicality in the circuit comes from constraints pertaining to the system of interest (here, system $A$), as opposed to any ancilla systems internal to the probing part of the circuit (here, system $B$).

Equivalently, one can simply define nonclassicality of a given process in terms of whether or not it can generate statistics that are not classically explainable when embedded in a circuit {\em that does not introduce any new quantum systems}. That is, the circuit cannot contain any quantum systems not connected directly to the process itself (for instance, it cannot contain a system such as $B$ in Figure~\ref{aux}).
This is consistent with the fact that noncontextuality is always defined relative to some specified systems. Consequently, one wishes to consider the possibility of classical explanations {\em for the systems in question} (those connected to the given process) {\em and no others}, and the stipulation just given is a simple heuristic to guarantee this.

We formalize all of this by defining the notion of a dual process to $P$. 
\begin{definition}
\label{def:dualprocess}
    For a given type of quantum process $P$, a {\em dual process} to $P$ is a quantum process that has a quantum input corresponding to every quantum output of $P$ and a quantum output corresponding to every quantum input of $P$,  and that has no other quantum inputs or outputs. 
\end{definition}
If one connects (that is, one ``wires together'')  each quantum output of $P$ to the corresponding quantum input of the dual process and each quantum input of $P$ to the corresponding quantum output of the dual process, one obtains a circuit with no open quantum wires and with which one can associate a probability (for each value of the setting variables and outcome variables of the process $P$). It follows that a dual process defines a linear functional on the process $P$.  We will refer to this wiring together of $D$ and $P$ to form $D\circ P$ as {\em composing} $D$ with $P$.

We depict examples of dual processes in Fig.~\ref{duals}. For the process of the preparation variety on $A$, i.e., a state, multi-state, source, or multi-source on system $A$, the dual process is an effect on $A$. Similarly, for the process of the measurement variety on $A$, i.e., an effect, measurement, or multi-measurement on system $A$, the dual process is a state on $A$.

Throughout the paper, we treat the variability in the states and effects of a dual process as {\em theoretical variability} (as discussed in Sec.~\ref{variation}).  Consequently, dual processes in our treatment do not have classical inputs nor classical outputs. 
As discussed in Sec.~\ref{variation}, there are many different ways of ranging over all dual processes of a given type through {\em physical variability} (i.e., by the variability in the values of classical inputs and classical outputs), but the definition of classicality is independent of the choice of how to do so.  This is discussed further in Appendix~\ref{labels}.

\begin{figure}[htb!]
\centering
\includegraphics[width=0.5\textwidth]{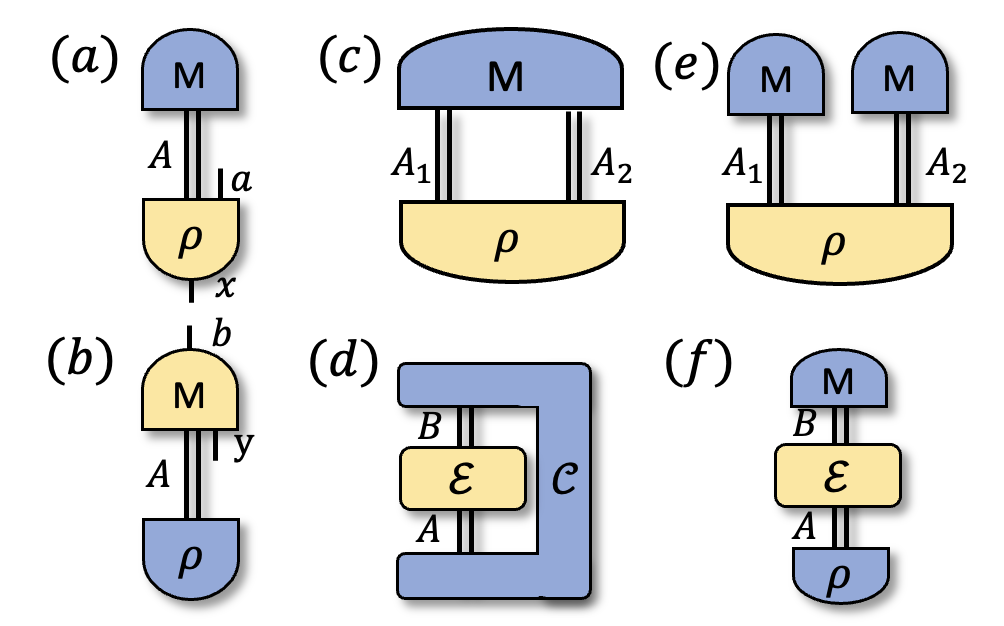}
\caption{
In (a)-(d), we depict examples of various types of quantum processes (in yellow) and their associated dual processes (in blue).  In (a), the process is a multi-source on a single system.  In (b), it is a multi-measurement on a single system.  In (c), it is a state on a bipartite system. In (d), it is a channel. In (e) and (f), the processes are the same as in (c) and (d), respectively, but we depict the associated dual processes when the latter are factorizing. 
}
\label{duals}
\end{figure}

For a state on a multipartite system, the dual process is an effect on that multipartite system, as depicted in \cref{duals}(c).  The same is true if one considers any process of the preparation variety, i.e., a multipartite multi-state, source, or multi-source. For processes with both nontrivial quantum inputs and nontrivial quantum outputs, the dual process is a comb~\cite{Chiribella2008, Chiribella2009}, as depicted in \cref{duals} (d).

With these notions in hand, we can present our proposed notion of classicality for an individual process. 

\begin{definition}
A quantum process is classical if and only if the statistics generated by the set of circuits that contract the process of interest with {\em any} dual process (i.e., where the set ranges over {\em all} choices of dual process) are classically explainable (in the sense of Definition~\ref{defn:classical}). 
\label{def:og_process}
\end{definition}

It is useful to provide a slightly more formal version of this definition. Suppose \texttt{T} denotes a variety of a process, such as a preparation on $A$, a preparation on $AB$, a measurement on $A$, a transformation from $A$ to $B$, etc. For a process $P$ of type $\texttt{T}$, let the set of all processes of the dual type be denoted ${\rm AllDuals}(\texttt{T})$, and recall that the circuit one obtains by contracting a process $P$ with a dual process $D$  is denoted $D\circ P$.  Relative to these notational conventions, the definition can be summarized as follows: a quantum process $P$ is classical if and only if the statistics generated by the set of circuits $\{ D\circ P | D \in {\rm AllDuals}(\texttt{T})\}$ are classically explainable.

Note that this basic definition generalizes naturally to theories beyond quantum theory; that is, to arbitrary GPTs. However, in this paper, we focus specifically on quantum theory. As a consequence, we can simplify the definition considerably.

It turns out that it is sufficient to focus on a special class of dual processes.  To define this class, we introduce the notion of a {\em factorizing process}.

\begin{definition}
A process is termed \emph{factorizing} if it consists of a product of local effects on its quantum inputs and a product of local states on its quantum outputs. 
\end{definition}
A few examples serve to illustrate the idea. A state on a bipartite system $AB$ is factorizing if it is a product of a state on $A$ and a state on $B$. A process with quantum inputs $A$ and $B$ and a single quantum output $C$ is factorizing if it is a product of an effect on $A$ and an effect on $B$, composed sequentially with the preparation of a state on $C$. A comb with quantum output $A$ and quantum input $B$ is factorizing if it consists of a state on $A$ and an effect on $B$.

With this definition, we prove (in Appendix~\ref{secdualsthm}) the following. 

\begin{theorem}\label{thmdual}
A quantum process is such that there exists some set of dual processes that it can be contracted with to obtain statistics that are not classically explainable, if and only if there exists some set of {\em factorizing} dual processes for which this is the case. 
\end{theorem}

Consequently, we can restrict our attention to factorizing dual processes.
For instance, for a bipartite state, it suffices to consider product effects, as depicted in \cref{duals}(e). For a channel, it suffices to consider factorizing combs, as depicted in \cref{duals}(f). With this result in hand, we can provide an alternative definition to Definition~\ref{defn:classical} of the classical-nonclassical divide for a single quantum process as follows.

\setcounter{definition}{3} 
\renewcommand\thedefinition{$3'$} 
\begin{definition}\label{maindefn}
  A quantum process is classical if and only if the statistics generated by the set of circuits that contract the process of interest with any factorizing dual process (i.e., where the set ranges over \emph{all} choices of factorizing dual process) are classically explainable (in the sense of Definition~\ref{defn:classical}). 
\end{definition}

\renewcommand\thedefinition{\arabic{definition}}

This is the central definition of our manuscript.

\section{ Processes of the preparation or measurement variety}
Let us start by reviewing the form of a classical explanation (an ontological model) for the statistics in the simplest prepare-measure scenario. We will consider the most general circuit in such a scenario, with a multi-source on the preparation side and a multi-measurement on the measurement side.

In this section, we will drop the superscript specifying the quantum system since only one system is involved. 
\subsection{Classical explainability of a prepare-measure scenario}
\label{ncstates}
In the prepare-measure scenario, we consider a multi-source $\msf{P}\coloneqq\{\{p(a|x)\rho_{a|x}\}_{a}\}_x$ and a multi-measurement $\msf{M} \coloneqq\{\{M_{b|y}\}_{b}\}_y$, which, when contracted together (as depicted in Fig.~\ref{PM}), generate the statistics 
\begin{equation}
p(ab|xy)=p(a|x)\tr[ M_{b|y}\rho_{a|x}].
\label{eq: noncontexual}
\end{equation}
In this simple context, the only constraints implied by diagram-preservation in this circuit are that states $\rho_{a|x}$ get represented as probability distributions $p(\lambda|ax)$ (the subnormalized state $p(a|x)\rho_{a|x}$ get represented as $p(a\lambda|x):=p(a|x)p(\lambda|ax)$), that effects get represented as response functions $p(b|y\lambda)$, and that one can reproduce the quantum statistics as
\begin{align}
&{p(ab|xy)=\sum_{\lambda} p(b|y\lambda)p(a\lambda|x)}
\label{eq:PMNC}
\end{align}
\begin{figure}[htb!]
\centering
\includegraphics[width=0.45\textwidth]{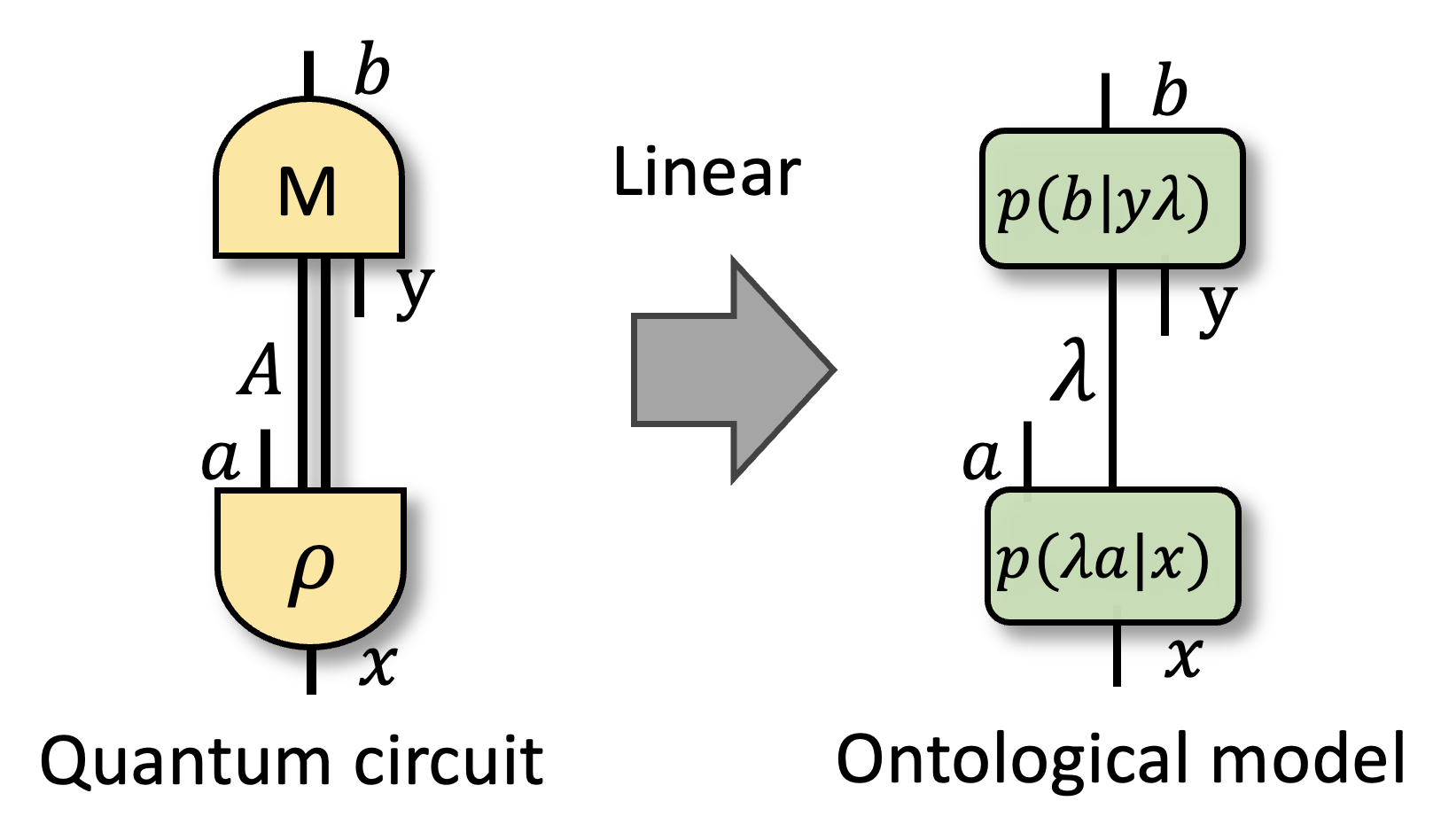}
\caption{A prepare-measure circuit (left) and an ontological model for it (right).}
\label{PM}
\end{figure}

The states and effects in such a circuit satisfy operational identities. We will identify these operational identities by their coefficients, and so we can denote the set of all operational identities for the states and that for the effects, respectively, by
\begin{subequations}
\begin{align}
 &{{\mc{O}}({\msf{P})}\coloneqq\{\{{\alpha}_{ax}\}_{ax} |\sum_{ax} {\alpha}_{ax} p(a|x){\rho}_{a|x}=\mbb{0} \}, \label{op:prep}} \\
 &\mc{O}({\msf{M}})\coloneqq\{\{\alpha_{by}\}_{by} |\sum_{by} \alpha_{by} M_{b|y}=\mbb{0}\}. \label{op:meas}
\end{align}
\end{subequations}

If one further demands that the ontological representation is linear, the conditional probabilities  $p(b|y\lambda)$ and $p(a\lambda|x)$ must satisfy the ontological identities
\begin{subequations}
\begin{align}
&{\sum_{ax} {\alpha}_{ax} p(a\lambda|x)=0 \quad  \forall \{\alpha_{ax}\}_{ax}\in{\mc{O}}({\msf{P})}, \label{nc:prep}}\\
 &\sum_{by} \alpha_{by} p(b|y\lambda)=0 \quad \forall \{\alpha_{by}\}_{by}\in\mc{O}(\msf{M}) \label{nc:meas}
\end{align}
\label{nc:PM}
\end{subequations}
for all $\lambda$.  Consequently, this circuit is classically explainable (i.e., admits a linear and diagram-preserving ontological representation)
if and only if the statistics can be reproduced by \cref{eq:PMNC} under the constraint of \cref{nc:PM}.

\subsection{ The classical-nonclassical divide for multi-sources}
Next, we characterize what states, multi-states, sources, and multi-sources are classical. Because all of these are special cases of the single concept of a multi-source, it suffices to define classicality for the latter. 
This is just a special case of Definition~\ref{maindefn}:
\begin{definition}
\label{def_multi-sourceNC}
A multi-source is classical if and only if the statistics generated by the set of circuits that contract it with any effect (i.e., where the set ranges over {\em all} possible choices of effects) are classically explainable (in the sense of Definition~\ref{defn:classical}).  
\end{definition}

The contraction of a multi-source with an effect results in a circuit of the prepare-measure form,   
as depicted in Fig.~\ref{duals}(a). 

Precisely \emph{which} multi-sources are classical is characterized by the following theorem, which refers to the notions of frame and dual frame. See Appendix~\ref{framerepn} for a review of these. 

\begin{theorem}\label{theoremprep}

For a multi-source $\msf{P}\coloneqq\{\{p(a|x)\rho_{a|x}\}_{a}\}_x$, the following statements are equivalent:\\
(0) It is classical (in the sense of Definition~\ref{def_multi-sourceNC}).\\
(1) Each of the subnormalized states 
$p(a|x)\rho_{a|x}$ can be decomposed as 
    \begin{align}
    \label{eq_rhodecom} {p(a|x)\rho_{a|x} = \sum_{\lambda} p(a\lambda|x)\sigma_{\lambda} \quad \forall a,x}
     \end{align}
    for a set of normalized states $\{\sigma_{\lambda}\}_{\lambda}$  and some conditional probability distribution $p(a\lambda|x)$ satisfying 
    \begin{align}
    \label{eq_paxOp}
&{\sum_{ax} \alpha_{ax} p(a \lambda|x)=0 \qquad     \forall \{\alpha_{ax}\}_{ax}\in\mc{O}(\msf{P})}
    \end{align}
for all $\lambda$, where $\mc{O}(\msf{P})$ is defined in Eq.~\eqref{op:prep}. 
\\(2) Each of the 
normalized states $\rho_{a|x}$ can be decomposed within a frame representation on $V=\text{Span}(\{\rho_{a|x}\}_{ax})$ as
\begin{align}
        \label{eq_rhodecFrame}
        \rho_{a|x} = \sum_{\lambda} \tr[H_{\lambda} \rho_{a|x}]\sigma_{\lambda}, \quad \forall a,x
\end{align}
where the frame is a fixed set of density operators 
$\{\sigma_{\lambda}\}_{\lambda}$ and the dual frame is a set of (possibly nonpositive) Hermitian operators, i.e., $\{H_{\lambda}\}_{\lambda}\subset \text{Herm}(\mc H)$ satisfying $\sum_{\lambda}H_{\lambda}=\mathcal{P}_V(\mbb{1})$ and $\tr[H_{\lambda} \rho_{a|x}]\in [0,1]$ for all $a,x,\lambda$, with $\mathcal{P}_V$ being the superoperator that projects onto the operator space $V$. This condition is depicted in \cref{framesource}.
\end{theorem}
\begin{remark}
\label{remark:Pv}

We introduce the superoperator $\mathcal{P}_V$ to reflect the fact that $\sum_{\lambda}H_{\lambda}$ need not equal identity operator $\mbb{1}$ on $\mc{H}$. (A sufficient condition for $\sum_{\lambda}H_{\lambda}=\mbb{1}$ is that $V$ coincides with  the operator space $\text{Herm}(\mc H)$.)

As an example, the multi-state $\rho_{\cdot|x}=\{\op{0}{0}, \op{+}{+}\}$ is classical with a frame representation of the form of \cref{eq_rhodecFrame} with $\sigma_0=\op{0}{0}, \sigma_1=\op{+}{+}$ and $H_0=\frac{4}{3}\op{0}{0}-\frac{2}{3}\op{+}{+}, H_1=\frac{4}{3}\op{+}{+}-\frac{2}{3}\op{0}{0}$. However, in this case, $\sum_{\lambda}H_{\lambda}=\mathcal{P}_V(\mbb{1})=\frac{2}{3}\op{0}{0}+\frac{2}{3}\op{+}{+}\ne\mbb{1}$. 
\end{remark}

The proof is given in Appendix~\ref{proofthmprep}. We will sometimes refer to this result as a {\em structure theorem} for classical multi-sources.
\begin{figure}[htb!]
\centering
\includegraphics[width=0.35\textwidth]{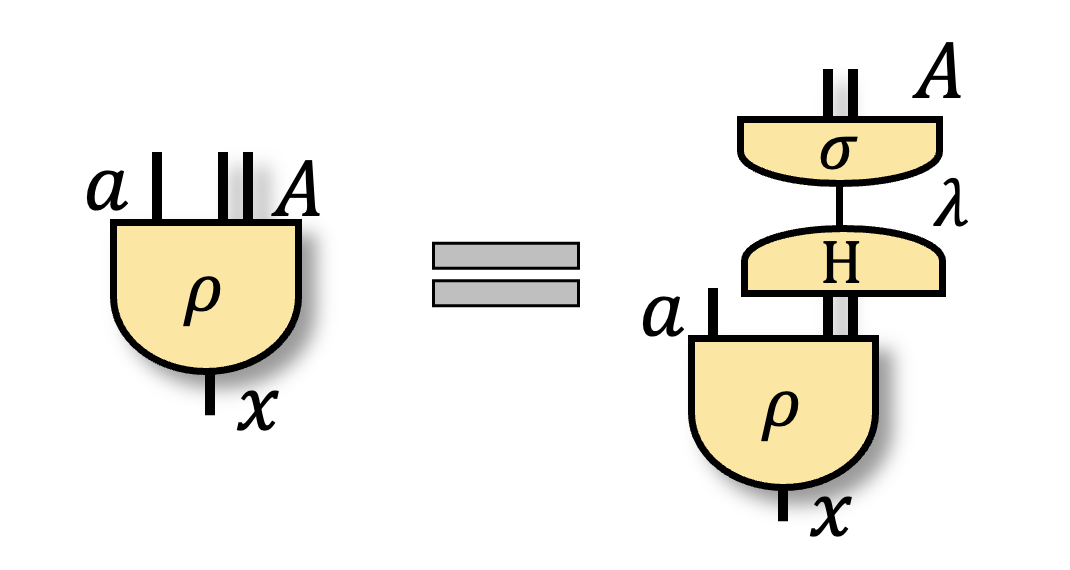}
\caption{A diagrammatic criterion for classicality of a multi-source $\{\{p(a|x)\rho_{a|x}\}_a\}_x$: a multi-source is classical if and only if it admits the frame representation as shown on the right, where the conditions on the frame and dual frame operators are given in Theorem~\ref{theoremprep}. }
\label{framesource}
\end{figure}
In a companion paper~\cite{zhang2024parellel}, we leverage this theorem to provide detailed results on the certification and quantification of the nonclassicality of general multi-sources. 

We can also derive some immediate insights from this theorem. First, consider multi-sources that arise in quantum steering experiments~\cite{Wiseman2007,zhang2025cost}. As shown on the left-hand side of Fig.~\ref{steering}, a pair of systems, $A_1$ and $A_2$, is prepared in an entangled state, and a multi-measurement is implemented on $A_1$. This yields an effective process that is a multi-source on $A_2$.   Such a multi-source always satisfies $\sum_a p(a|x)\rho_{a|x}=\sigma$ for some state $\sigma$ independent of $x$, and is typically termed an {\em assemblage}~\cite{Pusey2013}. An assemblage is said to be unsteerable~\cite{Wiseman2007} if its  unnormalized states can be decomposed as
\begin{equation}\label{unsteerableassemblage}
     p(a|x)  \rho_{a|x}=\sum_\lambda p(\lambda)p(a|x\lambda)\sigma_{\lambda},
\end{equation}
for some set $\{\sigma_{\lambda}\}_{\lambda}$ of density operators on $A_2$, some probability distribution $p(\lambda)$, and some conditional probability distribution $p(a|x\lambda)$, 
as depicted in Fig.~\ref{steering}.
This is often referred to as a local hidden state model of the assemblage.  (Note that in the resource theory of local operations and shared randomness~\cite{sq, Wolfe2020quantifyingbell, Schmid2020typeindependent, Schmid2023understanding}, the unsteerable assemblages are the free resources when one considers processes corresponding to a steering experiment~\cite{Zjawin2023quantifyingepr, Zjawin2023resourcetheoryof}.) If an assemblage fails to satisfy the condition for being unsteerable, it is said to be {\em steerable}. 

\begin{figure}[htb!]
\centering
\includegraphics[width=0.5\textwidth]{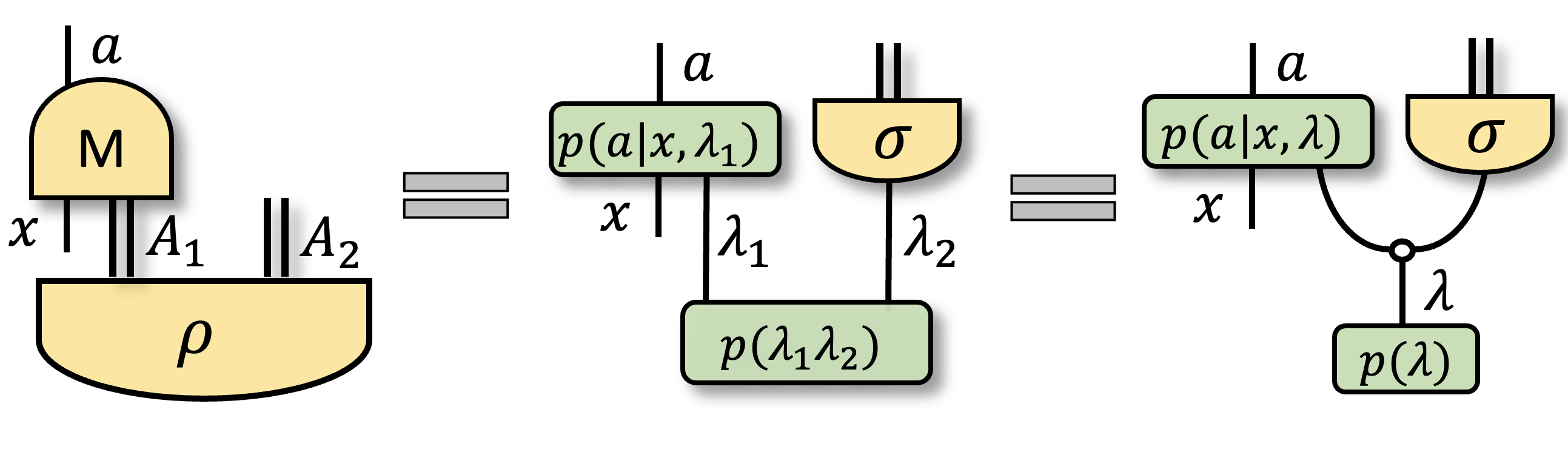}
\caption{ A steering experiment defines a multi-source $\{\{p(a|x)\rho_{a|x}\}_a\}_x$, often termed an assemblage. An assemblage is said to be unsteerable (or free) if and only if it admits a decomposition in terms of local operations and shared randomness, as $\sum_{\lambda_1,\lambda_2}p(a|x\lambda_1)p(\lambda_1\lambda_2)\sigma_{\lambda_2}$, as shown in the middle figure, or equivalently, as $\sum_{\lambda}p(\lambda)p(a|x \lambda)\sigma_{\lambda}$, as shown in the right-hand-side figure, where the dot represents the copying operation.}
\label{steering}
\end{figure}

Given the definition of steerability, Theorem~\ref{theoremprep} immediately implies the following result.
\begin{corollary}
    Every steerable assemblage is nonclassical. 
\end{corollary}
\begin{proof}
We prove the contrapositive that if the multi-source defined by an assemblage is classical, then it is unsteerable. The assumption of classicality is implied by the satisfaction of condition (1) of Theorem~\ref{theoremprep}. Summing Eq.~\eqref{eq_rhodecom} over the outcome index $a$, and using the fact that an assemblage satisfies the operational identity
\begin{align}
   \sum_a p(a|x)\rho_{a|x} = \sum_a p(a|x')\rho_{a|x'} \quad \text{for all } x,x'. 
\end{align}
By \cref{eq_paxOp}, it follows that
\begin{align}
   \sum_a p(a\lambda|x) = \sum_a p(a\lambda|x')
   \quad \forall\,x,x'.    
\end{align}
Defining $p(\lambda|x) \coloneqq \sum_a p(a\lambda|x)$, we conclude that $p(\lambda|x)$ is independent of $x$, i.e. we can denote,
\begin{align}\label{xindep}
   p(\lambda|x) = p(\lambda) \quad \forall x.
\end{align}

By Bayesian inversion,
\begin{align}
&p(a\lambda|x)=p(a|x\lambda)p(\lambda|x),
\end{align}
and so  Eq.~\eqref{xindep} implies
\begin{align}
&p(a\lambda|x)=p(a|x\lambda)p(\lambda).
\end{align}
Plugging this into Eq.~\eqref{eq_rhodecom}, we obtain Eq.~\eqref{unsteerableassemblage}, the condition for being an unsteerable assemblage. 
\end{proof}

In other words, an assemblage being classical (in the sense of Definition~\ref{def:og_process}) implies that it admits a local hidden state model and that it is a free resource relative to local operations and shared randomness.  It is worth noting that, in our framework, every steerable assemblage remains nonclassical even when it is viewed as a general multi-source, i.e.,  without necessarily assuming the lack of causal influence from the setting to the quantum output that is implied by the steering scenario. 
However, the converse is not true, as is established below in Example~\ref{steeringexample}.

Theorem~\ref{theoremprep}  also has consequences for the particular case of a multi-state where the states are linearly independent.
For any multi-state described by a set $\{\rho_{\cdot|x}\}_x$ of linearly independent states on a single system, the set $\mc{O}(\msf{P})$ of operational identities it satisfies is empty, so (by Eq.~\eqref{eq_paxOp}) there are no constraints on the set of distributions $\{ p(a|x)\}_x$ that are the ontological representations of these.
Consequently, one can always find a decomposition as in Eq.~\eqref{eq_rhodecom} (where $a$ is trivial) by defining $\sigma_{\lambda}\coloneqq\rho_{\cdot|\lambda}$ and $p(\lambda|x)=\delta_{x,\lambda}$. So one has the following result (also noted in Ref.~\cite{schmid2024addressing}):
\begin{corollary}
\label{lem:linearinde}
    A multi-state consisting of  a set of linearly independent states on a single system is classical.
\end{corollary} 

From this, one can prove that {\em any} multi-source defined by a set of three or fewer states on a single system is also classical. (The classical explainability of a prepare-measure scenario involving fewer than four states was first noted in Ref.~\cite{Pusey2018}.) The case where the states are linearly independent is covered by Lemma~\ref{lem:linearinde}, so it suffices to consider the case where they are linearly {\em dependent}.  The single-state case is trivial. For a pair of states, linear dependence implies that the states are equal, so we again have a trivial case.  For a triple of states, any possible linear dependence has the form $\rho_3=p\rho_1+(1-p)\rho_2$ for some $p\in [0,1]$ (for some suitable relabeling). There is consequently a single operational identity in the set $\mathcal{O}_{P}$. But in this case, one can find a decomposition 
as in Eq.~\eqref{eq_rhodecom},
where $a$ is trivial,  where $\lambda \coloneqq \{1,2\}$,  where $\sigma_1\coloneqq\rho_1$ and $\sigma_2\coloneqq\rho_2$, and where 
$p(\lambda|x)=\delta_{x,\lambda}$ for $x\in \{1,2\}$ and $p(\lambda|x)=p\delta_{\lambda,1}+(1-p)\delta_{\lambda,2}$ for $x=3$.  The last equation ensures that Eq.~\eqref{eq_paxOp} is satisfied and hence that condition (1) of Theorem \ref{theoremprep} holds. 

This argument generalizes to give the following sufficient condition for classicality of 
a set of states.
\begin{corollary} \label{simplexcor}
If a multi-state consists of a  set of states that fits inside a simplex within the quantum state space, then it is classical.
\end{corollary}

This is proven in Appendix~\ref{simpproof}.

For a set of distinct {\em pure} states, linear dependence is not only necessary 
for nonclassicality (per Corollary \ref{lem:linearinde}), but also sufficient. 
\begin{corollary}\label{purestatescorollary}
A multi-state consisting of a  set of distinct \emph{pure} states on a single system is classical if and only if the states are linearly independent. 
\end{corollary}
The proof is in \cref{app:purestatescorollary}.

So, for example, any set of five pure qubit states is nonclassical, as are any four pure states lying within a plane of the Bloch ball.  An example of such a set of four states arises in the BB84 protocol~\cite{BB84} and in parity-oblivious multiplexing~\cite{POM}: 
\begin{example}\label{BB84ex}
The multi-state $ \{\op{0}{0},\op{1}{1},\op{+}{+},\allowbreak \op{-}{-}\}$ is nonclassical.   
\end{example}
Follows immediately from Theorem~\ref{theoremprep} (condition 2), the qualitative nonclassicality of a multi-state is not affected by rescaling the states contained therein by any numbers in $(0,1]$ (this is also proven in Ref.~\cite{selby2023accessible}). Consequently, we infer the following result. 
\begin{example} \label{steeringexample}
    The source $\{p(a)\rho^{A}\}_a=\{\frac{1}{4}\op{0}{0},\allowbreak \frac{1}{4}\op{1}{1},\frac{1}{4}\op{+}{+},\frac{1}{4}\op{-}{-}\}$ is nonclassical.
\end{example}

Considered as an assemblage, one sees that it is unsteerable by noting that it has no setting variable (i.e., it is a single ensemble, rather than a set of ensembles) and so \cref{unsteerableassemblage} is satisfied trivially; this, therefore, is an example of an unsteerable assemblage that is \emph{nonclassical}.

Note that the converses of Corollaries~\ref{lem:linearinde} and~\ref{simplexcor} do not hold; there exist multi-states that exhibit nontrivial linear dependence relations, and moreover that do not fit inside a simplex within the quantum state space, but are nevertheless classical. This is illustrated by the following example. 
 
\begin{example}\label{simplexex}
 Consider the multi-states
\begin{equation} \label{eq:cube_prep}
  \left\{  \frac{1}{2}[\mbb{1}+\eta\hat{n}_x] \right\}_x,
\end{equation}
where $\{\hat{n}_x\}$ are unit vectors corresponding to the eight vertices of a cube inscribed in the Bloch sphere. This set is classical if and only if $\eta\le \frac{1}{\sqrt{3}}$. (See companion paper~\cite{zhang2024parellel} for the proof.)  However, these states are linearly dependent for all $\eta >0$; moreover, it can be verified geometrically that a shrunken cube with vertices $\{\frac{1}{\sqrt{3}}\hat{n}_x\}$ cannot be contained inside any tetrahedron that fits inside the Bloch sphere, and so the set of states when $\eta=\frac{1}{\sqrt{3}}$ does not fit inside any simplex within the qubit state space.
\end{example}
This example also highlights that our notion of classical multi-states is characterized by a frame representation (as defined in Theorem~\ref{theoremprep}), where the frame may be overcomplete, and the associated decomposition need not be simplicial (i.e., it need not correspond to a basis).

\subsection{The classical-nonclassical divide for multi-measurements}

Next, we turn to measurements. Recall that we use the term ``measurement'' to refer to a device with a quantum input and an outcome variable (no quantum output), and the term ``multi-measurement'' to refer to such a device when it has a setting variable as well.

The definition of nonclassicality of a multi-measurement  is again a special case of Definition~\ref{maindefn}:
\begin{definition}
\label{def_mtmtNC}
A multi-measurement is classical if and only if the statistics generated by the set of circuits where it is contracted with {\em any} state (i.e., where the set ranges over all states) are classically explainable (in the sense of Definition~\ref{defn:classical}).
\end{definition}

Precisely {\em which} multi-measurements are classical is characterized by the following theorem.  

\begin{theorem}\label{theoremmeas}

For a multi-measurement $\msf{M}=\{\{M_{b|y}\}_b\}_y$, the following statements are equivalent:\\
(0) It is classical (in the sense of Definition~\ref{def_mtmtNC}).\\
(1) Its effects can be decomposed as
    \begin{equation}
    \label{eq_Edecom}
        M_{b|y}=\sum_{\lambda} p(b|y\lambda)G_{\lambda} \quad \forall b,y
    \end{equation}
for some POVM $\{G_{\lambda}\}_{\lambda}$ and some conditional probability distribution satisfying 
    \begin{align}
    \label{eq_pbyOp}
\sum_{by} \beta_{by} p(b|y\lambda)=0 \qquad     \forall \{\beta_{by}\}_{by}\in\mc{O}(\msf{M})
    \end{align}
for all $\lambda$, where $\mc{O}(\msf{M})$  is defined in Eq.~\eqref{op:meas}.   \\ 
(2) Its effects can be decomposed within a frame representation on $V=\text{Span}(\{M_{b|y}\})$ as
    \begin{equation}
    \label{eq_MdecFrame}
        M_{b|y}=\sum_{\lambda} \tr[M_{b|y}F_{\lambda}] G_{\lambda} \quad \forall b,y
    \end{equation}
where $\{G_{\lambda}\}_{\lambda}$ is a POVM with $\sum_{\lambda}G_{\lambda}=\mc{P}_{V}(\mbb{1})=\mbb{1}$  and $\{F_{\lambda}\}_{\lambda}\subset\text{Herm}(\mc H)$ is a set of Hermitian operators satisfying $\tr[F_{\lambda}]=1$ and $\tr[M_{b|y}F_{\lambda}]\in [0,1]$ for all $b,y,\lambda$. This condition is depicted in \cref{framemeas}.
\end{theorem}

The proof of the theorem is provided in Appendix~\ref{proofthmmes}. We will sometimes refer to this as a structure theorem for classical multi-measurements.

\begin{remark}
Here, $\sum_{\lambda}G_{\lambda}=\mc{P}_{V}(\mbb{1})=\mbb{1}$ since the operator space $V=\text{Span}(\{M_{b|y}\})$ contains the identity operator (it suffices to recall that $\sum_b M_{b|y} = \mbb{1}$). Note the difference to the situation with the dual frame $\{H_{\lambda}\}_{\lambda}$ in Theorem~\ref{theoremprep}.
\end{remark}

\begin{figure}[htb!]
\centering
\includegraphics[width=0.32\textwidth]{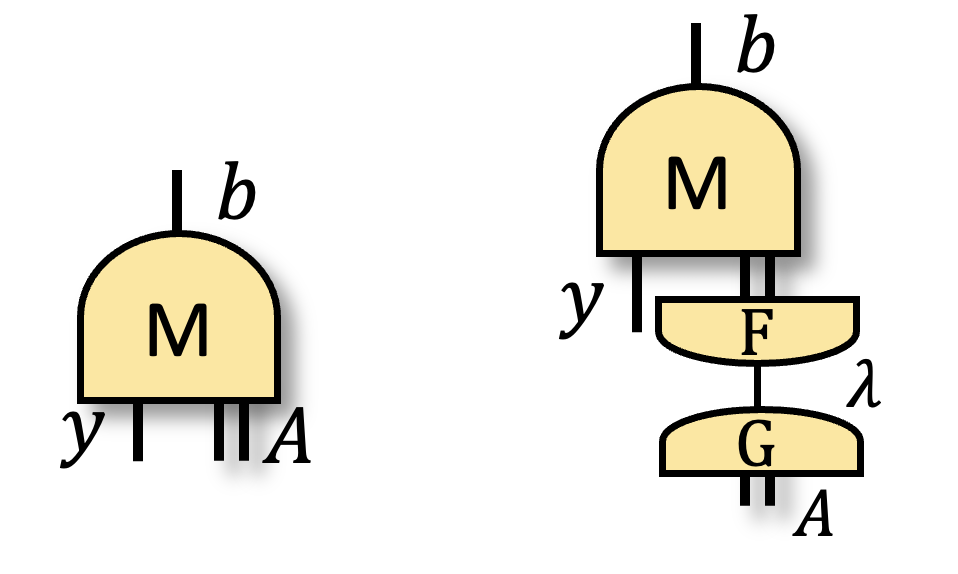}
\caption{A diagrammatic criterion for classicality of a multi-measurement $\{\{M_{b|y}\}_b\}_y$:  a multi-measurement is classical if and only if it admits a frame representation as shown on the right, where the conditions on the frame and dual frame operators are given in Theorem~\ref{theoremmeas}.}
\label{framemeas}
\end{figure}

It is worth noting that an alternative proof of Theorem~\ref{theoremmeas} is possible by leveraging Theorem~\ref{theoremprep} and Lemma~\ref{lem:Choi} from Section~\ref{general}.  Essentially, a multi-measurement defines a special type of multi-source (known as a steering assemblage) via the Choi isomorphism, and Lemma~\ref{lem:Choi} asserts that the verdict of nonclassicality is preserved under the Choi isomorphism.  Conditions (1) and (2) in Theorem~\ref{theoremmeas} are the translations through the Choi isomorphism of conditions (1) and (2) of Theorem~\ref{theoremprep}.

Several interesting insights follow from this theorem, most of which are direct analogs  of those in the previous section.  

First, recall that a set of measurements is compatible~\cite{Guhne2023} if and only if there is a single measurement
$\{G_\lambda\}_\lambda$ that can be post-processed to reproduce every measurement in the set. Representing this set of measurements as a multi-measurement $\{\{M_{b|y}\}_y\}_y$, the condition of compatibility is
\begin{equation}
M_{b|y}=\sum_\lambda p(b|y\lambda)G_{\lambda},
\end{equation}
as depicted in Fig.~\ref{incompatiblemmts}.
Any set of measurements for which the condition cannot be satisfied is said to be incompatible.

\begin{figure}[htb!]
\centering
\includegraphics[width=0.35\textwidth]{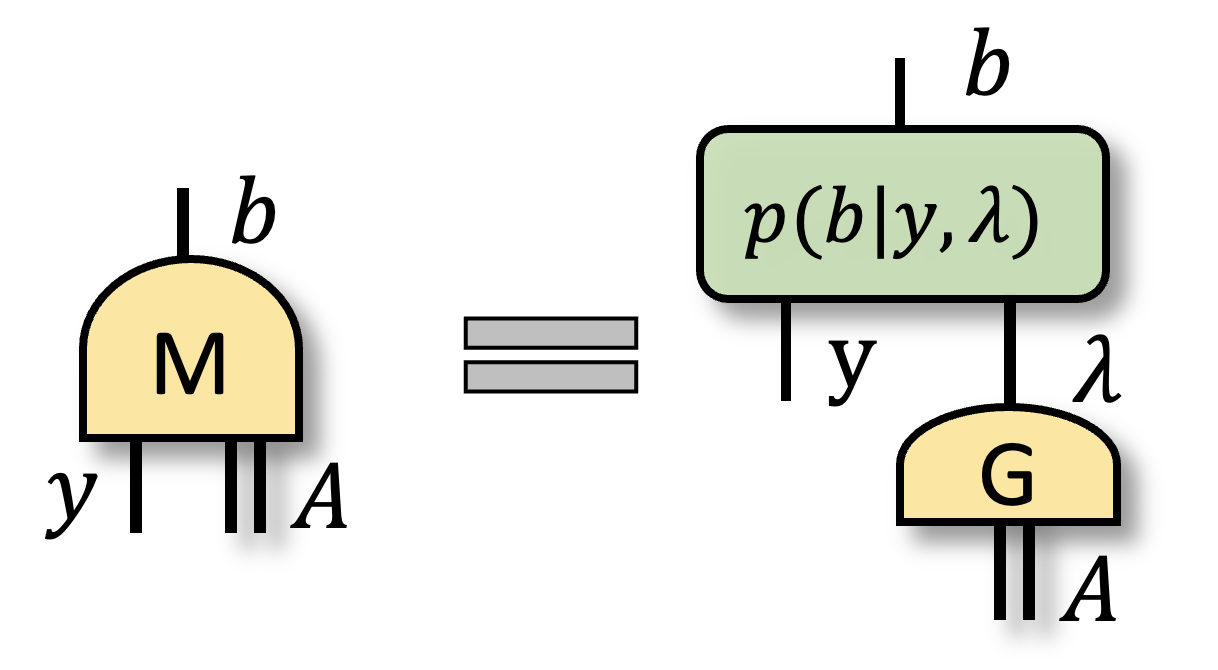}
\caption{ A set of measurements $\{\{M_{b|y}\}_b\}_y$ is compatible if and only if there is a single measurement which can be classically post-processed to simulate each element of the set. This implies that a multi-measurement ranging over those measurements decomposes as shown (with the single measurement denoted by $\{G_{\lambda}\}$).  }
\label{incompatiblemmts}
\end{figure}

Recalling this, Theorem~\ref{theoremmeas} immediately implies the following sufficient condition for nonclassicality.  
\begin{corollary}\label{corr:incompatible}
Every multi-measurement consisting of a  set of incompatible measurements is nonclassical. 
\end{corollary}

We prove this by demonstrating the contrapositive, that classicality implies compatibility. Specifically, classicality implies the satisfaction of condition 1 (Eq.~\eqref{eq_Edecom}) of Theorem~\ref{theoremmeas}, which implies that the POVM $\{G_{\lambda}\}_\lambda$ defines a single measurement that can simulate all of the measurements in the given set via post-processing~\cite{Guhne2023}, hence the set of measurements is compatible. The converse of Corollary~\ref{corr:incompatible} does not hold, due to the extra constraints in Eq.~\eqref{eq_pbyOp} of condition 1; there are sets of measurements that are all compatible and yet for which the multi-measurement they define is nonetheless nonclassical. 

A more exhaustive comparison between nonclassicality and incompatibility of a set of measurements is carried out in our companion paper~\cite{zhang2024parellel}.  

The analog for effects of our Corollary~\ref{purestatescorollary} holds  by an analogous proof: 
\begin{corollary}
\label{coro: measurement}
A measurement defined by a set of rank-1 effects  is classical
if and only if the effects are linearly independent. 
\end{corollary}

One significant implication of this corollary is 
that a measurement represented by a rank-1 symmetric informationally-complete POVM is classical~\footnote{We thank Thomas Galley for discussions in which this was first realized.}. 

Unlike a single state, a single measurement $\{M_b\}_b$ can be nonclassical, provided that there are nontrivial operational identities among its effects~\cite{Selby2023}. 
\begin{example}
The noisy Pentagon measurement $\{M^{\eta}_b\}_{b=0}^4$, where
\begin{equation}
    M^{\eta}_b=\frac{1}{5}(\mbb{1}+\eta(\cos\theta_b\sigma_x+\sin\theta_b\sigma_z)]~~~\theta_b=\frac{2\pi b}{5},
    \label{eq: planar}
\end{equation}
and $\eta \in [0,1]$ is a noise parameter (this multi-measurement was defined for $\eta=1$ in Ref.~\cite{Selby2023}).  This single measurement is nonclassical if and only if $\eta \gtrapprox 0.618$ as proven in Ref.~\cite{zhang2024parellel}.
\end{example}

\section{Processes of the bipartite preparation variety}
\label{sec:bipartite}

Recall that a single state on a unipartite system is classical, as a consequence of Corollary~\ref{lem:linearinde}.  However, once one considers multipartite systems, even a single state can be nonclassical. This is because the subsystem structure provides new opportunities for witnessing nonclassicality (as first recognized in Ref.~\cite{schmid2024addressing}).

When defining a linear and diagram-preserving ontological representation of a circuit that involves parallel composition of two or more systems (represented diagrammatically by a pair of wires), the assumption of diagram preservation implies that each system is associated with its own independent ontic state space. The ontic state space of the joint system is the Cartesian product of the ontic state spaces of the subsystems, an assumption that is sometimes called {\em ontic separability}.

Note that, if one instead wishes to entertain an ontological representation wherein the ontic state space of the joint system is {\em not} a Cartesian product, then one should rather represent the bipartite system in the quantum circuit as a single monolithic system---without a factorization into subsystems. In this case, the assumption of diagram preservation would not force one to assume ontic separability. Under such a choice, the joint system would be treated as a unipartite system and represented diagrammatically by a single wire rather than a pair of wires. 

Recall from Section~\ref{classicalexplainability} that assessments of nonclassicality are always made relative to a given circuit, which effectively embodies a hypothesis about the causal structure. It should be clear from the previous paragraph that for
a quantum system represented by a 4-dimensional Hilbert space, viewing it as (i) a single monolithic system or (ii) as a pair of qubits (according to some factorization of the system into a pair of 2-dimensional subsystems) constitutes two different hypotheses about the causal structure that are represented by different circuits.  Indeed, the latter view, wherein a privileged factorization is specified for the full system,  gives a more fine-grained perspective on the system. This is reflected in the fact that an ontological model associates to a monolithic system an ontic state space with no particular structure, while it associates to a bipartite system an ontic state space that has extra structure---it must be the Cartesian product of two ontic state spaces, one for each subsystem. 

Assessments of nonclassicality may depend on which view of the system one takes. In particular, the unipartite view may lead to an assessment of classicality, while the bipartite view may lead to an assessment of nonclassicality. ({Note that if the bipartite view leads to an assessment of classicality, then the unipartite view will as well.}) Indeed, we will see this in the following results, where we show that any entangled state on a bipartite system is nonclassical relative to the specified bipartition, whereas we saw in Section~\ref{ncstates} (from Corollary~\ref{lem:linearinde}, for example) that any individual state is classical when it is conceptualized as a state on a unipartite system with no specified subsystem structure.

To discuss the nonclassicality of a bipartite state, one must consider the circuit obtained by contracting it with its dual process, which is a bipartite effect, as depicted in Fig.~\ref{duals}(c).  Given Theorem~\ref{thmdual}, however, it suffices to consider the effects that factorize across the bipartition---hence, only {\em local} measurements. Thus, it suffices to consider only the type of circuits depicted in Fig.~\ref{duals}(e).   

\subsection{Classical explainability of a prepare-measure scenario on a bipartite system with local measurements}

As preparation for our discussion of nonclassicality of bipartite states, we consider the question of when a prepare-measure scenario involving a bipartite state and local measurements is classically explainable in the sense of Definition~\ref{defn:classical}. 

Specifically, we consider a circuit wherein the composite system $A_1 A_2$ is prepared in a particular  bipartite state $\msf{P}^{A_1A_1}:=\rho^{A_1A_2}$ and then $A_1$ is subjected to a multi-measurement $\msf{M}^{A_1}_1$ while $A_2$ is subjected to a multi-measurement $\msf{M}^{A_2}_2$, where $\msf{M}^{A_i}_i\coloneqq\{\{M_{b_i|y_i}^{A_i}\}_{b_i}\}_{y_i}$. In such an experiment, depicted in Fig.~\ref{fig:steer}(a), quantum theory predicts that
\begin{equation}
p(b_1b_2|y_1y_2)=\tr[\rho^{A_1A_2} M_{b_1|y_1}^{A_1}\otimes M_{b_2|y_2}^{A_2}].
\end{equation}

\begin{figure}[htb!]
\centering
\includegraphics[width=0.5\textwidth]{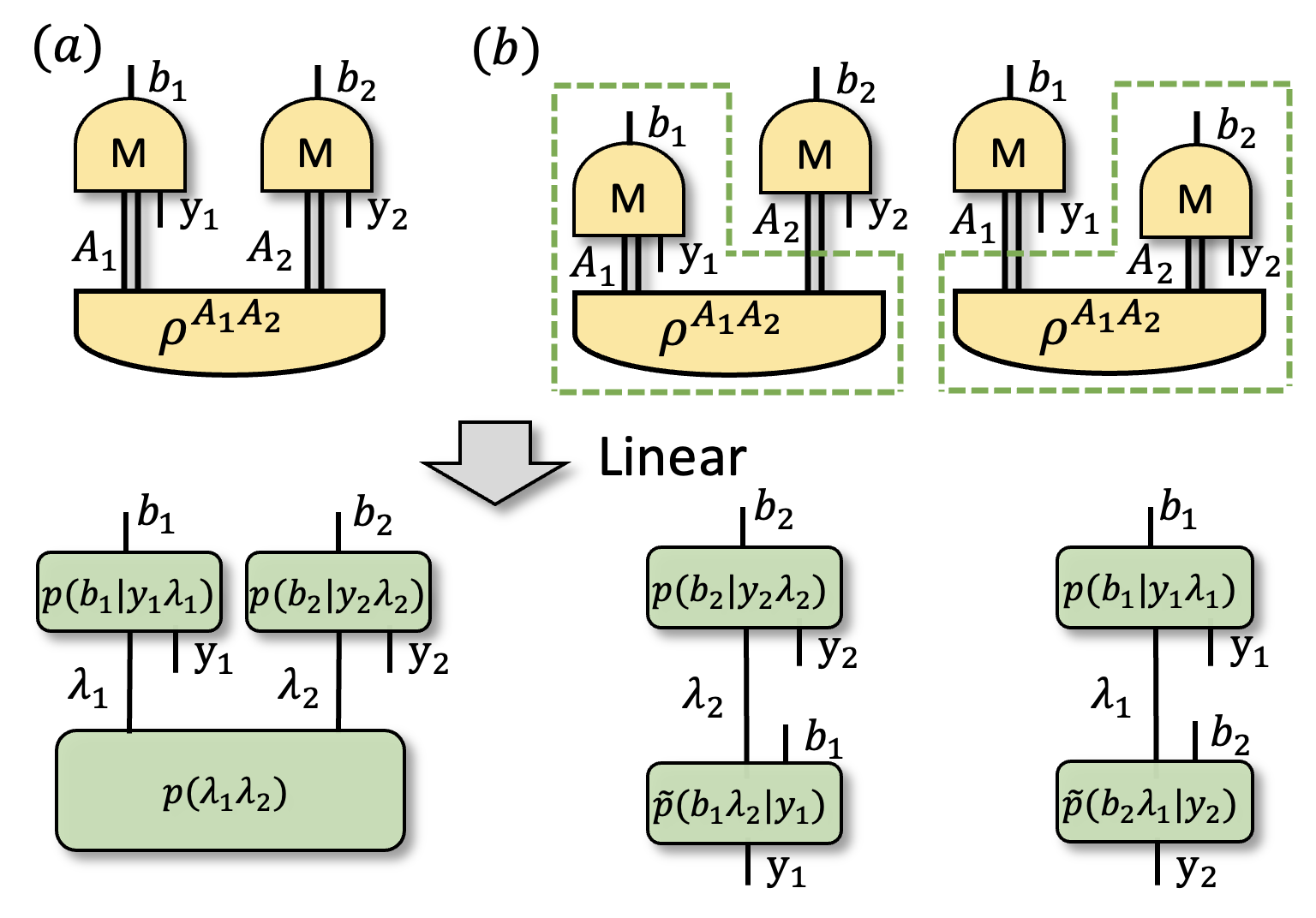}
\caption{a) Local measurements on a bipartite state. b) By boxing processes together, we can view this circuit as a set of preparations and measurements on $A_2$, or as a set of preparations and measurements on $A_1$. At the top of the figure, we depict the quantum circuits; at the bottom, we depict an ontological model of each circuit.}
\label{fig:steer}
\end{figure}

There may be operational identities among the effects on each system, namely
\begin{subequations}
\begin{align}
 &\mc{O}(\msf{M}^{}_1)=\{\{\alpha_{b_1y_1}\} |\sum_{b_1y_1} \alpha_{b_1y_1} M_{b_1|y_1}^{A_1}=\mbb{0}\},\label{op:meas1} \\ 
 &\mc{O}(\msf{M}_2)=\{\{\alpha_{b_2y_2}\} |\sum_{b_2y_2} \alpha_{b_2y_2} M_{b_2|y_2}^{A_2}=\mbb{0}\}.\label{op:meas2}
\end{align}
\label{op:bipartite1}
\end{subequations}
\par 
Note that for the sake of reducing notational clutter, we will sometimes drop the system-identifying superscripts on the names of processes when these appear as arguments inside $\mc{O}(\cdot)$. Thus we write $\mc{O}(\msf{M}^{}_i)$ instead of $\mc{O}(\msf{M}^{A_i}_i)$.

There may also be further operational identities in this circuit, namely, identities holding for the steering assemblages on $A_1$ and $A_2$ depicted in \cref{fig:steer}(b), the steering assemblage (i.e., multi-source) on $A_2$, for instance, is defined by compositing of a multi-measurement on system $A_1$ with the bipartite state on $A_1A_2$. Similarly for the steering assemblage (i.e., multi-source) on $A_1$. Formally, the multi-sources on $A_1$ and $A_2$ are defined by
\begin{subequations}
\begin{align}
\msf{M}^{A_1}_1\circ \msf{P}^{A_1A_2}:&=\{\{\tilde{\rho}_{b_1|y_1}^{A_2}\}_{b_1}\}_{y_1}  \label{eq:steer A} \\
&=  \{\{ \tr_{A_1}[(M_{b_1|y_1}^{A_1}\otimes \mbb{1}^{A_2})\rho^{A_1A_2}]\}_{b_1}\}_{y_1}, \notag \\
\msf{M}^{A_2}_2\circ \msf{P}^{A_1A_2}:&=\{\{\tilde{\rho}_{b_2|y_2}^{A_1}\}_{b_2}\}_{y_2} \label{eq:steer B}\\
&=  \{\{ \tr_{A_2}[(\mbb{1}^{A_1}\otimes M_{b_2|y_2}^{A_2})\rho^{A_1A_2}]\}_{b_2}\}_{y_2}. \notag 
\end{align} 
\label{eq:steer}
\end{subequations}
Note that here and elsewhere in the paper, we will use $\tilde{\rho}$ to denote a subnormalized state. The operational identities satisfied by these multi-sources are of the form:
\begin{subequations}
\begin{align}
 &\mc{O}(\msf{M}_1\circ \msf{P})=\{\{\alpha_{b_1y_1}\} |\sum_{b_1y_1} \alpha_{b_1y_1}\tilde{\rho}_{b_1|y_1}^{A_2}=\mbb{0}^{A_2}\}, \label{op:prep1}\\ 
 &\mc{O}(\msf{M}_2\circ \msf{P})=\{\{\alpha_{b_2y_2}\}|\sum_{b_2y_2}{\alpha}_{b_2y_2} \tilde{\rho}_{b_2|y_2}^{A_1}=\mbb{0}^{A_1}\}. \label{op:prep2}
\end{align}
\label{op:bipartite2}
\end{subequations}

Note that these last two sets of identities actually subsume those in Eq.~\eqref{op:meas1} and Eq.~\eqref{op:meas2}. To see that $\mc{O}(\msf{M}_1) \subseteq \mc{O}(\msf{M}_1\circ \msf{P})$ for instance, it suffices to note that $\sum_{b_1y_1} \alpha_{b_1y_1} M_{b_1|y_1}^{A_1}=\mbb{0}^{A_1}$ implies $\sum_{b_1y_1} \alpha_{b_1y_1}\tr_{1}[M_{b_1|y_1}^{A_1}\otimes\mbb{1}^{A_2}\rho^{A_1A_2}] =\mbb{0}^{A_2}$.

There is also a more subtle kind of operational identity, analogous to the one in \cref{eq:tranpa}. Specifically, the unnormalized states $\tilde{\rho}_{b_2|y_2}^{A_1}$ and $\tilde{\rho}_{b_1|y_1}^{A_2}$
 in the steering assemblage (i.e., in Eq.~\eqref{eq:steer}) can also be composed in parallel, giving rise to bipartite states on $A_1A_2$. Since such states are of the same type as the given bipartite state $\rho^{A_1A_2}$, there \textit{might} exist a new operational identity for some set of real coefficients $\{\alpha_{b_1b_2y_1y_2} \}_{b_1b_2y_1y_2}$ such that
\begin{equation}
 \rho^{A_1A_2} =
\sum_{{b_1b_2y_1y_2}}\alpha_{b_1b_2y_1y_2} \tilde{\rho}_{b_2|y_2}^{A_1}\otimes \tilde{\rho}_{b_1|y_1}^{A_2}.
\label{eq:OPnewbi}
\end{equation}
Such an identity exists if and only if $\rho^{A_1A_2}$ lies in the linear span of the set $\{\tilde{\rho}_{b_2|y_2}^{A_1}\otimes \tilde{\rho}_{b_1|y_1}^{A_2}\}_{b_1b_2y_1y_2}$. The operational identity of the form of Eq.~\eqref{eq:OPnewbi} will be denoted by
\begin{align}
&\mc{O}(\tilde{\msf{P}})\\
&=\{\{\alpha_{b_1b_2y_1y_2} \}| \rho^{A_1A_2} =
\sum_{{b_1b_2y_1y_2}}\alpha_{b_1b_2y_1y_2} \tilde{\rho}_{b_2|y_2}^{A_1}\otimes \tilde{\rho}_{b_1|y_1}^{A_2} \}. \notag
\end{align}

The circuit in Fig.~\ref{fig:steer}(a) is classically explainable (i.e., admits an ontological model) if and only if the following conditions are met. One can introduce  a separable ontic state space $\Lambda_1\times \Lambda_2$ over which one can reproduce the quantum statistics via 
\begin{align}
&  p(b_1b_2|y_1y_2)=\sum_{\lambda_1\lambda_2} p(\lambda_1\lambda_2)p(b_1|y_1\lambda_1)  p(b_2|y_2\lambda_2), 
\label{eq:cl-bipartite}
\end{align}
where $p(b_1|y_1\lambda_1)$ satisfies the ontological identities specified by $\mc{O}(\msf{M}_1)$, and $p(b_2|y_2\lambda_2)$ satisfies the ontological identities specified by $\mc{O}(\msf{M}_2)$. Moreover, defining compositions
\begin{subequations}
\begin{align}
\tilde{p}(b_1\lambda_2|y_1)&\coloneqq\sum_{\lambda_1} p(b_1|y_1\lambda_1)p(\lambda_1\lambda_2),  \\
\tilde{p}(b_2\lambda_1|y_2)&\coloneqq\sum_{\lambda_2} p(b_2|y_2\lambda_2)p(\lambda_1\lambda_2), 
\end{align}
\label{eq:steeronto}
\end{subequations}
we require that, $\tilde{p}(b_1\lambda_2|y_1)$ must satisfy the ontological identities specified  by $\mc{O}(\msf{M}_1\circ \msf{P})$, and $\tilde{p}(b_2\lambda_1|y_2)$ must satisfy the ontological identities specified by $\mc{O}(\msf{M}_2\circ \msf{P})$. Finally, we require that $p(\lambda_1\lambda_2)$ must satisfy the ontological identity specified  by $\mc{O}({\widetilde{\msf{P}}})$. All told, we require that, for all $\lambda_1,\lambda_2$,
\begin{subequations}
\begin{align}
 &\sum_{b_1y_1} \alpha_{b_1y_1} p(b_1|y_1\lambda_1) =0 \quad \forall \{\alpha_{b_1y_1}\}\in \mc{O}(\msf{M}_1),\label{ncmeas1} \\ 
 &\sum_{b_2y_2} \alpha_{b_2y_2} p(b_2|y_2\lambda_2)=0\quad \forall \{\alpha_{b_2y_2}\}\in \mc{O}(\msf{M}_2), \label{ncmeas2} \\ 
 &\sum_{b_1y_1} \alpha_{b_1y_1}\tilde{p}(b_1\lambda_2|y_1) =0 \quad \forall \{\alpha_{b_1y_1}\}\in \mc{O}(\msf{M}_1\circ \msf{P}), \label{ncprep1}\\ 
 &\sum_{b_2y_2}{\alpha}_{b_2y_2}\tilde{p}(b_2\lambda_1|y_2)=0 \quad \forall \{\alpha_{b_2y_2}\}\in \mc{O}(\msf{M}_2\circ \msf{P}), \label{ncprep2}\\
 &p(\lambda_1\lambda_2)=\sum_{{b_1b_2y_1y_2}} \alpha_{b_1b_2y_1y_2} \tilde{p}(b_2\lambda_1|y_2)\tilde{p}(b_1\lambda_2|y_1) \notag \\
&\quad\quad\quad\quad\quad\quad\quad\quad\quad\quad\quad  \forall \{\alpha_{b_1b_2y_1y_2}\}\in \mc{O}(\tilde{\msf{P}}). \label{ncnewstate}
\end{align}
\label{op:bipartiterep}
\end{subequations}

Note that the first two constraints in Eq.~\eqref{op:bipartiterep} follow from linearity, while the latter three follow from diagram-preservation together with linearity.

Importantly, since the ontological model for the steering assemblages has already been explicitly constructed in Eq.~\eqref{eq:steeronto}, this yields the following necessary condition for classical explainability:
\begin{proposition}\label{prop:bipartite}
Any circuit with local measurements on a bipartite state (as depicted in Fig.~\ref{fig:steer}(a)) is classically explainable only if the two prepare-measure circuits it induces by boxing processes together (in one or the other direction, as in Fig.~\ref{fig:steer}(b)) are classically explainable.
\end{proposition}

The intuition behind this result is straightforward. To determine the classical explainability of a prepare-measure scenario obtained by boxing operation, it suffices to check only a subset of the constraints in~\cref{op:bipartiterep} (specifically only \cref{ncprep1,ncmeas2} or only \cref{ncprep2,ncmeas1}), rather than the full set that must be checked in the original bipartite scenario. Consequently, if a bipartite scenario is classically explainable, then so is each of the unipartite prepare-measure scenarios it defines by steering.  

For completeness, a formal proof of this implication is provided in Appendix~\ref{steerproof}.

\subsection{ The classical-nonclassical divide for bipartite states
}
\label{sec:bipartitedivide}

We obtain the definition of classicality for a single bipartite state by particularizing Definition~\ref{maindefn} to this type of process:

\begin{definition}
A state on a bipartite system is classical if and only if the statistics generated by the set of circuits where it is contracted with {\em any} product effect (i.e., where the set ranges over all product effects) is classically explainable (in the sense of Definition~\ref{defn:classical}).
\label{def:bipartite}
\end{definition}

We therefore consider a prepare-measure circuit of the type depicted in Fig.~\ref{fig:steer}(a), with a single bipartite state and local measurements as depicted in Fig.~\ref{duals}(e), where the range of measurements includes every possible effect on $A_1$ and every possible effect on $A_2$. A necessary condition for classicality of a bipartite state can then be obtained as a simple corollary of Proposition~\ref{prop:bipartite}. 

\begin{proposition}
A bipartite state $\rho^{A_1A_2}$
is classical only if  every multi-source on $A_2$ that is obtained by contracting $\rho^{A_1A_2}$ with a multi-measurement  on $A_1$ is classical, and every multi-source on $A_1$ that is obtained by contracting $\rho^{A_1A_2}$ with a multi-measurement on $A_2$ is classical. 
\label{propremote}
\end{proposition}

\begin{proof}
The proof is simple. If the bipartite state is classical, then by Definition~\ref{def:bipartite}, the circuit where it is contracted with any product effect is classically explainable.  By Proposition~\ref{prop:bipartite}, this implies that the two prepare-measure circuits (as depicted in Fig.~\ref{fig:steer}(b)) are classically explainable when all local effects on each side are considered. But classical explainability for each prepare-measure circuit then implies that each corresponding set of steered states (the composition between the bipartite state and local effects on the other side) is classical according to Definition~\ref{def_multi-sourceNC}.
\end{proof}

Whether the classicality of the set of steered states on both sides is {\em sufficient} to establish the classicality of the bipartite state is still unknown. However, in the following, we will provide the necessary and sufficient conditions for a classical bipartite state. 

First, define the operator spaces
\begin{subequations}
\begin{align}
&V_1 \\
&\coloneqq\text{Span}\big(\{\tr_{A_2}[(\mbb{1}^{A_1}\otimes M^{A_2})\rho^{A_1A_2}]: M^{A_2}\in\mc M(\mc H^{A_2})\}\big),\notag \\
&V_2 \\
&\coloneqq\text{Span}\big(\{\tr_{A_1}[(M^{A_1}\otimes \mbb{1}^{A_2} )\rho^{A_1A_2}]: M^{A_1}\in\mc M(\mc H^{A_1})\}\big).\notag
\end{align}
\label{eq: bi-operator space}
\end{subequations}
Whether or not a bipartite state $\rho^{A_1A_2}$ is classical can then be characterized based on a product-frame decomposition, analogous to Theorems~\ref{theoremprep} and~\ref{theoremmeas}.

\begin{theorem}\label{ncbp}
A bipartite state $\rho^{A_1A_2}\in\mc S(\mc H^{A_1}\otimes\mc H^{A_2})$ is classical if and only if it admits a frame decomposition:
\begin{equation}\label{eq:bipartiteframe}
\rho^{A_1A_2}
=\sum_{\lambda_1,\lambda_2}
\tr\big[(H^{A_1}_{\lambda_1}\otimes H^{A_2}_{\lambda_2}) \rho^{A_1A_2}\big]
\sigma^{A_1}_{\lambda_1}\otimes\sigma^{A_2}_{\lambda_2},
\end{equation}
where the frame on the operator space $V_i$ is a fixed set of density operators $\{\sigma^{A_i}_{\lambda_i}\}_{\lambda_i}$, and the dual frame is a set of Hermitian operators $\{H^{A_i}_{\lambda_i}\}_{\lambda_i}$ 
satisfying $\sum_{\lambda_{i}}H^{A_i}_{\lambda_i}=\mc{P}_{V_i}(\mbb 1^{A_i})$ with $\mc{P}_{V_i}$ denoting the superoperator projecting onto $V_i$, and such that
\begin{equation}
\tr\big[(H^{A_1}_{\lambda_1}\otimes H^{A_2}_{\lambda_2}) \rho^{A_1A_2}\big]\ge 0
\quad \forall \lambda_1,\lambda_2.\label{eq:constraints}
\end{equation} 
\end{theorem}

\par 
\begin{figure}[htb!]
\centering
\includegraphics[width=0.4\textwidth]{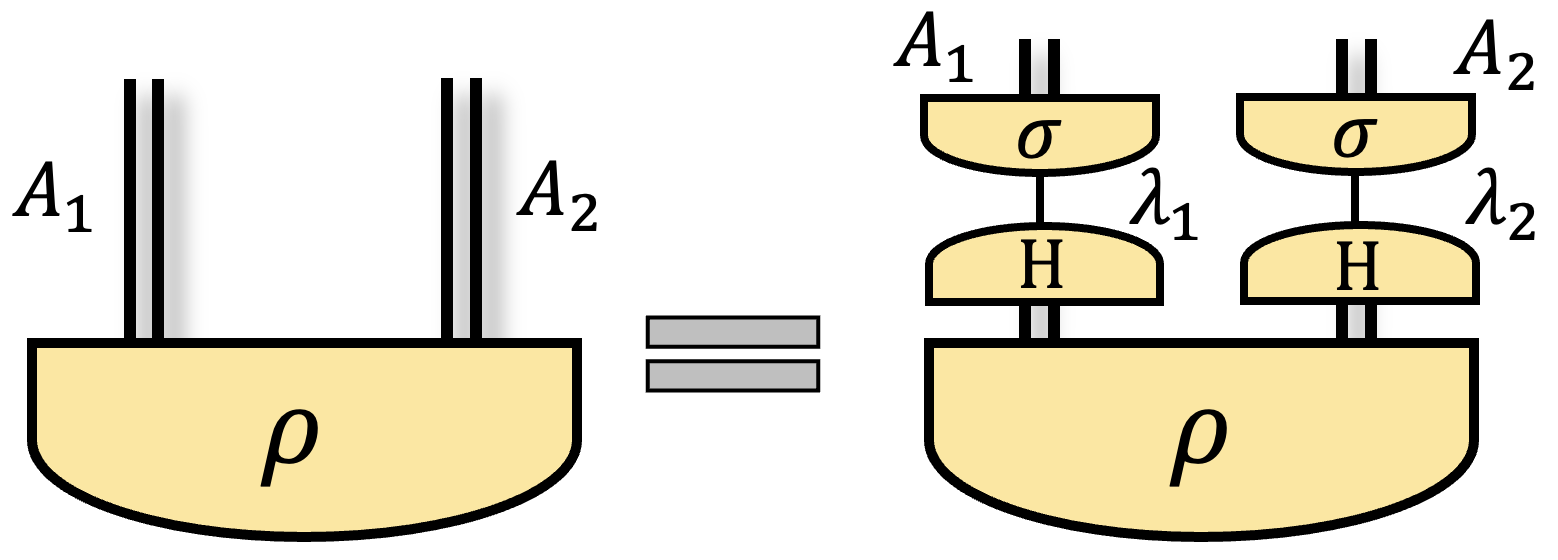}
\caption{A bipartite state is classical if and only if it admits a frame representation as shown on the right-hand  side of the diagrammatic equation, where the condition on the frame and dual frame operators are given in Theorem~\ref{ncbp}.}
\label{framebp}
\end{figure}
\par 
As a direct consequence of Theorem~\ref {ncbp}, it is easy to see that for a bipartite state, entanglement implies nonclassicality because satisfying \cref{eq:bipartiteframe,eq:constraints} in the condition for classicality articulated in Theorem~\ref{ncbp} implies that $\rho^{A_1A_2}$ must be separable. 
 
\begin{corollary}
If a bipartite state is entangled, then it is nonclassical.
\label{corr:bipartiteentangle}
\end{corollary}

Within the set of all bipartite states, the boundary between separable and entangled is distinct from the boundary between steerable and unsteerable, and both of these are also distinct from the boundary between local and nonlocal~\cite{Werner1989,wisesteer,zhang2024,Renner2024}.  As it turns out, the boundary between classical and nonclassical, as defined here, is also distinct from each of these. This follows from the fact that while separability is a {\em necessary} condition for classicality (the content of Corollary~\ref{corr:bipartiteentangle}), it is not a sufficient condition, as is demonstrated by the following example.

\begin{example}\label{singularexample}
The separable bipartite state
\begin{align}\label{eq:sep-nl}
\rho^{A_1A_2}&=\frac{1}{4}(\op{0}{0}\otimes\op{0}{0}+\op{1}{1}\otimes\op{1}{1} \notag \\
&+\op{2}{2}\otimes\op{+}{+}+\op{3}{3}\otimes\op{-}{-})
\end{align}
is nonclassical, because there exists a measurement on $A_1$, namely, the projector-valued measure $\{\op{0}{0},\op{1}{1},\op{2}{2},\op{3}{3}\}$ that steers $A_2$ to the ``BB84 source" $\{\frac{1}{4}\op{0}{0},\frac{1}{4}\op{1}{1},\frac{1}{4}\op{+}{+},\frac{1}{4}\op{-}{-}\}$, which was shown to be nonclassical in Example~\ref{steeringexample}. The claim then follows from Proposition 2. 
\end{example}

If a separable state can be nonclassical, it is natural to ask whether this nonclassicality constitutes a resource for achieving quantum advantages in any information-processing task. The answer is that it does (at least in some cases). Consider the separable state from Example~\ref{singularexample} above. It induces the nonclassical BB84 source on $A_2$,  which is known to be a resource for quantum key distribution~\cite{BB84} and for parity-oblivious multiplexing~\cite{Spekken2009}. That said, a complete understanding of the resourcefulness of separable states that are nonclassical---particularly whether their resourcefulness extends beyond their ability to realize nonclassical resources via steering---remains open.

It is worth noting that, nonclassicality of a bipartite state, in the sense introduced here, does not coincide with having nonzero discord~\cite{Modi2012}; for example, the state $\frac{1}{2}\op{0}{0}\otimes\op{0}{0}+\frac{1}{2}\op{1}{1}\otimes\op{+}{+}$ has nonzero discord but is classical by our definition. It is not hard to see that nonzero discord is necessary but not sufficient for a single bipartite state to be nonclassical.

A complete characterization of the boundary between classical and nonclassical for 
bipartite states is likely to be difficult to obtain. For one, the set of classical bipartite states is not convex, because all product states are classical, but not all separable states are.

However, in the special case of a two-qubit system, a complete characterization is possible to obtain.
\begin{proposition} 
\label{prop: biqubit}
A bipartite qubit state $\rho^{A_1A_2}$ is classical if and only if it is separable. 
\end{proposition}
\begin{proof}
The “only if” direction follows from Corollary~\ref{corr:bipartiteentangle}; we now prove the “if” direction. 

We first treat the special case where $\rho^{A_2}\coloneqq\tr_{A_1}[\rho^{A_1A_2}]=\frac{\mbb{1}}{2}$. By Ref.~\cite{Jevtic2014}, for any two-qubit separable state with $\rho^{A_2}=\frac{\mbb{1}}{2}$,  the set of states that can be steered on $A_1$ (i.e., the steering ellipsoid) is contained in a tetrahedron in the Bloch ball. Let $d$ denote the dimension of the steering ellipsoid of $\rho^{A_1A_2}$ (an integer less than or equal to 3), Ref.~\cite{Jevtic2014} shows constructively that ${\rho}^{A_1A_2}$ admits a decomposition 
\begin{equation}
{\rho}^{A_1A_2}=\sum_{\lambda_1,\lambda_2=0}^dp_{\lambda_1}\delta_{\lambda_1,\lambda_2}{\sigma}^{A_1}_{\lambda_1}\otimes {\sigma}^{A_2}_{\lambda_2},    
\end{equation}
where $\{{\sigma}^{A_i}_{\lambda_i}\}_{\lambda_i=0}^d$ are linearly independent sets of  states. Given the linear independence, there exist unique dual frames $\{ {H}_{\lambda_i}^{A_i}\}_{\lambda_i=0}^d$ on the space $V_i\coloneqq\text{Span}(\{{\sigma}_{\lambda_i}^{A_i}\}_{\lambda_i=0}^d)$ such that $\tr[{H}_{\lambda_i}^{A_i}{\sigma}^{A_i}_{\lambda'_i}]=\delta_{\lambda_i,\lambda'_i}$ (one can infer that $\sum_{\lambda_i}{H}_{\lambda_i}^{A_i}=\mc{P}_{V_i}(\mbb{1})$). Therefore, one can construct the following frame decomposition for ${\rho}^{A_1A_2}$:
\begin{align}
{\rho}^{A_1A_2}&=\sum_{\lambda_1,\lambda_2=0}^d\tr[{H}_{\lambda_1}^{A_1}\otimes {H}_{\lambda_2}^{A_2}{\rho}^{A_1A_2}]{\sigma}^{A_1}_{\lambda_1}\otimes {\sigma}^{A_2}_{\lambda_2}
\end{align}
 where $
 \tr[{H}_{\lambda_1}^{A_1}\otimes{H}_{\lambda_2}^{A_2} {\rho}^{A_1A_2}]=p_{\lambda_1}\delta_{\lambda_1,\lambda_2}\ge 0 $. 
We have shown that the conditions of Theorem~\ref{ncbp} are satisfied, and so it follows that any two-qubit separable state  ${\rho}^{A_1A_2}$ satisfying $\tr_{A_1}[\rho^{A_1A_2}]=\frac{\mbb{1}}{2}$ is classical.

Next, consider the case where $\rho^{A_2}\ne\frac{\mbb{1}}{2}$. If $\rho^{A_2}$ is rank-1, $\rho^{A_1A_2}$ must be of the form $\rho^{A_1A_2}=\rho^{A_1}\otimes\rho^{A_2}$ and thus is classical. If $\rho^{A_2}$ has full rank, we can obtain a locally filtered state $\tilde{\rho}^{A_1A_2}$ via:
\begin{equation}
    \tilde{\rho}^{A_1A_2}\coloneqq(\mbb{1}^{A_1}\otimes(2\rho^{A_2})^{-\frac{1}{2}})\rho^{A_1A_2}(\mbb{1}^{A_1}\otimes(2\rho^{A_2})^{-\frac{1}{2}})
\end{equation}
so that $\tr_{A_1}[\tilde{\rho}^{A_1A_2}]=\frac{\mbb{1}^{A_2}}{2}$. Since this local filtering preserves separability,  $\tilde{\rho}^{A_1A_2}$ is still separable. Now, given dual-frame operators $\{\tilde{\sigma}^{A_i}_{\lambda_i}\}_{\lambda_i}$ and $\{\tilde{H}_{\lambda_i}^{A_i}\}_{\lambda_i}$ for $\tilde{\rho}^{A_1A_2}$ defined on the operator space $\tilde{V}_i\coloneqq\text{Span}(\{\tilde{\sigma}_{\lambda_i}^{A_i}\}_{\lambda_i=0}^d)$, one can map these frame operators back to $\rho^{A_1A_2}$ by defining
\begin{subequations}
\begin{align}
&{\sigma}^{A_1}_{\lambda_1}\coloneqq\tilde{\sigma}^{A_1}_{\lambda_1}, \quad\quad\quad {H}_{\lambda_1}^{A_1}=\tilde{H}_{\lambda_1}^{A_1},\\
&{\sigma}^{A_2}_{\lambda_2}\coloneqq\frac{(2\rho^{A_2})^{\frac{1}{2}}\tilde{\sigma}^{A_2}_{\lambda_2}(2\rho^{A_2})^{\frac{1}{2}}}{\tr[\tilde{\sigma}^{A_2}_{\lambda_2}(2\rho^{A_2}))]} \label{eq:newH1},\\ 
&{H}_{\lambda_2}^{A_2}\coloneqq\tr[\tilde{\sigma}^{A_2}_{\lambda_2}(2\rho^{A_2}))]\big[(2\rho^{A_2})^{-\frac{1}{2}}\tilde{H}_{\lambda_2}^{A_2}(2\rho^{A_2})^{-\frac{1}{2}}\big]. \label{eq:newH}
\end{align}
\end{subequations}
Then, one has the frame representation 
\begin{align}
{\rho}^{A_1A_2}=\sum_{\lambda_1,\lambda_2=0}^d\tr[{H}_{\lambda_1}^{A_1}\otimes {H}_{\lambda_2}^{A_2} {\rho}^{A_1A_2}]{\sigma}^{A_1}_{\lambda_1}\otimes{\sigma}^{A_2}_{\lambda_2},    
\end{align}
It is easy to see that $ \tr[{H}_{\lambda_1}^{A_1}\otimes {H}_{\lambda_2}^{A_2} {\rho}^{A_1A_2}]\ge 0$, and moreover, $\{{H}^{A_1}_{\lambda_1}\}$ and $\{{\sigma}^{A_1}_{\lambda_1}\}$ defines a dual basis on $V_1$, and $
\sum_{\lambda_1}{H}^{A_1}_{\lambda_1}=\mc{P}_{V_1}(\mbb{1})$  since they are untouched.  In order to show that $
\sum_{\lambda_2}{H}^{A_2}_{\lambda_2}=\mc{P}_{V_2}(\mbb{1})$, it  suffices to note that from \cref{eq:newH1,eq:newH}, we have
\begin{align}
\tr[{H}_{\lambda_2}^{A_2}{\sigma}_{\lambda'_2}^{A_2}]&=\frac{\tr[\tilde{\sigma}^{A_2}_{\lambda_2}(2\rho^{A_2}))]}{\tr[\tilde{\sigma}^{A_2}_{\lambda'_2}(2\rho^{A_2}))]}\tr[\tilde{H}_{\lambda_2}^{A_2}\tilde{\sigma}_{\lambda'_2}^{A_2}]\notag \\
&=\frac{\tr[\tilde{\sigma}^{A_2}_{\lambda_2}(2\rho^{A_2}))]}{\tr[\tilde{\sigma}^{A_2}_{\lambda'_2}(2\rho^{A_2}))]}\delta_{\lambda_2,\lambda_2'}=\delta_{\lambda_2,\lambda_2'}
\end{align}
where the second equality holds because
$\{\tilde{\sigma}^{A_2}_{\lambda_2}\}_{\lambda_2}$ and $\{\tilde{H}^{A_2}_{\lambda_2}\}_{\lambda_2}$ are a pair of basis and dual basis. 

Again, having shown that the conditions of Theorem~\ref{ncbp} are satisfied, it follows that the two-qubit state $\rho^{A_1A_2}$ is classical. 
\end{proof}

\begin{remark}
The proof can be run in the opposite direction, i.e., by considering the steering ellipsoid on $A_2$.
\end{remark}
Using an analogous proof technique, we can find a sufficient condition for classicality for bipartite states, independent of the subsystem dimensions. 
\begin{corollary}
A bipartite state $\rho^{A_1A_2}$ is classical if it admits a separable decomposition of the form:
\begin{equation}
\rho^{A_1A_2}=\sum_{\lambda_1\lambda_2}p(\lambda_1\lambda_2)\sigma^{A_1}_{\lambda_1}\otimes \sigma^{A_2}_{\lambda_2},
\end{equation}
where each of $\{\sigma^{A_1}_{\lambda_1}\}_{\lambda_1}$ and $\{\sigma^{A_2}_{\lambda_2}\}_{\lambda_2}$ is a set of  linearly independent states. 
\end{corollary}
For subsystems having Hilbert space dimensionality greater than 2, there exist separable states that do not admit such a decomposition, as is illustrated in Example~\ref{singularexample}.  This fact allows the possibility of separable nonclassical states when one moves beyond qubits.   

\par
In the following, we briefly discuss a conjecture concerning the task of witnessing the nonclassicality of a bipartite state.  It is based on a property we term {\em nonsingularity}  (which we will define right after the conjecture).

\begin{conjecture}\label{SuffSepForNonsingular}
A nonsingular bipartite state is classical if and only if it is separable.
\end{conjecture}
The nonsingularity  property is defined as follows.

\begin{definition}\label{defn:nonsingularity}
A bipartite state $\rho^{A_1A_2}$ is said to be {\em nonsingular} if 
the {\em steering maps} induced by $\rho^{A_1A_2}$ are invertible.  The steering maps are the superoperators $\Psi^{\rho}_1$ and $\Psi^{\rho}_2$
defined by the expressions
\begin{align}
    &\Psi_1^{\rho}(X)=\tr_{A_1}[(X^T\otimes \mbb{1})\rho^{A_1A_2}],\\ 
    &\Psi_2^{\rho}(X)=\tr_{ A_2 }[(\mbb{1} \otimes X^T)\rho^{A_1A_2}],
\end{align}
\end{definition} 
where $T$ denotes the transpose. An equivalent characterization of what it is for a bipartite state to be  nonsingular can be given in terms of operational identities:  a bipartite state is  nonsingular if and only if the coefficients defining the linear identities holding among the measurement effects on $A_1$ are equivalent to those defining the linear identities holding among the steered states on $A_2$, and vice-versa, that is, if and only if $\mc{O}(\msf{M_1}\circ\msf{P})=\mc{O}(\msf{M_1})$ and $\mc{O}(\msf{M_2}\circ\msf{P})=\mc{O}(\msf{M_2})$ (recall that $\mc{O}(\msf{M_1})$, $\mc{O}(\msf{M_2})$, $\mc{O}(\msf{M_1}\circ\msf{P})$, and $\mc{O}(\msf{M_2}\circ\msf{P})$
are defined in Eqs.~\eqref{op:bipartite1} and ~\eqref{op:bipartite2} respectively).  One can summarize this characterization of nonsingularity as follows: a bipartite state is nonsingular if it does not induce any new operational identities when contracted with the local effects.  

That this characterization of nonsingularity  coincides with that of Definition~\ref{defn:nonsingularity} 
can be easily seen from the fact that $\mc{O}(\msf{M_i}\circ\msf{P})$ in Eq.~\eqref{op:bipartite2} can be expressed as 
\begin{align}
\mc{O}(\msf{M_i}\circ\msf{P})=
 \{\{\alpha_{b_iy_i}\}|\sum_{b_i,y_i} 
\alpha_{b_iy_i} \Psi_i^{\rho}([M^{A_i}_{b_i|y_i}]^T)
=0\},
\end{align}
but using the invertibility of $\Psi_i^{\rho}$ and the invertibility of the transpose, this implies 
\begin{align}
\mc{O}(\msf{M_i}\circ\msf{P})=
\{ \{\alpha_{b_iy_i}\}| \sum_{b_i,y_i} 
\alpha_{b_iy_i} M_{b_i|y_i}^{A_i}
=0\},
\end{align}
such that, recalling the definition of $\mc{O}(\msf{M_i})$ from Eqs.~\eqref{op:meas1} and ~\eqref{op:meas2}, we infer that $\mc{O}(\msf{M_i}\circ\msf{P})=\mc{O}(\msf{M_i})$.

Corollary~\ref{corr:bipartiteentangle} establishes that for bipartite states, being separable is a necessary condition for being classical. To establish Conjecture~\ref{SuffSepForNonsingular}, what remains to be proven is that for nonsingular bipartite states, being separable is also a sufficient condition for being classical. When considering what constraints the operational identities imply for a noncontextual ontological model, one only needs to consider the first two sets of  constraints 
in Eq.~\eqref{op:bipartiterep}, which are of the form
\begin{align}
\sum_{b_iy_i}\alpha_{b_iy_i} p(b_i|y_i,\lambda_i)=0~~~~\forall \{\alpha_{b_iy_i}\}\in \mc{O}(\msf{M_i}),
\end{align}
This is because the other 
two sets of constraints in Eq.~\eqref{ncprep1} and Eq.~\eqref{ncprep2}, namely,
\begin{align}
  \sum_{b_iy_i}\alpha_{b_iy_i} \sum_{\lambda_{j}}p(\lambda_i\lambda_j) p(b_i|y_i,\lambda_i)=0~\forall\{\alpha_{b_iy_i}\} \in \mc{O}(\msf{M_i}\circ\msf{P}),
\end{align}
are automatically satisfied since we have $\mc{O}(\msf{M_i}\circ\msf{P})=\mc{O}(\msf{M_i})$ for nonsingular bipartite states. 
\par

Now, let us discuss why Conjecture~\ref{SuffSepForNonsingular} might hold, and why we are unable to provide a complete proof. Consider a bipartite state; from the proof of Corollary~\ref{corr:bipartiteentangle}, the state is separable if and only if the statistics generated by the set of circuits where it is contracted with any product effect admits an ontological model that satisfies the first two sets of identities in Eq.~\eqref{op:bipartiterep}. If the bipartite state is also nonsingular, as we discussed above, we can disregard the third and fourth sets of constraints. In other words, statistics generated by a nonsingular bipartite state automatically satisfy the first four sets of identities in Eq.~\eqref{op:bipartiterep}.

What remains unclear is whether the fifth constraint in Eq.~\eqref{op:bipartiterep}—namely, the one implied by the operational identity of $\mc{O}(\tilde{\msf{P}})$—has nontrivial consequences for the characterization of classicality for bipartite states. Thus,  it is still an open question whether every nonsingular separable state is classical.

As discussed at the beginning of this section, diagram-preservation implies that an ontological model assigns a Cartesian product of ontic state spaces to a particular tensor product of Hilbert spaces. This leads to additional constraints on the ontological representation of quantum processes---constraints that can in turn imply nonclassicality of a process that would otherwise be judged classical.  This highlights how conclusions about nonclassicality only hold {\em relative to} a choice of factorization of the Hilbert space. For instance, as noted above, a 4-dimensional Hilbert space can be factorized into a pair of 2-dimensional Hilbert spaces in many ways, where the factor spaces correspond to virtual subsystems~\cite{zanardi2004quantum}. Because for a given pure state there is always some choice of factorization relative to which it is entangled, it follows from Corollary~\ref{corr:bipartiteentangle} that there is always some choice of factorization relative to which it is nonclassical.  In this article, the choice of factorization of Hilbert space relative to which nonclassicality is being judged is always specified by the circuit one is considering.
 
Note that the ideas and results in this section generalize straightforwardly to multipartite states, and that some of these generalizations appear in Section~\ref{general}. 

\section{Processes of the transformation variety }

To discuss the nonclassicality of a process of the transformation variety (a channel, multi-channel, instrument, or multi-instrument), one must consider the circuit obtained by contracting the process with its dual process, namely, a comb, as depicted in Fig.~\ref{duals}(d).  But given Theorem~\ref{thmdual}, we need not consider all combs. It suffices to consider those that factorize, meaning they can be understood as a preparation preceding the process and a measurement following it.  In other words, it suffices to consider circuits that have the form depicted in Fig.~\ref{PTMcoarse}(a), termed a {\em prepare-transform-measure scenario}.  

As in the previous sections, we begin by reviewing the form of a classical explanation (i.e., a diagram-preserving and linear ontological representation) of a prepare-transform-measure scenario. We consider the case where the preparation and measurement stages are general (a multi-state and a multi-measurement, respectively), but the transformation is simply a channel, rather than a general process of the transformation variety. We label the quantum system before and after the transformation by $A$ and $B$.

\subsection{Classical explainability of a prepare-transform-measure scenario}

We consider a  multi-source $\msf{P}^A\coloneqq\{\{p(a|x)\rho^A_{a|x}\}_{a}\}_x$ on $\mc{H}^{A}$, followed by a   
channel $\msf{T}^{B|A}\coloneqq\mc{E}^{B|A}$: $\mc{L}(\mc{H}^A)\mapsto \mc{L}(\mc{H}^B)$, and a multi-measurement $\msf{M}^B\coloneqq\{\{M_{b|y}^B\}_{b}\}_y$ on $\mc{H}^{B}$.  For such a circuit, quantum theory predicts that 
\begin{equation}
p(ab|xy)=p(a|x)\tr[M_{b|y}^B  \mc{E}^{B|A} (\rho_{a|x}^A)].
\end{equation}
In this circuit, there are operational identities among the states on the input system and among the effects on the output system, namely,
\begin{subequations}
\begin{align}
 &\mc{O}(\msf{M})=\{\{\alpha_{by}\} |\sum_{by} \alpha_{by} M_{b|y}^B=\mbb{0}^B\}, \\ 
  &\mc{O}(\msf{P})=\{\{\alpha_{ax}\} |\sum_{ax} \alpha_{ax} p(a|x)\rho_{a|x}^A=\mbb{0}^A\}.
\end{align}
 \label{eq:op-PTM1}
\end{subequations} 
There are also additional operational identities, as a consequence of the fact that the \textit{sequential} composition of the channel, together with the multi-source on its input, defines an effective multi-source, and the sequential composition of the channel, together with the multi-measurement on its output, defines an effective multi-measurement. Specifically, the effective multi-source and multi-measurement are, 
\begin{subequations}
\begin{align}
&\msf{M}^B\circ\msf{T}^{B|A}:= \{\{\tilde{M}_{b|y}^A\}_b\}_y=\{\{[\mc{E}^{B|A}]^{\dagger }(M_{b|y}^B)\}_b\}_y,  \label{eq:ch-seq-A} \\
&\msf{T}^{B|A}\circ \msf{P}^A:=\{\{\tilde{\rho}^B_{a|x}\}_a\}_x=\{\{p(a|x)\mc{E}^{B|A}(\rho_{a|x}^A)\}_a\}_x, \label{eq:ch-seq-B}
\end{align}
\label{eq:ch-seq-comp}
\end{subequations}
and the operational identities for these are of the form:
\begin{subequations}
\begin{align}
&\mc{O}(\msf{M}\circ\msf{T})=\{\{\alpha_{by}\}|\sum_{by}{\alpha}_{by} \tilde{M}_{b|y}^A =\mbb{0}^{A}\},  \\ 
&\mc{O}(\msf{T}\circ\msf{P})=\{\{\alpha_{ax}\}|\sum_{ax}{\alpha}_{ax} \tilde{\rho}^B_{a|x} =\mbb{0}^{B}\}.
\end{align}
 \label{eq:op-PTM2}
\end{subequations}
 
\begin{figure}[htb!]
\centering
\includegraphics[width=0.5\textwidth]{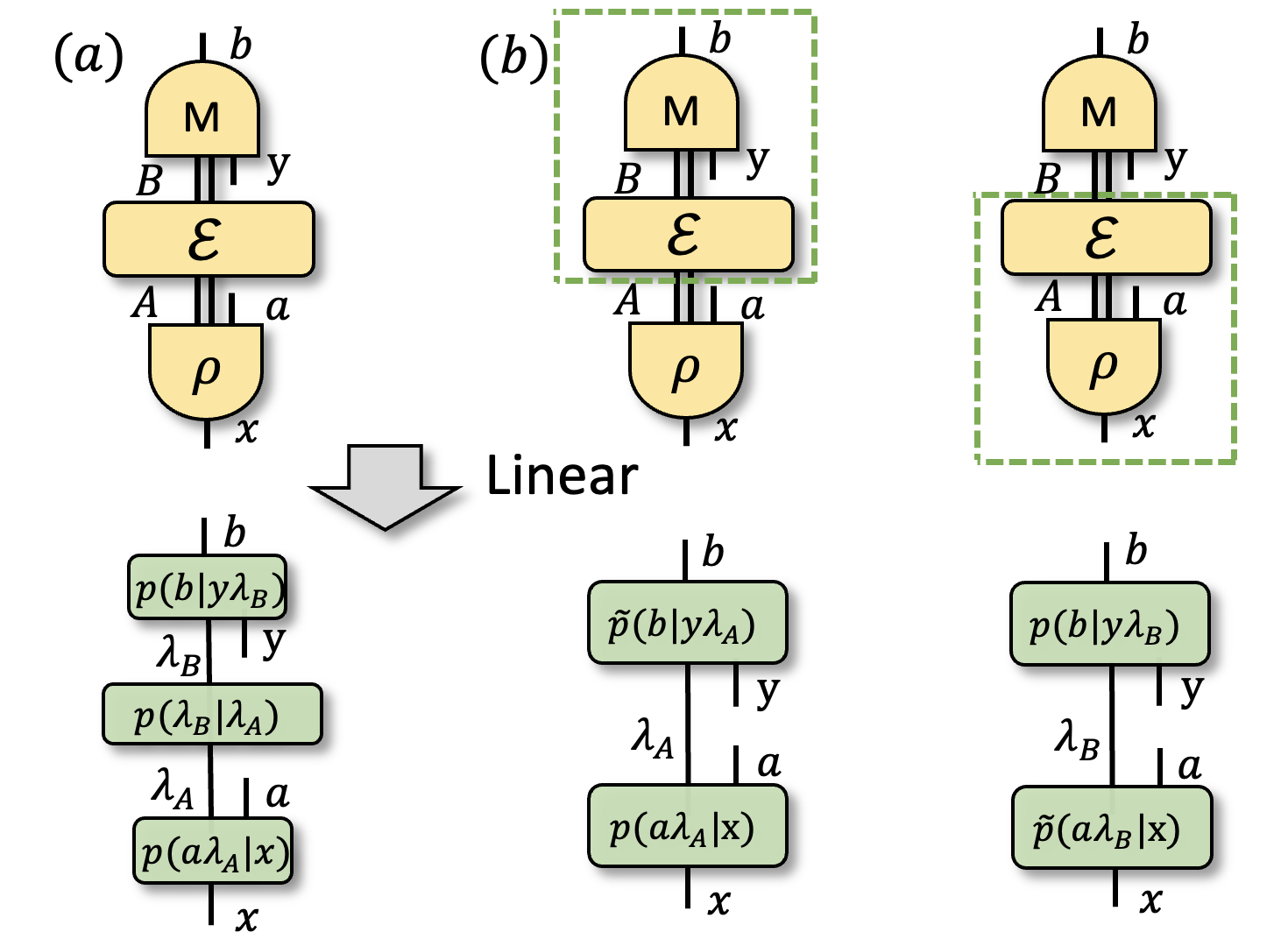}
\caption{a) A prepare-transform-measure circuit. b) By boxing the channel together with the preparation or the measurement, we can obtain a pair of prepare-measure circuits. At the top of the figure, we depict the quantum circuits; at the bottom, we depict an ontological model of each circuit.}
\label{PTMcoarse}
\end{figure}

\par 
Moreover, the sequential composition of the new multi-measurement on system $A$,  $\{\{\tilde{M}_{b|y}^{A}\}_b\}_y$, with the new multi-source on system $B$,  $\{\{\tilde{\rho}_{a|x}^B\}_a\}_x$, gives rise to a measure-and-prepare multi-instrument, namely,
\begin{equation}
\tilde{\mc{E}}^{B|A}_{ab|xy}(\cdot)\coloneqq\tilde{\rho}^B_{a|x}\tr \left[\tilde{M}_{b|y}^A(\cdot)\right].
\end{equation}
Since this composite process is of the same type as the given channel $\mc{E}^{B|A}$ making up the prepare-transform-measure scenario, there might exist a new operational identity in the form of 
\begin{equation}
\mc{E}^{B|A} =
\sum_{{abxy}}\alpha_{abxy}\tilde{\mc{E}}^{B|A}_{ab|xy}.
\label{eq:OPnewch}
\end{equation}
Such an identity exists if and only if $\mc{E}^{B|A}$ lies in the linear span of the set $\left\{\tilde{\mc{E}}^{B|A}_{ab|xy}(\cdot)=\tilde{\rho}^B_{a|x}\tr \left[\tilde{M}_{b|y}^A(\cdot)\right]\right\}_{abxy}$ with $\tilde{\rho}^B_{a|x}$ and $\tilde{M}_{b|y}^A$ given in Eq.~\eqref{eq:ch-seq-comp}. The operational identity of the form of Eq.~\eqref{eq:OPnewch} will then be denoted by
\begin{align}
\mc{O}(\widetilde{\msf{T}})=\{\{\alpha_{abxy}\}|\mc{E}^{B|A} =
\sum_{{abxy}}\alpha_{abxy}\tilde{\mc{E}}^{B|A}_{ab|xy}\}. 
\end{align}
Consequently, a prepare-transform-measure  circuit is classically explainable if and only if one can reproduce the quantum statistics via 
\begin{equation}
    p(ab|xy) =\sum_{\lambda_A\lambda_B}p(b|y\lambda_B)p(\lambda_B|\lambda_A)p(a\lambda_A |x),
    \label{eq: classical model}
\end{equation}  
where $p(b|y\lambda_B)$ respects the identities implied by $\mc{O}(\msf{M})$, and $p(a\lambda_A|x)$ respects the identities implied by $\mc{O}(\msf{P})$. Moreover, define
\begin{subequations}
\begin{align}
\tilde{p}(b|y\lambda_{A})&\coloneqq\sum_{\lambda_B}p(b|y\lambda_B)p(\lambda_B|\lambda_A), \\ 
\tilde{p}(a\lambda_{B}|x)&\coloneqq\sum_{\lambda_A}p(a\lambda_{A}|x)p(\lambda_B|\lambda_A),
\end{align}
\label{eq:channelonto}
\end{subequations}
Here, $\tilde{p}(b|y\lambda_{A})$ must satisfy the ontological identities implied by 
$\mc{O}(\msf{M}\circ \msf T)$, and $\tilde{p}(a\lambda_{B}|x)$ must satisfy the ontological identities implied by 
 $\mc{O}(\msf T \circ \msf{P})$, while $p(\lambda_B|\lambda_A)$ must satisfy the ontological identities implied by  $\mc{O}(\widetilde{\msf{T}})$.
 Explicitly, we require that for all $\lambda_A,\lambda_B$, 
\begin{subequations}
\begin{align}
 &\sum_{by} \alpha_{by} p(b|y\lambda_B) =0 \quad \forall \{\alpha_{by}\}\in \mc{O}(\msf{M}),\label{channelmeas1} \\ 
 &\sum_{ax} \alpha_{ax} p(a\lambda_A|x)=0 \quad \forall \{\alpha_{ax}\}\in \mc{O}(\msf{P}), \label{channelmeas2} \\ 
 &\sum_{by} \alpha_{by}\tilde{p}(b|y\lambda_{A}) =0 \quad \forall \{\alpha_{by}\}\in \mc{O}(\msf{M}\circ \msf{T}), \label{channelprep1}\\ 
 &\sum_{ax}{\alpha}_{ax} \tilde{p}(a\lambda_{B}|x)=0 \quad \forall \{\alpha_{ax}\}\in \mc{O}(\msf{T}\circ\msf{P}), \label{channelprep2} \\
 & {p(\lambda_B|\lambda_A)=\sum_{{abxy}}\alpha_{abxy} \tilde{p}(b|y\lambda_{A})\tilde{p}(a\lambda_{B}|x) }\notag \\
 &~~\quad\quad\quad\quad\quad\quad\quad\quad\quad \forall\{\alpha_{abxy}\}\in \mc{O}(\widetilde{\msf{T}}). \label{channel3}
\end{align}
\label{op:channel}
\end{subequations}
The first two constraints here follow from linearity, while the latter three follow from diagram-preservation together with linearity.
Finally, we present a result that will be useful for the characterization of nonclassical channels (the analog of Proposition~\ref{prop:bipartite}). 
\begin{proposition}
\label{prop:channel}
Any prepare-transform-measure circuit (as depicted in Fig.~\ref{PTMcoarse}(a)) is classically explainable only if the two prepare-measure circuits it induces by boxing the transformation with the preparation or with the measurements (the two cases depicted in Fig.~\ref{PTMcoarse}(b)) are classically explainable.
\end{proposition}

\subsection{The classical-nonclassical divide for channels }
\label{sec:channeldevide}
We turn now to the definition of classicality for a channel, again obtained by particularizing Definition~\ref{maindefn}.  Recalling that the most general dual process of a channel from $A$ to $B$ is a comb with output $A$ and input $B$ and that a factorizing such comb constitutes a preparation on $A$ and a measurement on $B$, we are led to consider a prepare-transform-measure circuit of the type depicted in Fig.~\ref{PTMcoarse}(a), but where the range of preparations includes every possible state on $A$ and the range of measurements includes every possible effect on $B$---a situation depicted in Fig.~\ref{duals}(f).

\begin{definition}
A channel is classical if and only if the statistics generated by the set of circuits where it is contracted with {\em any} state at its input and {\em any} effect at its output (i.e., where the set ranges over all states and all effects) is classically explainable (in the sense of Definition~\ref{defn:classical}).  
 \label{def:PTM}
\end{definition}

A useful result concerning the sufficient condition for nonclassical channels (necessary condition for classical channels) can be inferred from Proposition~\ref{prop:channel}.

\begin{proposition}\label{channelprop}

A channel $\mc{E}^{B|A}$ is classical only if the image of every multi-source 
on $A$ under $\mc{E}^{B|A}$ is a classical multi-source on $B$, and the image of every multi-measurement on $B$ under $[\mc{E}^{B|A}]^{\dagger}$ is a classical multi-measurement on $A$.

\end{proposition}

It is straightforward to give a proof that is analogous to that of Proposition~\ref{propremote}.  The result can also be understood as following from Proposition~\ref{propremote} by leveraging lemma~\ref{lem:Choi} 
(which will be presented in Section~\ref{general}) because the latter result shows that the verdict regarding the nonclassicality of a channel is the same as that regarding the bipartite state that is Choi-isomorphic to it.  

However the converse of proposition~\ref{channelprop} is yet unknown. Instead, we show that whether a quantum channel $\mc{E}^{B|A}$ is classical or not can be characterized by the following theorem, analogous to Theorems~\ref{theoremprep}, ~\ref{theoremmeas}, and \ref{ncbp}, which is proven in Appendix~\ref{app:proof7}.

Define the operator spaces
\begin{subequations}
\begin{align}
&V_B=\text{Span}(\{\mc{E}^{B|A}(\rho^A):\rho^A\in \mc S(\mc{H}^A)\}),  \\
&V_A=\text{Span}(\{[\mc{E}^{B|A}]^{\dagger }(M^B):M^B\in \mc M(\mc{H}^B)\}). 
\end{align}
Note that since $[\mc{E}^{B|A}]^{\dagger }$ is a unital channel, $V_A$ must also contains identity operator $\mbb{1}^A$.
\end{subequations}
\begin{theorem}
\label{ncchannelstruc}
A channel $\mc{E}^{B|A}$: $\mc{L}(\mc{H}^{A})\mapsto \mc{L}(\mc{H}^{B})$ is classical if and only if it admits a frame representation:
\begin{equation}
\mc{E}^{B|A}(\cdot)=\sum_{\lambda_A\lambda_B}\sigma^{B}_{\lambda_B}\tr[H_{\lambda_B}^{B}\mc{E}^{B|A}(F^{A}_{\lambda_A})]\tr[G_{\lambda_A}^{A}(\cdot)]
\label{eq:channelframe}
\end{equation}
where the frame on $V_A$ is a set of Hermitian operators $\{{F}^{A}_{\lambda_A}\}_{\lambda_A}$ each with unit trace, $\tr[{F}^{A}_{\lambda_A}]=1$, and the dual frame on $V_A$ is a POVM $\{G_{\lambda_A}^{A}\}_{\lambda_A}$ satisfying $\sum_{\lambda_A}G_{\lambda_A}^{A}=\mc{P}_{V_A}(\mbb{1}^A)=\mbb{1}^A$
, while the frame on $V_B$ is a set of density operators $\{\sigma^{B}_{\lambda_B}\}_{\lambda_B}$ and the dual frame on $V_B$  is a set of Hermitian operators $\{{H}_{\lambda_B}^{B}\}_{\lambda_B}$ satisfying $\sum_{\lambda_B}{H}_{\lambda_B}^{B}=\mc{P}_{V_B}(\mbb{1}^{B})$, with $\mc{P}_{V_{B(A)}}$ denoting the superoperator projecting onto $V_{B(A)}$, and such that   
\begin{align} 
\tr[H_{\lambda_B}^{B}\mc{E}^{B|A}(F^{A}_{\lambda_A})] \ge 0~~~~\forall \lambda_A, \lambda_B. \label{eq:cl-cons-all}
\end{align}
\end{theorem}

\begin{figure}[htb!]
\centering
\includegraphics[width=0.4\textwidth]{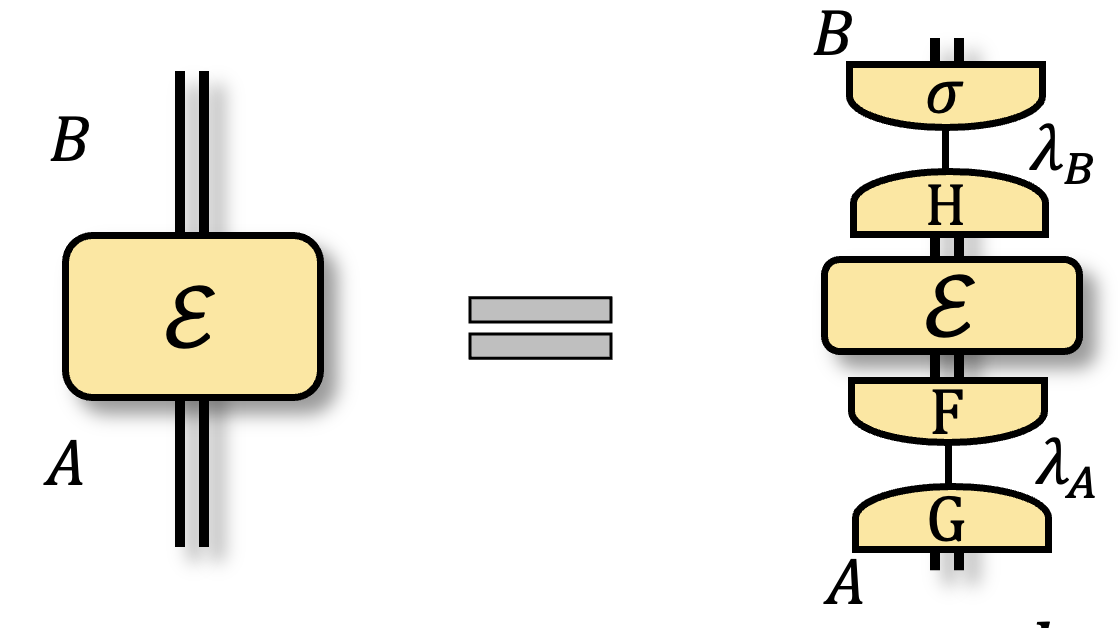}
\caption{A quantum channel is classical if and only if it admits a frame representation as shown on the right-hand side of the diagrammatic equation, where the condition on the frame and dual frame operators are given in Theorem~\ref{ncchannelstruc}.}
\label{framech}
\end{figure}
\par 
As a direct consequence, one can observe that a quantum channel being classical implies it is entanglement-breaking:
\begin{corollary}
If a channel is non-entanglement-breaking, then it is nonclassical.
\label{prop:tf_ENT}
\end{corollary}
This also follows immediately from the fact that if a state of a bipartite system is entangled, then it is nonclassical (Corollary~\ref{corr:bipartiteentangle}) together with the fact that nonclassicality of processes is preserved under the Choi isomorphism (Lemma~\ref{lem:Choi}, which will be presented further on in the article).

Although being entanglement-breaking is a necessary condition for classicality, it is not sufficient, as illustrated by the following examples.

\begin{example}
The entanglement-breaking channel
\begin{align}
\mc{E}(X)&=\tr\big[\op{0}{0}X\big]\op{0}{0}+\tr\big[\op{1}{1}X\big]\op{1}{1} \\
&+\tr\big[\op{2}{2}X\big]\op{+}{+}+\tr\big[\op{3}{3}X\big]\op{-}{-}\notag
\end{align}
is nonclassical.
This follows from Proposition~\ref{channelprop} and the fact that if the input to the channel is the source $\{\frac{1}{4}\op{i}{i}\}_{i=0}^3$, then its output is the source $\{\frac{1}{4}\op{0}{0},\frac{1}{4}\op{1}{1},\frac{1}{4}\op{+}{+},\frac{1}{4}\op{-}{-}\}$, which is the nonclassical source discussed in Example \ref{steeringexample}. (Note that there does not exist a set of effects that is mapped by the adjoint of this channel to a nonclassical set of effects, since any set of effects will get mapped to a set of effects that are all diagonal in the same eigenbasis, namely, $\{\op{i}{i}\}_{i=0}^3$.) 
\end{example}
\begin{example}

The entanglement-breaking channel 
\begin{align}
\mc{E}(X)&=\frac{1}{2}[\tr\big[\op{0}{0}X\big]\op{0}{0}+\tr\big[\op{1}{1}X\big]\op{1}{1} \\
&+\tr\big[\op{+}{+}X\big]\op{2}{2}+\tr\big[\op{-}{-}X\big]\op{3}{3}]\notag
\end{align}
is also nonclassical. This follows from Proposition~\ref{channelprop} and the fact that if the output to the channel is composed with the POVM $\{\op{i}{i}\}_{i=0}^3$, then it defines an effective POVM $\{\frac{1}{2}\op{0}{0},\frac{1}{2}\op{1}{1},\frac{1}{2}\op{+}{+},\frac{1}{2}\op{-}{-}\}$, which is nonclassical according to our Corollary~\ref{coro: measurement}. 

\label{example 6}
\end{example}
An alternative proof of the nonclassicality of this channel is to note that the channel is Choi-isomorphic to the nonclassical separable state in Example~\ref{singularexample} and to leverage the fact that nonclassicality is preserved under the Choi isomorphism (Lemma~\ref{lem:Choi}).

\section{Processes of arbitrary type }\label{general}

In this section, we derive results on the nonclassicality of a general process that extends those in the previous sections. We first present some necessary conditions for classicality for the process of arbitrary type. Next, we present a full characterization of such classicality in terms of a frame representation. Finally, we show that the nonclassicality of any process can be deduced from the nonclassicality of an associated multipartite state.

Here, the most general process we consider is a multipartite multi-instrument  $\{\{\mc{F}_{c|z}^{\vec{B}|\vec{A}}\}_c\}_z$, which has a classical setting variable $z$ and a classical outcome variable $c$, takes a set of quantum inputs labeled $\vec{A}\coloneqq(A_1,\cdots,A_n)$ and produces a set of quantum outputs labeled $\vec{B}\coloneqq(B_1,\cdots,B_m)$. This class subsumes multi-source, multi-measurement, bipartite, multi-instrument, and multipartite versions of each. We also refer the reader to Appendix~\ref{sec:appK} for the generalization to $N$-combs~\footnote{We have not developed our approach for processes with indefinite causal order, as it is not clear how to define an ontological model for these.}.

This section introduces substantial notation to treat processes of arbitrary type. Conceptually, however, the results are very similar to those in section~\ref{sec:bipartite} and~\ref{sec:channeldevide} for bipartite states and channels. Readers who wish to avoid the notational overhead may focus on the main results—Theorems~\ref{ncgen}, Prop.~\ref{prop:general}, Prop.~\ref{generalprop}, Coro.\ref{thmmulti}, and Coro.~\ref{thmmulti1} and consult the surrounding material only as needed.

\subsection{Classical explainability of a multipartite prepare-transform-measure scenario}

\begin{figure}[htb!]
\centering
\includegraphics[width=0.5\textwidth]{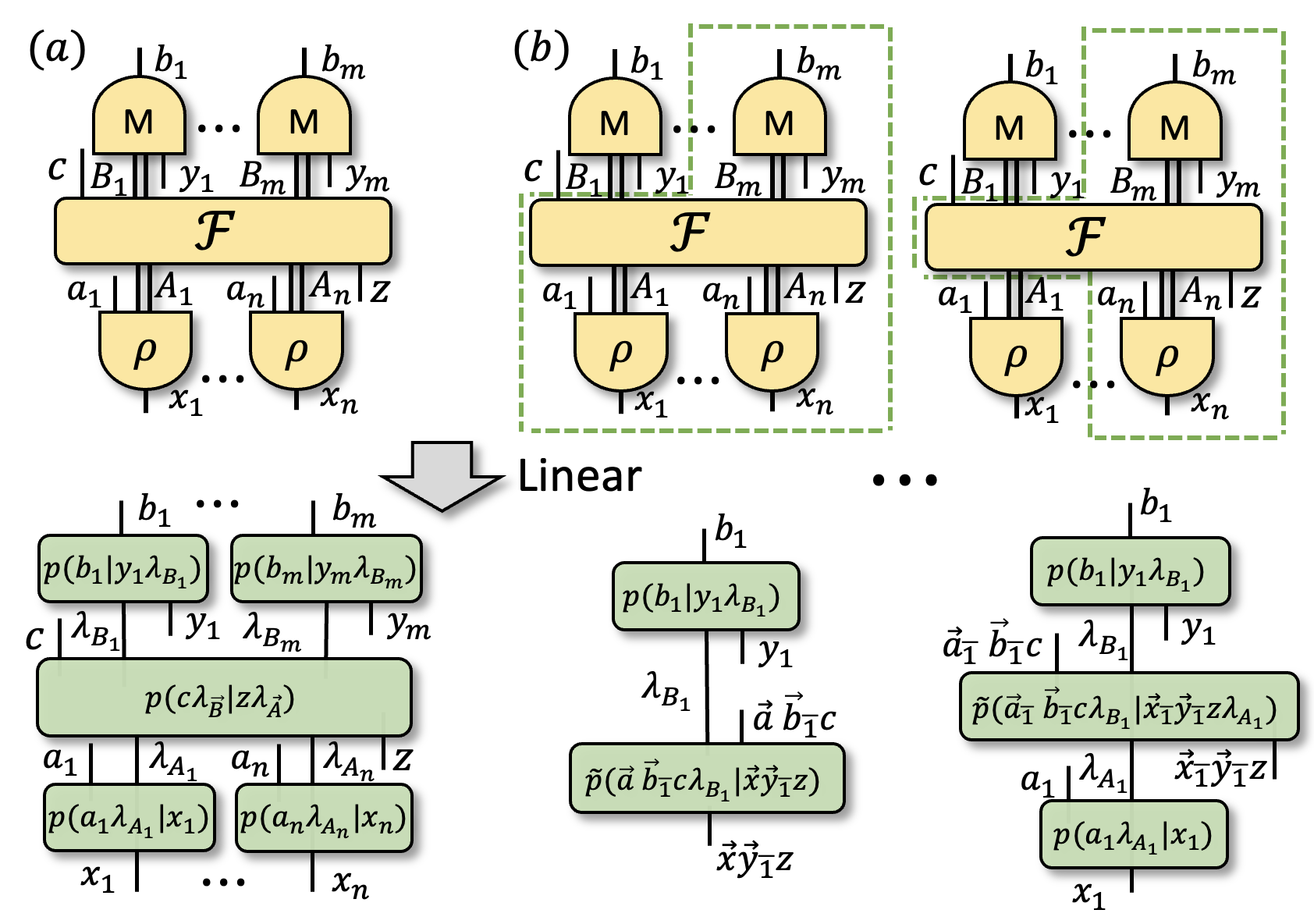}
\caption{(a)A multipartite prepare-transform-measure circuit with local multi-measurements, local multi-sources, and a multipartite multi-instrument. (b) By boxing processes together, we can, for instance, view this circuit as a unipartite prepare-measure circuit on $B_1 $, or as a unipartite prepare-transform-measure circuit from $A_1$ to $B_1$, where we denote $\vec{b}_{\bar{1}}\coloneqq\{b_2,\cdots b_m\}$ (and similarly for $\vec{a}_{\bar{1}}$, $\vec{x}_{\bar{1}}$ and $\vec{y}_{\bar{1}}$). At the top of the circuit, we depict the quantum circuit, and at the bottom, we depict their diagram-preserving ontological model. }
\label{fig:general}
\end{figure}

As preparation for our discussion of the nonclassicality of a general multipartite multi-instrument $\{\{\mc{F}_{c|z}^{\vec{B}|\vec{A}}\}_c\}_z$, we first consider the question of when a multipartite prepare-transform-measure circuit involving a multi-instrument, along with local preparations and local measurements on the multipartite system, is classically explainable. 

Specifically, we consider the circuit depicted in \cref{fig:general} where at the preparation stage, there is for each $i\in [n]\coloneqq\{1,\cdots,n\}$, a multi-source $\msf{P}_i^{A_i}\coloneqq\{\{p(a_i|x_i)\rho^{A_i}_{a_i|x_i}\}_{a_i}\}_{x_i}$ on $\mc{H}^{A_i}$,  at the transformation stage there is a multipartite multi-instrument $\msf{T}^{\vec{B}|\vec{A}}\coloneqq\{\{\mc{F}_{c|z}^{\vec{B}|\vec{A}}\}_c\}_z$ where  $\mc{F}_{c|z}^{\vec{B}|\vec{A}}$: $\mc{L}(\underset{i\in [n]}\bigotimes\mc{H}^{A_i})\mapsto \mc{L}(\underset{j\in [m]}{\bigotimes}\mc{H}^{B_{j}})$, and finally, at the measurement stage there is, for each $j\in [m]\coloneqq\{1,\cdots,m\}$, a local measurement $\msf{M}^{B_j}_j\coloneqq \{\{M_{b_j|y_j}^{B_j}\}_{b_j}\}_{y_j}$. For such a circuit, quantum theory predicts the statistics
\begin{equation}
p(\vec{a}\vec{b}c|\vec{x}\vec{y}z)=\tr[\underset{j\in [m]}{\big(\bigotimes} M_{b_j|y_j}^{B_j}\big)\mc{F}_{c|z}^{\vec{B}|\vec{A}} (\underset{i\in [n]}\bigotimes p(a_i|x_i)\rho^{A_i}_{a_i|x_i})].
\end{equation}

Similar to earlier discussions of simpler quantum circuits, one can enumerate all operational identities among the individual processes in the circuit. Specifically,
\begin{subequations}
\begin{align}
  &\mc{O}(\msf{P}_i)=\{\{\alpha_{a_ix_i}\} |\sum_{a_ix_i} \alpha_{a_ix_i} p(a_i|x_i)\rho^{A_i}_{a_i|x_i}=\mbb{0}^{A_i}\}, \\
 &\mc{O}(\msf{M}_j)=\{\{\alpha_{b_jy_j}\} |\sum_{b_jy_j} \alpha_{b_jy_j} M^{B_j}_{b_j|y_j}=\mbb{0}^{B_j}\}, \\ 
& \mc{O}(\msf{T})=\{\{\alpha_{cz}\} |\sum_{cz} \alpha_{cz} \mc{F}_{c|z}^{\vec{B}|\vec{A}}=\mbb{0}^{\vec{B}|\vec{A}}\}. \label{op-instr}
\end{align}
 \label{eq:op-gen1}
\end{subequations} 
\par There are additional operational identities in this circuit arising from sequential compositions of the multipartite multi-instrument with multi-states on different subsets of the input systems and with multi-measurements on different subsets of the output systems, respectively. 

Specifically, for any sets $\mc{I}\subseteq [n]$ and $\mc{J}\subseteq [m]$ with complementary sets $\bar{\mc I}:=[n]/\mc I$ and $\bar{\mc J}:=[m]/\mc J$, we denote $\vec{B}_{{\mc{J}}}\coloneqq(B_j)_{j\in \mc{J}}$ and $\vec{B}_{\bar{\mc{J}}}\coloneqq(B_j)_{j\in \bar{\mc{J}}}$ with analogous notation for $\vec{A}_{{\mc{I}}}$, $\vec{A}_{\bar{\mc{I}}}$, $\vec{b}_{\mc{J}}$ and $\vec{x}_{\mc{I}}$. We can define the effective process (a transformation from $\vec{A}_{\bar{\mc{I}}}$ to $\vec{B}_{\bar{\mc{J}}}$) by
\begin{align}
\msf{T}^{\vec{B}_{\bar{\mc{J}}}|\vec{A}_{\bar{\mc{I}}}}&(\cdot)=\bigotimes_{j\in \mc{J}}\msf{M}_j^{B_j}\circ\msf{T}^{\vec{B}|\vec{A}}\circ\bigotimes_{i\in\mc{I}}\msf{P}_i^{A_i} \label{eq:seq-comp}\\
\coloneqq &  \{\{\widetilde{\mc{F}}_{\vec{a}_{\mc{I}}\vec{b}_{\mc{J}}c|\vec{x}_{\mc{I}}\vec{y}_{\mc{J}}z}^{\vec{B}_{\bar{\mc{J}}}|\vec{A}_{\bar{\mc{I}}}}(\cdot)\}_{\vec{a}_{\mc{I}}\vec{b}_{\mc{J}}c}\}_{\vec{x}_{\mc{I}}\vec{y}_{\mc{J}}z}\notag \\
=&\{\{\tr_{\vec{B}_{\mc{J}}}[M_{\vec{b}_{\mc J}|\vec{y}_{\mc J}}^{\vec{B}}\mc{F}_{c|z}^{\vec{B}|\vec{A}}( \tilde{\rho}^{\vec{A}_{\mc{I}}}_{\vec{a}_{\mc I}|\vec{x}_{\mc I}}\otimes (\cdot) )]\}_{c\vec{a}_{\mc{I}}\vec{b}_{\mc{J}}}\}_{z\vec{x}_{\mc{I}}\vec{y}_{\mc{J}}},\notag
\end{align}
where we denote $M_{\vec{b}_{\mc J}|\vec{y}_{\mc J}}^{\vec{B}}\coloneqq\underset{j'\in \bar{\mc{J}}}{\bigotimes} \mbb{1}^{B_{j'}}\otimes\underset{j\in \mc{J}}{\bigotimes} M_{{b_j|y_j}}^{B_j}$ and $\tilde{\rho}^{\vec{A}_{\mc{I}}}_{\vec{a}_{\mc I}|\vec{x}_{\mc I}}\coloneqq\underset{i\in \mc{I}}{\bigotimes} p(a_i|x_i)\rho^{A_i}_{a_i|x_i}$. For these effective processes, their operational identities are of the form 
\begin{align}
&\mc{O}(\msf{T}^{\bar{\mc{J}}|\bar{\mc{I}}}) \\\label{eq:general-seq}
&=\{\{\alpha_{\vec{a}_{\mc{I}}\vec{b}_{\mc{J}}c}^{\vec{x}_{\mc{I}}\vec{y}_{\mc{J}}z}\} |\sum_{\substack{\vec{x}_{\mc{I}}\vec{y}_{\mc{J}}z\\ \vec{a}_{\mc{I}}\vec{b}_{\mc{J}}c}} \alpha_{\vec{a}_{\mc{I}}\vec{b}_{\mc{J}}c}^{\vec{x}_{\mc{I}}\vec{y}_{\mc{J}}z} \widetilde{\mc{F}}_{\vec{a}_{\mc{I}}\vec{b}_{\mc{J}}c|\vec{x}_{\mc{I}}\vec{y}_{\mc{J}}z}^{\vec{B}_{\bar{\mc{J}}}|\vec{A}_{\bar{\mc{I}}}}=\mbb{0}^{\vec{B}_{\bar{\mc{J}}}|\vec{A}_{\bar{\mc{I}}}}\}, 
\notag
\end{align}
of which the operational identities in Eq.~\eqref{op-instr} are a special case with $\mc{I}=\mc{J}=\emptyset$.

In addition to these sequential compositions, parallel compositions of the effective processes in \cref{eq:seq-comp} can also lead to nontrivial operational identities. For instance, given any disjoint collection of subsets of $[m]$, $\{\bar{\mc{J}}_{\nu}\}_{\nu}$, and any disjoint collection of subsets of $[n]$, $\{\bar{\mc{I}}_{\nu}\}_{\nu}$, such that $ \bigcup_{\nu}\bar{\mc{J}}_{\nu}=[m]$ and $\bigcup_{\nu}\bar{\mc{I}}_{\nu}= [n]$. Denote $\mc{I}_{\nu}:=[n]/\bar{\mc{I}}_{\nu}$ and $\mc{J}_{\nu}:=[m]/\bar{\mc{J}}_{\mu}$ as the corresponding set-theoretic complements, there might exist nontrivial operational identities of the form:
\begin{align}
&\sum_{\substack{cz}} \alpha_{cz} {\mc{F}}_{c|z}^{\vec{B}|\vec{A}}+  \\ \label{op:gencomp}
&\sum_{\substack{ \{\vec{x}^{\nu}_{\mc{I}_{\nu}}\vec{y}^{\nu}_{\mc{J}_{\nu}}z^{\nu}\}_{\nu} \\ \{\vec{a}^{\nu}_{\mc{I}_{\nu}}\vec{b}^{\nu}_{\mc{J}_{\nu}}c^{\nu} \}_{\nu}}}\alpha^{\{\vec{x}^{\nu}_{\mc{I}_{\nu}}\vec{y}^{\nu}_{\mc{J}_{\nu}}z^{\nu}\}_{\nu}}_{\{\vec{a}^{\nu}_{\mc{I}_{\nu}}\vec{b}^{\nu}_{\mc{J}_{\nu}}c^{\nu} \}_{\nu}}\underset{\nu}{\bigotimes}\widetilde{\mc{F}}_{\vec{a}^{\nu}_{\mc{I}_{\nu}}\vec{b}^{\nu}_{\mc{J}_{\nu}}c^{\nu}|\vec{x}^{\nu}_{\mc{I}_{\nu}}\vec{y}^{\nu}_{\mc{J}_{\nu}}z^{\nu}}^{\vec{B}_{\bar{\mc{J}_{\nu}}}|\vec{A}_{\bar{\mc{I}}_{\nu}}}=\mbb{0}^{\vec{B}|\vec{A}}, \notag
\end{align}
where the superscript $\nu$ in $c^{\nu}$, $\vec{a}^{\nu}$ is an index to distinguish different settings and outcomes in different parts of the composition. One can think of this as a generalization of \cref{eq:OPnewbi,eq:OPnewch}.

Since the new identity obviously includes the one in \cref{op-instr} as a special case, the sets of new operational identities are
\begin{align}
&\mc{O}(\widetilde{\msf{T}}):=\{\{\alpha_{cz}\}\cup\{\alpha^{\{\vec{x}^{\nu}_{\mc{I}_{\nu}}\vec{y}^{\nu}_{\mc{J}_{\nu}}z^{\nu}\}_{\nu}}_{\{\vec{a}^{\nu}_{\mc{I}_{\nu}}\vec{b}^{\nu}_{\mc{J}_{\nu}}c^{\nu} \}_{\nu}}\}|  \sum_{\substack{cz}} \alpha_{cz} {\mc{F}}_{c|z}^{\vec{B}|\vec{A}}+\\
&\sum_{\substack{ \{\vec{x}^{\nu}_{\mc{I}_{\nu}}\vec{y}^{\nu}_{\mc{J}_{\nu}}z^{\nu}\}_{\nu} \\ \{\vec{a}^{\nu}_{\mc{I}_{\nu}}\vec{b}^{\nu}_{\mc{J}_{\nu}}c^{\nu} \}_{\nu}}}\alpha^{\{\vec{x}^{\nu}_{\mc{I}_{\nu}}\vec{y}^{\nu}_{\mc{J}_{\nu}}z^{\nu}\}_{\nu}}_{\{\vec{a}^{\nu}_{\mc{I}_{\nu}}\vec{b}^{\nu}_{\mc{J}_{\nu}}c^{\nu} \}_{\nu}}\underset{\nu}{\bigotimes}\widetilde{\mc{F}}_{\vec{a}^{\nu}_{\mc{I}_{\nu}}\vec{b}^{\nu}_{\mc{J}_{\nu}}c^{\nu}|\vec{x}^{\nu}_{\mc{I}_{\nu}}\vec{y}^{\nu}_{\mc{J}_{\nu}}z^{\nu}}^{\vec{B}_{\bar{\mc{J}_{\nu}}}|\vec{A}_{\bar{\mc{I}}_{\nu}}}=\mbb{0}^{\vec{B}|\vec{A}}\}.  \notag
\end{align}
where the coefficients $\{\alpha_{cz}\}$ and $\{\alpha^{\{\vec{x}^{\nu}_{\mc{I}_{\nu}}\vec{y}^{\nu}_{\mc{J}_{\nu}}z^{\nu}\}_{\nu}}_{\{\vec{a}^{\nu}_{\mc{I}_{\nu}}\vec{b}^{\nu}_{\mc{J}_{\nu}}c^{\nu} \}_{\nu}}\}$ are not both uniformly zero. 

\par 
Consequently, the circuit in \cref{fig:general} is classically explainable if and only if one can reproduce the quantum statistics via
\begin{align}
&p(\vec{a}\vec{b}c|\vec{x}\vec{y}z) \label{eq:cl-general}\\
&=\sum_{\lambda_{\vec{B}}\lambda_{\vec{A}}} p(c\lambda_{\vec{B}}|z\lambda_{\vec{A}})\prod_{j}p(b_j|y_j\lambda_{B_j})\prod_{i}p(a_i\lambda_{A_i}|x_i), \notag
\end{align} 
where $p(c\lambda_{\vec{B}}|z\lambda_{\vec{A}})$, $\{ p(b_j|y_j\lambda_{B_j})\}_j$, and $\{p(a_i\lambda_{A_i}|x_i)\}_i$ 
satisfy all the relevant ontological identities. Moreover, for any choice of $\mc{I}\subseteq [n]$ and $\mc{J}\subseteq [m]$ if we can define the effective ontological processes associated to $\msf{T}^{\vec{B}_{\bar{\mc{J}}}|\vec{A}_{\bar{\mc{I}}}}$ as 
\begin{align}
&p(\vec{a}_{{\mc{I}}}\vec{b}_{{\mc{J}}}c{\lambda
}_{\vec{B}_{\bar{\mc{J}}}}|\vec{x}_{{\mc{I}}}\vec{y}_{{\mc{J}}}z\lambda_{\vec{A}_{\bar{\mc{I}}}})\\
&=\sum_{\lambda_{\vec{B}_{{\mc{J}}}}\lambda_{\vec{A}_{{\mc{I}}}}}p(c\lambda_{\vec{B}}|z\lambda_{\vec{A}})\prod_{j\in \mc{J}}p(b_j|y_j\lambda_{B_j})\prod_{i\in\mc{I}}p(a_i\lambda_{A_i}|x_i), \notag
\end{align}
then they must satisfy the following ontological identities for all $\lambda_{\vec{B}}, \lambda_{\vec{A}}$,
\begin{subequations}
\begin{align}
&\sum_{b_jy_j} \alpha_{b_jy_j} p(b_j|y_j\lambda_{B_j}) =0,\label{instmeas1} \\ 
&\sum_{a_ix_i}  \alpha_{a_ix_i} p(a_i\lambda_{A_i}|x_i)=0, \label{instprep1} \\ 
&\sum_{\substack{\vec{x}_{\mc{I}}\vec{y}_{\mc{J}}z\\ \vec{a}_{\mc{I}}\vec{b}_{\mc{J}}c}} \alpha_{\vec{a}_{\mc{I}}\vec{b}_{\mc{J}}c}^{\vec{x}_{\mc{I}}\vec{y}_{\mc{J}}z}p(\vec{a}_{{\mc{I}}}\vec{b}_{{\mc{J}}}c{\lambda
}_{\vec{B}_{\bar{\mc{J}}}}|\vec{x}_{{\mc{I}}}\vec{y}_{{\mc{J}}}z\lambda_{\vec{A}_{\bar{\mc{I}}}})=0, \label{instrcomp} \\
&\sum_{\substack{ \{\vec{x}^{\nu}_{\mc{I}_{\nu}}\vec{y}^{\nu}_{\mc{J}_{\nu}}z^{\nu}\}_{\nu} \\ \{\vec{a}^{\nu}_{\mc{I}_{\nu}}\vec{b}^{\nu}_{\mc{J}_{\nu}}c^{\nu} \}_{\nu}}}\alpha^{\{\vec{x}^{\nu}_{\mc{I}_{\nu}}\vec{y}^{\nu}_{\mc{J}_{\nu}}z^{\nu}\}_{\nu}}_{\{\vec{a}^{\nu}_{\mc{I}_{\nu}}\vec{b}^{\nu}_{\mc{J}_{\nu}}c^{\nu} \}_{\nu}}\prod_{\nu}p(\vec{a}^{\nu}_{{\mc{I}_{\nu}}}\vec{b}^{\nu}_{{\mc{J}_{\nu}}}c^{\nu}{\lambda
}_{\vec{B}_{\bar{\mc{J}_{\nu}}}}|\vec{x}^{\nu}_{{\mc{I}_{\nu}}}\vec{y}^{\nu}_{{\mc{J}_{\nu}}}z^{\nu}\lambda_{\vec{A}_{\bar{\mc{I}}_{\nu}}})\notag \\
&\quad\quad\quad\quad+\sum_{cz} \alpha_{cz}p(c\lambda_{\vec{B}}|z\lambda_{\vec{A}})=0, \label{instrgen}
\end{align}
\label{op:general}
\end{subequations}
for all $\{\alpha_{b_jy_j}\}\in \mc{O}(\msf{M}_j)$, $\{\alpha_{a_ix_i}\}\in \mc{O}(\msf{P}_i)$, $ \{\alpha_{\vec{a}_{\mc{I}}\vec{b}_{\mc{J}}c}^{\vec{x}_{\mc{I}}\vec{y}_{\mc{J}}z}\}\in \mc{O}(\msf{T}^{\bar{\mc{J}}|\bar{\mc{I}}})$, and $\{\alpha_{cz}\}\cup\{\alpha^{\{\vec{x}^{\nu}_{\mc{I}_{\nu}}\vec{y}^{\nu}_{\mc{J}_{\nu}}z^{\nu}\}_{\nu}}_{\{\vec{a}^{\nu}_{\mc{I}_{\nu}}\vec{b}^{\nu}_{\mc{J}_{\nu}}c^{\nu} \}_{\nu}}\} \in \mc{O}(\widetilde{\msf{T}})$.

It is, in general, very difficult to check all of these identities and decide if the circuit is classically explainable; however, the following proposition gives a necessary condition for it to be classically explainable, which is similar to the ones in Prop.~\ref{prop:bipartite} and Prop.~\ref{prop:channel}.
\begin{proposition}
\label{prop:general}
Any multipartite prepare-transform-measure circuit (of the type depicted in Fig.~\ref{fig:general}(a)) is classically explainable only if any circuit it induces by boxing the transformation with any subset of local preparations and/or local measurements (as depicted in Fig.~\ref{fig:general}(b)) is classically explainable.
\end{proposition}

\subsection{The classical-nonclassical divide for a general process}

One can always box the multipartite multi-instrument together with local preparations on all input subsystems and local measurements on all but one of the output subsystems, as depicted in   Fig.~\ref{fig:general}(b), to obtain an effective preparation on the remaining output subsystem.
Similarly, one can always box the multipartite multi-instrument with local measurements on all output subsystems and local preparations on all but one of the input subsystems to obtain an effective measurement on the remaining input subsystem. In this way, a multipartite prepare–transform–measure scenario defines a number of effective prepare–measure scenarios. Applying Prop.~\ref{prop:general}, we obtain the following necessary condition for nonclassicality of a multipartite multi-instrument.

\begin{proposition}\label{generalprop}
A multipartite multi-instrument $\{\{\mc{F}_{c|z}^{\vec{B}|\vec{A}}\}_c\}_z$ is classical only if every multi-source on quantum output $B_j$ that is obtained by contracting it with a multi-measurement on each of the remaining quantum outputs and a multi-source on each of the quantum inputs is classical, and every multi-measurement on quantum input $A_i$ that is obtained by contracting it with a multi-source on each of the remaining quantum inputs and a multi-measurement on each of the quantum outputs is classical.
\end{proposition}

The proof is similar to that of Prop.~\ref{propremote}. If the multipartite multi-instrument is classical, then the circuit where it is contracted with the full set of product effects and the full set of product states  (depicted in Fig.~\ref{fig:general}(a)) is classically explainable.  By Proposition~\ref{prop:general}, this implies that all prepare-measure circuits induced by boxing the prepare-transform-measure circuit (e.g., Fig.~\ref{fig:general}(b)) are classically explainable, which indicates that the multi-states and multi-measurements that appear in the prepare-measure circuit are classical in the sense of Definition~\ref{def_multi-sourceNC} and Definition~\ref{def_mtmtNC}. However, whether the converse is true is yet unknown. 

Similar to Theorems~\ref{theoremprep}, \ref{theoremmeas}, \ref{ncbp}, and  \ref{ncchannelstruc}, we here provide the necessary and sufficient characterization of classicality for multipartite multi-instruments, which is given by the following theorem. 

Define $V_{B_j}$ as the operator space spanned by all possible states on $\mc{H}^{B_j}$ that can be generated by $\{\{\mc{F}_{c|z}^{\vec{B}|\vec{A}}\}_c\}_z$ when it is boxed with local preparations on all subsystems $\{A_i\}_{i=1}^n$ and local measurements on all subsystems but subsystem $B_j$.  Specifically, 
\begin{align}
&V_{B_j}:=\text{Span}(\{\tr_{\vec{B}_{\bar{j}}}[(\mbb{1}^{B_j}\otimes M^{\vec{B}_{\bar{j}}})\mc{F}_{c|z}^{\vec{B}|\vec{A}}( \rho^{\vec{A}})]\}_{c,z}),  
\end{align}
with $\rho^{\vec{A}}=  \underset{i\in [n]}\bigotimes \rho^{A_i} $ ranging over all
$ \rho^{A_i} \in
\mc S(\mc H^{A_i})$, and  $M^{\vec{B}_{\bar{j}}}=\underset{j'\ne j}\bigotimes M^{B_{j'}}$ ranging over all $M^{B_{j'}}\in \mc M(\mc H^{B_{j'}})$. 

Similarly, define $V_{A_i}$ as the operator space spanned by all possible effects on $\mc{H}^{A_i}$ that can be generated by $\{\{\mc{F}_{c|z}^{\vec{B}|\vec{A}}\}_c\}_z$ when it is boxed with local measurements on all subsystems $\{B_j\}_{j=1}^m$ and local preparations on all subsystems but  $A_i$. Specifically, 
\begin{align}
&V_{A_i}:=\text{Span}(\{\tr_{\vec{A}_{\bar{i}}}[(\mbb{1}^{A_i}\otimes \rho^{\vec{A}_{\bar{i}}})[\mc{F}_{c|z}^{\vec{B}|\vec{A}}]^{\dagger}( M^{\vec{B}})]\}_{c,z}), 
\end{align}
with $\rho^{\vec{A}_{\bar{i}}}=  \underset{i'\ne i}\bigotimes \rho^{A_{i'}} $ ranging over all
$ \rho^{A_{i'}} \in
\mc S(\mc H^{A_{i'}})$, and  $M^{\vec{B}}=\underset{j\in [m]}\bigotimes M^{B_{j}}$ ranging over all $M^{B_j}\in \mc M(\mc H^{B_{j}})$. 

\begin{theorem}\label{ncgen}
A multipartite multi-instrument $\{\{\mc{F}_{c|z}^{\vec{B}|\vec{A}}\}_c\}_z$,
where $\mc{F}_{c|z}^{\vec{B}|\vec{A}}$
: $\mc{L}(\bigotimes_{i=1}^n\mc{H}^{A_i})\mapsto \mc{L}(\bigotimes_{j=1}^m\mc{H}^{B_{j}})$, is classical if and only if it admits a frame representation: 
\begin{align}
\label{eq:generalframe}
&\mc{F}_{c|z}^{\vec{B}|\vec{A}}(\cdot)= \\
&\sum_{\lambda_{\vec{A}}\lambda_{\vec{B}}}\underset{j}{\bigotimes} \sigma_{\lambda_{B_j}}^{B_j}\tr[(\underset{j}{\bigotimes} H_{\lambda_{B_j}}^{B_j})\mc{F}_{c|z}^{\vec{B}|\vec{A}}(\underset{i}{\bigotimes} F^{A_i}_{\lambda_{A_i}})]\tr[(\underset{i}{\bigotimes} G_{\lambda_{A_i}}^{A_i})(\cdot)], \nonumber
\end{align}
where the frame on the operator space $V_{B_j}$ is a fixed set of density operators $\{\sigma^{B_j}_{\lambda_{B_j}}\}_{\lambda_{B_j}}$, and the dual frame is a set of Hermitian operators $\{{H}_{\lambda_{B_j}}^{B_j}\}_{\lambda_{B_j}}$ satisfying $\sum_{\lambda_{B_j}}H^{B_j}_{\lambda_{B_j}}=\mc{P}_{V_{B_j}}(\mbb 1^{{B_j}})$; the frame on the operator space $V_{A_i}$ is a set of Hermitian operators $\{{F}^{A_i}_{\lambda_{A_i}}\}_{\lambda_{A_i}}$ satisfying $\tr[{F}^{A_i}_{\lambda_{A_i}}]=1$, and the dual frame is a POVM $\{G_{\lambda_{A_i}}^{A_i}\}_{\lambda_{A_i}}$ satisfying $\sum_{\lambda_{A_i}}G_{\lambda_{A_i}}=\mc{P}_{V_{A_i}}(\mbb{1}^{A_i})=\mbb{1}^{A_i}$, 
with $\mc{P}_{V_{B_j(A_i)}}$ denoting the superoperator projecting onto $V_{B_j(A_i)}$, and such that
\begin{align}  
\tr[(\underset{j}{\bigotimes} H_{\lambda_{B_j}}^{B_j})\mc{F}_{c|z}^{\vec{B}|\vec{A}}(\underset{i}{\bigotimes} F^{A_i}_{\lambda_{A_i}})] \ge 0,~~\forall \lambda_{\vec{A}}, \lambda_{\vec{B}},c,z. \label{eq:cl-gen-cons}
\end{align}
\end{theorem}
\begin{figure}[htb!]
\centering
\includegraphics[width=0.5\textwidth]{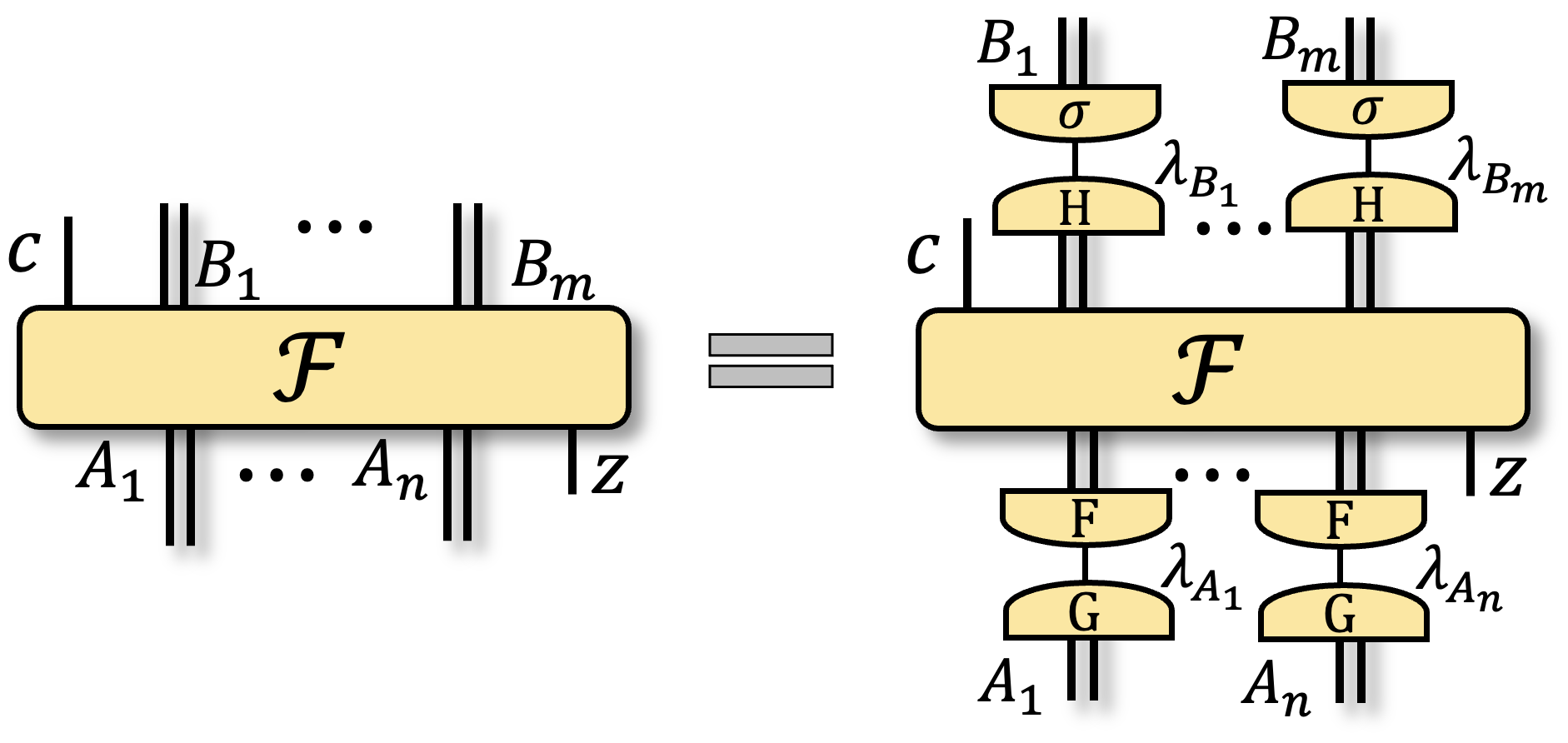}
\caption{A multipartite multi-instrument is classical if and only if it admits a frame representation as shown on the right-hand side of the equation, where the condition on the frames and dual frames are given in Theorem~\ref{ncgen}.}
\label{framegen}
\end{figure}

Such a frame decomposition is depicted diagrammatically in Fig.~\ref{framegen}.
The proof of this theorem is similar to that of Theorems~\ref{ncbp} and~\ref{ncchannelstruc}, and is outlined in detail in Appendix~\ref{app:proofmpmi} for completeness.

As a direct corollary of Theorem~\ref{ncgen}, we obtain the following necessary condition for the classicality (sufficient condition for the nonclassicality) of a multipartite state.
\begin{corollary}\label{thmmulti}
    A multipartite state is classical only if it is separable.
\end{corollary}

Moreover, as another direct corollary of Theorem~\ref{ncgen}, another necessary condition for the classicality of a multipartite state can be found in a manner analogous to Proposition~\ref {generalprop}:

\begin{corollary}\label{thmmulti1}
A multipartite state $\rho^{A_1\cdots A_N}$ is classical only if, for every party $A_i$, any multi-sources on $A_i$ obtained by contracting the multipartite state $\rho^{A_1\cdots A_N}$ with any multi-measurements on the other parties $A_{i'}\ne A_i$ is classical.  
\end{corollary}

In the multipartite case (just as in the bipartite case), separability is necessary but not sufficient for classicality. That is, there exist multipartite states that are not entangled across any partition but that are nonetheless nonclassical according to our proposal. It suffices to consider multipartite generalizations of Example~\ref{singularexample}. 

One can also find tripartite qubit separable states whose steering assemblages are nonclassical on every party (recall that by contrast, as shown in Proposition~\ref{prop: biqubit}, any bipartite qubit separable state is classical). 

\begin{example}
Given a tripartite state defined as
\begin{align}
\rho^{A_1A_2A_3}&=\frac{1}{4}(\op{000}{000}+\op{-+1}{-+1}  \\
&+\op{1-+}{1-+}+\op{+1-}{+1-}),
\end{align}
there exist measurements on any two subsystems that steer the remaining system to a set of states that is nonclassical.
In particular, if any two parties each perform the BB84 measurement, which is associated to the POVM  $\{\frac{1}{2}\op{0}{0},\frac{1}{2}\op{1}{1},\frac{1}{2}\op{+}{+},\frac{1}{2}\op{-}{-}\}$,  then they steer the third party's system to the ensemble containing the unnormalized states 
\begin{align}
&\left\{\frac{3}{64}(\frac{2}{3}\rho_x+\frac{1}{3}\tau_y),\frac{3}{64}(\frac{1}{3}\rho_x+\frac{2}{3}\tau_y)\right\}_{x,y}\notag \\
\cup&\left\{\frac{5}{64}(\frac{3}{5}\rho_x+\frac{1}{5}\mbb{1}),\frac{5}{64}(\frac{1}{5}\rho_x+\frac{2}{5}\mbb{1})\right\}_x\notag \\
\cup&\left\{\frac{5}{64}(\frac{3}{5}\tau_y+\frac{1}{5}\mbb{1}),\frac{5}{64}(\frac{1}{5}\tau_y+\frac{2}{5}\mbb{1})\right\}_y, 
\end{align}
where $x$ and $y$ are binary variables with $\{\rho_x\}_x=\{\op{0}{0},\op{1}{1}\}$ and $\{\tau_y\}_y=\{\op{+}{+},\op{-}{-}\}$. This ensemble can be shown to be nonclassical using the semidefinite program given in Ref.~\cite{zhang2024parellel}.   \end{example}

The tripartite state above is the uniform mixture of the four product states that have been shown to exhibit ``nonlocality without entanglement''~\cite{Bennett1999}. These product states form an unextendible product basis (UPB) and cannot be perfectly distinguished by LOCC measurements.

An interesting open question is whether, when assessing the nonclassicality of a multipartite state, it suffices to consider the nonclassicality of the multi-sources to which one can steer on each subsystem of the full partition. We have not managed to settle this question one way or the other. 
If considering single-party steering possibilities {\em does} suffice, this will greatly simplify the task of characterizing nonclassicality of multipartite states (and hence general processes). If it does not, then this will constitute a new way in which the nonclassicality of a state can be judged to be {\em genuinely multipartite}~\cite{PhysRevLett.125.240505,Schmid2023understanding,Seevinck2002,Toth2005}.

\subsection{Nonclassicality of a process via the nonclassicality of an associated multipartite state}
Now, we will demonstrate how the nonclassicality of a multipartite multi-instrument can be deduced by analyzing the nonclassicality of an associated multipartite state.

\begin{theorem}
\label{thm:associated state}
For any quantum process, define its unique associated multipartite state by:
\begin{itemize}
    \item applying the Choi isomorphism to all of its quantum inputs,
    \item applying flag-convexification to all of its classical inputs (setting variables),
    \item reconceptualizing all of its classical outputs (outcomes) as dephased quantum outputs.
\end{itemize}
The process is nonclassical if and only if its associated multipartite state is nonclassical.
\end{theorem}

This theorem follows from three lemmas.

The first lemma states that verdicts concerning the
nonclassicality of a given process is  preserved under the Choi isomorphism. ({This is analogous to how any resource in the resource theory of local operations and shared randomness (LOSR) can always be transformed via LOSR operations into a process with only classical outputs and quantum inputs in such a way that the property of being nonfree is preserved~\cite{Schmid2020typeindependent}.})

\begin{lemma}
\label{lem:Choi}
The state that is Choi-isomorphic to a given process is nonclassical if and only if the process itself is nonclassical.
\end{lemma}

This is proven in Appendix~\ref{appendixL}. 

\begin{figure}[htb!]
\includegraphics[width=0.75\linewidth]{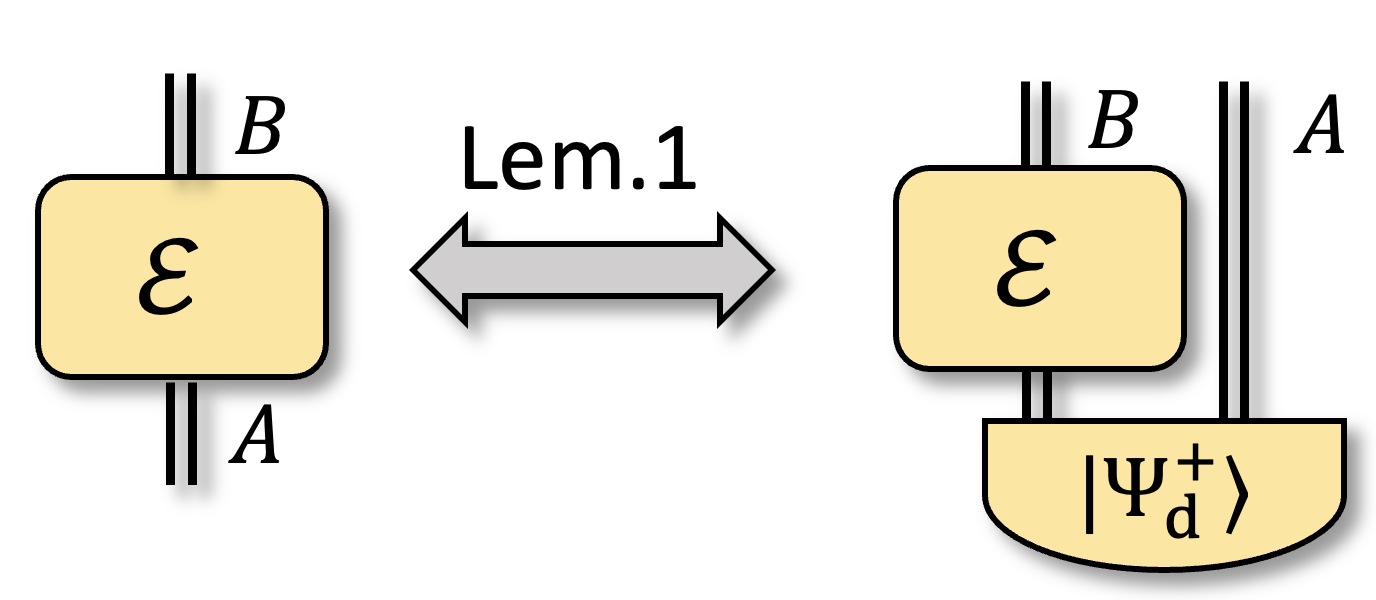}
    \caption{As an example of Lemma~\ref{lem:Choi}, a channel $\mc{E}^{B|A}$ is nonclassical if and only if the Choi-isomorphic bipartite state $\rho^{BA}_{\mc{E}} = (\mc{E}^{B|A}\otimes \mbb{1}^{A})\op{\Psi^+}{\Psi^+}^{AA}$ 
    is nonclassical. (Here, the second $A$ denotes the isomorphic copy used in the Choi construction)}
    \label{fig:Choi}
\end{figure}

Hence, to decide nonclassicality, one can first map any process to one with no quantum inputs using the Choi-isomorphism. This is depicted in Fig.~\ref{fig:Choi}. This enables, e.g., characterizations of nonclassicality for bipartite states to be transferred to channels, and those for unipartite sources to be transferred to unipartite measurements. 

The second lemma states that the verdict about the nonclassicality of a given process is preserved under mapping a classical input (a setting variable) into a classical output (an outcome variable) via {\em flag-convexification}, a notion introduced in Ref.~\cite{selby2023incompatibility}. In flag-convexification, the classical input variable is sampled from a full-support probability distribution, then copied, and the copy becomes a classical output variable. We depict an example of this on the right-hand side of Fig.~\ref{fig:lem5}. (In this example and henceforth, we take the probability distribution in question to be the uniform distribution, although our results hold just as well for any other full-support distribution.)

\begin{lemma}
\label{lem:flag-convexified}
The image of a given process under flag-convexification of its classical input is nonclassical if and only if the process itself is nonclassical.
\end{lemma}
\begin{figure}[htb!]
\includegraphics[width=0.75\linewidth]{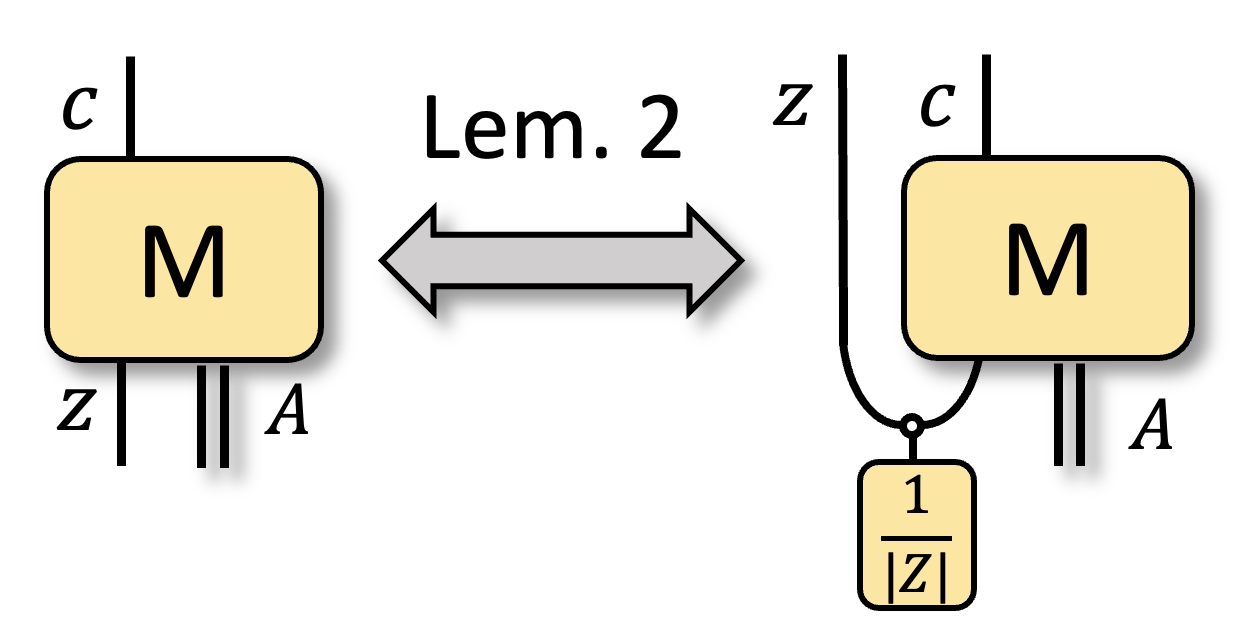}
\caption{As an example of Lemma~\ref{lem:flag-convexified}, a multi-measurement $\{\{M_{c|z}\}_c\}_z$ is nonclassical if and only if the single measurement $\{\tilde{M}_{c,z}\}_{c,z}$
obtained from flag-convexification of the setting variable $z$ is nonclassical. Assuming the uniform distribution for sampling, we have $\tilde{M}_{c,z}=\frac{1}{|Z|}M_{c|z}$.}\label{fig:lem5}
\end{figure}

The proof follows from facts about the ontological representation of flag-convexification described in
Refs.~\cite{selby2023incompatibility, schmid2024PTM}. Specifically, it follows from the fact that all that changes when one moves from a given process to its flag-convexified counterpart is that it is rescaled by a constant factor, and so an ontological representation of the given process can simply be rescaled by the same constant factor to give an ontological representation of the flag-convexified process. This lemma also follows from the fact that the frame-representation characterization of the classicality of a process (as in Theorem~\ref{ncgen}) is preserved under flag-convexification of the setting variable.

The third lemma states that the nonclassicality of a process does not depend on whether or not one reconceptualizes its classical output (its outcome) as a quantum output that is dephased in a fixed basis.

\begin{figure}[htb!]
\includegraphics[width=0.9\linewidth]{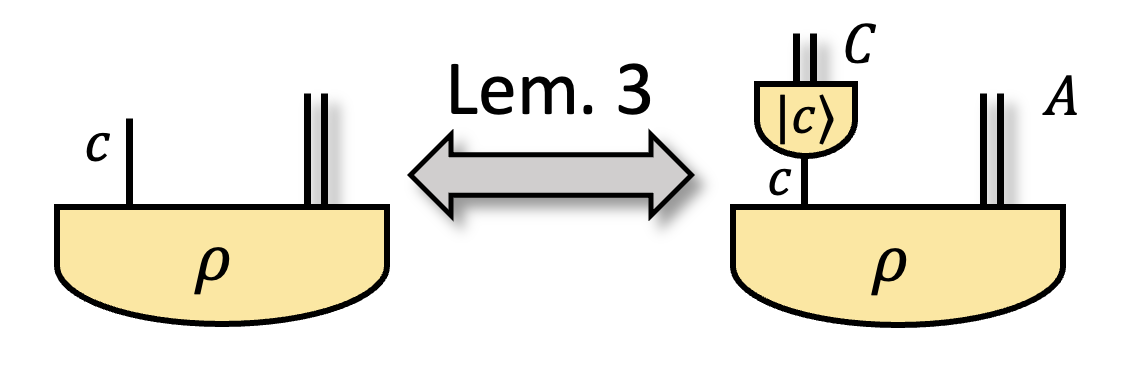}
\caption{As an example of Lemma~\ref{lem:absorb_classical_handle}, a source $\{p(c)\rho_{c}^A\}_c$ is nonclassical if and only if the bipartite state 
$\tilde{\rho}^{CA}=\sum_c p(c)\op{c}{c}^C\otimes \rho_{c}^{A}$, obtained by treating its classical outcome as a dephased quantum output, is nonclassical.
Here, the dephased quantum output is represented by preparing one of the orthogonal
set of states $\{\op{c}{c}\}_c$ on system $C$ conditional on the classical outcome $c$. }\label{fig:view0}
\end{figure}

\begin{lemma}
\label{lem:absorb_classical_handle}
Whether or not a process is nonclassical  (in the sense of Definition~\ref{maindefn})  is independent of whether one reconceptualizes its classical outputs as dephased quantum outputs. 
\end{lemma}
This is proven in Appendix~\ref{appendixL} and an example is depicted in Fig.~\ref{fig:view0}.

(In the next section, we show that the analogous result is also true for classical inputs.)

Theorem~\ref{thm:associated state} follows immediately by combining these three lemmas as depicted in Fig.~\ref{fig:Choi-all}. This implies that characterizations of nonclassicality of arbitrary multipartite states can be transferred to the characterization of the nonclassicality of \emph{arbitrary} processes in quantum theory.

\begin{figure}[htb!]
\includegraphics[width=0.9\linewidth]{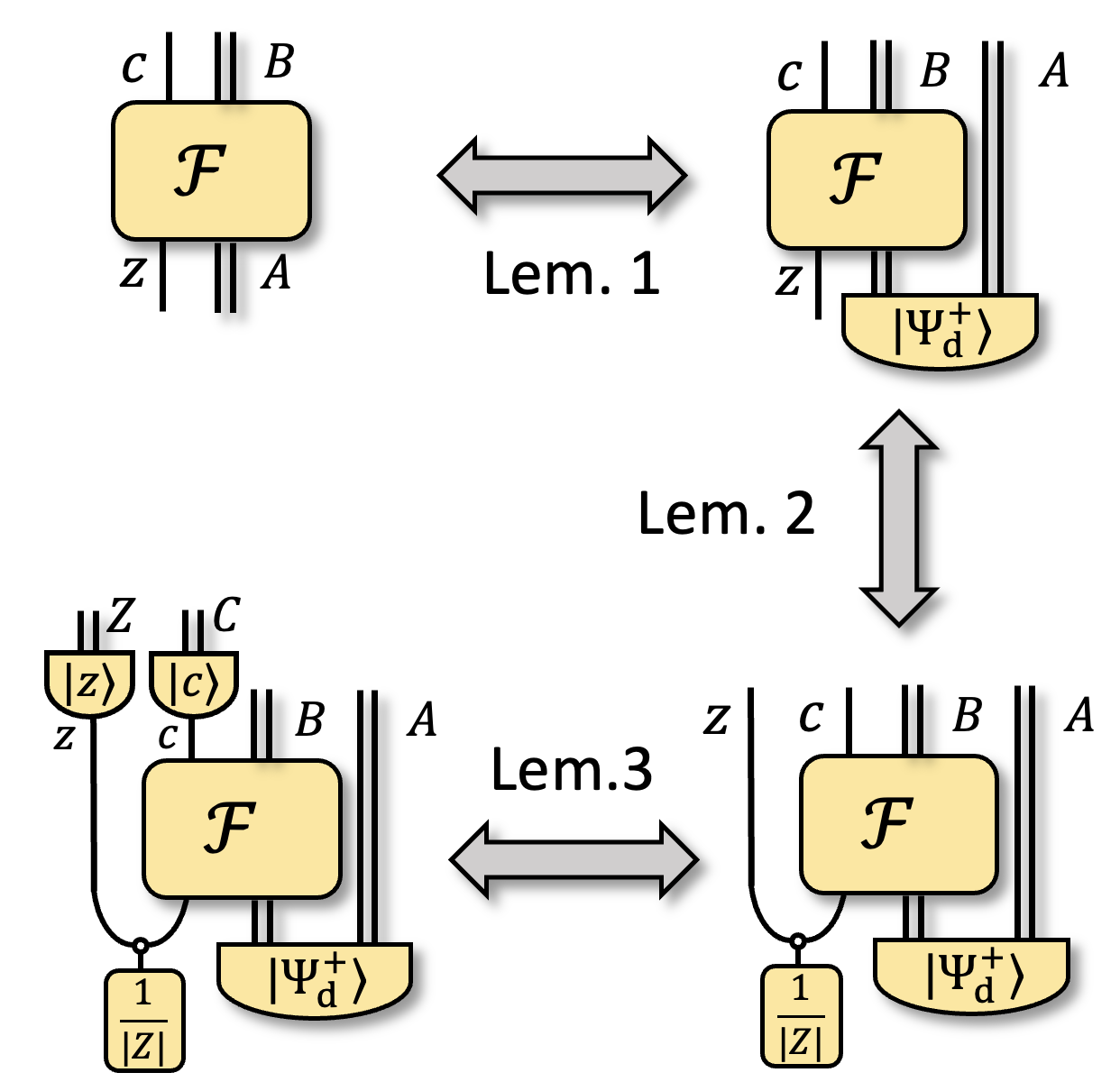}
    \caption{A multi-instrument $\{\{\mc{F}_{c|z}^{B|A}\}_c\}_z$ is nonclassical if and only if the associated multipartite state $\rho^{ZCBA}=\sum_{cz}\frac{1}{|Z|}p(c|z)\op{z}{z}^Z\otimes\op{c}{c}^C\otimes\rho_{c|z}^{BA}$ is nonclassical, where $p(c|z)\rho_{c|z}^{BA} = (\mc{F}_{c|z}^{B|A}\otimes \mbb{1}^{A})\op{\Psi^+}{\Psi^+}_{AA}$. }
    \label{fig:Choi-all}
\end{figure}

\section{Consistency under reconceptualizing settings and outcomes as dephased quantum systems}
\label{consistency}

\begin{figure}[htb!]

\includegraphics[width=0.9\linewidth]{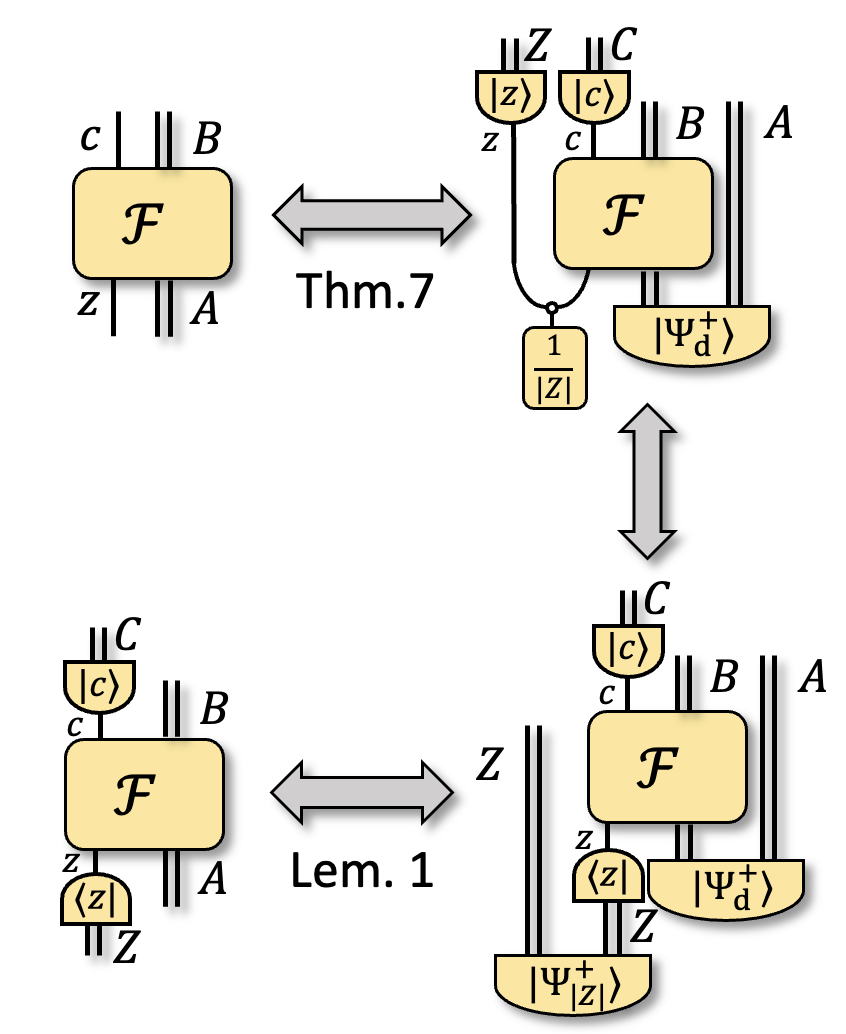}
    \caption{The proof that nonclassicality is preserved under reconceptualizing classical inputs and outputs as dephased quantum outputs. The implication described by the vertical arrow follows from the identity: $\tr_2[(\mbb{1}\otimes\op{z}{z})\op{\Psi^+_{|Z|}}{\Psi^+_{|Z|}}]=\frac{1}{|Z|}\op{z}{z}$, where $\ket{\Psi^+_{|Z|}}=\frac{1}{\sqrt{|Z|}}\sum_{z=1}^{|Z|}\ket{z}\ket{z}$. }\label{fig:outputstoo}
\end{figure}

A basic consistency condition for any notion of classicality is that it should not depend on the theorist's representational choices. In particular, it should not depend on whether one represents the settings and outcomes of a given process in the conventional manner, as classical variables, or as quantum systems that are dephased in an appropriate basis. This is an instance of the principle of cut motility which is introduced and defined in Ref.~\cite{SpekkensCut} as follows: For any physical theory purporting to have universal applicability, describing an experimental phenomenon typically involves making a cut between (i) the systems that are modeled explicitly within the theory and are the objects of interventions (preparations, measurements, transformations), and (ii) the systems that make up the devices that implement these interventions. A theory is said to satisfy the principle of cut-motility if its predictions are independent of the placement of the cut.


The notion of classicality we have introduced here satisfies this consistency condition. In the case of reconceptualization of classical outputs, Lemma~\ref{lem:absorb_classical_handle} already gives the result, and we can use the three lemmas from the previous section to extend this result to reconceptualization of classical settings, following the logic shown in Fig.~\ref{fig:outputstoo}.

Although this kind of consistency may seem quite basic, na\"{i}ve attempts to stitch together different notions of classicality found in the literature can fail to be consistent in this manner, and so can fail to provide a unified notion of classicality.

Consider a multi-state on a bipartite system, as shown in the left-hand side of Fig.~\ref{fig:view}. If one endorses the view that entanglement is necessary for nonclassicality of a bipartite state (call this the {\em nonclassical-only-if-entangled view}), it is reasonable to endorse the view  that the nonclassical bipartite multi-states are all and only those for which there is 
at least one value of the setting variable that yields an entangled state. 
Now consider any bipartite multi-state
for which the output system $C$ is dephased in some basis.
Taking the nonclassical-only-if-entangled view would imply that such a ``one-output-dephased'' process must be considered classical.

But the type of process just described can also be conceptualized as a multi-{\em source} on a {\em uni}partite system.  Specifically, the dephased quantum system $C$ can be reconceptualized as the classical outcome $c$ of a multi-source on $A$. This is depicted on the right-hand side of Fig.~\ref{fig:view}. Now recall that there is a conventional notion of nonclassicality for a special class of multi-sources (those obtained in a steering scenario), namely, steerability~\cite{Wiseman2007}.  Consequently, some processes of this type (namely, the steerable ones,  or equivalently, the LOSR-nonfree ones~\cite{Schmid2020typeindependent,Zjawin2023quantifyingepr}) will be judged to be nonclassical under this view.

In this way, the consistency criterion may be violated. Note that regardless of whether anyone has ever endorsed this pair of type-specific notions of nonclassicality, they appear very reasonable at first glance and yet they fail the consistency test. A concrete example of this type is as follows.

\begin{example}

Consider the following multi-source on a unipartite system (i.e., assemblage):  $$\{\{\frac{1}{2}\op{0}{0},\frac{1}{2}\op{1}{1}\},\{\frac{1}{2}\op{+}{+},\frac{1}{2}\op{-}{-}\}\}.$$ 
Note that it is a steerable assemblage.
If one  treats the classical outcome variable as a dephased quantum output, the process is reconceptualized as a multi-state on a bipartite system, namely, the one associated to the  set of bipartite  states $\{\rho_1, \rho_2\}$, where
\begin{align}
\rho_1&=\tfrac{1}{2}\op{0}{0}\otimes\op{0}{0}+\tfrac{1}{2}\op{1}{1}\otimes\op{1}{1},\\
\rho_2&=\tfrac{1}{2}\op{0}{0}\otimes\op{+}{+}+\tfrac{1}{2}\op{1}{1}\otimes\op{-}{-}.
\end{align}

Note that both of these states are separable.
If one takes the view that the classical-nonclassical divide for a multi-source is the unsteerable-steerable divide, then one will judge the process to be nonclassical. 
Meanwhile, if one takes the view that a multi-state on a bipartite system is nonclassical only if it yields an entangled state for at least one value of the setting variable, then one will judge the process to be classical~\footnote{Defining nonclassicality via steering is usually done for the particular class of multi-sources that arise in steering scenarios, namely, those wherein there is no causal influence of the setting variable $z$ on the quantum output $A$, even though a generic multi-source allows such an influence. One might wonder whether this nuance could impact the point we are trying to make here. Specifically, might refining the class of bipartite multi-states in the same way, namely, to those wherein there is no causal influence of $z$ on $A$, change the assessment of classicality of such a multi-state? Not if the assessment is made based on the notion that entanglement is necessary for nonclassicality regardless of how the bipartite state is prepared. The latter strikes us as being the conventional view, and so the nuance in question does not seem to impact our conclusion.}.  

\label{ex:dephase-bipartite}
\end{example}


This, therefore, provides an example of how seemingly reasonable views about the classical-nonclassical divide for different types of processes need not satisfy our consistency criterion. It follows that achieving such consistency is not trivial, and that if one is proposing a unified notion of nonclassicality, one should check for such consistency. ({It is worth noting that the notion of resourcefulness in the LOSR resource theory is indeed consistent in this sense~\cite{Schmid2020typeindependent}.})

For the unified notion of nonclassicality proposed in this paper, the process in Example~\ref{ex:dephase-bipartite} above is found to be nonclassical regardless of how it is conceptualized. 
Indeed, the requirement of consistency in 
this example provides further evidence that 
containing only separable states is not a sufficient condition for a bipartite multi-state to be considered classical.

\begin{figure}[htb!]
\includegraphics[width=0.8\linewidth]{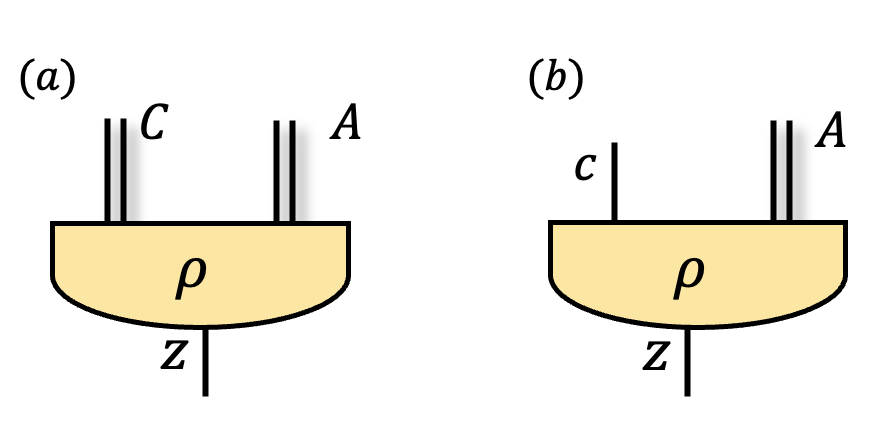}
    \caption{
    Example of a challenge to the consistency criterion when reconceptualizing a classical outcome as a dephased quantum output. (a) A multi-state on a bipartite system. (b) A multi-source on a unipartite system.  We can obtain an instance of the former from an instance of the latter by reconceptualizing the classical outcome as a dephased quantum output.   
     Every bipartite multi-state obtained in this way
    contains only separable states, such that if one adopted the conventional view that entanglement is necessary for nonclassicality, it would be deemed classical.  Meanwhile, if the unipartite multi-source is of the steerable variety, then under the conventional view that steerability is sufficient for nonclassicality, the same process would be judged to be nonclassical.
     }\label{fig:view}
\end{figure}

 A second example of how seemingly reasonable type-specific notions of nonclassicality can fail to satisfy the consistency criterion is as follows. 

Consider a multi-channel, as shown in the left-hand side of Fig.~\ref{fig:view1}. If one holds the view that for bipartite states, being entangled is necessary for nonclassicality, then it is also reasonable to take the view that for channels, being non-entanglement-breaking is necessary for nonclassicality. Furthermore, it is reasonable to then endorse the following view of nonclassicality for {\em multi-channels}:  
the nonclassical multi-channels 
are all and only those  for which the set of channels accessible by variation of the setting variable 
contains at least one that is non-entanglement-breaking.

Now consider any multi-channel from $A$ to $C$ 
for which the output system $C$ is dephased in some basis.
Under the view just described, such a process must be considered classical.  

But the type of process just described can also be conceptualized as a multi-{\em measurement} on a {\em uni}partite system.  Specifically, the dephased quantum system $C$ can be reconceptualized as the classical outcome $c$ of a multi-measurement on $A$. This is depicted on the right-hand side of Fig.~\ref{fig:view1},

Now recall that a conventional view about the classical-nonclassical divide for multi-measurements (equivalently, sets of measurements) is that incompatibility is necessary for nonclassicality.  Consequently, some of the multi-measurements that correspond to a multi-channel with a dephased output will be judged to be nonclassical under this view.

In this way, the consistency criterion may be violated. The following provides a concrete example.

\begin{figure}[htb!]
\includegraphics[width=0.7\linewidth]{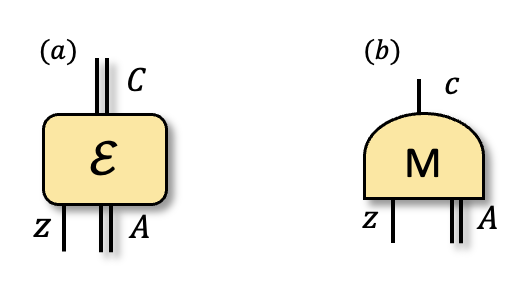}
    \caption{
        Another example of a challenge to the consistency criterion when reconceptualizing a classical outcome as a dephased quantum output.  (a) A multi-channel. (b) A multi-measurement. 
          We can obtain an instance of the former from an instance of the latter by reconceptualizing the classical outcome as a dephased quantum output. 
     Every multi-channel obtained in this way
    contains only entanglement-breaking channels, such that if one adopted the conventional view that for channels, being non-entanglement-breaking is necessary for nonclassicality, the multi-channel would be deemed classical.  Meanwhile, if the multi-measurement contains incompatible measurements, then, under the conventional view that such incompatibility is sufficient for nonclassicality, the same process would be judged to be nonclassical. 
    }\label{fig:view1}
\end{figure}
\begin{example}
Consider the following multi-measurement (equivalently, set of measurements):  $$\{\{\op{0}{0},\op{1}{1}\},\{\op{+}{+},\op{-}{-}\}\}.$$ 
Note that the two measurements in this set are incompatible.
If one  treats the classical outcome variable as a dephased quantum output, the process is reconceptualized as a multi-channel, namely, the one associated to the  set of channels  $\{\mc E_1, \mc E_2\}$, where
\cmt
\begin{align}
\mc E_1(\cdot)&=\tr[\op{0}{0}(\cdot)]\op{0}{0}+\tr[\op{1}{1}(\cdot)]\op{1}{1},\\
\mc E_2(\cdot)&=\tr[\op{+}{+}(\cdot)]\op{0}{0}+\tr[\op{-}{-}(\cdot)]\op{1}{1}. 
\end{align}

Note that both of these channels are entanglement-breaking. 

If one takes the view that the classical-nonclassical divide for a multi-measurement is the compatible-incompatible divide, then one will judge the process to be nonclassical. 
Meanwhile, if one takes the view that a multi-channel is nonclassical only if it yields a non-entanglement-breaking channel for at least one value of the setting variable, then one will judge the process to be classical.
\label{ex:dephased-channel}
\end{example}

Again, we see that seemingly reasonable views about the classical-nonclassical divide for different types of processes need not satisfy our consistency criterion.

For the unified notion of nonclassicality proposed here, the process in Example~\ref{ex:dephased-channel} is found to be nonclassical regardless of how it is conceptualized. This possibility follows from the fact that, relative to our notion, incompatibility in a multi-measurement is a sufficient condition for nonclassicality, and a multi-channel can be deemed nonclassical even if all the channels it contains are entanglement-breaking~\footnote{We note that the set of channels in question would also be considered nonclassical if one took  the notion of channel incompatibility~\cite{Heinosaari2014} to be a sufficient condition for nonclassicality.}.

\section{Related works}\label{related}

Another approach to defining a classical-nonclassical divide in a manner that follows a unified principle which applies across a wide range of types of processes is the type-independent resource theory presented in Refs.~\cite{Schmid2020typeindependent,Schmid2023understanding}. This approach defines a notion of {\em nonclassicality of common cause} which applies to processes such as multipartite states, multipartite boxes (i.e., processes with only classical inputs and classical outputs), assemblages, and so on. However, the classical-nonclassical divide defined by this approach is trivial for unipartite processes, and it is moreover not clear how to generalize this approach to arbitrary causal structures (although first steps were discussed in Appendix A.3 of Ref.~\cite{Wolfe2020quantifyingbell}.

It is nonetheless useful to contrast this approach with the one we propose here. 
One way to understand the relationship between the two is that classicality, as defined in the LOSR approach, demands that the ontological accounts of the statistics have a causal structure that mirrors that of the quantum circuit (an instance of the assumption of diagram preservation),  while the approach described here demands, in addition, that the ontological representation map be linear. 
This stronger assumption is motivated by the principle of generalized noncontextuality and implies that the set of processes that is found to be nonclassical is broader than the set that is found when one does not assume linearity.   For instance, there are processes on unipartite systems that are judged to be nonclassical.

One may also wonder how our results are consistent with certain foils to quantum theory, such as the stabilizer subtheory of qutrits and Gaussian quantum mechanics, which exhibit entanglement, incompatibility, and steering despite being classically explainable~\cite{bartlett2012reconstruction,epistricted} (i.e., despite being consistent with a generalized-noncontextual ontological model). This does not conflict with our results because these foil theories are subtheories of quantum theory that contain a strict subset of the states and measurements allowed in quantum theory.
and the notion of classicality proposed here depends on considering  {\em all possible dual processes} allowed in the theory, which will be different in these subtheories compared to quantum theory.  

Extending our analysis of nonclassicality from quantum theory to more general GPTs, and using this to characterize different theories---in the spirit of the GPT reconstructions of quantum theory---is a natural and interesting direction for future work.

Several previous works~\cite{Plavala2024,PhysRevResearch.2.013011,Wrigh2023} also studied the properties of a bipartite state that are pertinent to its ability to support a proof of generalized contextuality, in the following sense.  In a steering scenario for a bipartite system $A_1A_2$, if one considers implementing any local measurement on $A_1$, the bipartite state on $A_1A_2$ induces a multi-source on $A_2$.  This multi-source, together with an arbitrary measurement on $A_2$ constitutes a prepare-measure circuit, and one can ask whether the statistics of this prepare-measure circuit admit of a noncontextual ontological model.  In particular, these prior works have studied the connection between these sorts of properties of a bipartite state and their entanglement properties, steering properties, and nonlocality properties.

However, none of these works considered the full set of operational identities in the bipartite scenario (listed in Eq.~\eqref{op:bipartiterep}), which is essential if one wishes to define or fully characterize the classicality of a bipartite state in the sense of finding necessary and sufficient conditions. Essentially, all these works consider the prepare-measure scenario induced by steering from one party to another. Consequently, they only consider the operational identities $\mc{O}(\msf{M}_1\circ \msf{P})$ and $\mc{O}(\msf{M}_2)$ (or a subset thereof), while omitting the identities that arise when considering the preparation steered in the opposite direction, from $A_2$ to $A_1$, namely $\mc{O}(\msf{M}_2\circ \msf{P})$ and $\mc{O}(\msf{M}_1)$, and while also omitting
the identities $\mc{O}(\tilde{\msf{P}})$ defined in~Eq.~\eqref{eq:OPnewbi}.

To be more specific, the work by Tavakoli and Uola \cite{Tavakoli2020} considers only the operational identities in $\mc{O}(\msf{M}_2)$ and a {\em subset} of the operational identities in $\mc{O}(\msf{M}_1\circ \msf{P})$, namely, those implied by the no-signaling principle. 
The results of Wright and Farkas \cite{Wright2023invertible} rely on a still smaller set of operational identities, including those following from no-signaling, together with a subset of the operational identities in $\mc{O}(\msf{M}_2)$, specifically those arising from the normalization condition $\sum_bM_{b|y}=\sum_bM_{b|y'}$.
Plavala and Gühne \cite{Plavala2024}, in comparison, do consider all operational identities in the prepare-measure scenario induced by steering from $A_1$ to $A_2$, but they do not consider all operational identities on the full bipartite scenario, in particular, not the operational identities $\mc{O}(\tilde{\msf{P}})$ defined in Eq.~\eqref{eq:OPnewbi}.

Because these earlier works impose fewer constraints on the ontological model than we do here, the results they derive do not fully characterize the classical-nonclassical divide for 
bipartite states according to our proposal. This is consistent with statements in many of those works, which explicitly state that the results merely provide necessary but not sufficient conditions for noncontextual realizability.
The one exception to this is Theorem 2 in~\cite{PhysRevResearch.2.013011}, whose statement suggests that it leverages {\em all} ontological constraints that follow from the assumption of noncontextuality in the scenario under study, when in fact it does not. We discuss this further in Ref.~\cite{zhang2024parellel}.

In a companion paper~\cite{zhang2025forthcoming}, we also consider the sufficient conditions for nonclassicality implied by different subsets of the set of all constraints that arise from noncontextuality in bipartite Bell scenarios, and we show that for some particularly simple choices of these subsets, one singles out local, unsteerable, and entangled states.

\section{Conclusion}

In this work, we proposed a definition of the classical-nonclassical divide 
for individual quantum processes.  There are many avenues for future work. 

If one seeks only sufficient conditions for nonclassicality of a  given process, then rather than considering the classical explainability of the set of circuits induced by composing it with {\em all} dual processes, it suffices to show the failure of classical explainability for a subset of dual processes. 
One can consider the partial order of such subsets as a means to understand what is required to witness nonclassicality. For instance, one could seek to find the least elements in this partial order.

Another important question which we study in forthcoming work~\cite{zhang2026forthcoming} is how classicality behaves under composition. For example, we prove therein that the composition of two classical processes (either in parallel or in sequence) is always classical; that is, there exists no superactivation of nonclassicality.

This work focused on the decision problem regarding the classical-nonclassical divide: assessing whether a given process is classical or nonclassical. An important open problem is that of {\em quantifying} the degree of nonclassicality according to the approach introduced here.  

In a companion paper~\cite{zhang2024parellel}, we give some basic tools for doing so for multi-sources and multi-measurements, based on specific measures of how robust a given process is to noise. 

The most principled and informative approach to the quantification problem, however, is to pursue it within the framework of resource theories~\cite{coecke2016mathematical}.  Roughly speaking, one expects that the correct choice of the divide between free and nonfree resources in the resource theory should follow from the divide between classical and nonclassical processes.

One can already recognize some aspects that such a resource theory will satisfy. For example, it will not be convex, since (for example) any set of product states is classical, while some separable states are nonclassical. Similarly, the convex combination of two classical measurements can result in a nonclassical measurement. One way to understand this nonconvexity is to note that the process of mixing can introduce new operational identities, and hence new constraints on the ontological model,  which in turn can activate nonclassicality. 

Moreover, this work opens new possibilities for constructing a resource theory of nonclassicality based on nonrealizability by a noncontextual ontological model; for instance, it could be used to define a resource theory in which the objects under study (the resources) are themselves individual quantum processes. Prior to this work, it seemed that the resources would need to be more complicated mathematical objects (such as sets of states together with sets of effects, or sets of states together with operational equivalences among these; see, for example, Ref.~\cite{catani2024resource}). 

An example of a potential application of such a resource theory is to give a resource-theoretic account of the simultaneous self-testing of the nonclassicality in the bipartite state and in the local measurements of a Bell scenario. This cannot be achieved in the LOSR approach to quantifying nonclassicality, because although it can yield a resource-theoretic account of self-testing for the bipartite state, as described in Ref.~\cite{Schmid2023understanding}, it cannot yield a resource-theoretic characterization of self-testing for the local measurements, since all local measurements are taken to be free in the LOSR resource theory.

In the introduction, we noted that the problem of understanding the scope of quantum-over-classical advantages for various tasks motivates understanding the classical-nonclassical divide.  In recent years, it has become apparent that the framework of resource theories provides a particularly powerful perspective on this problem. Understanding quantum-over-classical advantages, therefore, provides another motivation for developing a unified 
resource theory of nonclassicality, i.e., one wherein the notion of freeness is defined in a manner that can be particularized to any type of process. Furthermore, by leveraging the similarity of the structure of nonclassical resources for processes of different types (for instance, multi-sources and multi-measurements), a known quantum-over-classical advantage in one task might be used to infer an advantage in a distinct task.

{\em Acknowledgements---}  We thank two anonymous referees for catching some mistakes in an earlier version of the article and for comments that led us to make many improvements in the presentation. 
RWS thanks Gilad Gour and Teiko Heinosaari for early discussions on whether incompatibility of measurements should be viewed as a type of nonclassicality. We thank John Selby, Ana Bel\'en Sainz, Roberto D. Baldijão, Rafael Wagner, and Martin Pl\'avala for helpful discussions. YZ was supported by the Natural Sciences and Engineering Research Council of Canada
(NSERC) and Canada Foundation for Innovation (CFI), YZ is grateful for the hospitality of Perimeter Institute, where this work was carried out.  DS, RWS, and YY were supported by Perimeter Institute for Theoretical Physics. Research at Perimeter Institute is supported in part by the Government of Canada through the Department of Innovation, Science and Economic Development and by the Province of Ontario through the Ministry of Colleges and Universities. YY was also supported by the Natural Sciences and Engineering Research Council of Canada (Grant No. RGPIN-2024-04419).

\bibliographystyle{apsrev4-2}
\bibliography{ref}

\section{Appendix}
\appendix
\renewcommand{\addcontentsline}[3]{}
\section{Frame representations}\label{framerepn}
Consider a complex Hilbert space $\mc{H}$ with some finite
dimension $d$. The Hermitian operators $\text{Herm}(\mc{H})$ on $\mc{H}$ form a real vector space of dimension $d^2$ with the Hilbert-Schmidt inner product $\langle A,B\rangle \coloneqq\tr[AB]$ for $A,B \in \text{Herm}(\mc{H})$. \par 

Therefore, some standard vector space concepts can be applied to the operator space. For instance, it is well-known that a basis of a vector space uniquely defines a dual basis in the dual vector space (i.e., the space of functionals on the vector space). We will leverage this fact, but in a slightly different form: 
\begin{lemma}\label{lemmaunique}
Given any basis $\{F_\lambda\}_\lambda$ for a $d^2$-dimensional real vector space $\text{Herm}(\mc{H})$ of Hermitian operators on a Hilbert space $\mathcal{H}$, there is a unique set $\{D_\lambda\}_\lambda$ of $d^2$ Hermitian operators satisfying
\begin{equation}
\label{eq:dualFrame1}
\tr(D_{\lambda'} F_\lambda)=\delta_{\lambda,\lambda'},
\end{equation}
and $\{D_\lambda\}_\lambda$ also forms a basis for $\text{Herm}(\mc{H})$.
\end{lemma}
 \begin{proof} 
Consider any basis $\{ F_\lambda\}_\lambda$ of $\text{Herm}(\mc{H})$. It uniquely specifies a basis $\{\mathcal{D}_\lambda\}_\lambda$ of the dual vector space $\text{Herm}(\mc{H})^*$, where $\{\mathcal{D}_\lambda\}_\lambda$ are linear functionals satisfying $\mathcal{D}_{\lambda'}(F_\lambda)=\delta_{\lambda,\lambda'}$. {To see that this is unique, consider a linear functional $\mathcal{D}'_{\lambda'}$ satisfying ${D}'_{\lambda'}(F_\lambda) = \delta_{\lambda,\lambda'}$ for all $\lambda$. Since a linear functional is fully specified by its action on a basis, $\mathcal{D}'_{\lambda'}$ is the exact same functional as $\mathcal{D}_{\lambda'}$.} Now, in order to obtain again a set of Hermitian operators $\{D_\lambda\}_\lambda$, we use the Riesz representation Theorem~\cite{riesz1914demonstration}, which states that each of these linear functionals $\mathcal{D}_\lambda$ can be written as the Hilbert-Schmidt inner product with a unique Hermitian operator $D_\lambda$, namely
$\mathcal{D}_{\lambda}(\cdot)= \tr{ [ (\cdot) D_{\lambda}] }.$
This picks out a unique basis $\{D_{\lambda}\}_\lambda$ which satisfies Eq.~\eqref{eq:dualFrame1}. 
\end{proof}
More generally, one can define a {\em frame and a dual frame} as a generalization of the notion of a basis and a dual basis~\cite{Ferrie2008}, where the vectors in both the frame and the dual frame are assumed to be spanning, but may or may not be linearly independent~\cite {Ferrie2008}. If $\{F_{\lambda}\}_{\lambda}$ and $\{D_{\lambda}\}_{\lambda}$ denote a frame and a dual frame in $\text{Herm}(\mc{H})$, respectively, then it follows that, for any $A\in \text{Herm}(\mc{H})$, one has
\begin{align}
 & A=\sum_{\lambda}\tr[A D_{\lambda}]F_{\lambda} \notag \\
 & A=\sum_{\lambda}\tr[A F_{\lambda}]D_{\lambda}.
 \label{eq:statereconstr}
\end{align}
These reconstruction equations show explicitly how every vector in the space can be expanded in terms of the spanning frame operators or the spanning dual frame operators.
Note that, unlike dual bases, dual frames are not necessarily unique. A detailed exploration of frames, dual frames, and their applications in quantum information science lies beyond the scope of this letter. For a more comprehensive introduction, we refer the reader to Ref.~\cite{Ferrie2008}.

\section{Operational identities in a quantum circuit}
\label{sec:gauge}

In Subsection~\ref{classicalexplainability} of the main text, we discussed the classical explainability of a given prepare-transform-measure circuit, an analysis that relies on enumerating all operational identities in the quantum circuit. In this appendix, we make explicit the assumptions underlying that enumeration and justify why, under these assumptions, the operational identities listed in Section~\ref{classicalexplainability} are the only ones that are relevant. We also briefly explain how this family of operational identities would change under alternative modeling choices, and why we do not adopt those alternatives in our analysis.

When assessing classical explainability, one must account for all operational identities implied by linearity and diagram preservation. Concretely, one should (i) consider all quantum processes obtainable via finite parallel and sequential compositions of the set of processes under consideration, since, in principle, each such composite process can be inferred from the experimental statistics, and (ii) group these processes according to their quantum input and output system types, since an operational equivalence condition can only hold among processes that have the same quantum inputs and outputs.

The prepare-transform-measure scenario discussed in Section~\ref{classicalexplainability} contains a multi-state on system $A$: $\{\rho_{\cdot|x}^A\}_x$, a multi-channel from system $A$ to $B$: $\{\mc{E}_{\cdot|z}^{B|A}\}_z$, and a measurement on system $B$: $\{M_{b}^{B}\}_b$, as depicted in Figure~\ref{PTM}. Considering arbitrary compositions of these processes leads one to consider also a multi-state on system $B$: $\{\tilde{\rho}^B_{\cdot|xz}\}_{xz}$, a multi-measurement on system $A$: $\{\{\tilde{M}^{A}_{b|z}\}_b\}_z$ and a measure-and-prepare instrument from system $A$ to $B$: $\{\{\tilde{\mc{E}}^{B|A}_{b''|z'x'z''}\}_b\}_{xz}$. The operational identities among these processes are described in \cref{eq:opidentity,eq:op-PTM-t,eq:op-PTM-m,eq:op-PTM-s,eq:op-PTM-pr} in Section~\ref{classicalexplainability}.

Under the assumption that $A$ and $B$ are distinct system types, the list above is indeed the set of all operational identities. In particular, it \textit{does not} include any that refer to 
measure-and-prepare channels constructed by composing states $\rho^A_{\cdot|x}$ and effects $M^B_b$. This is the case since the only operational identities on these measure-and-prepare channels from $B$ to $A$ are implied by those holding among the individual states and among the individual effects. Unlike in the case of channels from $A$ to $B$, there are no other possible channels from $B$ to $A$ within this scenario that could be combined with these to lead to new operational identities.

In principle, one could also consider the case where $A$ and $B$ have the same system type. In this case, additional operational identities can arise, for example, relating $\rho^A_x$ and $\mc{E}^{B|A}_{\cdot|z}(\rho^A_{\cdot|x})$,  $M_b^A$ and $[\mc{E}^{B|A}_{\cdot|z}]^{\dagger }(M_b^A)$, or involving sequential compositions of the channels themselves (note that such sequential compositions are not possible if the system type $A$ and $B$ are distinct). 

Throughout the main text, we assume that any two distinct quantum inputs and outputs belong to different system types. This choice is motivated by the problem of gauge freedom from GPT tomography and gate-set tomography~\cite{zhang2025forthcoming, Mazurek2021, Nielsen2021}. Specifically, in self-consistent tomography, one does not obtain a unique output of the processes in a quantum circuit; rather, one obtains an equivalence class of representations related by invertible linear transformations on each system, which are referred to as tomographic {\em gauge freedoms}, following Ref.~\cite{Nielsen2021}. Without additional assumptions beyond the circuit data, there is no canonical identification between the representations for different quantum systems; therefore, treating two systems as the same ``type" is itself a gauge choice. 

Thus, by treating distinct quantum systems in the circuit as distinct types, we restrict ourselves to operational identities that are gauge-independent. 
Therefore, our result can be seen as a gauge-independent assessment of the nature of the process, which is readily experimentally implementable using schemes similar to various self-consistent tomographies~\cite{Nielsen2021, Mazurek2021, zhang2025forthcoming}. 

\section{Why the form of variability over dual processes is irrelevant} \label{labels}

In the main text (especially Section~\ref{variation}), we noted that there are many ways of ranging over all quantum states. For example, one can range over all quantum states using a multi-state process, where the setting variable ranges over all possible quantum states, or by post-selecting on the outcome of a source associated with the unitarily covariant ensemble of states, or via a multi-source with both setting and outcome variables. We noted that the different choices one can make in this regard make no difference for qualitative assessments of nonclassicality.  
This is because every multi-source defines a collection of (possibly subnormalized) quantum states, and the ontological representation of a collection of states defines the associated collection of representations of the states; moreover, by the linearity of any ontological model, one can uniquely extend the ontological representation of the normalized states to an ontological representation of subnormalized versions of those same states.

The story is similar for measurements. An ontological model for some set of measurements constitutes a collection of representations of all of the effects appearing in those measurements. These response functions serve to define a valid and empirically adequate ontological model for any other set of measurements that contains the same set of effects. Consequently, the existence of a linear and empirically adequate mapping from the set of all quantum effects to their ontological representation as response functions implies the existence of such a mapping for any multi-measurement one could construct with those effects.

It is worth noting that the way of understanding why verdicts of nonclassicality are unchanged under flag-convexification (see lemma~\ref{lem:flag-convexified}). Whether one fixes a certain input setting value, or postselects on the corresponding flag variable taking that same value, one is simply choosing between two different physical instantiations of the same variability, so an ontological representation for one directly gives an ontological representation for the other.

\section{Proof of the sufficiency of considering factorizing duals }\label{secdualsthm}

We now prove Theorem~\ref{thmdual}, that the definition of nonclassicality obtained by considering all dual processes and the one obtained by considering all factorizing dual processes are the same. In other words, we prove that Definition~\ref{def:og_process} is equivalent to Definition~\ref{maindefn}.

In the following two subsections, we prove the result for bipartite states and for transformations, respectively. The generalization of the proof to arbitrary process types is straightforward.

Although the proof is lengthy, its basic idea is simple. First, we note that every quantum process with multipartite inputs and/or outputs can be written as a linear combination of processes that factorize across the inputs and outputs. From this and the assumption that an ontological model is diagram-preserving and linear, one can then uniquely extend any given ontological model defined for all product processes in quantum theory to a quasiprobabilistic representation of {\em all} processes in quantum theory, including the nonproduct ones. It remains only to show that for the relevant processes in question, this quasiprobabilistic representation is in fact a valid probabilistic representation (e.g., is positive and normalized appropriately), a fact which follows easily from frame representation theory~\cite{christensen2016introduction}.

\subsection{Factorizing duals for bipartite states}

We first prove the theorem for the special case where the given process is a bipartite state. That is, we prove that the definition of nonclassicality for a bipartite state is the same whether it is defined relative to all dual processes or to all factorizing dual processes.
\begin{lemma}
The statistics generated by a bipartite state together with {\em all} bipartite effects are classical explainable (in the sense of definition~\ref{defn:classical}) if and
only if the statistics generated by it with all {\em product} effects are classically explainable.
\label{lem:bi_global}
\end{lemma}
\begin{proof}
The proof for the \emph{only if} direction is simply that factorizing product bipartite effects are special cases of general bipartite effects. So if the statistics generated by a bipartite state with {\em all} bipartite effects are classically explainable, then those generated by that state with all {\em product} effects are classically explainable.

Now we prove the \emph{if} direction. If the statistics generated by a bipartite state $\rho^{A_1A_2}$ with every product bipartite effect $M^{A_1}\otimes M^{A_2}$ are classically explainable, then there exists an ontological model for the statistics such that 
\begin{align}
\tr[M^{A_1}&\otimes M^{A_2} \rho^{A_1A_2}] \notag \\
&=\sum_{\lambda_1\lambda_2} p(\lambda_1\lambda_2)p(M^{A_1}|\lambda_1)  p(M^{A_2}|\lambda_2)
\end{align}
where $p(M^{A_i}|\lambda)$s depends linearly on the effect $M^{A_i} \in \mc{M}(\mc{H}^{A_i})$. Moreover since $p(M^{A_i}|\lambda)\in [0,1]$ and $p(\mbb{1}^{A_i}|\lambda)=1$, by the generalized Gleason Theorem~\cite{Busch2003}, there exists a unique set of density operators $\{ \sigma^{A_i}_{\lambda_i}\}_{\lambda_i}$ such that $p(M^{A_i}|\lambda_i)=\tr[M^{A_i}\sigma^{A_i}_{\lambda_i}]$, and so we have 
\begin{align}
&\tr[M^{A_1}\otimes M^{A_2} \rho^{A_1A_2}]\notag \\
=&\sum_{\lambda_1\lambda_2} p(\lambda_1\lambda_2)\tr[M^{A_1}\sigma^{A_1}_{\lambda_1}]\tr[M^{A_2}\sigma^{A_2}_{\lambda_2}]\notag  \\
=&\tr[(M^{A_1}\otimes M^{A_2} )\sum_{\lambda_1\lambda_2} p(\lambda_1\lambda_2)\sigma^{A_1}_{\lambda_1}\otimes \sigma^{A_2}_{\lambda_2}]
\label{eq_prodMM}
\end{align}
for any $M^{A_i} \in \mc{M}(\mc{H}^{A_i})$. Because the set of all product effects is informationally complete, \cref{eq_prodMM} implies that 
\begin{equation}
    \rho^{A_1A_2}=\sum_{\lambda_1\lambda_2} p(\lambda_1\lambda_2)\sigma^{A_1}_{\lambda_1}\otimes \sigma^{A_2}_{\lambda_2},
\end{equation}

and so the state $\rho^{A_1A_2}$ must be separable.

In quantum theory, any bipartite effect $M^{A_1A_2}\in \mc{M}({\mc{H}^{A_1}\otimes \mc{H}^{A_2}})$ can always be decomposed into a linear combination of product effects (see Fig.~\ref{Globaleffect}), as
\begin{align}
    M^{A_1A_2}=\sum_{k_1,k_2} r_{k_1k_2} M_{k_1}^{A_1}\otimes  M_{k_2}^{A_2}
    \label{eq:joint_decomp}
\end{align}
with some real coefficients $\{r_{k_1k_2}\}_{k_1k_2}$. Consequently, we can define a unique quasiprobability representation for $M^{A_1A_2}$ as the unique linear extension of the ontological representation of the product effects, via
\begin{align}
p(M^{A_1A_2}|\lambda_1\lambda_2)\coloneqq&\sum_{k_1k_2}r_{k_1k_2} p(M_{k_1}^{A_1}|\lambda_1)  p(M^{A_2}_{k_2}|\lambda_2)\notag \\ 
    =& \sum_{k_1,k_2} r_{k_1k_2}\tr[M^{A_1}_{k_1}\sigma^{A_1}_{\lambda_1}]\tr[M^{A_2}_{k_2}\sigma^{A_2}_{\lambda_2}]\notag  \\
  =&\tr[M^{A_1A_2}\sigma^{A_1}_{\lambda_1}\otimes \sigma^{A_2}_{\lambda_2}]
    \label{entangrepn}
\end{align}

The form of Eq.~\eqref{entangrepn} immediately guarantees that the representation map is linear and diagram-preserving. Furthermore, since $\sigma^{A_1}_{\lambda_1}$ and $\sigma^{A_2}_{\lambda_2}$ are valid quantum states, $p(M^{A_1A_2}|\lambda_1\lambda_2)\in[0,1]$ defines a valid stochastic representation.

\begin{figure}[htb!]
\centering
\includegraphics[width=0.45\textwidth]{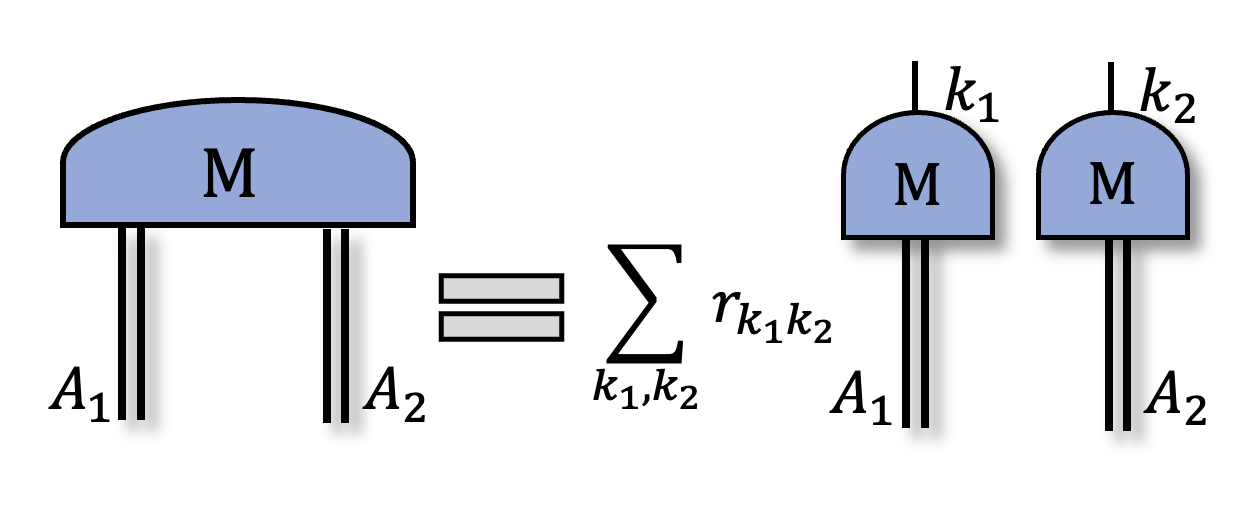}
\caption{In quantum theory, every effect can be decomposed as a linear combination of product effects.}
\label{Globaleffect}
\end{figure}

Finally, it is straightforward to verify that the ontological model for the given effect and state $\rho^{A_1A_2}$ reproduces the quantum statistics:
\begin{align}
& \sum_{\lambda_1\lambda_2}
p(M^{A_1A_2}|\lambda_1\lambda_2) p(\lambda_1\lambda_2)  \notag \\
=&  \sum_{\lambda_1\lambda_2}p(\lambda_1\lambda_2)\sum_{k_1,k_2} r_{k_1k_2}\tr[M^{A_1}_{k_1}\sigma^{A_1}_{\lambda_1}]\tr[M^{A_2}_{k_2}\sigma^{A_2}_{\lambda_2}]\notag  \\
=&\tr[M^{A_1A_2} \rho^{A_1A_2}].
\label{def:op-multi}
\end{align}
Consequently, any statistics generated by $\rho^{A_1A_2}$ with any bipartite effect (not necessarily factorizing) are classically explainable. This concludes the proof of the if direction.
\end{proof}

\subsection{Factorizing duals for transformations }

Next, we prove that the statistics for the set of circuits generated by contracting the transformation of interest with all dual combs are classically explainable if and only if this is the case with all factorizing dual combs. 

Before proving this, we include an aside about the ontological representation of quantum combs. To our knowledge, this has not been explicitly discussed in the literature. Nevertheless, the fully compositional definition of the ontological representation map given in Ref.~\cite{Schmid2024structuretheorem} can be applied to this special case, and we do so here. 

We consider the specific class of combs which are dual processes for a given transformation $\cal E_{B|A}: \mc{L}(\mc{H}^A)\mapsto \mc{L}(\mc{H}^B)$, as shown in the top left of Figure~\ref{Globalcomb}. These combs are denoted $\mc{C}^{A|B}$. $\mc{C}^{A|B}$ is a special case of 2-combs defined in Ref.~\cite{Chiribella2009}, which only has a quantum output on $\mc{L}(\mc{H}^{A})$ and a quantum input on $\mc{L}(\mc{H}^B)$.

An ontological model of the circuit in \cref{Globalcomb} associates to $\cal E$ a stochastic map $p_{\mc{E}}(\lambda_B|\lambda_A)$ and to $\mc{C}$ a sub stochastic map $\tilde{p}_{\mc{C}}(\lambda_A|\lambda_B)$ (note that in our case,  a comb is define by a single effect so that its ontological model is subnormalized, i.e.,  $\sum_{\lambda_A}\tilde{p}_{\mc{C}}(\lambda_A|\lambda_B)$ need not to be 1). 
As always for an ontological model, the map from the set of quantum processes of a given type(in this case, from the set of combs) to their ontological representations must be linear, and moreover, the composition of the processes within quantum theory must equal the composition of their ontological representations (so that the empirical predictions of quantum theory are reproduced). Linearity implies that the $\tilde{p}_{\mc{C}}(\lambda_A|\lambda_B)$ obey the constraints implied by all operational identities among combs, while the latter fact implies that
\begin{equation}
    \mc{C}^{A|B}\circ \mc{E}^{B|A} = \sum_{\lambda_A\lambda_B}\tilde{p}_{\mc{C}}(\lambda_A|\lambda_B)p_{\mc{E}}(\lambda_B|\lambda_A).
    \label{eq:OM_comb}
\end{equation}

\begin{figure}[htb!]
\centering
\includegraphics[width=0.5\textwidth]{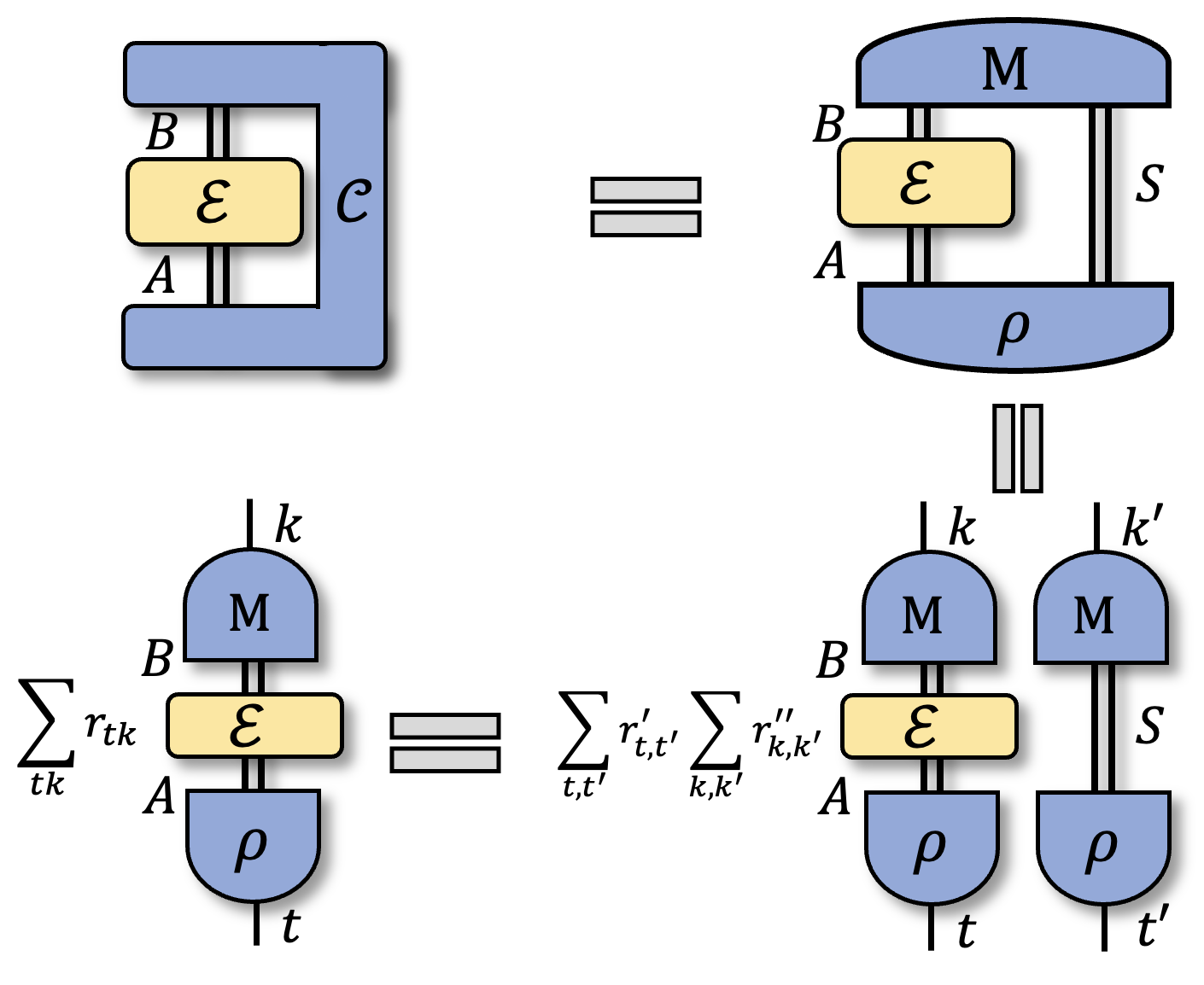}
\caption{A comb can be decomposed using state-effect pairs. }
\label{Globalcomb}
\end{figure}

We can now give the main argument. Note that a factorizing comb with quantum output $A$ and input $B$ is simply a state on $A$ and an effect on $B$.
\begin{lemma}
The statistics of the set of circuits obtained by composing a quantum channel together with every possible comb is classically explainable if and only if the statistics of the set of circuits obtained by composing it with every possible factorizing comb is classically explainable. 
\label{lem:comb_global}
\end{lemma}
\begin{proof}
The \textit{only if} direction is trivial, as a factorizing comb is a special case of a comb.

Now we prove the \emph{if} direction. If the statistics generated by the circuit consists of a quantum channel $\mc{E}^{B|A}$ and a factorizing comb defined by a state $\rho^A$ and effect $M^{B}$ that are classically explainable, then we have an ontological model for the operational statistics 
\begin{align}
\tr[M^{B}\mc{E}^{B|A}(\rho^A)]&=\sum_{\lambda_A\lambda_B}p(M^{B}|\lambda_B)p(\lambda_B|\lambda_A)p(\lambda_A|\rho^A)  
\end{align}
where $p(M^{B}|\lambda_B)$ depends linearly on $M^{B}$ and $p(\lambda_A|\rho^A) $ depends linearly on all $\rho^A$, respectively. Since $p(M^{B}|\lambda_B)\in[0,1]$ with $p(\mbb{1}^{B}|\lambda_B)=1$, from generalized Gleason's theorem (and the Riesz representation theorem), there exists a unique POVM $\{G_{\lambda_A}^{A}\}_{\lambda_A}$ such that $ p(\lambda_A|\rho^A)=\tr[\rho^AG_{\lambda_A}^{A}]$ and that there exists a unique set of density operator $\{\sigma^{B}_{\lambda_B}\}_{\lambda_B}$ such that $p(M^{B}|\lambda_B)=\tr[M^{B}\sigma^{B}_{\lambda_B}]$. Therefore: 
\begin{align}
\tr[M^{B}&\mc{E}^{B|A}(\rho^A)] \label{eq:combrep}\\
=&\sum_{\lambda_A\lambda_B}p(\lambda_B|\lambda_A)\tr[\rho^AG_{\lambda_A}^{A}]\tr[M^{B}\sigma^{B}_{\lambda_B}] \notag \\
=& \tr\big[M^{B} \sum_{\lambda_A\lambda_B}p(\lambda_B|\lambda_A)\tr[G_{\lambda_A}^{A}\rho^A]\sigma^{B}_{\lambda_B}\big]\notag
\end{align}
Since this holds for all quantum effects $M^{B}\in\mc{M}(\mc{H}^B)$, we have
\begin{align}
\label{eq:holevo form}
    \mc{E}^{B|A}(\rho^A)=  \sum_{\lambda_A\lambda_B}p(\lambda_B|\lambda_A)\tr[G_{\lambda_A}^{A}\rho^A]\sigma^{B}_{\lambda_B}
\end{align}
for all $\rho^A\in\mc{S}(\mc H^A)$, which means that $\mc{E}^{B|A}$ can be expressed in the Holevo form and therefore must be entanglement breaking~\cite{Horodecki2003a}. 

Now, we will use these ontological representations to show that we can also construct an ontological representation of the comb $\cal C$. 

First, we will rewrite the comb $\mc{C}^{A|B}$ within quantum theory as a linear combinations of factorizing combs as shown in \cref{Globalcomb}; then, we will use the above representations of the states and effects appearing in the factorizing comb to construct a quasiprobabilistic representation of $\mc{C}^{A|B}$. This representation is then proven to be nonnegative and therefore to define a valid ontological representation. The detailed steps are as follows. 

Specifically, it has been shown~\cite{Chiribella2009} that any comb $\mc{C}^{A|B}$ can be decomposed as a bipartite state $\rho^{AS}$ followed by a bipartite measurement (effect) $M^{BS}$ with an identity channel acting on system $S$. That is
\begin{align}
\mc{C}^{A|B}(\mc{E}^{B|A}):&=\tr[M^{BS}  (\mc{E}^{B|A} \otimes \mbb{1}^{S|S} ) \rho^{AS}].
\end{align}
Furthermore, any bipartite quantum state and bipartite measurement can be decomposed as 
\begin{subequations}
\begin{align}
&\rho^{AS}=\sum_{tt'} r_{tt'}' \rho_{t}^A\otimes \rho^S_{t'}\\
&M^{BS}=\sum_{kk'} r_{kk'}^{''}M_{k}^{B}\otimes M_{k'}^{S} 
\end{align}
\end{subequations}
with some real coefficients $\{r_{tt'}'\}_{tt'}$ and $\{r''_{kk'}\}_{kk'}$. As such, the action of the comb $\mc{C}^{A|B}$ on the channel ${\cal E}^{B|A}$ can be further decomposed as
\begin{align}
&\mc{C}^{A|B}(\mc{E}^{B|A})\coloneqq \\
&\tr[(\sum_{kk'} r_{kk'}^{''}M_{k}^{B}\otimes M_{k'}^{S})  (\mc{E}^{B|A} \otimes \mathbb{1} ) (\sum_{tt'} r_{tt'}' \rho_{t}^A\otimes \rho^S_{t'})]. \nonumber
\end{align}
Defining the new real coefficients
\begin{align}
    r_{tk}\coloneqq\sum_{tt'kk'}r_{tt'}'r''_{kk'}\tr[ M_{k'}^{S} \rho^S_{t'}],
\end{align}
the expression for the comb $\mc{C}^{A|B}$ acting on $\mc{E}^{B|A}$ can be further simplified to
\begin{align}
\mc{C}^{A|B}(\mc{E}^{B|A}) = \sum_{tk} r_{tk}\tr[M_{k}^{B}\mc{E}^{B|A}(\rho_{t}^A)]. 
\end{align}
This final decomposition of the comb $\mc{C}^{A|B}$ is shown on the bottom left diagram of \cref{Globalcomb}.

Now, given our choice to represent states and effects as in Eq.~\eqref{eq:combrep}, the unique linear and diagram-preserving representation of the comb $\mc{C}^{A|B}$ can be defined as
\begin{align}
    \tilde{p}_{\mc{C}}(\lambda_A|\lambda_B)\coloneqq&\sum_{tk}r_{tk} p(\lambda_A|\rho_{t}^A)p(M_{k}^{B}|\lambda_B)\notag \\ 
    =& \sum_{tk}r_{tk} \tr[G_{\lambda_{A}}^{A}\rho_{t}^A]\tr[M_{k}^{B}\sigma^{B}_{\lambda_B}]\notag \\
    =&\tr[G_{\lambda_{A}}^{A}\mc{C}^{A|B}(\sigma^{B}_{\lambda_B})].
    \label{combOM}
\end{align}
Furthermore, since $G_{\lambda_A}^{A}$ are valid effect and $\sigma_{\lambda
}^{B}$ are valid quantum states, $\tilde{p}_{\mc{C}}(\lambda_A|\lambda_B)\in[0,1]$ defines a valid substochastic representation.

Finally, it is straightforward to verify that these representations reproduce the statistics:
\begin{align}
 \sum_{\lambda_A\lambda_B}&\tilde{p}_{\mc{C}}(\lambda_A|\lambda_B)p_{\mc{E}}(\lambda_B|\lambda_A) \notag \\
 =& \sum_{\lambda_A\lambda_B}p_{\mc{E}}(\lambda_B|\lambda_A)\tr[G_{\lambda_{A}}^{A}\mc{C}^{A|B}(\sigma^{B}_{\lambda_B})] \notag \\
=&\mc{E}^{B|A}\circ\mc{C}^{A|B}. 
\label{eq:op-multi}
\end{align}
where in the second equality, we used the Holevo form of the channel ${\cal E}_{B|A}$ from \cref{eq:holevo form} (also see Fig.~\ref{Combonto}).

\begin{figure}[htb!]
\centering
\includegraphics[width=0.5\textwidth]{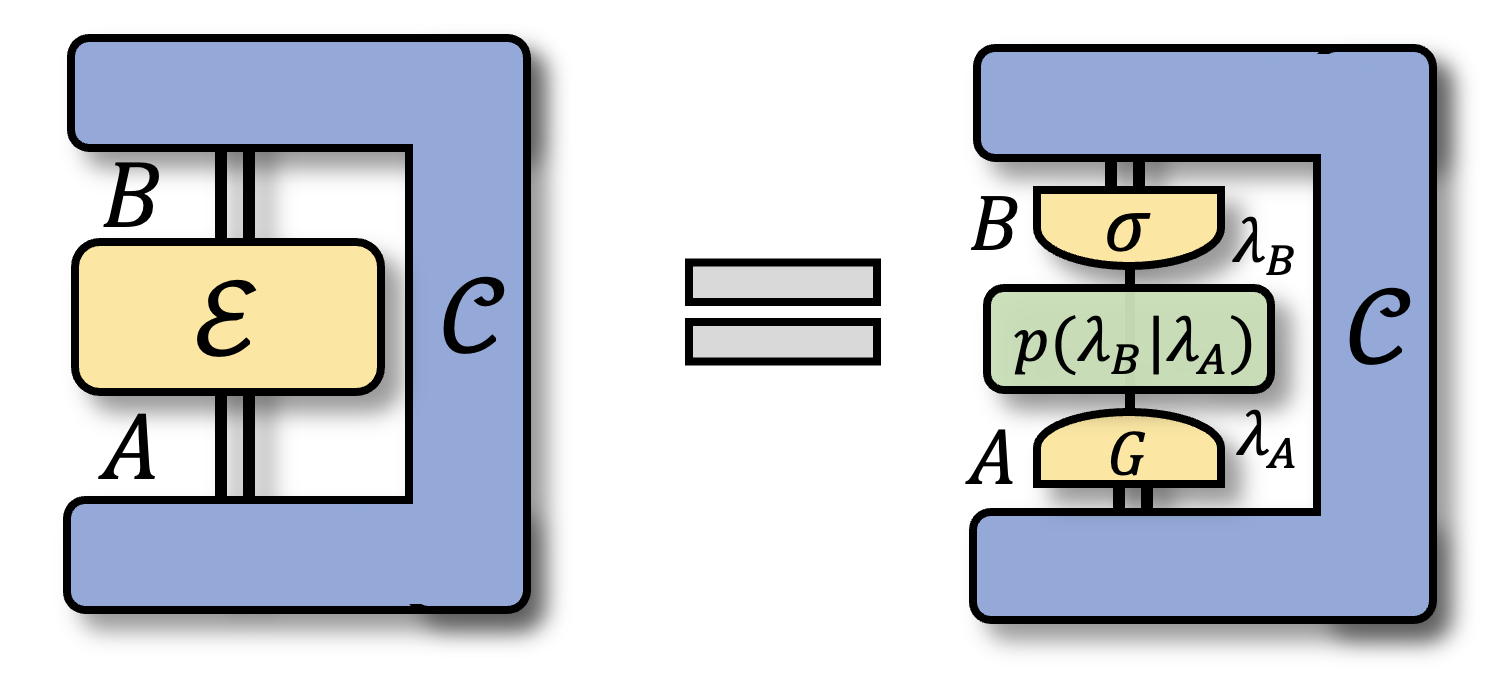}
\caption{An entanglement-breaking channel and a dual.}
\label{Combonto}
\end{figure}

Thus, $\tilde{p}_{\mc{C}}(\lambda_A|\lambda_B)$ is a valid ontological representation for a generic quantum comb $\mc{C}^{A|B}$. That is, we have seen that whenever the statistics generated by ${\cal E}^{B|A}$ with all factorizing pairs of effect and state are classically explainable, one can write down such a representation for any comb $\mc{C}^{A|B}$ (possibly nonfactorizing). It follows that the statistics generated by ${\cal E}^{B|A}$ with any comb must also be classically explainable. This concludes the if direction of the proof. 
\end{proof}

\subsection{Factorizing duals for arbitrary process}
\setcounter{theorem}{0}

With the two special cases discussed in the previous subsections, we now prove our main theorem analogously for the most general quantum process-- a multipartite multi-instrument. 
\begin{theorem}
A quantum process is such that there exists some set of dual processes that it can be contracted with to obtain statistics that are not classically explainable, if and only if there exists some set of {\em factorizing} dual processes for which this is the case. 
\end{theorem}
\begin{proof}
The core idea of the proof is similar to the ones for the ``if'' direction in Lemmas \ref{lem:bi_global} and \ref{lem:comb_global}. That is, if one can construct an ontological model for the given quantum process contracted with every possible factorizing dual process, then one can obtain an ontological representation for every dual process (possibly nonfactorizing) by building it from the representations of the factorizing dual processes in a linear, diagram-preserving manner.  

In this paper, the most general process we consider is a multipartite multi-instrument $\{\{\mc{F}_{c|z}^{\vec{B}|\vec{A}}\}_c\}_z$: $\mc{L}(\bigotimes_{i=1}^n\mc{H}^{A_i})\mapsto \mc{L}(\bigotimes_{j=1}^m\mc{H}^{B_{j}})$, i.e., each  $\mc{F}_{c|z}^{\vec{B}|\vec{A}}$ is a CP trace-nonincreasing map such that $\sum_c\mc{F}_{c|z}^{\vec{B}|\vec{A}}$ defines quantum channel. \par

If the statistics generated by a multipartite multi-instrument $\{\{\mc{F}_{c|z}^{\vec{B}|\vec{A}}\}_c\}_z$ is classical for factorizing duals composed of the set of states $\{\rho^{A_i}\}$ and the set of effects $\{M^{B_j}\}$, then there necessarily exists an ontological model for the operational statistics such that
\begin{align}
\tr&\left[\left(\bigotimes_{j=1}^m M^{B_j}\right)\mc{F}_{c|z}^{\vec{B}|\vec{A}} \left(\bigotimes_{i=1}^n\rho^{A_i}\right)\right] \label{eq:fdual-general} \\
=&\sum_{\lambda_{\vec{B}}\lambda_{\vec{A}}} p(c\lambda_{\vec{B}}|z\lambda_{\vec{A}})\prod_{j}p(M^{B_j}|\lambda_{B_j})\prod_{i}p(\lambda_{A_i}|\rho^{A_i}),\notag \\
=&\sum_{\lambda_{\vec{B}}\lambda_{\vec{A}}} p(c\lambda_{\vec{B}}|z\lambda_{\vec{A}})\prod_{j}\tr[M^{B_j}\sigma^{B_j}_{\lambda_{B_j}}]\prod_{i}\tr[G^{A_i}_{\lambda_{A_i}}\rho^{A_i}],\notag 
\end{align}
where in the second equality we write $p(M^{B_j}|\lambda_{B_j})=\tr[M^{B_j}\sigma^{B_j}_{\lambda_{B_j}}]$ and $p(\lambda_{A_i}|\rho^{A_i})=\tr[G^{A_i}_{\lambda_{A_i}}\rho^{A_i}]$ using Gleason's theorem and Riesz representation theorem since these ontological processes depend linearly on all quantum effects $M^{B_j}$ and all quantum states $\rho^{A_i}$, respectively. Because the set of all product effects and states is informationally complete for characterizing the multipartite multi-instrument, it follows that
\begin{align} \mc{F}_{c|z}^{\vec{B}|\vec{A}} (\cdot)=\sum_{\lambda_{\vec{B}}\lambda_{\vec{A}}} p(c\lambda_{\vec{B}}|z\lambda_{\vec{A}})\bigotimes_{j}\sigma^{B_j}_{\lambda_{B_j}}\tr[\bigotimes_{i}G^{A_i}_{\lambda_{A_i}}(\cdot)]
\end{align}

Similar to the construction of ontological models for general dual processes in Lemmas \ref{lem:bi_global} and \ref{lem:comb_global}, we can now proceed to define the ontological model for a general dual process $\mc{C}^{\vec{A}|\vec{B}}$: $\mc{L}(\bigotimes_{i=1}^n\mc{H}^{A_i})\mapsto \mc{L}(\bigotimes_{j=1}^m\mc{H}^{B_{j}})$ for the target process $\{\{\mc{F}_{c|z}^{\vec{B}|\vec{A}}\}_c\}_z$.  To do so, we will need to first find a decomposition of the general dual process $\mc{C}^{\vec{A}|\vec{B}}$ (a multipartite comb) in terms of product states and product effects as we did in \cref{Globaleffect} and~\ref{Globalcomb}. Such a decomposition always exists because quantum theory satisfies the principle of tomographic locality. That is, for any channel $\mc{E}^{\vec{B}|\vec{A}}$, we have
\begin{align}
\mc{C}^{\vec{A}|\vec{B}}\circ \mc{E}^{\vec{B}|\vec{A}}\coloneqq\sum_{\vec{t}\vec{k}}r_{_{\vec{t}\vec{k}}}\tr[\bigotimes_j M_{k_j}^{B_j}\mc{E}^{\vec{B}|\vec{A}}(\bigotimes_i \rho_{t_i}^{A_i})]
\end{align}
for some real coefficients $\{r_{_{\vec{t}\vec{k}}}\}$. Now, given the ontological representation we have for states and effects in Eq.~\eqref{eq:fdual-general}, we can construct the ontological representations for the general dual process $\mc{C}^{\vec{A}|\vec{B}}$ as:
\begin{align}
\tilde{p}_{\mc{C}}(\lambda_{\vec{A}}|\lambda_{\vec{B}}):&=\sum_{\vec{t}\vec{k}}r_{_{\vec{t}\vec{k}}}\prod_{j}p(M_{k_j}^{B_j}|\lambda_{B_j})\prod_{i}p(\lambda_{A_i}|\rho_{t_i}^{A_i})\notag \\
&=\sum_{\vec{t}\vec{k}}r_{_{\vec{t}\vec{k}}}\prod_{j}\tr[M^{B_j}\sigma^{B_j}_{\lambda_{B_j}}]\prod_{i}\tr[G^{A_i}_{\lambda_{A_i}}\rho^{A_i}],\notag \\
&=\tr[\bigotimes_iG^{A_i}_{\lambda_{A_i}}\mc{C}^{\vec{A}|\vec{B}}(\bigotimes_i\sigma^{B_j}_{\lambda_{B_j}})].
\end{align}
Since $G^{A_i}_{\lambda_{A_i}}$ and $\sigma^{B_j}_{\lambda_{B_j}}$ are valid effects and states, we have $p_{\mc{C}}(\lambda_{\vec{A}}|\lambda_{\vec{B}})\in[0,1]$ defining a substochastic representation for any valid dual process $\mc{C}^{\vec{A}|\vec{B}}$. 

Therefore,  once one has constructed an ontological representation of a multipartite multi-instrument along with all factorizing combs dual to it,  one can uniquely extend this ontological representation to all combs (including those that are nonfactorizing).
\end{proof}

\section{Proof of Theorem~\ref{theoremprep}}\label{proofthmprep}

We now repeat and then prove Theorem~\ref{theoremprep}.

\setcounter{theorem}{1}
\begin{theorem}\label{theoremprep_app}
For a multi-source $\msf{P}\coloneqq\{\{p(a|x)\rho_{a|x}\}_{a}\}_x$, the following statements are equivalent:\\
(0) It is classical (in the sense of Definition~\ref{def_multi-sourceNC}).\\
(1) Each of the unnormalized states 
$p(a|x)\rho_{a|x}$ can be decomposed as 
    \begin{align}
    \label{eq_rhodecom_app} {p(a|x)\rho_{a|x} = \sum_{\lambda} p(a\lambda|x)\sigma_{\lambda} \quad \forall a,x}
     \end{align}
    for some fixed set of normalized states $\{\sigma_{\lambda}\}_{\lambda}$  and some conditional probability distribution $p(a\lambda|x)$ satisfying 
    \begin{align}
    \label{eq_paxOp_app}
&{\sum_{ax} \alpha_{ax} p(a \lambda|x)=0 \qquad     \forall \{\alpha_{ax}\}_{ax}\in\mc{O}(\msf{P})}
    \end{align}
for all $\lambda$, where $\mc{O}(\msf{P})$ is defined in Eq.~\eqref{op:prep}. 
\\(2) Each of the 
normalized states $\rho_{a|x}$ can be decomposed within a frame representation on $V=\text{Span}(\{\rho_{a|x}\}_{ax})$ as
\begin{align}
        \label{eq_rhodecFrame_app}
        \rho_{a|x} = \sum_{\lambda} \tr[H_{\lambda} \rho_{a|x}]\sigma_{\lambda}, \quad \forall a,x
\end{align}
where the frame is a fixed set of density operators 
$\{\sigma_{\lambda}\}_{\lambda}$ and the dual frame is a set of (possibly nonpositive) Hermitian operators, i.e., $\{H_{\lambda}\}_{\lambda}\subset \text{Herm}(\mc H)$ satisfying $\sum_{\lambda}H_{\lambda}=\mathcal{P}_V(\mbb{1})$ and $\tr[H_{\lambda} \rho_{a|x}]\in [0,1]$ for all $a,x,\lambda$, with $\mathcal{P}_V$ being the superoperator that projects onto the operator space $V$. This condition is depicted in \cref{framesource}.
\end{theorem}

\begin{proof}
We first prove that condition (0) implies $(2)$. From \cref{def_multi-sourceNC}, we know that the  statistics generated by a classical multi-source $\{\{p(a|x)\rho_{a|x}\}_a\}_x$ with any quantum effect can be expressed as 
\begin{align*}
\tr[M\rho_{a|x}]=\sum_{\lambda}p(M|\lambda)p(\lambda|a,x)~~~\forall a,x,M\in\mc{M}(\mc{H})
\end{align*}
where $\{p(\lambda|a,x)\}_{ax}$ depends linearly on  $\{\rho_{a|x}\}_{ax}$  and  $p(M|\lambda)$ depends linearly on all quantum effects $M\in\mc{M}(\mc{H})$. \par 

Since $p(M|\lambda)\in[0,1]$ with $p(\mbb{1}|\lambda)=1$ for unit effect $\mbb{1}$, we can use the generalized Gleason's Theorem~\cite{Busch2003} to infer that $p(M|\lambda)=\tr{[M\sigma_{\lambda}]}$ for a unique normalized density matrix $\sigma_{\lambda}$. Similarly, for any linear map from states $\{\rho_{a|x}\}_{ax}$ to probability distributions $\{p(\lambda|a,x)\}_{ax}$, there always exists a Hermitian operator $H_{\lambda}$ for each $\lambda$ such that $p(\lambda|a,x)=\tr [H_{\lambda}\rho_{a|x}]$.
Given that $p(\lambda|a,x)\in[0,1]$ and $\sum_{\lambda}p(\lambda|a,x)=1$ for all $\rho_{a|x}$, we infer that $\sum_{\lambda}H_{\lambda}=\mc{P}_V(\mbb{1})$ where $\mc{P}_V$ is the superoperator that projects $\mbb{1}$ onto the operator space $V=\text{Span}(\{\rho_{a|x}\}_{a.x})$.  We thus have, for any given $a$ and $x$,
\begin{align}
    \tr[M\rho_{a|x}] &=
\sum_{\lambda}\tr{[M\sigma_{\lambda} ]}\tr [H_{\lambda}\rho_{a|x}] \\
&= \tr\big[M\sum_{\lambda}\tr [H_{\lambda}\rho_{a|x}]\sigma_{\lambda} \big]
, \quad \forall M, \nonumber
\label{eq:allM}
\end{align}
Because \cref{eq:allM} must hold for all $M$, it follows that
\begin{align}
    \rho_{a|x} &=\sum_{\lambda}\tr [H_{\lambda}\rho_{a|x}]\sigma_{\lambda}.
\end{align}

To show $(2)\rightarrow (1)$,  we can simply define
\begin{align}
    &p(\lambda|a,x) \coloneqq \tr[H_{\lambda}\rho_{a|x}].\notag \\
   \Rightarrow &p(a\lambda|x) \coloneqq p(a|x)p(\lambda|a,x)=\tr[H_{\lambda}p(a|x)\rho_{a|x}]
\end{align}
It follows that $\{p(a\lambda|x)\}_{ax}$ depends linearly on $\{p(a|x)\rho_{a|x}\}_{ax}$ and consequently respects the operational identities $\mc{O}(\msf{P})$.

To show that condition $(1)$ implies condition $(0)$, we can define $p(M|\lambda)\coloneqq\tr[\sigma_{\lambda}M]$ and thus we have:
\begin{align}
    &\tr[M \rho_{a|x}]
= \sum_{\lambda}\tr{[M\sigma_{\lambda} ]} p(\lambda|a,x), \quad \forall M, a,x,\\
    &p(a|x)\tr[M\rho_{a|x}]
= \sum_{\lambda}\tr{[M\sigma_{\lambda} ]} p(a\lambda|x), \quad \forall M, a,x,  \notag 
\end{align}
where $\{p(a\lambda|x)\}_{ax}$ respects $\mc{O}(\msf{P})$ and $p(M|\lambda)$ defines a stochastic map that is a linear functional on the set of all quantum effects $M\in\mc{M}(\mc{H})$. Thus, there exists an ontological model for the statistics generated by the multi-source contracted with any quantum effect;
it follows that $\{\{p(a|x)\rho_{a|x}\}_a\}_x$ is classical. 
\end{proof}

\section{Proof of Theorem~\ref{theoremmeas}}\label{proofthmmes}

We now repeat and then prove Theorem~\ref{theoremmeas}.

\begin{theorem}\label{theoremmeas_app}
For a multi-measurement $\msf{M}=\{\{M_{b|y}\}_b\}_y$, the following statements are equivalent:\\
(0) It is classical (in the sense of Definition~\ref{def_mtmtNC})\\
(1) Its effects can be decomposed as
    \begin{equation}
    \label{eq_Edecom_app}
        M_{b|y}=\sum_{\lambda} p(b|y\lambda)G_{\lambda} \quad \forall b,y
    \end{equation}
for some POVM $\{G_{\lambda}\}_{\lambda}$ and some conditional probability distribution satisfying 
    \begin{align}
    \label{eq_pbyOp_app}
\sum_{by} \beta_{by} p(b|y\lambda)=0 \qquad     \forall \{\beta_{by}\}_{by}\in\mc{O}(\msf{M}).
    \end{align}
for all $\lambda$, where $\mc{O}(\msf{M})$  is defined in Eq.~\eqref{op:meas}.    \\
(2) Its effects can be decomposed within a frame representation on $V=\text{Span}(\{M_{b|y}\})$ as
    \begin{equation}
    \label{eq_MdecFrame_app}
        M_{b|y}=\sum_{\lambda} \tr[M_{b|y}F_{\lambda}] G_{\lambda} \quad \forall b,y
    \end{equation}
where $\{G_{\lambda}\}_{\lambda}$ is a POVM with $\sum_{\lambda}G_{\lambda}=\mc{P}_{V}(\mbb{1})=\mbb{1}$  and $\{F_{\lambda}\}_{\lambda}\subset\text{Herm}(\mc H)$ is a set of Hermitian operators satisfying $\tr[F_{\lambda}]=1$ and $\tr[M_{b|y}F_{\lambda}]\in [0,1]$ for all $b,y,\lambda$. This condition is depicted in \cref{framemeas}.
\end{theorem}

\begin{proof}
We first prove that condition $(0)$ implies condition $(2)$. From \cref{def_mtmtNC}, we know that the statistics generated by a classical multi-measurement $\{\{M_{b|y}\}_b\}_y$ contracted with every possible quantum state can be expressed as 
\begin{align}
\tr[M_{b|y}\rho]=\sum_{\lambda} p(b|y\lambda)p(\lambda|\rho)~~~~\forall b,y, \forall\rho\in \mc{S}(\mc{H})
\end{align}
where $\{p(b|y\lambda)\}_{by}$ depends linearly on $\{M_{b|y}\}_{by}$ (and so respects $\mc{O}(\msf{M})$) and $p(\lambda|\rho)$ depends linearly on the quantum state $\rho\in \mc{S}(\mc{H})$. 
\par 
Using the  Riesz representation theorem, we can write this latter linear function as $p(\lambda|\rho)=\tr{[G_{\lambda}\rho]}$ for some effect $G_{\lambda}\in \mc{M}(\mc H)$. Furthermore, since $\sum_{\lambda}p(\lambda|\rho)=1$, we must have that $\sum_{\lambda}G_{\lambda}=\mc{P}_{V}(\mbb{1})=\mbb{1}$ . Thus, $\{G_{\lambda}\}_{\lambda}$ is a POVM. Similarly, we can always write $p(b|y\lambda)=\tr [M_{b|y}F_{\lambda}]$ for some Hermitian operator $F_{\lambda}$. The fact that $\sum_{b}p(b|y\lambda) =1$ implies that $\tr[F_{\lambda}]=1$ . Thus, we have:
\begin{align}
\tr[M_{b|y}\rho ]&=\sum_{\lambda} \tr[M_{b|y}F_{\lambda}]\tr[G_{\lambda}\rho] \nonumber \\
&=\tr\big[\sum_{\lambda} \tr[M_{b|y}F_{\lambda}] G_{\lambda}\rho\big]
\quad \forall \rho \in \mc{S}(\mc{H})
\label{eq:allrho}
\end{align}
where $\tr[M_{b|y}F_{\lambda}]=p(b|y\lambda)\in [0,1]$ (as $p(b|y\lambda)\in[0,1]$) and $\tr[F_{\lambda}]=1$ (as $\sum_{b}p(b|y\lambda) =1$) for all $b,y\lambda$.  Since \cref{eq:allrho} must hold for all $\rho$, it follows that 
\begin{align}
    M_{b|y}=\sum_{\lambda} \tr[M_{b|y}F_{\lambda}]G_{\lambda}.
\end{align}

To show $(2)\rightarrow (1)$,  we can simply define 
\begin{align}
    p(b|y\lambda) \coloneqq \tr[M_{b|y}F_{\lambda}], \notag 
\end{align}
It follows that $p(b|y\lambda)$ depends linearly on $M_{b|y}$, which implies that it respects the operational identities in the set ${\cal O}_{\msf{M}}$. Furthermore, since $\tr[M_{b|y}F_{\lambda}]\in[0,1]$ and $\sum_b \tr[M_{b|y}F_{\lambda}]=1$, we have that the $p(b|y\lambda)$s are conditional probabilities.

To show that condition $(1)$ implies condition $(0)$,  we can define $p(\lambda|\rho)\coloneqq\tr[G_{\lambda} \rho ]$, i.e., define a stochastic map that is linear on all quantum states $\rho$. Thus $\forall b,y,\rho\in \mc{S}(\mc{H})$, we have:
\begin{align}
    \tr[M_{b|y} \rho]=  \sum_{\lambda}p(b|y\lambda)\tr{[G_{\lambda} \rho ]}
\end{align}
where $\{p(b|y\lambda)\}_{by}$ respects $\mc{O}(\msf{M})$, and $p(\lambda|\rho)$ respects $\mc O_{\msf{P}}$. Thus, the statistics of the set of circuits obtained by contracting the multi-measurement with all possible quantum states are classically explainable. Hence the multi-measurement$\{\{M_{b|y}\}_b\}_y$ is classical. 
\end{proof}

\section{Proof of Corollary~\ref{simplexcor}} \label{simpproof}
We now repeat and then prove Corollary~\ref{simplexcor}.
\setcounter{corollary}{2}
\begin{corollary} 
If a multi-state consists of a  set of states that fits inside a simplex within the quantum state space, then it is classical.
\end{corollary}

\begin{proof}
Assume the existence of such a simplex and denote its vertices by $\{\sigma_{\lambda}\}_{\lambda}$ (by definition, they are affinely independent trace-1 operators, thus linearly independent). As each element of the multi-state $\{\rho_{\cdot|x}\}_x$ is in the convex hull of the $\{\sigma_{\lambda}\}_{\lambda}$, it follows that the span of the former is contained in the span of the latter. Thus the linearly independent set $\{\sigma_{\lambda}\}_{\lambda}$ forms a basis for the span of the $\{\rho_{\cdot|x}\}$, and so there exists a dual basis $\{H_{\lambda}\}_{\lambda}$ such that (recalling Eq.~\eqref{eq:statereconstr})
\begin{align}
        \rho_{\cdot|x} = \sum_{\lambda}\tr[H_{\lambda}\rho_{\cdot|x}]\sigma_{\lambda}.
\end{align}
Thus, we  can define an ontological representation of $\rho_{a|x}$ by
\begin{align} 
p(\lambda|x)\coloneqq& \tr[H_{\lambda}\rho_{\cdot|x}]. 
\end{align}
Given that this is  a linear functional on the set of quantum states, it follows that the set $\{p(\lambda|x)\}_x$ satisfies all the operational identities in $\mc{O}(\msf{P})$ defined by $\{\rho_{\cdot|x}\}_x$. By condition 1 of Theorem~\ref{theoremprep}, one concludes that the multi-state $\{\rho_{\cdot|x}\}_x$ is classical.
\end{proof}

\section{Proof of \cref{purestatescorollary}}
\label{app:purestatescorollary}
\setcounter{corollary}{3}
\begin{corollary}
A multi-state consisting of a set of distinct \emph{pure} states on a single system is classical if and only if the states are linearly independent. 

\end{corollary}

\begin{proof}
The proof for the \emph{if} direction follows directly from \cref{lem:linearinde}, since a set of linearly independent states forms a simplex. 
The proof for the \emph{only if} direction can be shown by contradiction. Suppose the states in the set $\{\rho_{\cdot|x}\}_x$ constituting the multi-state are linearly dependent, then we must have at least one equation of the following form:
\begin{align}
    \sum_x \alpha_x \rho_{\cdot|x} =0,
\end{align}
where $\alpha_x \neq 0$ for some $x$. If the set is classical, then there exists a set of states $\{\sigma_{\lambda}\}_{\lambda}$ such that $\rho_{\cdot|x}$ can be expressed as a convex mixture of the form $\rho_{\cdot|x}=\sum_{\lambda}p(\lambda|x)\sigma_{\lambda}$ where
\begin{align}
    \sum_x\alpha_x p(\lambda|x)=0.
    \label{eq:alpha}
\end{align}
However, since the $\rho_{\cdot|x}$s are assumed to be pure, the convex mixture must be trivial, that is, one must have $p(\lambda|x)=\delta_{\lambda, x}$ and $\sigma_{x}=\rho_{\cdot|x}^A$. Then, \cref{eq:alpha} reduces to
\begin{align}
    \sum_x \alpha_x \delta_{\lambda, x}=0,
\end{align}
implying that $\alpha_x=0$ for all $x$, contradicting our earlier assumption of linear dependence. Hence $\{\rho_{\cdot|x}\}_x$ must be a linearly independent set. 
\end{proof}

\section{Proof of Proposition~\ref{prop:bipartite}}\label{steerproof}

We now repeat and then prove Proposition~\ref{prop:bipartite}. We note that because the proof is relatively simple, and because the proof of Proposition~\ref {prop:channel} and Proposition~\ref {prop:general} are analogous, we do not provide the details of the latter, but rather trust that the reader will see how to generalize the proof of the Proposition~\ref{prop:bipartite} to proofs of Proposition~\ref {prop:channel} and Proposition~\ref {prop:general}.

\setcounter{proposition}{0}
\begin{proposition}
Any circuit  with local measurements on a bipartite state (as depicted in Fig.~\ref{fig:steer}(a)) is classically explainable only if the two prepare-measure circuits  it induces by boxing to obtain a steering scenario (in one or the other direction, as in Fig.~\ref{fig:steer}(b)) are classically explainable.
\end{proposition}

\begin{proof}
If the statistics generated by a bipartite state $\rho^{A_1A_2}$ with any multi-measurements $\{\{M_{b_1|y_1}^{A_1}\}_{b_1}\}_{y_1}$ and $\{\{M_{b_2|y_2}^{A_2}\}_{b_2}\}_{y_2}$ are classically explainable, then they can be expressed as
\begin{align}
&  p(b_1b_2|y_1y_2)=\sum_{\lambda_1\lambda_2} p(\lambda_1\lambda_2)p(b_1|y_1\lambda_1)  p(b_2|y_2\lambda_2), 
\end{align}
where $p(b_1|y_1\lambda_1)$ respects $\mc{O}(\msf{M}_1)$, $p(b_2|y_2\lambda_2)$ respects $\mc{O}(\msf{M}_2)$, the composite
$\sum_{\lambda_1} p(b_1|y_1\lambda_1)p(\lambda_1\lambda_2)$ respects $\mc{O}(\msf{M}_1\circ \msf{P})$, the composite $\sum_{\lambda_2} p(b_2|y_2\lambda_2)p(\lambda_1\lambda_2)$ respects $\mc{O}(\msf{M}_2\circ \msf{P})$, and $p(\lambda_1\lambda_2)$ respects $\mc{O}({ \msf{\tilde{P}}})$ as in \cref{op:bipartiterep} in the main text.

The fact that $\lambda_1$ screens-off the dependence of $b_1$ to $\lambda_2$ in the bipartite circuit in Fig.~\ref{fig:steer}(a) allows us to define 
\begin{align}
   \tilde{p}(b_1 \lambda_2|y_1) \coloneqq \sum_{\lambda_1} p(b_1|y_1\lambda_1)p(\lambda_1\lambda_2),
\end{align}
we thus have that
\begin{align}
p(b_1b_2|y_1y_2)=\sum_{\lambda_2}p(b_2|y_2\lambda_2) \tilde{p}(b_1\lambda_2|y_1),
\end{align}
which recovers the statistics of a prepare-measure  scenario involving the  following pair of processes: (1) the preparation $\{\{\tilde{\rho}^{A_2}_{b_1|y_1}\}_{b_1}\}_{y_1}$ on $\mc{H}^{A_2}$, where
\begin{align}\label{eq:B2state}
&\tilde{\rho}^{A_2}_{b_1|y_1}=\tr_{A_1}[(M_{b_1|y_1}^{A_1}\otimes\mbb{1}^{A_2})\rho^{A_1A_2}], 
\end{align}
and (2) the measurement $\{\{M_{b_2|y_2}^{A_2}\}_{b_2}\}_{y_2}$ on $\mc{H}^{A_2}$. Evidently, $\tilde{p}(b_1 \lambda_2|y_1)$ respects the identities satisfied by $\{\{\tilde{\rho}_{b_1|y_1}\}_{b_1}\}_{y_1}$, namely $\mc{O}(\msf{M}_1\circ \msf{P})$, and $p(b_2|y_2\lambda_2)$ respects the identities satisfied by the multi-measurement, namely $\mc{O}(\msf{M}_2)$. Thus, such prepare-measure statistics are classically explainable. Using analogous arguments, we infer that the statistics of the prepare-measure scenario on $A_1$ is also classically explainable.
\end{proof}

\section{Proof of frame decomposition}
\subsection{Proof of Theorem~\ref{ncbp}}
\label{appendixK}

\setcounter{theorem}{3}
\begin{theorem}\label{ncbp_app}
A bipartite state $\rho^{A_1A_2}\in\mc S(\mc H^{A_1}\otimes\mc H^{A_2})$ is classical if and only if it admits a frame decomposition:
\begin{equation}\label{eq:bipartiteframe_app}
\rho^{A_1A_2}
=\sum_{\lambda_1,\lambda_2}
\tr\big[(H^{A_1}_{\lambda_1}\otimes H^{A_2}_{\lambda_2}) \rho^{A_1A_2}\big]
\sigma^{A_1}_{\lambda_1}\otimes\sigma^{A_2}_{\lambda_2},
\end{equation}
where the frame on the operator space $V_i$ is a fixed set of density operators $\{\sigma^{A_i}_{\lambda_i}\}_{\lambda_i}$, and the dual frame is a set of Hermitian operators $\{H^{A_i}_{\lambda_i}\}_{\lambda_i}$ 
satisfying $\sum_{\lambda_{i}}H^{A_i}_{\lambda_i}=\mc{P}_{V_i}(\mbb 1^{A_i})$ with $\mc{P}_{V_i}$ denoting the superoperator projecting onto $V_i$, and such that
\begin{equation}\label{eq:constraints_app}
\tr\big[(H^{A_1}_{\lambda_1}\otimes H^{A_2}_{\lambda_2}) \rho^{A_1A_2}\big]\ge 0
\quad \forall \lambda_1,\lambda_2.
\end{equation} 
\end{theorem}
\begin{proof}

We first prove the \enquote{only if} direction. Assume $\rho^{A_1A_2}$ is classical in the sense of Definition~\ref{def:bipartite} in the main text. Then, for the statistics generated by the circuits consisting of the bipartite state $\rho^{A_1A_2}$ on $\mc{H}^{A_1A_2}$ and every possible multi-measurement $\{M^{A_1}_{b_1|y_1}\}_{b_1,y_1}$ on $\mc{H}^{A_1}$ and $\{M^{A_2}_{b_2|y_2}\}_{b_2,y_2}$ on $\mc{H}^{A_2}$, there exists an ontological model:
\begin{subequations}
\begin{align}
    \tr[&(M^{A_1}_{b_1|y_1}\otimes M^{A_2}_{b_2|y_2})\rho^{A_1A_2}]\notag \\    &=\sum_{\lambda_1\lambda_2}p(\lambda_1\lambda_2)p(b_1|y_1\lambda_1)p(b_2|y_2\lambda_2),  \label{eq:bp-onto}  \\&=\sum_{\lambda_1}p(b_1|y_1\lambda_1)\tilde{p}(b_2\lambda_1|y_2), \label{eq:bp-onto1} 
    \\&=\sum_{\lambda_2}p(b_2|y_2\lambda_2)\tilde{p}(b_1\lambda_2|y_1 ), \label{eq:bp-onto2}  
\end{align}
\end{subequations}
where $\tilde{p}(b_2\lambda_1|y_2)$ and $\tilde{p}(b_1\lambda_2|y_1 )$ are defined in \cref{eq:steeronto} and where the ontological identities in \cref{op:bipartiterep} are satisfied. In particular, these constraints imply that:

(1) From \cref{ncmeas1}, for each $\lambda_1$, the map $M^{A_1}_{b_1|y_1}\mapsto p(b_1|y_1\lambda_1)$
is a normalized positive linear functional on the space of all effects on $A_1$, i.e., $\mc{M}(\mc H^{A_1})$. {The normalization $p(\mbb{1}^{A_1}|\lambda_1)=1$ follows from the fact that $\tr[(\mbb{1}^{A_1}\otimes \mbb{1}^{A_2})\rho^{A_1A_2}]=1$, and by \cref{eq:bp-onto}, $p(\mbb{1}^{A_i}|\lambda_i)$ must equal $1$.} Therefore, by the generalized Gleason theorem, there exists a valid quantum state $\sigma^{A_1}_{\lambda_1}$ such that \begin{equation}
p(b_1|y_1\lambda_1)=\tr[M^{A_1}_{b_1|y_1}\sigma^{A_1}_{\lambda_1}].  \label{eq:frbp1}
\end{equation}

(2) From \cref{ncmeas2}, for each $\lambda_2$, the map $M^{A_2}_{b_2|y_2}\mapsto p(b_2|y_2\lambda_2)$
is a normalized positive linear functional on the space of all effects on $A_2$, i.e., $\mc{M}(\mc H^{A_2})$. Therefore, by the generalized Gleason theorem, there exists a valid quantum state $\sigma^{A_2}_{\lambda_2}$ such that: 
\begin{equation}
p(b_2|y_2\lambda_2)=\tr[M^{A_2}_{b_2|y_2}\sigma^{A_2}_{\lambda_2}].\label{eq:frbp2}
\end{equation}

(3) From \cref{ncprep1}, for each $\lambda_2$, the map $\tilde\rho^{A_2}_{b_1|y_1}\mapsto \tilde p(b_1\lambda_2|y_1)$ (where $\tilde\rho^{A_2}_{b_1|y_1}$ is the steered state on $A_2$, defined in \cref{eq:steer A}) is a positive linear functional on the operator space $V_2:=\text{Span}(\{\tilde\rho^{A_2}_{b_1|y_1}\})$, 
hence there exists a Hermitian operator $H^{A_2}_{\lambda_2}$ on $\mc{H}^{A_2}$ such that:
\begin{align}
\tilde{p}(b_1\lambda_2|y_1) =\tr[H_{\lambda_2}^{A_2}\tilde{\rho}^{A_2}_{b_1|y_1}]. \label{eq:frbp3}
   \end{align}

(4) From \cref{ncprep2}, for each $\lambda_1$ the map $\tilde{\rho}^{A_1}_{b_2|y_2}\mapsto \tilde p(b_2\lambda_1|y_2)$
(where $\tilde\rho^{A_1}_{b_2|y_2}$ is the steered state on $A_1$, defined in \cref{eq:steer B}) is a positive linear functional on the operator space $V_1=\text{Span}(\{\tilde\rho^{A_1}_{b_2|y_2}\})$, 
hence there exists a Hermitian operator $H^{A_1}_{\lambda_1}$ on $\mc{H}^{A_1}$ such that:
\begin{align}
\tilde{p}(b_2\lambda_1|y_2) =\tr[H_{\lambda_1}^{A_1}\tilde{\rho}^{A_1}_{b_2|y_2}]. \label{eq:frbp4}
\end{align}

Plugging \cref{eq:frbp1} and \cref{eq:frbp4} into  \cref{eq:bp-onto1} yields
\begin{align}
\tr[(M^{A_1}_{b_1|y_1}\otimes& M^{A_2}_{b_2|y_2})\rho^{A_1A_2}] \\
&=\sum_{\lambda_1}\tr[M^{A_1}_{b_1|y_1}\sigma^{A_1}_{\lambda_1}]\tr[H_{\lambda_1}^{A_1}\tilde{\rho}^{A_1}_{b_2|y_2}].\notag 
\end{align}
Recalling the definition of $\tilde{\rho}^{A_1}_{b_2|y_2}$ provided in \cref{eq:steer B}, we have
\begin{align}
\tr[(M^{A_1}_{b_1|y_1}\otimes& M^{A_2}_{b_2|y_2})\rho^{A_1A_2}]=\tr[M^{A_1}_{b_1|y_1}\tilde{\rho}^{A_1}_{b_2|y_2}],
\end{align}
and combining these two expressions, we obtain 
\begin{align}
\tr[M^{A_1}_{b_1|y_1}\tilde{\rho}^{A_1}_{b_2|y_2}]=\sum_{\lambda_1}\tr[M^{A_1}_{b_1|y_1}\sigma^{A_1}_{\lambda_1}]\tr[H_{\lambda_1}^{A_1}\tilde{\rho}^{A_1}_{b_2|y_2}].\notag 
\end{align}

Since this holds for all $M^{A_1}_{b_1|y_1}\in \mc{M}(\mc H^{A_1})$ and these are tomographically complete, it follows that
\begin{align}
\tilde{\rho}^{A_1}_{b_2|y_2}=\sum_{\lambda_1}\tr[\tilde{\rho}^{A_1}_{b_2|y_2}H^{A_1}_{\lambda_1}] \sigma^{A_1}_{\lambda_1}.
\end{align}
That is, $\{ \sigma^{A_1}_{\lambda_1}\}_{\lambda_1}$ and $\{H^{A_1}_{\lambda_1}\}_{\lambda_1}$ define a frame and its dual frame on $V_1=\text{Span}(\{\tilde{\rho}^{A_1}_{b_2|y_2}\})$. Moreover, taking the trace of both sides yields
\begin{align}
\tr[\tilde{\rho}^{A_1}_{b_2|y_2}]=\tr[\tilde{\rho}^{A_1}_{b_2|y_2} \sum_{\lambda_1} H^{A_1}_{\lambda_1}].
\end{align}
Noting that this holds for all $\tilde{\rho}^{A_1}_{b_2|y_2} \in V_1$ implies that
$\sum_{\lambda_1}H^{A_1}_{\lambda_1}= \mathcal{P}_{V_1}(\mbb{1}^{A_1})$, where we denote the superoperator that projects onto the operator space $V_1$ by $\mathcal{P}_{V_1}$.

Similarly, substituting \cref{eq:frbp2}, and \cref{eq:frbp3} into  \cref{eq:bp-onto2}, we can prove that $\{\sigma^{A_2}_{\lambda_2}\}_{\lambda_2}$ and the set of Hermitian operators $\{H^{A_2}_{\lambda_2}\}_{\lambda_2}$ (with $\sum_{\lambda_2}H^{A_2}_{\lambda_2}= \mathcal{P}_{V_2}(\mbb{1}^{A_2})$) defines a pair of frame and dual frame on $V_2=\text{Span}(\{\tilde{\rho}^{A_2}_{b_1|y_1}\})$, and in particular
\begin{align}
\tilde{\rho}^{A_2}_{b_1|y_1}=\sum_{\lambda_2}\tr[\tilde{\rho}^{A_2}_{b_1|y_1}H^{A_2}_{\lambda_2}] \sigma^{A_2}_{\lambda_2},
\end{align}

(5) Finally, using the ontological identities in \cref{ncnewstate} and substituting \cref{eq:frbp3,eq:frbp4}, we obtain
\begin{align}
p(\lambda_1\lambda_2)&=\sum_{{b_1b_2y_1y_2}} \alpha_{b_1b_2y_1y_2} \tilde{p}(b_1\lambda_2|y_1)\tilde{p}(b_2\lambda_1|y_2)  \\
&=\sum_{{b_1b_2y_1y_2}} \alpha_{b_1b_2y_1y_2} \tr[H_{\lambda_2}^{A_2}\tilde{\rho}^{A_2}_{b_1|y_1}]\tr[H_{\lambda_1}^{A_1}\tilde{\rho}^{A_1}_{b_2|y_2}].\notag
 \end{align}
Recalling the operational identity in \cref{eq:OPnewbi}, we thus have 
\begin{align}
p(\lambda_1\lambda_2)=\tr[(H_{\lambda_1}^{A_1}\otimes H_{\lambda_2}^{A_2})\rho^{A_1A_2}]. \label{eq:bipositive}
\end{align}

Combining it with \cref{eq:frbp1,eq:frbp2,eq:bp-onto} gives
\begin{align}
    &\tr[(M^{A_1}_{b_1|y_1}\otimes M^{A_2}_{b_2|y_2})\rho^{A_1A_2}]\\ &=\sum_{\lambda_1\lambda_2}\tr[(H_{\lambda_1}^{A_1}\otimes H_{\lambda_2}^{A_2})\rho^{A_1A_2}]\tr[M^{A_1}_{b_1|y_1}\sigma^{A_1}_{\lambda_1}]\tr[M^{A_2}_{b_2|y_2}\sigma^{A_2}_{\lambda_2}],  \notag 
\end{align}
Since this holds for all $M^{A_i}_{b_i|y_i}\in \mc{M}(\mc H^{A_i})$, which are tomographically complete, we conclude that
\begin{align}
\rho^{A_1A_2}=\sum_{\lambda_1\lambda_2}\tr[(H_{\lambda_1}^{A_1}\otimes H_{\lambda_2}^{A_2})\rho^{A_1A_2}]\sigma_{\lambda_1}^{A_1}\otimes \sigma_{\lambda_2}^{A_2}.
\end{align}
Moreover, the constraint $\tr[(H_{\lambda_1}^{A_1}\otimes H_{\lambda_2}^{A_2})\rho^{A_1A_2}]\ge 0$ in \cref{eq:constraints_app} holds by \cref{eq:bipositive} and the existence of the ontological model in \cref{eq:bp-onto}. \par

Now we prove the \enquote{if} direction. Suppose there exists a frame decomposition of $\rho^{A_1A_2}$ as in \cref{eq:bipartiteframe_app} and satisfying the requisite constraints. For the bipartite state $\rho^{A_1A_2}$, any effect $M_{b_1|y_1}^{A_1}\in\mc{M}(\mc H^{A_1})$ and any effect $M_{b_2|y_2}^{A_2}\in\mc{M}(\mc H^{A_2})$, we can define ontological representations of each of these via the linear functionals induced by the frames and dual frames:
\begin{subequations}
\begin{align} 
p(\lambda_1\lambda_2)&\coloneqq\tr[(H_{\lambda_1}^{A_1}\otimes H_{\lambda_2}^{A_2})\rho^{A_1A_2}], \label{eq: birecon1}\\
p(b_1|y_1\lambda_1)&=\tr[\sigma^{A_1}_{\lambda_1}M^{A_1}_{b_1|y_1}], \label{eq: birecon2}\\
p(b_2|y_2\lambda_2)&=\tr[\sigma^{A_2}_{\lambda_2}M^{A_2}_{b_2|y_2}]\label{eq: birecon3}.
\end{align}
The nonnegativity and normalization of these ontological representations follow from \cref{eq:constraints_app} together with the fact that the frame $\{\sigma^{A_i}_{\lambda_i}\}_{\lambda_i}$ is a set of density operators.

Next, we show that the effective ontological representations obtained by sequential composition depend linearly on the corresponding effective quantum processes. This follows from the fact that $\{\sigma^{A_i}_{\lambda_i}\}_{\lambda_i}$ and $\{H_{\lambda_i}^{A_i}\}_{\lambda_i}$ form a frame and its dual frame on $V_i$. In particular, 
\begin{align}
\tilde{p}&(b_1\lambda_2|y_1 )\coloneqq\sum_{\lambda_1}p(\lambda_1\lambda_2)p(b_1|y_1\lambda_1)\notag \\&=\sum_{\lambda_1}\tr[(H_{\lambda_1}^{A_1}\otimes H_{\lambda_2}^{A_2})\rho^{A_1A_2}]\tr[M^{A_1}_{b_1|y_1}\sigma^{A_1}_{\lambda_1}] \notag \\
&=\tr[(M^{A_1}_{b_1|y_1}\otimes H_{\lambda_2}^{A_2})\rho^{A_1A_2}]=\tr[H_{\lambda_2}^{A_2}\tilde{\rho}^{A_2}_{b_1|y_1}], \label{eq:dualframe2}
\end{align}
where we use \cref{eq: birecon1,eq: birecon2} in the first equality, and the definition of $\tilde{\rho}^{A_2}_{b_1|y_1}$, provided in \cref{eq:steer A}, in the last equality.  Similarly, we can show that
\begin{align}
\tilde{p}&(b_2\lambda_1|y_2 )\coloneqq\sum_{\lambda_2}p(\lambda_1\lambda_2)p(b_2|y_2\lambda_2)\notag \\
&=\sum_{\lambda_2}\tr[(H_{\lambda_1}^{A_1}\otimes H_{\lambda_2}^{A_2})\rho^{A_1A_2}]\tr[M^{A_2}_{b_2|y_2}\sigma^{A_2}_{\lambda_2}] \notag \\
&=\tr[(H_{\lambda_1}^{A_1}\otimes M^{A_2}_{b_2|y_2})\rho^{A_1A_2}]=\tr[H_{\lambda_1}^{A_1}\tilde{\rho}^{A_1}_{b_2|y_2}], \label{eq:dualframe1}
\end{align}
\end{subequations}
using \cref{eq: birecon1,eq: birecon3} in the first equality, and the definition of $\tilde{\rho}^{A_1}_{b_2|y_2}$, provided in \cref{eq:steer B}, in the last equality.

Finally, the ontological model reproduces the empirical statistics:
\begin{align}
&\sum_{\lambda_1\lambda_2}p(\lambda_1\lambda_2)p({b_1|y_1}\lambda_1)p({b_2|y_2}\lambda_2) \notag \\
&=\sum_{\lambda_1\lambda_2}\tr[(H_{\lambda_1}^{A_1}\otimes H_{\lambda_2}^{A_2})\rho^{A_1A_2}]\tr[M^{A_1}_{b_1|y_1}\sigma^{A_1}_{\lambda_1}]\tr[M^{A_2}_{b_2|y_2}\sigma^{A_2}_{\lambda_2}]\notag \\
&=\sum_{\lambda_2}\tr[(M^{A_1}_{b_1|y_1}\otimes H_{\lambda_2}^{A_2})\rho^{A_1A_2}]\tr[M^{A_2}_{b_2|y_2}\sigma^{A_2}_{\lambda_2}]\notag \\
&=\tr[(M^{A_1}_{b_1|y_1}\otimes M^{A_2}_{b_2|y_2})\rho^{A_1A_2}],
\end{align}
where in the last two steps, we again use the properties of frame representations.

Since all these ontological representations depend linearly on the corresponding physical processes and are nonnegative, the circuit is classically explainable. Moreover, since we allowed all possible factorizing dual processes, namely, all local multi-measurements, it follows that the bipartite state $\rho^{A_1A_2}$ is classical by Definition~\ref{def:bipartite}. 
\end{proof}
\subsection{Proof of Theorem~\ref{ncchannelstruc}}
\label{app:proof7}

\setcounter{theorem}{4}
\begin{theorem}
\label{ncchannelstruc_app}
A channel $\mc{E}^{B|A}$: $\mc{L}(\mc{H}^{A})\mapsto \mc{L}(\mc{H}^{B})$ is classical if and only if it admits a frame representation:
\begin{equation}
\mc{E}^{B|A}(\cdot)=\sum_{\lambda_A\lambda_B}\sigma^{B}_{\lambda_B}\tr[H_{\lambda_B}^{B}\mc{E}^{B|A}(F^{A}_{\lambda_A})]\tr[G_{\lambda_A}^{A}(\cdot)]
\label{eq:channelframe_app}
\end{equation}
where the frame on $V_A$ is a set of Hermitian operators $\{{F}^{A}_{\lambda_A}\}_{\lambda_A}$ each with unit trace, $\tr[{F}^{A}_{\lambda_A}]=1$, and the dual frame on $V_A$ is a POVM $\{G_{\lambda_A}^{A}\}_{\lambda_A}$ satisfying $\sum_{\lambda_A}G_{\lambda_A}^{A}=\mc{P}_{V_A}(\mbb{1}^A)=\mbb{1}^{A}$, while the frame on $V_B$ is a set of density operators $\{\sigma^{B}_{\lambda_B}\}_{\lambda_B}$ and the dual frame on $V_B$  is a set of Hermitian operators $\{{H}_{\lambda_B}^{B}\}_{\lambda_B}$ satisfying $\sum_{\lambda_B}{H}_{\lambda_B}^{B}=\mc{P}_{V_B}(\mbb{1}^{B})$, with $\mc{P}_{V_{B(A)}}$ denoting the superoperator projecting onto $V_{B(A)}$, and such that  
\begin{align} 
\tr[H_{\lambda_B}^{B}\mc{E}^{B|A}(F^{A}_{\lambda_A})] \ge 0~~~~\forall \lambda_A, \lambda_B. \label{eq:cl-cons-all_app}
\end{align}
\end{theorem}
\begin{proof}
We first prove the `only if' direction. Assume $\mc{E}^{B|A}$ is classical in the sense of Definition~\ref{def:PTM} in the main text. Then, for the statistics generated in the prepare-transform-measure circuits consisting of the channel $\mc{E}^{B|A}$ and every possible multi-source $\{\{p(a|x)\rho^A_{a|x}\}_a\}_x$ and multi-measurement $\{\{M^B_{b|y}\}_b\}_y$, there exists an ontological model: 
\begin{subequations}
\begin{align}
\tr[&M^B_{b|y}\mc{E}^{B|A}(p(a|x)\rho^A_{a|x})]\notag \\
&=\sum_{\lambda_A\lambda_B}p(b|y\lambda_B)p(\lambda_B|\lambda_A)p(a\lambda_A|x)     \label{eq:chonto} \\
&=\sum_{\lambda_B}p(b|y\lambda_B)\tilde{p}(a\lambda_B|x)     \label{eq:chonto1} \\
&=\sum_{\lambda_A}\tilde{p}(b|y\lambda_A)p(a\lambda_A|x),     \label{eq:chonto2}
\end{align}
\end{subequations}
where $\tilde{p}(a\lambda_B|x)$ and $\tilde{p}(b|y\lambda_A)$ are defined in \cref{eq:channelonto}, and where the ontological identities in \cref{op:channel} are respected. In particular, these constraints imply that:

(1) From \cref{channelmeas1}, for each $\lambda_B$, the map   
$M^{B}_{b|y}\mapsto p(b|y\lambda_B)$
is a normalized positive linear functional on the space of all effects on $B$, i.e., $\mc{M}(\mc H^{B})$. {Normalization $p(\mbb{1}^{B}|\lambda_B)=1$ follows from $\tr[\mbb{1}^{B}\mc{E}^{B|A}(\rho^A)]=1$, and by \cref{eq:chonto}, $p(\mbb{1}^{B}|\lambda_B)$ must equal $1$.} Thus, by the generalized Gleason theorem, there exists a quantum state $\sigma_{\lambda_B}^B$ such that
\begin{align}
    p(b|y\lambda_B)=\tr[M^B_{b|y}\sigma_{\lambda_B}^B]. \label{eq:frch2}
\end{align}

(2) From \cref{channelmeas2}, for each $\lambda_A$, the map $p(a|x)\rho^A_{a|x}\mapsto p(a\lambda_A|x) $ is a positive linear functional on the space of all unnormalized states on $A$. Therefore, by the Riesz representation theorem, there exists a positive semi-definite operator $G_{\lambda_A}^A$ such that
\begin{align}
    p(a\lambda_A|x)=\tr[G_{\lambda_A}^Ap(a|x)\rho^A_{a|x}]. \label{eq:frch1}
\end{align}

(3)From \cref{channelprep1},  for each $\lambda_A$, the map   
$\tilde M^{A}_{b|y}\mapsto \tilde p(b|y\lambda_A)$
is a normalized positive linear functional on $V_A=\text{Span}(\{\tilde M^A_{b|y} \})$ with $\tilde M^A_{b|y}=[\mc{E}^{B|A}]^{\dagger}(M^{B}_{b|y})$, therefore, there exists a Hermitian operator $F_{\lambda_A}^A$ such that
\begin{align}
    \tilde p(b|y\lambda_A)=\tr[\tilde M^A_{b|y}F_{\lambda_A}^A]. \label{eq:frch3}
\end{align}
Since $\sum_b\tr[\tilde M^A_{b|y}F_{\lambda_A}^A]=1$, $F_{\lambda_A}^A$ is trace-1.

(4) From \cref{channelprep2},  for each $\lambda_B$, the map   
$\tilde \rho^{B}_{a|x}\mapsto \tilde{p}(a\lambda_B|x)$ is a positive linear functional on $V_B=\text{Spam}(\{\tilde \rho^{B}_{a|x} \})$ with  $\tilde \rho^{B}_{a|x}=\mc{E}^{B|A}(p(a|x)\rho^{A}_{a|x})$. Thus, there exists a Hermitian operator $H_{\lambda_B}^B$ such that
\begin{align}
    \tilde{p}(a\lambda_B|x)=\tr[H_{\lambda_B}^B\tilde \rho^{B}_{a|x}]. \label{eq:frch4}
\end{align}

Substituting \cref{eq:frch2,eq:frch4} into \cref{eq:chonto1} yields:
\begin{align}
\tr[M^B_{b|y}&\mc{E}^{B|A}(p(a|x)\rho^A_{a|x})]\notag \\
&=\sum_{\lambda_B}\tr[M^B_{b|y}\sigma_{\lambda_B}^B]\tr[H_{\lambda_B}^B \tilde \rho^{B}_{a|x}].
\end{align}
Recalling the definition of $\tilde \rho^B_{a|x}$ in \cref{eq:ch-seq-B}, we have
\begin{align}
\tr[M^B_{b|y}&\mc{E}^{B|A}(p(a|x)\rho^A_{a|x})]:=\tr[M^B_{b|y}\tilde \rho^B_{a|x}],
\end{align}
and combining these two expressions, we obtain
\begin{align}
\tr[M^B_{b|y}\tilde \rho^B_{a|x})]=\sum_{\lambda_B}\tr[M^B_{b|y}\sigma_{\lambda_B}^B]\tr[H_{\lambda_B}^B \tilde \rho^{B}_{a|x}].
\end{align}

Importantly, since this holds for all $M^{B}_{b|y}\in \mc{M}(\mc H^{B})$, which are tomographically complete, this implies
\begin{align}
\tilde \rho^B_{a|x}=\sum_{\lambda_B}\tr[H_{\lambda_B}^B \tilde \rho^{B}_{a|x}]\sigma_{\lambda_B}^B.
\end{align}
That is, $\{ \sigma^{B}_{\lambda_B}\}_{\lambda_B}$ and $\{H^{B}_{\lambda_B}\}_{\lambda_B}$ define a frame and its dual frame on $V_B=\text{Span}(\{\tilde{\rho}^{B}_{a|x}\})$. Moreover, taking the trace of both sides yields
\begin{align}
\tr[\tilde \rho^B_{a|x}]=\sum_{\lambda_B}\tr[H_{\lambda_B}^B \tilde \rho^{B}_{a|x}].
\end{align}
Noting that, this holds for all $\tilde \rho^B_{a|x}$ on $V_B$ implies that $\sum_{\lambda_B}H^{B}_{\lambda_B}=\mc{P}_{V_B}(\mbb{1}^{B})$, where we denote the superoperator that projects onto $V_B$ by $\mc{P}_{V_B}$.

Similarly, substituting \cref{eq:frch1,eq:frch3} into \cref{eq:chonto2}, we have:
\begin{align}
\tr[M^B_{b|y}&\mc{E}^{B|A}(p(a|x)\rho^A_{a|x})]\notag \\
&=\sum_{\lambda_A}\tr[\tilde{M}^A_{b|y}F_{\lambda_A}^A]\tr[G_{\lambda_A}^A p(a|x) \rho^A_{a|x})].
\end{align}
Recalling the definition of $\tilde{M}^A_{b|y}$ in \cref{eq:ch-seq-B}, we have
\begin{align}
\tr[M^B_{b|y}&\mc{E}^{B|A}(p(a|x)\rho^A_{a|x})]=\tr[\tilde{M}^A_{b|y}p(a|x)\rho^A_{a|x}],
\end{align}
and combining these two expressions, we obtain
\begin{align}
\tr[\tilde{M}^A_{b|y}p(a|x)\rho^A_{a|x}]=\sum_{\lambda_B}\tr[\tilde{M}^A_{b|y}F_{\lambda_A}^A]\tr[G_{\lambda_A}^A p(a|x) \rho^A_{a|x}]
\end{align}

Importantly, since this hold for all unnormalized state $p(a|x) \rho^A_{a|x}$ with $\rho^{A}_{a|x}\in \mc{S}(\mc{H^A})$, which are tomographically complete, we have 
\begin{align}
\tilde M^A_{b|y}=\sum_{\lambda_A}\tr[\tilde{M}^A_{b|y}F_{\lambda_A}^A]G_{\lambda_A}^A,
\end{align}
That is $\{F_{\lambda_A}^A\}_{\lambda_A}$ 
and $\{G_{\lambda_A}\}_{\lambda_A}$ defines a frame and dual frame on $V_A=\text{Span}(\{\tilde{M}^{A}_{b|y}\})$. Moreover, taking the summation over $b$ on both sides yields
\begin{align}
\sum_b \tilde  M^A_{b|y}&=\sum_b\sum_{\lambda_A}\tr[\tilde{M}^A_{b|y}F_{\lambda_A}^A]G_{\lambda_A}^A\notag \\
&=\sum_{\lambda_A}G_{\lambda_A}^A=\mc P_V(\mbb{1}^{A})=\mbb{1}^{A}.
\end{align}
Therefore, $\{G_{\lambda_A}\}_{\lambda_A}$ is a valid POVM. 

(5)Finally, plugging \cref{eq:frch3,eq:frch4} into the assumed decomposition in \cref{channel3}, we obtain
\begin{align}
p(\lambda_B|\lambda_A)&=\sum_{{abxy}}\alpha_{abxy} \tilde{p}(b|y\lambda_{A})\tilde{p}(a\lambda_{B}|x) \notag \\
&=\sum_{{abxy}}\alpha_{abxy} \tr[\tilde M^A_{b|y}F_{\lambda_A}^A]\tr[H_{\lambda_B}^B\tilde \rho^{B}_{a|x}].
\end{align}
Recalling the operational identity in \cref{eq:OPnewch}, we thus have
\begin{align}
p(\lambda_B|\lambda_A)=\tr[H_{\lambda_B}^B\mc{E}^{B|A}(F_{\lambda_A}^A)].
\end{align}

Combining it with\cref{eq:frch2,eq:frch1,eq:chonto} yields
\begin{align}
&\tr[M^B_{b|y}\mc{E}^{B|A}(p(a|x)\rho^A_{a|x})]\notag \\
&=\sum_{\lambda_A\lambda_B}\tr[H_{\lambda_B}^B\mc{E}^{B|A}(F_{\lambda_A}^A)]\tr[M^B_{b|y}\sigma_{\lambda_B}^B]\tr[G_{\lambda_A}^Ap(a|x)\rho^A_{a|x}],
\end{align}
since this holds for all $M^B_{b|y}$ and $p(a|x)\rho^A_{a|x}$, which are tomographically complete, we have
\begin{align}
\mc{E}^{B|A}(\cdot)=\sum_{\lambda_A\lambda_B}\tr[H_{\lambda_B}^B\mc{E}^{B|A}(F_{\lambda_A}^A)]\tr[G_{\lambda_A}^A(\cdot)]\sigma_{\lambda_B}^B.
\end{align}

Now let us prove the ``if'' part. Suppose there exists a frame decomposition of the channel of $\mc{E}^{B|A}$ as in Eq.~\eqref{eq:channelframe_app}, then for the channel $\mc{E}^{B|A}$, any effect $M^{B}_{b|y}$ and any unnormalized state $p(a|x)\rho^A_{a|x}$, we can define ontological representations of each of these via the linear functionals induced by the frames and dual frames:
\begin{subequations}
\begin{align} 
p(\lambda_B|\lambda_A)&\coloneqq\tr[H^B_{\lambda_B}\mc{E}^{B|A}(F^A_{\lambda_A})]\label{eq: chrecon1},\\
p({b|y}\lambda_B)&\coloneqq\tr[M^{B}_{b|y}\sigma^{B}_{\lambda_B}]\label{eq: chrecon2},\\
p(a\lambda_A|x)&\coloneqq\tr[G^A_{\lambda_A}p(a|x)\rho^A_{a|x}]\label{eq: chrecon3}, 
\end{align}
where the nonnegativity and normalization are ensured by \cref{eq:cl-cons-all_app}, $\sigma^{B}_{\lambda_B}$ being a density operator and $\{G_{\lambda_A}\}_{\lambda_A}$ being a POVM. 

We can also show that the effective ontological processes obtained by sequential compositions are also linear to the corresponding effective quantum processes obtained by sequential compositions. This is because $\{\sigma^{A_i}_{\lambda_i}\}_{\lambda_i}$ and $\{H_{\lambda_i}^{A_i}\}_{\lambda_i}$ are a frame and its dual frame on $V_i$. In particular, 
\begin{align}
\tilde p(a\lambda_B|x):&=\sum_{\lambda_A}p(\lambda_B|\lambda_A)p(a\lambda_A|x)\label{PTMcom1}\\&=\sum_{\lambda_A}\tr[H^B_{\lambda_B}\mc{E}^{B|A}(F^A_{\lambda_A})]\tr[G^A_{\lambda_A}p(a|x)\rho^{A}_{a|x}]  \notag \\
&=\tr[H^B_{\lambda_B}\mc{E}^{B|A}(p(a|x)\rho^{A}_{a|x})]=\tr[H^B_{\lambda_B}\tilde{\rho}^B_{a|x}],  \notag 
\end{align}
where we use \cref{eq: chrecon1,eq: chrecon3} in the first equality, and the definition of $\tilde{\rho}^{B}_{a|x}$ in the last equality.  Similarly, using \cref{eq: chrecon1,eq: chrecon2}, we can show that
\begin{align}
\tilde{p}(b|y\lambda_A):&=\sum_{\lambda_B}p(b|y\lambda_B)p(\lambda_B|\lambda_A)\label{PTMcom2} 
\\&=\sum_{\lambda_B}\tr[M^{B}_{b|y}\sigma^{B}_{\lambda_B}]\tr[H^B_{\lambda_B}\mc{E}^{B|A}(F^A_{\lambda_A})]  \notag \\
    &=\tr[M^{B}_{b|y}\mc{E}^{B|A}(F^A_{\lambda_A})]=\tr[\tilde{M}^A_{b|y}F^A_{\lambda_A}]. \notag 
\end{align}
\end{subequations}

Finally, the ontological model also reproduces the empirical statistics:
\begin{align}
&\sum_{\lambda_A\lambda_B}p(b|y\lambda_B)p(\lambda_B|\lambda_A)p(a\lambda_A|x) \\
&=\sum_{\lambda_A\lambda_B}\tr[M^{B}_{b|y}\sigma^{B}_{\lambda_B}]\tr[H^B_{\lambda_B}\mc{E}^{B|A}(F^A_{\lambda_A})]\tr[G^A_{\lambda_A}p(a|x)\rho^{A}_{a|x}] \notag \\
&=\sum_{\lambda_A}\tr[M^B_{b|y}\mc{E}^{B|A}(F^A_{\lambda_A})]\tr[G^A_{\lambda_A}p(a|x)\rho^{A}_{a|x}] \notag \\
&=\tr[M^{B}_{b|y}\mc{E}^{B|A}(p(a|x)\rho^{A}_{a|x})]. \notag
\end{align}
In the last two equations, we again use properties of frame representations. 

Since all these ontological representations depend linearly on the corresponding physical processes and are nonnegative, the circuit is classically explainable. Moreover, since we allowed all possible factorizing dual processes, this, in turn, implies that the channel $\mc{E}^{B|A}$ is classical by Definition~\ref{channelprop}.

\end{proof}

\subsection{Proof of Theorem~\ref{ncgen}}
\label{app:proofmpmi}
\setcounter{theorem}{5}
\begin{theorem}\label{ncgen_app}
A multipartite multi-instrument $\{\{\mc{F}_{c|z}^{\vec{B}|\vec{A}}\}_c\}_z$,
where $\mc{F}_{c|z}^{\vec{B}|\vec{A}}$
: $\mc{L}(\bigotimes_{i=1}^n\mc{H}^{A_i})\mapsto \mc{L}(\bigotimes_{j=1}^m\mc{H}^{B_{j}})$, is classical if and only if it admits a frame representation: 
\begin{align}
\label{eq:generalframe_app}
&\mc{F}_{c|z}^{\vec{B}|\vec{A}}(\cdot)= \\
&\sum_{\lambda_{\vec{A}}\lambda_{\vec{B}}}\underset{j}{\bigotimes} \sigma_{\lambda_{B_j}}^{B_j}\tr[(\underset{j}{\bigotimes} H_{\lambda_{B_j}}^{B_j})\mc{F}_{c|z}^{\vec{B}|\vec{A}}(\underset{i}{\bigotimes} F^{A_i}_{\lambda_{A_i}})]\tr[(\underset{i}{\bigotimes} G_{\lambda_{A_i}}^{A_i})(\cdot)], \notag
\end{align}
where the frame on the operator space $V_{B_j}$ is a fixed set of density operators $\{\sigma^{B_j}_{\lambda_{B_j}}\}_{\lambda_{B_j}}$, and the dual frame is a set of Hermitian operators $\{{H}_{\lambda_{B_j}}^{B_j}\}_{\lambda_{B_j}}$ satisfying $\sum_{\lambda_{B_j}}H^{B_j}_{\lambda_{B_j}}=\mc{P}_{V_{B_j}}(\mbb 1^{{B_j}})$; the frame on the operator space $V_{A_i}$ is a set of Hermitian operators $\{{F}^{A_i}_{\lambda_{A_i}}\}_{\lambda_{A_i}}$ satisfying $\tr[{F}^{A_i}_{\lambda_{A_i}}]=1$, and the dual frame is a POVM $\{G_{\lambda_{A_i}}^{A_i}\}_{\lambda_{A_i}}$ satisfying $\sum_{\lambda_{A_i}}G_{\lambda_{A_i}}=\mc{P}_{V_{A_i}}(\mbb{1}^{A_i})=\mbb{1}^{{A_i}}$, and such that
\begin{align}  
\tr[(\underset{j}{\bigotimes} H_{\lambda_{B_j}}^{B_j})\mc{F}_{c|z}^{\vec{B}|\vec{A}}(\underset{i}{\bigotimes} F^{A_i}_{\lambda_{A_i}})] \ge 0,~~\forall \lambda_{\vec{A}}, \lambda_{\vec{B}},c,z. \label{eq:cl-gen-cons_app}
\end{align}

\end{theorem}
\begin{proof}
We first prove the `only if' direction. Assume the multipartite multi-instrument $\{\{\mc{F}_{c|z}^{\vec{B}|\vec{A}}\}_c\}_z$ is classical. Then, for the statistics generated in the multipartite prepare-transform-measure circuit consisting of the multipartite multi-instrument $\{\{\mc{F}_{c|z}^{\vec{B}|\vec{A}}\}_c\}_z$
and every possible multi-source $\{\{p(a_i|x_i)\rho^{A_i}_{a_i|x_i}\}_{a_i}\}_{x_i}$ on $\mc{H}^{A_i}$ and multi-measurement $\{\{M^{B_j}_{b_j|y_j}\}_{b_j}\}_{y_j}$ on $\mc{H}^{B_j}$, there exists an ontological model: 
\begin{subequations}
\begin{align}
&\tr[(\underset{j\in [m]}{\bigotimes} M_{b_j|y_j}^{B_j})\mc{F}_{c|z}^{\vec{B}|\vec{A}} (\underset{i\in [n]}\bigotimes p(a_i|x_i)\rho^{A_i}_{a_i|x_i})]\notag \\
&=\sum_{\lambda_{\vec{B}}\lambda_{\vec{A}}} p(c\lambda_{\vec{B}}|z\lambda_{\vec{A}})\prod_{j}p(b_j|y_j\lambda_{B_j})\prod_{i}p(a_i\lambda_{A_i}|x_i)    \label{eq:genonto} \\
&=\sum_{\lambda_{B_j}}p(b_j|y_j\lambda_{B_j})\tilde{p}(\vec{a}\vec{b}_{\bar{j}}c\lambda_{B_j}|\vec{x}\vec{y}_{\bar{j}}z)     \label{eq:genonto1} \\
&=\sum_{\lambda_{A_i}}p(a_i\lambda_{A_i}|x_i)\tilde{p}(\vec{a}_{\bar{i}}\vec{b}c|\vec{x}_{\bar{i}}\vec{y}z\lambda_{A_i}),     \label{eq:genonto2}
\end{align}
\end{subequations}
where the ontological identities in \cref{op:general} are respected. Here we denote $\bar{j}=\{j'|j'\ne j\}$ and $\bar{i}=\{i'|i'\ne i\}$ and define
\begin{subequations}
\begin{align}
\tilde{p}&(\vec{a}\vec{b}_{\bar{j}}c\lambda_{B_j}|\vec{x}\vec{y}_{\bar{j}}z) \\
&=\sum_{\lambda_{\vec{B}_{\bar{j}}}\lambda_{\vec{A}}} p(c\lambda_{\vec{B}}|z\lambda_{\vec{A}})\prod_{j'\ne j}p(b_{j'}|y_{j'}\lambda_{B_{j'}})\prod_{i}p(a_i\lambda_{A_i}|x_i), \notag \\
\tilde{p}&(\vec{a}_{\bar{i}}\vec{b}c|\vec{x}_{\bar{i}}\vec{y}z\lambda_{A_i}) \\
&=\sum_{\lambda_{\vec{B}}\lambda_{\vec{A}_{\bar{i}}}} p(c\lambda_{\vec{B}}|z\lambda_{\vec{A}})\prod_{j}p(b_j|y_j\lambda_{B_j})\prod_{i'\ne i}p(a_{i'}\lambda_{A_{i'}}|x_{i'}).  \notag
\end{align}
\label{eq:gen-state-effect}
\end{subequations}
Specifically, the constraints in \cref{op:general} then imply that:

(1) From \cref{instmeas1}, for each $j\in [m]$ each $\lambda_{B_j}$, the map  
$M^{B_j}_{b_j|y_j}\mapsto p(b_j|y_j,\lambda_{B_j})$
is a normalized positive linear functional on the space of all effects on $B_j$, i.e., $\mc{M}(\mc H^{B_j})$.  Thus, by the generalized Gleason theorem, there exists a quantum state $\sigma_{\lambda_{B_j}}^{B_j}$ such that
\begin{align}
p(b_j|y_j\lambda_{B_j})=\tr[M^{B_j}_{b_j|y_j}\sigma_{\lambda_{B_j}}^{B_j}]. \label{eq:frinst2}
\end{align}

(2) From \cref{instprep1}, for each $i\in [n]$, for each $\lambda_{A_i}$, the map $p(a_i|x_i)\rho^{A_i}_{a_i|x_i}\mapsto p(a_i\lambda_{A_i}|x_i) $ is a positive linear functional on the space of all unnormalized states. Therefore, there exists a positive semi-definite operator $G_{\lambda_{A_i}}^{A_i}$ such that
\begin{align}
p(a_i\lambda_{A_i}|x_i)=\tr[G_{\lambda_{A_i}}^{A_i}p(a_i|x_i)\rho^{A_i}_{a_i|x_i}]. \label{eq:frinst1}
\end{align}

There are many ontological identities included in \cref{instrcomp}; however, it turns out that it is sufficient to consider those for the ontological representation in \cref{eq:gen-state-effect} defined for the effective physical processes 
\begin{subequations}
\begin{align}
\tilde{\rho}^{B_j}_{\vec{a}\vec{b}_{\bar{j}}c|\vec{x}\vec{y}_{\bar{j}}z} &=\tr_{\vec{B}_{\bar{j}}}[M_{\vec{b}_{\bar{j}}|\vec{y}_{\bar{j}}}^{\vec{B}}\mc{F}_{c|z}^{\vec{B}|\vec{A}}(\tilde{\rho}^{\vec{A}}_{\vec{a}|\vec{x}} )], \label{eq:gen-state}\\
\tilde{M}^{A_i}_{\vec{a}_{\bar{i}}\vec{b}c|\vec{x}_{\bar{i}}\vec{y}z} &=\tr_{\vec{A}_{\bar{i}}}[\tilde{\rho}^{\vec{A}_{\bar{i}}}_{\vec{a}_{\bar{i}}|\vec{x}_{\bar{i}}}[\mc{F}_{c|z}^{\vec{B}|\vec{A}}]^{\dagger}(M_{\vec{b}|\vec{y}}^{\vec{B}})], \label{eq:gen-effect}
\end{align}
\label{eq:gen-state-effect_q}
\end{subequations}
where we denote $M^{\vec{B}}_{\vec{b}|\vec{y}}=\underset{j}{\bigotimes} M_{{b_{j}|y_{j}}}^{B_{j}}$ and $M^{\vec{B}}_{\vec{b}_{\bar{j}}|\vec{y}_{\bar{j}}}\coloneqq\mbb{1}^{B_{j}}\otimes\underset{j'\ne j}{\bigotimes} M_{{b_{j'}|y_{j'}}}^{B_{j'}}$, similarly 
$\tilde{\rho}^{\vec{A}}_{\vec{a}|\vec{x}}\coloneqq\underset{i}{\bigotimes} p(a_i|x_i)\rho^{A_i}_{a_i|x_i}$ and $\tilde{\rho}^{\vec{A}_{\bar{i}}}_{\vec{a}_{\bar{i}}|\vec{x}_{\bar{i}}}\coloneqq\mbb{1}^{A_i}\otimes\underset{i'\ne i}{\bigotimes} p(a_{i'}|x_{i'})\rho^{A_{i'}}_{a_{i'}|x_{i'}}$.

(3) As a special case of \cref{instrcomp}, for each $j\in [m]$ and each $\lambda_{B_j}$, the map $\tilde{\rho}^{B_j}_{\vec{a}\vec{b}_{\bar{j}}c|\vec{x}\vec{y}_{\bar{j}}z}\mapsto\tilde{p}(\vec{a}\vec{b}_{\bar{j}}c\lambda_{B_j}|\vec{x}\vec{y}_{\bar{j}}z)$ is a positive linear functional on the space of all states $\tilde{\rho}^{B_j}_{\vec{a}\vec{b}_{\bar{j}}c|\vec{x}\vec{y}_{\bar{j}}z}\in V_{B_j}$ in \cref{eq:gen-state}. Thus, there exists a Hermitian operator $H^{B_j}_{\lambda_{B_j}}$ such that
\begin{align}
\tilde{p}(\vec{a}\vec{b}_{\bar{j}}c\lambda_{B_j}|\vec{x}\vec{y}_{\bar{j}}z)=\tr[H^{B_j}_{\lambda_{B_j}}\tilde{\rho}^{B_j}_{\vec{a}\vec{b}_{\bar{j}}c|\vec{x}\vec{y}_{\bar{j}}z}].\label{eq:frinst3}
\end{align}

(4) As another special case of \cref{instrcomp}, for each $i\in [n]$ and each $\lambda_{A_i}$, the map $\tilde{M}^{A_i}_{\vec{a}_{\bar{i}}\vec{b}c|\vec{x}_{\bar{i}}\vec{y}z} \mapsto\tilde{p}(\vec{a}_{\bar{i}}\vec{b}c|\vec{x}_{\bar{i}}\vec{y}z\lambda_{A_i})$ is a positive linear functional on the space of all effects $\tilde{M}^{A_i}_{\vec{a}_{\bar{i}}\vec{b}c|\vec{x}_{\bar{i}}\vec{y}z}\in V_{A_i}$ in \cref{eq:gen-effect}. Thus, by the Riesz representation theorem, there exists a Hermitian operator $F^{A_i}_{\lambda_{A_i}}$ such that
\begin{align}
\tilde{p}(\vec{a}_{\bar{i}}\vec{b}c|\vec{x}_{\bar{i}}\vec{y}z\lambda_{A_i})=\tr[F^{A_i}_{\lambda_{A_i}}\tilde{M}^{A_i}_{\vec{a}_{\bar{i}}\vec{b}c|\vec{x}_{\bar{i}}\vec{y}z} ].\label{eq:frinst4}
\end{align}

Plugging \cref{eq:frinst2,eq:frinst3} into \cref{eq:genonto1} yields
\begin{align}
&\tr[(\underset{j\in [m]}{\bigotimes} M_{b_j|y_j}^{B_j})\mc{F}_{c|z}^{\vec{B}|\vec{A}} (\underset{i\in [n]}\bigotimes p(a_i|x_i)\rho^{A_i}_{a_i|x_i})]\\
&=\sum_{\lambda_{B_j}}\tr[M^{B_j}_{b_j|y_j}\sigma_{\lambda_{B_j}}^{B_j}]\tr[H^{B_j}_{\lambda_{B_j}}\tilde{\rho}^{B_j}_{\vec{a}\vec{b}_{\bar{j}}c|\vec{x}\vec{y}_{\bar{j}}z}]. \notag 
\end{align}
Combining it with \cref{eq:gen-state}, we have
\begin{align}
&\tr[M_{b_j|y_j}^{B_j}\tilde{\rho}^{B_j}_{\vec{a}\vec{b}_{\bar{j}}c|\vec{x}\vec{y}_{\bar{j}}z} ]\\
&=\sum_{\lambda_{B_j}}\tr[M^{B_j}_{b_j|y_j}\sigma_{\lambda_{B_j}}^{B_j}]\tr[H^{B_j}_{\lambda_{B_j}}\tilde{\rho}^{B_j}_{\vec{a}\vec{b}_{\bar{j}}c|\vec{x}\vec{y}_{\bar{j}}z}]. \notag 
\end{align}
Importantly, since this holds for all $M_{b_j|y_j}^{B_j}\in\mc M(\mc H^{B_j})$, which are tomographically complete, this implies
\begin{align}
\tilde{\rho}^{B_j}_{\vec{a}\vec{b}_{\bar{j}}c|\vec{x}\vec{y}_{\bar{j}}z} =\sum_{\lambda_{B_j}}\tr[H^{B_j}_{\lambda_{B_j}}\tilde{\rho}^{B_j}_{\vec{a}\vec{b}_{\bar{j}}c|\vec{x}\vec{y}_{\bar{j}}z}]\sigma_{\lambda_{B_j}}^{B_j}.
\end{align}
That is, $\{\sigma_{\lambda_{B_j}}^{B_j}\}_{\lambda_{B_j}}$ and $\{H^{B_j}_{\lambda_{B_j}}\}_{\lambda_{B_j}}$ defines a frame and its dual frame on $V_{B_j}=\text{Span}(\{\tilde{\rho}^{B_j}_{\vec{a}\vec{b}_{\bar{j}}c|\vec{x}\vec{y}_{\bar{j}}z} \})$. Moreover, taking the trace of both sides yields
\begin{align}
\tr[\tilde{\rho}^{B_j}_{\vec{a}\vec{b}_{\bar{j}}c|\vec{x}\vec{y}_{\bar{j}}z}] =\sum_{\lambda_{B_j}}\tr[H^{B_j}_{\lambda_{B_j}}\tilde{\rho}^{B_j}_{\vec{a}\vec{b}_{\bar{j}}c|\vec{x}\vec{y}_{\bar{j}}z}].
\end{align}
Noting that this hold for all $\tilde{\rho}^{B_j}_{\vec{a}\vec{b}_{\bar{j}}c|\vec{x}\vec{y}_{\bar{j}}z}\in V_{B_j}$. we thus have $\sum_{\lambda_{B_j}}H^{B_j}_{\lambda_{B_j}}=\mc{P}_{V_{B_j}}(\mbb{1}^{B_j})$, where $\mc{P}_{V_{B_j}}$ is the superoperator that projects onto the operator space $V_{B_j}$.

Similarly, by plugging \cref{eq:frinst1,eq:frinst4} into \cref{eq:genonto2}, we can prove that $\{F^{A_i}_{\lambda_{A_i}}\}_{\lambda_{A_i}}$ and a POVM $\{G_{\lambda_{A_i}}^{A_i}\}_{\lambda_{A_i}}$ defines a frame and dual frame on $V_{A_i}=\text{Span}(\{\tilde{M}^{A_i}_{\vec{a}_{\bar{i}}\vec{b}c|\vec{x}_{\bar{i}}\vec{y}z}\})$, and in particular
\begin{align}
\tilde{M}^{A_i}_{\vec{a}_{\bar{i}}\vec{b}c|\vec{x}_{\bar{i}}\vec{y}z}=\sum_{\lambda_{A_i}}\tr[F^{A_i}_{\lambda_{A_i}}\tilde{M}^{A_i}_{\vec{a}_{\bar{i}}\vec{b}c|\vec{x}_{\bar{i}}\vec{y}z}]G_{\lambda_{A_i}}^{A_i}.
\end{align}

(5) Finally, we now consider the ontological identities in \cref{instrgen}. Again, there are many ontological identities therein. However, it turns out that it is sufficient to consider some special ones. Specifically, for any fixed $c$ and $z$, we consider operational identity of the form of:
\begin{widetext}
\begin{align}
\mc{F}^{\vec{B}|\vec{A}}_{c|z}(\cdot)&=\sum_{\substack{ \{c^{i}\vec{b}^{i}\vec{y}^{i}\}_{i}\{c^{j}\vec{b}^{j}_{\bar{j}}\vec{y}^{j}_{\bar{j}}\}_{j}\\ \{z^{i}\vec{a}^{i}_{{\bar{i}}}\vec{x}^{i}_{{\bar{i}}}\}_{{i}}\{{z^{j}\vec{a}^{j}\vec{x}^{j}\}_{{j}}}}}\alpha_{\{c^{i}\vec{b}^{i}\vec{y}^{i}\}_{{i}}\{c^{j}\vec{b}^{j}_{\bar{j}}\vec{y}^{j}_{\bar{j}}\}_{j}}^{\{z^{i}\vec{a}^{i}_{{\bar{i}}}\vec{x}^{i}_{{\bar{i}}}\}_{{i}}\{{z^{j}\vec{a}^{j}\vec{x}^{j}\}_{{j}}}}\tr[(\bigotimes_i\tilde{M}^{A_i}_{\vec{a}_{\bar{i}}\vec{b}c|\vec{x}_{\bar{i}}\vec{y}z})(\cdot)]\bigotimes_j \tilde{\rho}^{B_j}_{\vec{a}\vec{b}_{\bar{j}}c|\vec{x}\vec{y}_{\bar{j}}z}, 
\label{eq:gen-op-sp}
\end{align}
where the superscript $i$ and $j$ are indices to distinguish different settings and outcome variables. The associated ontological identities are then given as:
\begin{align}
p(c\lambda_{\vec{B}}|z\lambda_{\vec{A}})
&=\sum_{\substack{ \{c^{i}\vec{b}^{i}\vec{y}^{i}\}_{i}\{c^{j}\vec{b}^{j}_{\bar{j}}\vec{y}^{j}_{\bar{j}}\}_{j}\\ \{z^{i}\vec{a}^{i}_{{\bar{i}}}\vec{x}^{i}_{{\bar{i}}}\}_{{i}}\{{z^{j}\vec{a}^{j}\vec{x}^{j}\}_{{j}}}}}\alpha_{\{c^{i}\vec{b}^{i}\vec{y}^{i}\}_{{i}}\{c^{j}\vec{b}^{j}_{\bar{j}}\vec{y}^{j}_{\bar{j}}\}_{j}}^{\{z^{i}\vec{a}^{i}_{{\bar{i}}}\vec{x}^{i}_{{\bar{i}}}\}_{{i}}\{{z^{j}\vec{a}^{j}\vec{x}^{j}\}_{{j}}}}\prod_{i}\tilde{p}(\vec{a}^i_{\bar{i}}\vec{b}^ic^i|\vec{x}^i_{\bar{i}}\vec{y}^iz^i\lambda_{A_i})\times \prod_{j}\tilde{p}(\vec{a}^j\vec{b}^j_{\bar{j}}c^j\lambda_{B_j}|\vec{x}^j\vec{y}^j_{\bar{j}}z^j)\notag\\
&=\sum_{\substack{ \{c^{i}\vec{b}^{i}\vec{y}^{i}\}_{i}\{c^{j}\vec{b}^{j}_{\bar{j}}\vec{y}^{j}_{\bar{j}}\}_{j}\\ \{z^{i}\vec{a}^{i}_{{\bar{i}}}\vec{x}^{i}_{{\bar{i}}}\}_{{i}}\{{z^{j}\vec{a}^{j}\vec{x}^{j}\}_{{j}}}}}\alpha_{\{c^{i}\vec{b}^{i}\vec{y}^{i}\}_{{i}}\{c^{j}\vec{b}^{j}_{\bar{j}}\vec{y}^{j}_{\bar{j}}\}_{j}}^{\{z^{i}\vec{a}^{i}_{{\bar{i}}}\vec{x}^{i}_{{\bar{i}}}\}_{{i}}\{{z^{j}\vec{a}^{j}\vec{x}^{j}\}_{{j}}}}\prod_{i}\tr[F^{A_i}_{\lambda_{A_i}}\tilde{M}^{A_i}_{\vec{a}_{\bar{i}}\vec{b}c|\vec{x}_{\bar{i}}\vec{y}z} ]\times \prod_{j}\tr[H^{B_j}_{\lambda_{B_j}}\tilde{\rho}^{B_j}_{\vec{a}\vec{b}_{\bar{j}}c|\vec{x}\vec{y}_{\bar{j}}z}], \label{eq:gen-on-sp}
\end{align}
where the equality follows from \cref{eq:frinst3,eq:frinst4}.
\end{widetext}
plugging the operational identity in \cref{eq:gen-op-sp} back into \cref{eq:gen-on-sp}, we thus have
\begin{align}
p(c\lambda_{\vec{B}}|z\lambda_{\vec{A}})=\tr[(\bigotimes_j H^{B_j}_{\lambda_{B_j}})\mc{F}^{\vec{B}|\vec{A}}_{c|z}(\bigotimes_iF^{A_i}_{\lambda_{A_i}})].\label{eq:frinst5}
\end{align}

Now, substituting \cref{eq:frinst1,eq:frinst2,eq:frinst5} into \cref{eq:genonto}, and using the fact that we are considering every possible local measurements and preparations that are tomographically  complete,  we have:
\begin{widetext}
\begin{align}
&\tr[\underset{j\in [m]}{\bigotimes} (M_{b_j|y_j}^{B_j})\mc{F}_{c|z}^{\vec{B}|\vec{A}} (\underset{i\in [n]}\bigotimes p(a_i|x_i)\rho^{A_i}_{a_i|x_i})]=\sum_{\lambda_{\vec{B}}\lambda_{\vec{A}}}\tr[(\bigotimes_j H^{B_j}_{\lambda_{B_j}})\mc{F}^{\vec{B}|\vec{A}}_{c|z}(\bigotimes_iF^{A_i}_{\lambda_{A_i}})]\prod_j\tr[M^{B_j}_{b_j|y_j}\sigma_{\lambda_{B_j}}^{B_j}]\prod_i\tr[G_{\lambda_{A_i}}^{A_i}p(a_i|x_i)\rho^{A_i}_{a_i|x_i}] \notag \\
&\Rightarrow \mc{F}_{c|z}^{\vec{B}|\vec{A}}(\cdot)=\sum_{\lambda_{\vec{A}}\lambda_{\vec{B}}}\underset{j}{\bigotimes} \sigma_{\lambda_{B_j}}^{B_j}\tr[(\underset{j}{\bigotimes} H_{\lambda_{B_j}}^{B_j})\mc{F}_{c|z}^{\vec{B}|\vec{A}}(\underset{i}{\bigotimes} F^{A_i}_{\lambda_{A_i}})]\tr[\underset{i}{(\bigotimes} G_{\lambda_{A_i}}^{A_i})(\cdot)].
\end{align}
\end{widetext}

Now let us prove the ``if'' part. If there exists a frame decomposition of the multipartite multi-instrument $\{\{\mc{F}_{c|z}^{\vec{B}|\vec{A}}\}_c\}_z$  then for any element of $\{\{\mc{F}_{c|z}^{\vec{B}|\vec{A}}\}_c\}_z$, any element of a multi-measurement $M^{B}_{b|y}$ and any element of a multi-source $p(a|x)\rho^A_{a|x}$, one can define the individual ontological representations induced by the frame operators
\begin{subequations}
\begin{align} 
&p(c\lambda_{\vec{B}}|z\lambda_{\vec{A}})=\tr[(\bigotimes_jH^{B_j}_{\lambda_{B_j}})  \mc{F}_{c|z}^{\vec{B}|\vec{A}}(\bigotimes_i F^{A_i}_{\lambda_{A_i}} )], \label{eq:instrc1}
\\
&p(a_i\lambda_{A_i}|x_i)\coloneqq\tr[G_{\lambda_{A_i}}^{A_i}p(a_i|x_i)\rho^{A_i}_{a_i|x_i}],\label{eq:instrc2} \\
&p(b_j|y_j\lambda_{B_j})\coloneqq\tr[\sigma_{\lambda_{B_j}}^{B_j}M^{B_j}_{b_j|y_j}], \label{eq:instrc3}
\end{align}
where the nonnegativity and normalization are ensured by Eq.~\ref{eq:cl-gen-cons_app}, $\sigma^{B_j}_{\lambda_{B_j}}$ being a density operator and $\{G_{\lambda_{A_i}}^{A_i}\}_{\lambda_{A_i}}$ being a POVM. Moreover, each is a linear functional on its corresponding quantum process. 

One can also show that the effective ontological processes obtained by sequential compositions are also linear to the corresponding effective quantum processes obtained by sequential compositions. This is because $\{F^{A_i}_{\lambda_{A_i}}\}_{\lambda_{A_i}}$ and $\{G_{\lambda_{A_i}}^{A_i}\}_{\lambda_{A_i}}$, $\{\sigma^{B_j}_{\lambda_{B_j}}\}_{\lambda_{B_j}}$ and $\{H^{B_j}_{\lambda_{B_j}}\}_{\lambda_{B_j}}$ are a frame and its dual frame on $V_{A_i}$ and $V_{B_j}$, respectively. In particular, for any choice of $\mc{I}\subseteq [1,\cdots,n]$ and $\mc{J}\subseteq [1,\cdots,m]$, one can verify by plugging in \cref{eq:instrc1,eq:instrc2,eq:instrc3} and using the frame-dual frame properties that
\begin{align}
&p(c\vec{a}_{{\mc{I}}}\vec{b}_{{\mc{J}}}{\lambda
}_{\vec{B}_{\bar{\mc{J}}}}|z\vec{x}_{{\mc{I}}}\vec{y}_{{\mc{J}}}\lambda_{\vec{A}_{\bar{\mc{I}}}})\\
&=\sum_{\lambda_{\vec{B}_{{\mc{J}}}}\lambda_{\vec{A}_{{\mc{I}}}}}p(c\lambda_{\vec{B}}|z\lambda_{\vec{A}})\prod_{j\in \mc{J}}p(b_j|y_j\lambda_{B_j})\prod_{i\in\mc{I}}p(a_i\lambda_{A_i}|x_i), \notag \\
&=\tr[\underset{j\in\mc{J}}{\bigotimes} M^{B_j}_{b_j|y_j} \underset{j\in \bar{\mc{J}}}{\bigotimes}H^{B_j}_{\lambda_{B_j}}  \mc{F}_{c|z}^{\vec{B}|\vec{A}}(\underset{i\in\mc{I}}{\bigotimes}p(a_i|x_i)\rho^{A_i}_{a_i|x_i}\underset{i\notin\mc{I}}{\bigotimes}F^{A_i}_{\lambda_{A_i}} )]. \notag
\end{align}
\end{subequations}
Since all these ontological representations depend linearly on the corresponding physical processes and are nonnegative, we have shown that the circuit is classically explainable. Since we allowed all possible factorizing dual processes (i.e., all product multi-measurements and product multi-sources), this, in turn, implies that the multipartite multi-instrument $\{\{\mc{F}_{c|z}^{\vec{B}|\vec{A}}\}_c\}_z$ is classical by Definition~\ref{maindefn}.
\end{proof}

\section{Frame representation of $2$-combs}
\label{sec:appK}
As we defined in the introduction, a generic process, in the context of this paper, is simply a transformation with various input and output systems, which may be quantum or classical. Note that we did not specify the internal causal structure of a given quantum process. That is, here we did not stipulate any particular causal influence relations among the input and output systems~\footnote{For more discussions on the internal causal structure of a general circuit, see Ref.~\cite[Appendix F.1]{ying2025quantumtheoryneedscomplex}.}. For example, any multipartite quantum channel is treated as a monolithic object with multiple inputs and multiple outputs without constraints on its internal causal structure, so any input system may influence any output system. As a consequence, in our framework, a multipartite channel and an $N$-comb (as defined in~\cite{Chiribella2009}) are treated as the same type of process, even though they are physically distinct, with the latter subject to additional no-signaling constraints arising from its internal causal structure. Nevertheless, as we will see by considering frame representations of $N$-combs, although we never specify the internal causal structure of a given process, if an $N$-comb is classical, its corresponding ontological model nevertheless automatically satisfies the required no-signaling constraints.

Let us take a 2-comb $\mathcal{C}^{B_2B_1|A_2A_1}$ as an example. For $\mathcal{C}^{B_2B_1|A_2A_1}$ to be a valid 2-comb, it must be completely positive and trace preserving, and it must also satisfy a no-signaling constraint between input $A_2$ and output $B_1$. That is, there exists a channel $\mathcal{E}^{B_1|A_1}$ such that
\begin{equation}
\tr_{B_2}\big[\mathcal{C}^{B_2B_1|A_2A_1}(\rho^{A_1}\otimes \sigma^{A_2})\big]
    = \mathcal{E}^{B_1|A_1}(\rho^{A_1})
\end{equation}
for all states $\rho^{A_1}$ and $\sigma^{A_2}$.

Now apply Theorem~\ref{ncgen} to this 2-comb by viewing it as a structure-less process. A classical 2-comb then must admit a frame representation of the form:
\begin{align}
&\mc{C}^{B_2B_1|A_2A_1}(\cdot)=\notag \\
&\sum_{\substack{\lambda_{A_1}\lambda_{A_2}\\ \lambda_{B_1}\lambda_{B_2}}}
    p(\lambda_{B_1}\lambda_{B_2}|\lambda_{A_1}\lambda_{A_2})\,
    \mathrm{tr}\!\left[(G^{A_1}_{\lambda_{A_1}}\otimes G^{A_2}_{\lambda_{A_2}})(\cdot)\right]\,
    \sigma^{B_1}_{\lambda_{B_1}}\otimes \sigma^{B_2}_{\lambda_{B_2}},
\end{align}
where
\begin{align}
&p(\lambda_{B_1}\lambda_{B_2}|\lambda_{A_1}\lambda_{A_2})\notag \\
&\coloneqq\mathrm{tr}\!\left[(
    H^{B_1}_{\lambda_{B_1}}\otimes H^{B_2}_{\lambda_{B_2}})
    \mathcal{C}^{B_2B_1|A_2A_1}(F^{A_1}_{\lambda_{A_1}}\otimes F^{A_2}_{\lambda_{A_2}})
\right]
\ge 0
\end{align}
is its corresponding ontological representation.

One can now directly verify that $p(\lambda_{B_1}\lambda_{B_2}|\lambda_{A_1}\lambda_{A_2})$ satisfies the appropriate no-signaling constraint, i.e., 
\begin{align}
&\sum_{\lambda_{B_2}} p(\lambda_{B_1}\lambda_{B_2}|\lambda_{A_1}\lambda_{A_2}) \notag \\
&=\sum_{\lambda_{B_2}} \mathrm{tr}\!\left[(
     H^{B_1}_{\lambda_{B_1}}\otimes H^{B_2}_{\lambda_{B_2}})
     \mathcal{C}^{B_2B_1|A_2A_1}(F^{A_1}_{\lambda_{A_1}}\otimes F^{A_2}_{\lambda_{A_2}})
\right] \notag \\
&=\mathrm{tr}\!\left[(
     H^{B_1}_{\lambda_{B_1}}\otimes \mathbb{1}^{B_2})
     \mathcal{C}^{B_2B_1|A_2A_1}(F^{A_1}_{\lambda_{A_1}}\otimes F^{A_2}_{\lambda_{A_2}})
\right]  \notag \\
&=\mathrm{tr}_{B_1}\!\left[
    H^{B_1}_{\lambda_{B_1}}\,
    \mathcal{E}^{B_1|A_1}\!\left(F^{A_1}_{\lambda_{A_1}}\right)
\right]  \notag \\
&= p(\lambda_{B_1}|\lambda_{A_1}) \ge 0,
\end{align}
where we used $\sum_{\lambda_{B_2}} H^{B_2}_{\lambda_{B_2}} = \mathbb{1}^{B_2}$ in the second equality and the no-signaling condition of the 2-comb in the third equality. The nonnegativity of the final expression is automatic.

Therefore, although our framework does not assume any internal causal structure of a given process, and only treats it as a structureless object by specifying only its input and output systems, our framework can naturally be applied to those higher-order processes with fixed internal causal structure. Intuitively, this is because the corresponding no-signaling constraints for any process are automatically satisfied in its ontological representation due to the structure theorems we provide.

However, we do not claim that our framework can be directly generalized to treat arbitrary higher-order quantum processes. In particular, it is not clear how to define an ontological model for a process with indefinite causal order. Explicitly specifying the internal causal structure as an operational constraint might be important for studying arbitrary higher-order processes. We view this as an interesting direction for future research, which will likely require further investigation.

\section{Proof of Lemma~\ref{lem:Choi}}
\label{appendixL}

\begin{figure}[htb!]
\includegraphics[width=0.8\linewidth]{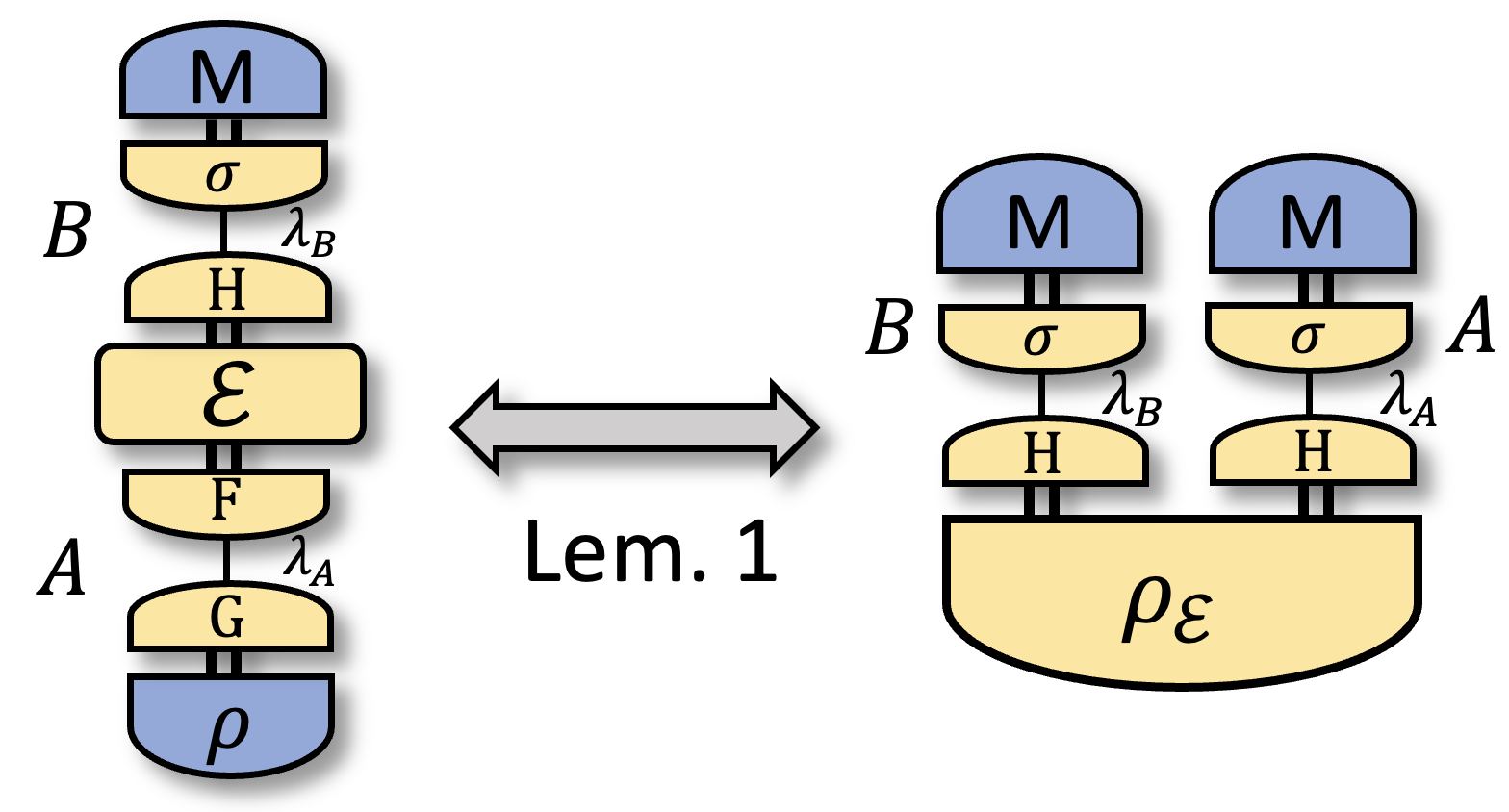}
    \caption{Connecting a prepare-transform-measure circuit with a prepare-measure circuit on a bipartite system involving local measurements}
    \label{fig:Choi_app}
\end{figure}

We now repeat and then prove Lemma~\ref{lem:Choi}.
\setcounter{lemma}{0}

\begin{lemma}
The Choi state isomorphic to a given process is classical if and only if the process is classical.
\end{lemma}

\begin{proof}
Let us first show that if a quantum channel $\mc{E}^{B|A}$ is classical, then so is its Choi state 
$\rho_{\mc{E}}^{BA}=(\mc{E}^{B|A}\otimes \mbb{1}^{A} )\op{\Psi^+}{\Psi^+}^{AA}$ .  \par 

By Theorem~\ref{ncchannelstruc}, a classical channel admits a frame decomposition
\begin{equation}
\mc{E}^{B|A}(\cdot)=\sum_{\lambda_A\lambda_B}\sigma_{\lambda_B}^{B}\tr[H^B_{\lambda_B}\mc{E}^{B|A}(F^A_{\lambda_A})]\tr[G^{A}_{\lambda_A}(\cdot)]
\end{equation}
where $\{\sigma^{B}_{\lambda_B}\}_{\lambda_B}$ is a set of quantum states defining a frame on $V_B\coloneqq\text{Span}(\{\mc{E}^{B|A}(\rho^A)\})$, and $\{G_{\lambda_A}^{A}\}_{\lambda_A}$ is a POVM defining a frame on $V_A\coloneqq\text{Span}(\{[\mc{E}^{{B|A}}]^{\dagger}(M^{B})\})$, and where $\{{H}_{\lambda_B}^{B}\}_{\lambda_B}$ and $\{{F}^{A}_{\lambda_A}\}_{\lambda_A}$ are sets of Hermitian operators forming the dual frames that satisfy $\sum_{\lambda_B}{H}_{\lambda_B}^{B}=\mc P_{_{V_B}}(\mbb{1}^{B})$, $\tr[{F}^{A}_{\lambda_A}]=1$ and moreover, 
\begin{align} 
\tr[H^B_{\lambda_B}\mc{E}^{B|A}(F^A_{\lambda_A})] \ge 0~~~~\forall \lambda_A, \lambda_B.
\end{align}
We can then write the Choi state $\rho_{\mc{E}}^{BA}$ of the quantum channel $\mc{E}^{A\rightarrow B}$ in terms of this frame decomposition as
\begin{align}
&\rho^{BA}_{\mc{E}}=(\mc{E}\otimes \mbb{1}^{A} )\op{\Psi^+}{\Psi^+}^{AA} \notag \\
&=\sum_{\lambda_A\lambda_B}\sigma_{\lambda_B}^{B}\tr[H^B_{\lambda_B}\mc{E}^{B|A}(F^A_{\lambda_A})]\tr_A[G^{A}_{\lambda_A}\otimes\mbb{1}^{A} \op{\Psi^+}{\Psi^+}^{AA}] \notag \\
&=\frac{1}{d}\sum_{\lambda_A\lambda_B}\sigma_{\lambda_B}^{B}\otimes [G^{A}_{\lambda_A}]^T\tr[H^B_{\lambda_B}\mc{E}^{B|A}(F^A_{\lambda_A})] \notag \\
&=\sum_{\lambda_A\lambda_B}\sigma_{\lambda_B}^{B}\otimes [G^{A}_{\lambda_A}]^T\tr\left[H^B_{\lambda_B}\otimes [F^{A}_{\lambda_A}]^T \rho^{BA}_{\mc{E}}\right],
\label{channel-to-choi}
\end{align}
where all transposes are with respect to the basis defining $\ket{\Psi^+}\coloneqq\frac{1}{\sqrt{d}}\sum_i\ket{ii}$, and we used the standard Choi identities $\tr_A[X\otimes \mbb{1}\op{\Psi^+}{\Psi^+}]=\frac{1}{d}X^T$ and $\frac{1}{d}\tr\!\left[Y^B\mathcal{E}^{B|A}(X^A)\right]
= \tr\!\left[\left(Y^B\otimes [X^A]^{\mathsf T}\right)\rho^{BA}_{\mathcal{E}}\right]$.

We can now rescale $\{G^{A}_{\lambda_A}\}_{\lambda_A}$ and $\{F^{A}_{\lambda_A}\}_{\lambda_A}$ to define a new frame and its dual frame
\begin{align}
\sigma_{\lambda_A}^{A}\coloneqq\frac{[G^{A}_{\lambda_A}]^T}{\tr[G^{A}_{\lambda_A}]}, ~~~~~~~ H_{\lambda_A}^{A}\coloneqq[F^{A}_{\lambda_A}]^T\tr[G^{A}_{\lambda_A}], \label{eq:choiupdate1}
\end{align}
such that $\{\sigma_{\lambda_A}^{A}\}_{\lambda_A}$ is a set of quantum states and $\{H_{\lambda_A}^{A}\}_{\lambda_A}$ is a set of Hermitian operators such that $\sum_{\lambda_A}H_{\lambda_A}^{A}=\mc P_{_{V_A}}(\mbb{1}^{A})=\mbb{1}^{A}$ since $V_A=\text{Span}(\{[\mc{E}^{{B|A}}]^{\dagger}(M^{B})\})$ must contain the identity operator $\mbb{1}^{A}$ as the map $[\mc{E}^{{B|A}}]^{\dagger}$ is unital; 
moreover, the nonnegativity conditions are respected since
\begin{align}
&\tr[H^B_{\lambda_B}\otimes H^{A}_{\lambda_A} \rho_{\mc{E}}^{BA}] \nonumber\\
=&\frac{1}{d}\tr[G_{\lambda_A}^{A}]\tr[H_{\lambda_B}^B\mc{E}^{B|A}(F_{\lambda_A}^A)] \ge 0.
\end{align}
Therefore, by \cref{ncbp}, $\rho_{\mc{E}}^{BA}$ is classical. \par 

Now we prove the reverse direction. A classical bipartite state $\rho_{\mc{E}}^{BA}$  admits a frame decomposition
\begin{equation}
\rho_{\mc{E}}^{BA}=\sum_{\lambda_A\lambda_B}\tr[(H^{B}_{\lambda_B}\otimes H^{A}_{\lambda_A})\rho_{\mc{E}}^{BA}]\sigma^{B}_{\lambda_B}\otimes \sigma^{A}_{\lambda_A},
\end{equation}
where $\{\sigma^{B}_{\lambda_B}\}$ and $\{\sigma^{A}_{\lambda_A}\}$  are sets of quantum states defining two respective frames on $V_{B(A)}=\text{Span}(\{\sigma^{B(A)}_{\lambda_{B(A)}}\})$, while $\{{H}^{B}_{\lambda_B}\}$, $\{{H}^{A}_{\lambda_A}\}$ are sets of Hermitian operators defining the respective dual frames and satisfying $\sum_{\lambda_B}{H}^{B}_{\lambda_B}=\mc P_{_{V_B}}(\mbb{1}^{B})$, and $\sum_{\lambda_A}{H}^{A}_{\lambda_A}=\mc P_{_{V_A}}(\mbb{1}^{A})=\mbb{1}^{A}$ where
\begin{align} 
\tr[(H^{B}_{\lambda_B}\otimes H^{A}_{\lambda_A})\rho_{\mc{E}}^{BA}] \ge 0~~~~\forall \lambda_B, \lambda_A. 
\label{eq:ch-to-bp}
\end{align}
Since a Choi state must satisfy $\tr_{B}[\rho_{\mc{E}}^{BA}]=\frac{1}{d}\mbb{1}^{A}$, we have
\begin{align}
\tr_B[\rho_{\mc{E}}^{BA}]&=\sum_{\lambda_A}\tr[(\mc P_{_{V_B}}(\mbb{1}^{B})\otimes H^{A}_{\lambda_A})\rho_{\mc{E}}^{BA}] \sigma^{A}_{\lambda_A} \notag \\
&=\sum_{\lambda_A}\tr_A[H^{A}_{\lambda_A}\tr_B[\rho_{\mc{E}}^{BA}]]\sigma^{A}_{\lambda_A}\notag \\
&=\sum_{\lambda_A}\frac{1}{d}\tr_A[H^{A}_{\lambda_A}]\sigma^{A}_{\lambda_A}=\frac{1}{d}\mbb{1}^{A},\label{eq:ch-margi1}
\end{align}
where the second equality is implied by the fact that $\mc P_{_{V_B}}(\mbb{1}^{B})$ acts as an identity operator on the span of $V_B$, and the last equality holds since $\{H^{A}_{\lambda_A}\}_{\lambda_A}$ and $\{\sigma^{A}_{\lambda_A}\}_{\lambda_A}$ defines a pair of frame and dual frame on $V_A$, which contains $\mbb{1}^{A}$.

Moreover, the constraints in \cref{eq:ch-to-bp} imply
\begin{align} 
\sum_{\lambda_B}\tr[(H^{B}_{\lambda_B}\otimes H^{A}_{\lambda_A})\rho_{\mc{E}}^{BA}]=\frac{1}{d}\tr[H^{A}_{\lambda_A}]  \ge 0~~\forall \lambda_A. 
\label{eq:ch-margi2}
\end{align}

We can now reconstruct the quantum channel $\mc{E}^{B|A}$ from its Choi state $\rho_{\mc{E}}^{BA}$ as 
\begin{align}
\mc{E}^{B|A}(\cdot)&=d\tr_{A}[(\mbb{1}^{B}\otimes (\cdot)^T )\rho_{\mc{E}}^{BA}] \\
&=d\sum_{\lambda_B\lambda_A}\tr[(H^{B}_{\lambda_B}\otimes H^{A}_{\lambda_A})\rho_{\mc{E}}^{BA}]\tr[\sigma^{A}_{\lambda_A}(\cdot)^T]\sigma^{B}_{\lambda_B}\notag \\
&=\sum_{\lambda_B\lambda_A}\tr[H^{B}_{\lambda_B}\mc{E}^{B|A}(H^{A}_{\lambda_A})^T]\tr[(\sigma^{A}_{\lambda_A})^T(\cdot)]\sigma^{B}_{\lambda_B}. \notag 
\end{align}
We then rescale $\{G^{A}_{\lambda_A}\}_{\lambda_A}$ and $\{F^{A}_{\lambda_A}\}_{\lambda_A}$ to define a new frame and its dual frame by
\begin{align}
F_{\lambda_A}^{A}\coloneqq\frac{(H^{A}_{\lambda_A})^T}{\tr[H^{A}_{\lambda_A}]},~~~~~~G_{\lambda_A}^{A}\coloneqq\tr[H^{A}_{\lambda_A}](\sigma^{A}_{\lambda_A})^T,\label{eq:choiupdate2}
\end{align}
such that $\tr[F_{\lambda_A}^{A}]=1$ and $\sum_{\lambda_A}G_{\lambda_A}^{A}=\sum_{\lambda_A}\tr[H^{A}_{\lambda_A}](\sigma^{A}_{\lambda_A})^T=\mc P_{V_A}(\mbb{1}^{A})=\mbb{1}^{A}$ using \cref{eq:ch-margi1}. Thus, we have
\begin{align}
\tr[H_{\lambda_B}^{B}\mc{E}^{B|A}(F_{\lambda_A}^{A})]=\frac{d}{\tr[H_{\lambda_A}^{A}]}\tr[(H_{\lambda_B}^{B}\otimes H_{\lambda_A}^{A})\rho_{\mc E}^{BA}]\ge 0,\notag
\end{align}
where $\frac{d}{\tr[H^{A}_{\lambda_A}]}\ge 0$ is ensured by \cref{eq:ch-margi2}. 

Therefore, a quantum channel $\mc{E}^{B|A}$ with a classical Choi state $\rho_{\mc{E}}^{BA}$ must be classical. 
\end{proof}
\begin{remark}
To show that for any general process with several quantum inputs and/or outputs, classicality is preserved under the Choi isomorphism, it suffices to apply the frame and dual-frame rescaling rules in \cref{eq:choiupdate1} or \cref{eq:choiupdate2} to each quantum input or output. All the nonnegativity and normalization constraints are preserved, and no additional modifications are required.
\end{remark}

\section{Proof of Lemma~\ref{lem:absorb_classical_handle}}
\setcounter{lemma}{2}
\par 
\begin{figure}[htb!]
  \includegraphics[width=0.9\linewidth]{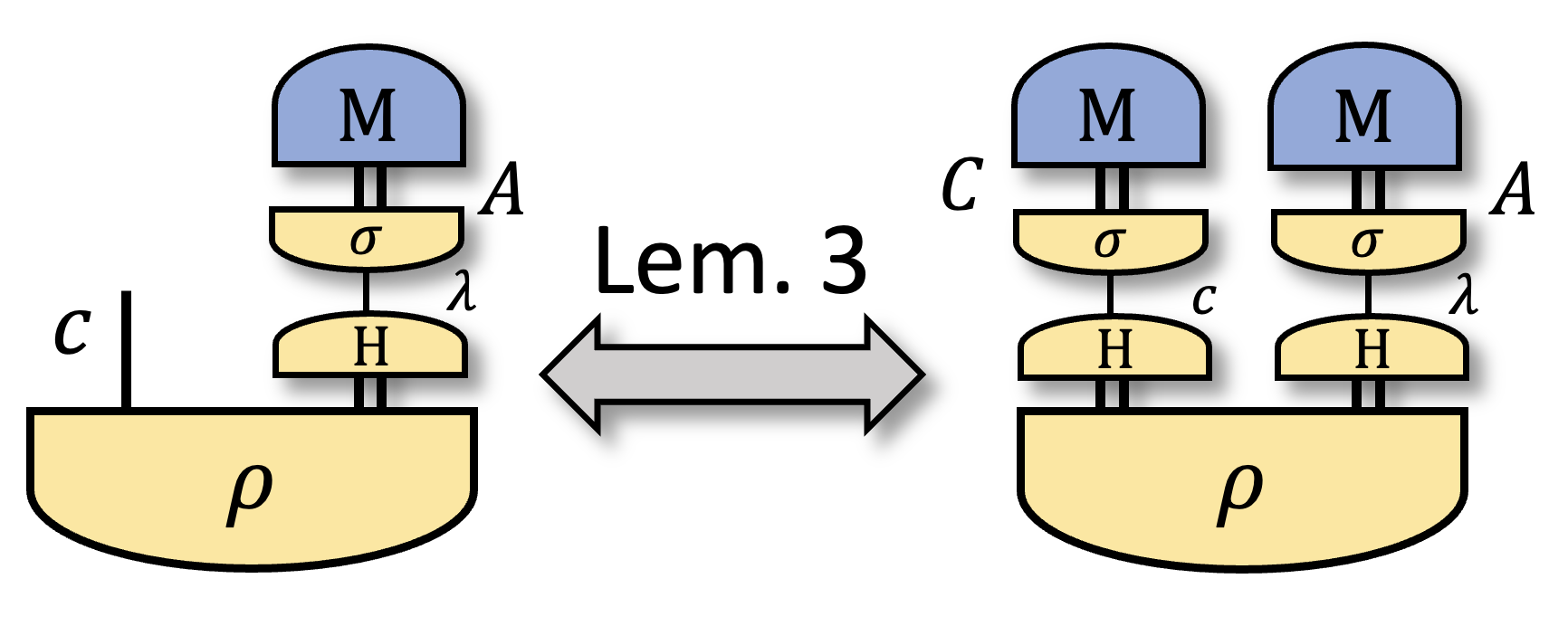}
    \caption{Connecting a prepare-measure circuit for a source with a prepare-measure circuit on a bipartite system involving local measurements}
    \label{fig:Cut_app}
\end{figure}
We now prove Lemma~\ref{lem:absorb_classical_handle}.
\begin{lemma}
Whether or not a process is nonclassical is independent of whether one reconceptualizes its classical outputs as dephased quantum outputs.
\end{lemma}

\begin{proof}
We prove the claim for a source $\{p(c)\rho_c^A\}$, which can also be viewed as a bipartite state
\begin{align}
\label{eq:rhoCA}
   \rho^{CA}\coloneqq\sum_cp(c)\op{c}{c}^C\otimes\rho_c^A 
\end{align}
when the classical outcome is reconceptualized as the dephased quantum output $C$. The generalization to arbitrary quantum processes is straightforward. \par 

For one direction, if the bipartite state $\rho^{CA}$ is classical, following Proposition~\ref{propremote}, then any source on $A$ obtainable by a measurement on $C$ is classical. In particular,  measuring $\{M_c=\op{c}{c}^C\}_c$ on $\rho^{CA}$ yields the source $\{p(c)\rho^A_c\}$, which must therefore be classical.\par 

For the other direction, if the source $\{p(c)\rho_c^A\}$ is classical, according to Theorem~\ref{theoremprep}, there exists a frame decomposition such that
\begin{equation}
\label{eq:rhoC}
\rho^A_c=\sum_{\lambda}\tr[H^A_{\lambda}\rho^A_c]\sigma^A_{\lambda},
\end{equation}
where $\{\sigma^A_{\lambda}\}_{\lambda}$ is a set of quantum states and $\{H^A_{\lambda}\}$ is a set of Hermitian operators with $\sum_{\lambda}H_{\lambda}^A=\mc P_{V}(\mbb{1}^A)$ with $V=\text{span}(\{\rho_c\}_c)$, and  $\tr[H^A_{\lambda}\rho^A_c]\ge 0$ for all $\lambda, c$.

One can then define another pair of frame and dual frame as $\{\sigma_{c}^C\coloneqq\op{c}{c}^C\}_c$ and $\{H_{c}^C\coloneqq\op{c}{c}^C\}_c$ with $\tr[H_{c}^C\sigma_{c'}^C]=|\bra{c}c'\rangle|^2 =\delta_{c,c'}$. With that, one can find a frame decomposition of $\rho^{CA}$ as
\begin{equation}
\rho^{CA}=\sum_{c,\lambda}\tr[(H_{c}^C\otimes H^A_{\lambda})\rho^{CA}]\sigma_{c}^C\otimes \sigma_{\lambda}^A.
\end{equation}
Moreover, since $\tr[H^A_{\lambda}\rho^A_c]\ge 0$ for all $\lambda$, we have
\begin{align*}
\tr[(H_{c}^C\otimes H^A_{\lambda})\rho^{CA}]=p(c)\tr[H^A_{\lambda}\rho_c^A]\ge 0~~~~\forall c, \lambda. \\
\end{align*}
Therefore, the bipartite state $\rho^{CA}$ is classical.
\end{proof}

\begin{remark}
For any general process with several quantum inputs and outputs, classicality remains preserved when classical outcomes are reconceptualized as dephased quantum outputs. This follows from choosing the eigenbasis $\{\op{c}{c}\}_c$ of each dephased system to serve simultaneously as the frame and its dual, ensuring that all associated nonnegativity and normalization constraints are automatically respected.
\end{remark}

\end{document}